\documentclass[letterpaper,english,11pt]{scrartcl}

\newif\ifreview\reviewfalse

\usepackage{standalone}

\usepackage{amsmath,amsfonts,amsthm,amssymb}

\allowdisplaybreaks

\usepackage{relsize}
\usepackage{mathtools}
\usepackage{braket}
\usepackage{aligned-overset}

\usepackage[margin=1in,bottom=1.3in]{geometry}
\usepackage{authblk}
\usepackage{enumitem,framed}

\usepackage{longtable}

\usepackage[linesnumbered,ruled,nokwfunc,longend]{algorithm2e}
\DontPrintSemicolon
\SetKwInOut{Input}{Input}\SetKwInOut{Output}{Output}
\SetKwProg{fetry}{foreach }{ try}{} 
\SetKwBlock{catch}{catch }{end}

\usepackage{graphicx}

\usepackage{tikz}
\usetikzlibrary{calc}

\usepackage{pgfplots}
\pgfplotsset{compat=1.18}
\usepgfplotslibrary{fillbetween}

\tikzstyle{vertex} = [shape=circle,draw=black]
\tikzstyle{namedVertex} = [shape=circle,draw=black]
\tikzstyle{vertexF} = [shape=circle,draw=black,fill=white]
\tikzstyle{namedVertexF} = [shape=circle,draw=black,fill=white]
\tikzstyle{namedVertexW} = [shape=circle,draw=white,fill=white]
\tikzstyle{edge} = [draw,->,ultra thick]
\tikzstyle{labeledNodeS}=[circle, color=black!75!white, draw, inner sep = 0.1em, minimum size = 1.5em, scale=1.25]
\tikzstyle{normalEdge}=[very thick, >=stealth]

\definecolor{colA}{rgb}{.8,0,0}
\definecolor{colB}{rgb}{0,0,.8}
\definecolor{colC}{rgb}{0,.5,.3}
\definecolor{colD}{rgb}{.8,0,0}
\definecolor{colE}{rgb}{.8,.8,0}
\definecolor{colF}{rgb}{.8,0,.8}

\colorlet{fcolAr}{colA!70!white}
\colorlet{fcolBr}{colB!70!white}
\colorlet{fcolCr}{colC!70!white}
\colorlet{fcolDr}{colD!70!white}
\colorlet{fcolEr}{colE!70!white}
\colorlet{fcolFr}{colF!70!white}

\newcommand{\colName}[1]{{%
    \ifthenelse{\equal{#1}{A}}{\color{red}XXX}{}%
    \ifthenelse{\equal{#1}{B}}{\color{colB}blue}{}%
    \ifthenelse{\equal{#1}{C}}{\color{colC}green}{}%
    \ifthenelse{\equal{#1}{D}}{\color{colD}red}{}%
    \ifthenelse{\equal{#1}{E}}{\color{colE}yellow}{}%
    \ifthenelse{\equal{#1}{F}}{\color{colF}pink}{}%
}}

\newcommand{\colNo}[1]{{%
    \ifthenelse{\equal{#1}{A}}{\color{red}XXX}{}%
    \ifthenelse{\equal{#1}{B}}{\color{colB}1}{}%
    \ifthenelse{\equal{#1}{C}}{\color{colC}2}{}%
    \ifthenelse{\equal{#1}{D}}{\color{colD}3}{}%
    \ifthenelse{\equal{#1}{E}}{\color{colE}4}{}%
    \ifthenelse{\equal{#1}{F}}{\color{colF}5}{}%
}}

\tikzstyle{fcolA}=[color=colA,opacity=.7]
\tikzstyle{fcolB}=[color=colB,opacity=.7]
\tikzstyle{fcolC}=[color=colC,opacity=.7]
\tikzstyle{fcolD}=[color=colD,opacity=.7]
\tikzstyle{fcolE}=[color=colE,opacity=.7]
\tikzstyle{fcolF}=[color=colF,opacity=.7]

\providecommand{\QUEUEWIDTH}{18pt}
\providecommand{\QUEUEUHEIGHT}{10pt}
\providecommand{\FLOWUWIDTH}{3.5pt}

\newif\ifdebug\debugfalse

\usepackage{ifthen}

\usepackage{inputenc}
\usepackage{listofitems}

\definecolor{DFcolA}{rgb}{.8,0,0}
\definecolor{DFcolB}{rgb}{0,0,.8}
\definecolor{DFcolC}{rgb}{0,.8,0}
\definecolor{DFcolD}{rgb}{.8,.8,0}
\definecolor{DFcolE}{rgb}{.8,0,.8}
\definecolor{DFcolF}{rgb}{0,.8,.8}

\colorlet{DFfcolAr}{DFcolA!70!white}
\colorlet{DFfcolBr}{DFcolB!70!white}
\colorlet{DFfcolCr}{DFcolC!70!white}
\colorlet{DFfcolDr}{DFcolD!70!white}
\colorlet{DFfcolEr}{DFcolE!70!white}
\colorlet{DFfcolFr}{DFcolF!70!white}

\tikzstyle{DFfcolA}=[color=DFcolA,opacity=.7]
\tikzstyle{DFfcolB}=[color=DFcolB,opacity=.7]
\tikzstyle{DFfcolC}=[color=DFcolC,opacity=.7]
\tikzstyle{DFfcolD}=[color=DFcolD,opacity=.7]
\tikzstyle{DFfcolE}=[color=DFcolE,opacity=.7]
\tikzstyle{DFfcolF}=[color=DFcolF,opacity=.7]

\providecommand{\flowcolorlist}{DFfcolCr,DFfcolFr,DFfcolCr,DFfcolDr,DFfcolEr,DFfcolFr}

\providecommand{\QUEUEWIDTH}{.3}
\providecommand{\QUEUEUHEIGHT}{1}
\providecommand{\FLOWUWIDTH}{4pt}

\providecommand{\FLOWSTYLE}{parallel}

\newcommand{\drawFlow}[7][]{%
	\setsepchar[/]{,}
	\readlist\flowcoms{#7}
	\readlist\collist\flowcolorlist
	
	\def\totwidth{0}
	\foreachitem \width \in \flowcoms{\pgfmathadd{\totwidth}{\width}\xdef\totwidth{\pgfmathresult}}
	\def\cumwidth{\totwidth}
	
	\foreachitem \width \in \flowcoms{
		\ifthenelse{\equal{\width}{0.0} \OR \equal{\width}{0}}{}{
			\itemtomacro\collist[\widthcnt]\currentcol
            \ifthenelse{\equal{\FLOWSTYLE}{parallel}}{
                \path(#2)to[#1]coordinate[pos=0](tempEdgeStart)coordinate[pos=1](tempEdgeEnd) (#3);
                \draw[\currentcol,line width=\width*\FLOWUWIDTH,dash pattern={on 1pt off #5*#4 on #6*#4 off #4}, dash phase=1pt]($(tempEdgeStart)!{(.5*\totwidth-\cumwidth+.5*\width)*\FLOWUWIDTH}!90:(tempEdgeEnd)$) to[#1] ($(tempEdgeEnd)!{(.5*\totwidth-\cumwidth+.5*\width)*\FLOWUWIDTH}!-90:(tempEdgeStart)$);
            }{
                \draw[\currentcol,line width=\cumwidth*\FLOWUWIDTH,dash pattern={on 1pt off #5*#4 on #6*#4 off #4}, dash phase=1pt](#2) to[#1] (#3);
            }
			\pgfmathsubtract{\cumwidth}{\width}
			\xdef\cumwidth{\pgfmathresult}
		}
	}
}

\newcommand{\drawEdgeFlow}[5][]{%
	\ifthenelse{\equal{#5}{}}{}{
		\setsepchar[,]{I}
		\readlist\flowblocks#5 
		\foreachitem\flowblock\in\flowblocks{%
			\setsepchar[,]{/}
			\readlist\flowparts\flowblock 
			\itemtomacro\flowparts[3]\flowsplit
			\drawFlow[#1]{#2}{#3}{#4}{\flowparts[1]}{\flowparts[2]}{\flowsplit}
		}
	}
}

\newcommand{\drawQueueSegment}[3]{
	\setsepchar[/]{,}
	\readlist\flowcoms{#3}
	\readlist\collist\flowcolorlist
	
	\coordinate(tempA)at($(#1)+(-.5*\QUEUEWIDTH,0)$);
	\foreachitem \width \in \flowcoms {
		\itemtomacro\collist[\widthcnt]\currentcol
		\fill[\currentcol](tempA) rectangle +(\width*\QUEUEWIDTH,#2*\QUEUEUHEIGHT);        
		\coordinate(tempA)at($(tempA)+(\width*\QUEUEWIDTH,0)$);
	}
}

\newenvironment{placeQueue}[2]{%
	\begin{scope}[rotate around={#2:(#1)}]
		\coordinate(queuebase)at(#1);%
	}{%
	\end{scope}%
}

\newcommand{\drawEdgeQueue}[4][]{%
	\ifthenelse{\equal{#4}{}}{}{
		\begin{placeQueue}{#2}{#3}
			\setsepchar[,]{I}
			\readlist\queueblocks#4 
			\foreachitem\queueblock\in\queueblocks{%
				\setsepchar[,]{/}
				\readlist\queueparts\queueblock 
				\itemtomacro\queueparts[2]\queuesplit
				\drawQueueSegment{queuebase}{\queueparts[1]}{\queuesplit}
				\coordinate (queuebase)at($(queuebase)+(0,\queueparts[1]*\QUEUEUHEIGHT)$);
			}                      
		\end{placeQueue}
	}
}

\ifdebug
\renewcommand{\drawEdgeFlow}[5][]{%
	\drawFlow[#1]{#2}{#3}{#4}{0}{1}{1}%
}
\fi

\usepackage{csvsimple}

\NewDocumentCommand{\CharF}{O{}}{%
	\ifthenelse{\equal{#1}{}}%
	{1}%
	{1_{#1}}%
}

\usepackage{thm-restate}
\usepackage{xspace}
\usepackage{color, colortbl}
\usepackage{hyperref}
\usepackage{cleveref}

\newcommand{\refsym}[1]{%
	\ensuremath{(%
		\ifthenelse{\equal{#1}{1}}{\ast}%
		{\ifthenelse{\equal{#1}{2}}{\#}%
		{\ifthenelse{\equal{#1}{3}}{\triangle}%
		{\ifthenelse{\equal{#1}{4}}{\bigcirc}%
		{\ifthenelse{\equal{#1}{5}}{\diamond}
		{\color{red}NaN}}}}}%
	)}%
}
\newcommand{\symoverset}[2]{\toverset{\refsym{#1}}{#2}}
\newcommand{\toverset}[2]{\overset{\text{#1}}{#2}}

\newcommand{\Croverset}[2]{\overset{\text{\Crefshort{#1}}}{#2}}
\newcommand{\Crefshort}[1]{%
	{%
		\Crefname{definition}{Def.}{Def.}%
		\Crefname{lemma}{Lem.}{Lem.}%
		\Crefname{proposition}{Prop.}{Prop.}%
		\Crefname{corollary}{Cor.}{Cor.}%
		\Crefname{theorem}{Thm.}{Thm.}%
		\Crefname{observation}{Obs.}{Obs.}%
		\Crefname{claim}{Cl.}{Cl.}%
		\Crefname{section}{Sec.}{Sec.}%
		\Crefname{subsection}{Sec.}{Sec.}%
		\Crefname{example}{Ex.}{Ex.}%
		\Crefname{equation}{Eq.}{Eq.}%
		\Cref{#1}%
	}%
}

\makeatletter
\DeclareRobustCommand{\Crefnosort}[1]{%
  \begingroup\@cref@sortfalse\Cref{#1}\endgroup
}
\makeatother

\usepackage{stmaryrd}
\usepackage[font={small,it}]{caption}

\definecolor{darkgreen}{rgb}{0.0, 0.5, 0.0}

\newcommand{\R}{\mathbb{R}}
\newcommand{\Rnn}{\mathbb{R}_{\geq 0}}
\newcommand{\N}{\mathbb{N}}

\newcommand{\comment}[1]{}
\newcommand{\infunc}{B}
\newcommand{\func}{A}
\newcommand{\Inter}{\R}
\newcommand{\objfunc}{\vartheta}
\newcommand{\confunc}{\xi}
\newcommand{\emptyarg}{\,\cdot\,} 

\newcommand{\eflow}{y}
\newcommand{\wflow}{f}
\NewDocumentCommand{\sysop}{O{\ }}{system-optimal{#1}}
\NewDocumentCommand{\umini}{O{\ }}{\ensuremath{\ell^u}-minimal{#1}}
\NewDocumentCommand{\mini}{O{\ }}{\ensuremath{\ell}-minimal{#1}}

\newcommand{\GAS}{{\tilde{\GA}}}
\newcommand{\GVS}{{\tilde{\GV}}}
\newcommand{\ssource}{{\tilde{\source}}}
\newcommand{\sdest}{{\tilde{\dest}}}

\newcommand{\q}{q}

\newcommand{\decwttime}{\hat{\Psi}}
\NewDocumentCommand{\Nl}{O{\trav(\cdot,\cdot)}}{\ell^{#1}} 
\newcommand{\startint}{{(-\infty,\,}}

\usepackage{bm}

\definecolor{LightCyan}{rgb}{0.88,1,1}

\makeatletter
\def\mybig#1{{\hbox{$\left#1\vbox to23\p@{}\right.\n@space$}}}
\makeatother

\newcommand{\dup}[2]{\langle#1,#2\rangle}

\theoremstyle{definition}
\newtheorem{definition}{Definition}[section]
\newtheorem{assumption}[definition]{Assumption}

\theoremstyle{plain}
\newtheorem{theorem}[definition]{Theorem}

\newtheorem{lemma}[definition]{Lemma}
\newtheorem{corollary}[definition]{Corollary}
\newtheorem{claim}{Claim}
\newtheorem{subclaim}{Subclaim}[claim]

\Crefname{claim}{Claim}{Claims}
\Crefname{assumption}{Assumption}{Assumptions}

\theoremstyle{remark}

\newtheorem{example}[definition]{Example}

\newenvironment{introthm}[2][]{\begin{framed}\textbf{\Cref{#2}}\ifthenelse{\equal{#1}{}}{}{ (#1)}:\\}{\end{framed}}

\newlist{thmparts}{enumerate}{1}
\setlist[thmparts]{
	label=\alph*)
}

\usepackage{xstring}
\usepackage{etoolbox}
\makeatletter
\apptocmd{\cref@getref}{\xdef\@lastusedlabel{#1}}{}{error}

\Crefformat{thmpart}{%
	\StrCount{\@lastusedlabel}{:}[\LastColonPos]%
	\StrBefore[\LastColonPos]{\@lastusedlabel}{:}[\ThmLabel]%
	\nameCref{\ThmLabel}~#2\ref*{\ThmLabel}#1#3%
}
\Crefmultiformat{thmpart}{%
	\StrCount{\@lastusedlabel}{:}[\LastColonPos]%
	\StrBefore[\LastColonPos]{\@lastusedlabel}{:}[\ThmLabel]%
	\nameCref{\ThmLabel}~#2\ref*{\ThmLabel}#1#3%
}{ and~#2#1#3}{, #2#1#3}{ and~#2#1#3}
\makeatother

\newenvironment{proofClaim}[1][]{\ifthenelse{\equal{#1}{}}{\begin{proof}}{\begin{proof}[#1]}}{\end{proof}}

\newlist{proofbycases}{enumerate}{1}
\setlist[proofbycases]{
	leftmargin=0em,
	labelwidth=-.5em,
    parsep=0pt,
    listparindent=\parindent,
	label=\boldmath\bfseries\sffamily\arabic*. Case: \protect\casedescr:,
	ref=\arabic*,
	align=left
}

\newcommand{\proofitem}[1]{\def\pidescr{#1}%
	\item}
\newlist{structuredproof}{enumerate}{3}
\setlist[structuredproof]{
	leftmargin=1em,
    labelwidth=.3em,
    parsep=0pt,
    listparindent=\parindent,
	label=\boldmath\bfseries\sffamily\protect\pidescr:,
	align=left
}
 \setlist[structuredproof,2]{
     leftmargin=1em
 }
 \setlist[structuredproof,3]{
     leftmargin=1em
 }

\newenvironment{proofbyinduction}{\begin{description}[leftmargin=0em,parsep=0pt,listparindent=\parindent]}{\end{description}}
\newcommand{\inductionclaim}{\item[Induction Claim: ]}
\newcommand{\basecase}[1]{\item[Base Case ({\boldmath\bfseries#1}):]}
\newcommand{\basecases}[1]{\item[Base Cases ({\boldmath\bfseries#1}):]}
\newcommand{\inductionstep}[1]{\item[Induction Step ({\boldmath\bfseries#1}):]}

\DeclareFontFamily{OT1}{pzc}{}
\DeclareFontShape{OT1}{pzc}{m}{it}{<-> s * [1.10] pzcmi7t}{}
\DeclareMathAlphabet{\mathpzc}{OT1}{pzc}{m}{it}
 
\newcommand{\myparagraph}[1]{\paragraph{#1.}}

\makeatletter
\newcommand{\fixed@sra}{$\vrule height 2\fontdimen22\textfont2 width 0pt\shortrightarrow$}
\newcommand{\shortarrow}[1]{%
  \mathrel{\text{\rotatebox[origin=c]{\numexpr#1*45}{\fixed@sra}}}
}
\makeatother

\newcommand{\norm}[1]{\lVert #1 \rVert}
\newcommand{\abs}[1]{\lvert#1\rvert}

\NewDocumentCommand{\wlg}{O{\ }}{w.l.o.g.#1} 

\newcommand{\wrt}{w.r.t.\ }

\newcommand{\eps}{\varepsilon}

\usepackage{comment}

\newcommand{\di}{\;\mathrm{d}}

\newcommand{\prices}{p}

\newcommand{\arc}{e}

\newcommand{\g}{g}
\newcommand{\vot}{\gamma}
\newcommand{\exit}{T}
\newcommand{\arr}{A}
\newcommand{\trav}{D}

\newcommand{\hori}{\R}

\newcommand{\Routes}{\mathcal{W}}
\newcommand{\wa}{w}
\newcommand{\inflow}{r}

\newcommand{\wir}{\Lambda}
\newcommand{\cl}[1]{\mathrm{cl}(#1)}

\newcommand{\id}{\text{id}}

\newcommand{\GA}{E}
\newcommand{\GV}{V}
\newcommand{\source}{s}
\newcommand{\dest}{{d}}
\NewDocumentCommand{\sink}{O{\ }}{destination#1}
\NewDocumentCommand{\Sink}{O{\ }}{Destination#1}
\NewDocumentCommand{\wellposed}{O{\ }}{well-posed#1}
\newcommand{\edgesFrom}[1]{\delta^+(#1)}
\newcommand{\edgesTo}[1]{\delta^-(#1)}
\NewDocumentCommand{\MeasFuncUInt}{O{u}}{L^{0,#1}_+(\hori)}

\NewDocumentCommand{\stwalk}{O{\source}O{\dest}}{\ensuremath{#1,#2}-walk}
\NewDocumentCommand{\stpath}{O{\source}O{\dest}}{\ensuremath{#1,#2}-path}
\newcommand{\stwalki}{\stwalk[\source_i][\dest_i]}

\newcommand{\refrunind}[2]{%
\ifthenelse{\equal{#2}{1}}{
	\ensuremath{%
		\ifthenelse{\equal{#1}{1}}{n}%
		{\ifthenelse{\equal{#1}{2}}{m}%
		{\ifthenelse{\equal{#1}{3}}{j}%
		{\ifthenelse{\equal{#1}{4}}{l}%
		{\ifthenelse{\equal{#1}{5}}{i}
		{\color{red}NaN}}}}}%
	}%
 }
{
\ensuremath{%
		\ifthenelse{\equal{#1}{1}}{N}%
		{\ifthenelse{\equal{#1}{2}}{M}%
		{\ifthenelse{\equal{#1}{3}}{J}%
		{\ifthenelse{\equal{#1}{4}}{L}%
		{\ifthenelse{\equal{#1}{5}}{I}
		{\color{red}NaN}}}}}%
	}%
}
}
 
\newcommand{\n}[1]{\ifthenelse{\equal{#1}{}}{\refrunind{1}}{\refrunind{#1}{1}}}
\newcommand{\capn}[1]{\ifthenelse{\equal{#1}{}}{\refrunind{1}}{\refrunind{#1}{2}}}

\newcommand{\Pf}{P}
\newcommand{\op}{\nabla}

\newcommand{\diffe}{\Delta}
\newcommand{\ain}{\rho}

\newcommand{\wto}{\rightharpoonup}

\NewDocumentCommand{\seql}{O{1}O{\Routes}O{L(\hori)}}{\otimes^{#1}_{#2}#3}
\newcommand{\const}{C}

\newcommand{\ofeas}{\mathcal{M}}
\NewDocumentCommand{\edom}{mO{}}{\mathcal H_{#1}^{#2}}

\newcommand{\wttime}{{\Psi}}
\newcommand{\ttime}{\hat{\Psi}}
\NewDocumentCommand{\auto}{O{\ }}{autonomous#1}
\NewDocumentCommand{\Aauto}{O{\ }}{An autonomous#1}
\newcommand{\aauto}{an autonomous }
\newcommand{\Auto}{Autonomous }
\newcommand{\mto}[1]{\overset{#1}{\to}}
\NewDocumentCommand{\Dwg}{O{\wa}O{\g}}{\mathfrak D_{#1}^{#2}}

\newcommand{\DestCyc}{\mathcal{C}^\dest}
\usepackage{ifthen}

\usepackage{xr}
\makeatletter
\newcommand{\oref}[1]{%
  \@ifundefined{r@#1}{%
   {\cite[\Cref*{FD-ext-#1}]{GHS24FD}}%
  }{%
   \Cref{#1}%
  }%
}
\makeatother

\newcommand{\ourref}[1]{%
\ifstrequal{#1}{lem: elluExistenceProperties}{%
\cite[\Cref*{FD-ext-lem: elluExistencePropertiesShort}]{GHS24FD}%
}%
{%
\cite[\Cref*{FD-ext-#1}]{GHS24FD}%
}%
}

\Crefname{rsttheorem}{Theorem}{Theorems}
\Crefname{parttheorem}{Theorem}{Theorems}
\Crefname{innerrsttheorem}{Theorem}{Theorems}

 \usepackage{apptools}

\ifreview
    \usepackage[disable]{todonotes}
\else
    \usepackage[textsize=tiny,textwidth=2cm,shadow,loadshadowlibrary]{todonotes}
\fi
\usepackage{booktabs}
\newcommand{\lgcom}[2][]{\todo[color=green!70!blue!60]{LG: #2}} 
\newcommand{\jscom}[1]{\todo[color=blue!50!white]{JS: #1}} 

\ifreview
    \newcommand{\BigPicture}[2][0]{#2}
\else
    
    \newcommand{\BigPicture}[2][0]{#2}
\fi

\title{Tolls for Dynamic Equilibrium Flows}

\ifreview
    \author{}
    \date{\vspace{-2cm}}
\else
    \author{Lukas Graf, Tobias Harks and Julian Schwarz} 

    \affil{\small University of Passau, Faculty of Computer Science and Mathematics, 94032 Passau\\
    \href{mailto:julian.schwarz@uni-passau.de}{\{\texttt{lukas.graf,tobias.harks,julian.schwarz\}@uni-passau.de}}}
\fi

 \newcommand{\leb}{\sigma}
  \newcommand{\eqperdef}{\overset{\text{def}}{=}}
 \newcommand{\defpereq}{\eqperdef}
 
 \newcommand{\SimpCyc}{\mathcal{C}^{\mathrm{simp}}}
 \newcommand{\Indi}{\mathbf{1}}
 \newcommand{\tEnd}{t_f}
 
\newtheorem{innerrstassumption}{Assumption}
\newenvironment{rstassumption}[1]
{\renewcommand\theinnerrstassumption{#1}\innerrstassumption}
{\endinnerrstassumption}
\usepackage{nicefrac}
 
\begin{document}

\maketitle
\begin{abstract}   
 We consider dynamic network flows and study the following  question:
Which dynamic edge flows can be implemented as
tolled dynamic equilibrium flows?
We study this question for the ``heterogeneous-user'' model, where the flow particles 
are partitioned into populations characterized by their own 
source,destination-pairs and a cost function associating with any walk and departure time some costs. 
As our two main results, we first provide 
a duality-based characterization of implementability of dynamic edge
flows for the multi-source, multi-\sink case. 
Secondly, we derive both, 
a combinatorial and duality-based characterization of implementability of dynamic edge
flows for the multi-source, single-\sink case. 
Both results are derived under a fairly general network loading model. 
For the proof, we make several
technical contributions:
We formulate a novel infinite dimensional optimization problem, where the goal is to minimize the aggregated costs of the particles with respect to the fixed network loading induced by the given edge flow. This requires the recently introduced concept
of \auto network loadings (cf.~Graf et al.~\cite{GHS24FD}) for which we show several new structural insights. 
In particular, compared to~\cite{GHS24FD}, we give an alternative (tighter) characterization of the existence of \auto network loadings for our setting 
by deriving a generalization of a result of M.A.~Zarecki\u{\i} on the Lusin~$N^{-1}$ property of absolutely continuous monotone functions which may also be of independent interest. 
These insights  allow us to prove the stated characterizations under the assumption of strong duality. 
Finally, for the case of a single-\sink[,] we are able to provide a non-trivial proof that this assumption is always fulfilled for finitely supported edge flows with costs representing weighted travel times.\footnote{Parts of these results have been published in the proceedings of of the 2025 Annual ACM-SIAM Symposium on Discrete Algorithms (SODA25)~\cite{graf2025tolls} for the more restricted setting of single-source, single-\sink networks and weighted travel times as private costs. }

\end{abstract} 

\ifreview 
\vspace{2cm}
\else
\clearpage
\tableofcontents
\clearpage
\fi

\section{Introduction}   
Traffic congestion in urban areas has been recognized
as a critical factor not only affecting the transportation sector but also the economic and social life of people. 
It causes increased travel delays and higher chances of traffic accidents but it also aggravates environmental pollution in terms of increased energy consumption and increasing pollutant emissions. A prominent mechanism
to alleviate congestion effects is to impose congestion prices. The idea here is to charge anonymous tolls for traversing
edges of the network in order to influence the resulting equilibrium flow so that certain desiderata such as low travel times, low emissions, or enough capacities for
emergency safeways are guaranteed; see Yang et al.~\cite{Yang2020} for recent work quantifying the effect of congestion prices on traffic speeds and welfare gains for the Bejing traffic district. Congestion or toll pricing has a long history in the 
economics literature as a mechanism to implement equilibria with desirable properties (see Pigou~\cite{Pigou20} and Knight~\cite{Knight1924} for early works on the subject) but in the context of traffic flows it has been mainly applied 
to static traffic assignment problems, implementing Wardrop equilibria in static flows. 
For this model, the seminal results by Cole, Dodis and Roughgarden~\cite{ColeDR03}, Fleischer, Jain and Mahdian~\cite{Fleischer04}, Karakostas and Kollioupous~\cite{Karakostas04} and Yang and Huang~\cite{Yang04}
give complete characterizations of static edge flows that are implementable as toll-based Wardrop equilibria even if the users have heterogeneous preferences over travel times and paid tolls.\footnote{A toll-based Wardrop equilibrium is an ordinary Wardrop equilibrium with respect to changed travel time functions, where the static edge tolls along the chosen path are added to the path travel time.}

In the domain of traffic modeling it has been recognized  that static flow models are often too coarse to accurately model realistic traffic congestion phenomena, especially because the inherent tempo-spatial propagation of vehicles through the network is neglected in static models.
In this regard,  \emph{dynamic flows}  have emerged as the predominant framework in the field as an accurate but often very complex model to describe traffic flows.  
Since their introduction by by Ford and Fulkerson~\cite{Ford62}, 
there has been an exciting development in the theoretical computer science and mathematics community in terms of our understanding of 
dynamic flows in the context of optimization (see e.g.~\cite{AndersonP84,Gale59,FleischerT98,Philpott90}) as well as from a game theoretic perspective (see e.g.~\cite{CominettiCL15,CominettiCO22,GHKM23,GHS18,Koch11,OlverSK21,OlverSK26,Vickrey69}).

\subsection{Related Work}
In spite of the aforementioned advancement and in contrast to the static case, only few papers investigate
the power of (dynamic) tolls for dynamic flows  from a purely theoretical perspective. 
While the surveys by Yang and Huang \cite{Yang2005} and Lombardi et al.~\cite{Lombardi21} characterize the literature on dynamic road pricing as extensive -- particularly within traffic engineering literature -- Lombardi et al.\ point out that 
``[i]n almost all approaches evaluated in this survey, the overall structure consists of a
pricing strategy followed by a simulation study that evaluates the pricing strategy.''
In this regard, we focus in the following on works that, 
at a minimum, provide a mathematical model together with formal statements. 

Frascaria and Olver~\cite{FrascariaO22} consider single-source, single-destination networks with travel times based on the Vickrey model 
and a fixed volume of homogeneous users with free departure time choice. 
They tackle the problem of computing a corresponding dynamic flow that minimizes the total costs, which are the sum of travel time and scheduling costs where the scheduling costs are incurred based on the arrival times at the destination. Moreover, they show that the optimal flow that they compute is implementable as tolled dynamic equilibrium by defining an infinite dimensional linear program and then using feasible dual variables as tolls. 
Note that for their proof, it is essential that the underlying travel times are constant for the optimal flow.
In addition, they also show how their result carries over to the case of fixed departure times.  
The latter result  falls within our model, but   only demonstrates the implementability of a specific flow 
for a set of homogeneous users with a specific cost function and common source and destination. 
In contrast, we allow for heterogeneous users with generic costs and characterize the implementability of any flow 
within multi-source, multi-destination networks.

 Rosner, Schröder, and Vargas Koch~\cite{VargasKoch2025nashflowstimetolls} recently investigated the structure of equilibria arising in networks equipped with constant edge tolls within the Vickrey queuing model under constant network inflow rates.
They propose a procedure based on a linear programming formulation to compute steady states, that is, phases of an equilibrium flow that persist for all future times once reached. 
Moreover, they show that the resulting tolled equilibria need not be unique \wrt induced costs and may fail to reach a steady-state phase within finite time.

Wie and Tobin~\cite{WieTobin1998} considered a dynamic traffic assignment model based on an optimal control formulation, where the goal is to minimize the total travel time. The state variables of this formulation are the edge flow volumes and the control variables are the edge inflows. Under a strong differentiability condition on the queue exit-flow function and a convexification of the non-linear problem, they
prove that an optimal solution satisfies a standard constraint qualification condition which yields the existence of associated dual variables that are used as dynamic tolls. Note, that these tolls depend on the users' destination (i.e.\ are not anonymous) and only apply to the case of homogeneous users. Moreover, the Vickrey model does not satisfy the differentiability condition and, thus, their results only apply to their specific model.

 Ma, Ban and Szeto~\cite{Ma2017} considered a so-called double queue formulation (queues at the entry and exit of an edge) and also argued
 that an optimal solution with respect to a slightly different objective (accounting also  for emissions) can be implemented by tolls. They argue that a corresponding optimal control formulation can be used to derive dual variables. However, no formal proof of the existence of an optimal control solution, nor
of any regularity condition  is given to support these claims.

Yang and Meng~\cite{Yang1998}  and also Yang and Huang~\cite[Chapter 13]{Yang2005} studied a model where time is discretized leading to a space-time expanded network. This reduces the toll problem to a static problem and they derived tolls implementing an optimal (discrete-time) flow using the standard 
finite-dimensional linear programming formulation as used in~\cite{Fleischer04, Karakostas04, Yang04}.
For static models, Harks and Schwarz~\cite{NonconvexPricing} presented a general duality framework which generalized the previously mentioned works~\cite{ColeDR03,Fleischer04,Karakostas04,Yang04}. Their framework, however, only applies to finite-dimensional resource allocation problems and is, thus, not applicable to the model considered in this paper. There are (strong) duality results for variants of system-optimal dynamic flow problems known, see
Anderson and Philpott~\cite{AndersonP84}, Philpott~\cite{Philpott90} and Koch et al.~\cite{KochN14,KochNS11}.
All these models, however, assume constant flow-independent travel times.

We  recently introduced in~\cite{GHS24FD}  the concept of \auto network loadings, where walk inflow rates are loaded in the network according to a fixed (i.e.~flow independent) travel time function. 
 We then showed that these \auto network loadings are a useful tool for studying flow models with flow dependent travel times as well: More precisely, we considered the same dynamic flow model as in this paper and, first, gave  a characterization of those  walk inflow rates that lead to feasible dynamic edge flows after loading according to the fixed travel times. 
 Based on this and several other structural results on \auto flows, we ultimately 
  derived a flow
decomposition theorem for dynamic $s$,$d$ edge flows with general flow dependent travel times.

At various places in our proofs in this work, we will use the concept of \auto network loadings. 
In this regard, we use the results which we derived in~\cite{GHS24FD} but also derive new structural results on \auto flows.

\subsection{Our Results}
As described above, for dynamic equilibrium flows -- regardless of which network loading model is chosen --  not much is known 
regarding the power of dynamic congestion pricing.
In this paper, we consider the fundamental question of 
which dynamic network flows are implementable
as a tolled dynamic equilibrium, i.e.~for which flows there exist anonymous dynamic tolls on the edges such that a tolled dynamic equilibrium exists that induces this flow.  
We will consider this basic question for heterogeneous populations, where
each population  is characterized by a source,destination-pair and a private flow- and time-dependent
cost function for each walk.  
This includes in particular the case in which all users aim to minimize their travel time but 
  different populations
come with  different valuations of time (VoT), that is, they value spent money (due to to potential tolls)
and perceived travel times differently. 
Moreover, all of our results are based on a quite general network-loading model which includes
the linear edge-delay model and the Vickrey queuing model as special cases.
 
 We summarize our main results in the following. As a key step in our approach, we 
introduce an infinite dimensional program~\eqref{opt: Master} parameterized in the flow $u \in L_+(\hori)^\GA$ which we aim to implement. 
Under natural assumptions, we then show two characterizations: One for the multi-destination case (where different populations may have different destinations) and one for single-destination case (where all populations share the same destination):

\begin{introthm}[Multi-Destination - Duality Based Characterization]{thm: mainMSMS}  
For heterogeneous users in multi-source, multi-\sink networks with \eqref{opt: Master} admitting strong duality \wrt $L_+^\infty(\hori)^\GA$, the following statements are equivalent:
    \begin{enumerate}
        \item $u$ is implementable via bounded tolls $\prices \in L_+^\infty(\hori)^\GA$.
        \item  There exists a walk inflow $h^*$ which is optimal for  \eqref{opt: Master} and induces $u$. 
 \end{enumerate}
\end{introthm}

\begin{introthm}[Single-Destination -- Duality Based Characterization]{thm: mainMSSS}  
For heterogeneous users in multi-source, single-\sink networks with \eqref{opt: Master} admitting strong duality \wrt $\MeasFuncUInt$, the following statements are equivalent:
    \begin{enumerate}
        \item $u$ is implementable via tolls $\prices \in \MeasFuncUInt$.
        \item  There exists a walk inflow $h^*$ which is optimal for  \eqref{opt: Master}  and induces $u$. 
 \end{enumerate}
\end{introthm}
In the above, $\MeasFuncUInt$ denotes the set of all admissible tolls, that is, nonnegative real-valued tolls that, while not necessarily being bounded themselves, lead to finite total toll-induced costs $\sum_{\arc\in\GA}\int_\R\prices_\arc\cdot u_\arc\di\leb<\infty$ for the flow $u$. 
In this regard, note that the prerequisite of strong duality \wrt $L_+^\infty(\hori)^\GA$ in the multi-destination case is a stronger condition than the prerequisite of strong duality \wrt $\MeasFuncUInt$ in the single-destination case as  $L_+^\infty(\hori)^\GA$ can be seen as a subset of  $\MeasFuncUInt$. 

In the single-\sink case, we are also able to extend the above duality-based characterization via a combinatorial condition:

\begin{introthm}[Single-Destination -- Combinatorial Characterization]{thm: mainSingleSink}   
For heterogeneous users in multi-source, single-\sink networks with \eqref{opt: Master} admitting strong duality \wrt $\MeasFuncUInt$, the following statements are equivalent:
    \begin{enumerate}
        \item $u$ is implementable via tolls $\prices \in \MeasFuncUInt$.
        \item  There exists a walk inflow $h^*$ which is optimal for  \eqref{opt: Master}  and induces $u$. 
        \item Any flow-carrying cycle starting at the destination must have zero private costs for every population.
 \end{enumerate}
\end{introthm}
For the above situation, this result has, in particular, the following consequences: First, 
if there are no outgoing edges from the common destination, 
every  multi-source, single-\sink flow is implementable. 
Second, as long as free-flow travel times are strictly  positive, any dynamic multi-source, single-\sink flow~$u$ is implementable iff it has no outflow from the destination. 
Note that for homogeneous populations the latter statement can also be shown more directly
by adopting the approach of Fotakis and Spirakis~\cite{FotakisS08} 
for the static case.
Indeed, based on their approach, 
we show in  \Cref{thm:ImplementabilityHomogeneous} how to explicitly construct implementing tolls in this setting. 
Those tolls are then continuous and can even be chosen piece-wise linear and computed in finite time 
if the travel times are piecewise-wise linear as well. 
 

Finally, since our characterizations heavily rely on the master problem \eqref{opt: Master} admitting strong duality, 
we derive a sufficient condition for this duality to hold. 
\begin{introthm}[Single-Destination - Strong Duality]{thm: MasterZeroDual}   
For heterogeneous users in multi-source, single-\sink networks, the master problem \eqref{opt: Master} admits strong duality \wrt $L_+^\infty(\hori)^\GA$ 
if there exists a constant 
such that for any feasible solution the costs of any used walk are bounded by this constant.
\end{introthm}
Remark that the sufficient condition for strong duality given in the above theorem is in particular fulfilled if the private costs represent weighted travel times and $u$ has finite support. This is, because any feasible solution for \eqref{opt: Master} may only use walks and departure times for which the corresponding arrival time at the destination is contained in the support of~$u$.

\subsection{Technical Contributions}
The proofs of our main results heavily rely on the formulation of the master problem~\eqref{opt: Master}. 
Given a dynamic edge-load vector~$u \in L_+(\hori)^\GA$, 
we  define the following infinite dimensional optimization problem of the abstract form

 \begin{align}
    \inf_{h} \; & \sum_{\wa \in \Routes}\dup{\wttime_\wa (u,\cdot)}{h_\wa }\tag{${\mathrm{P}}(u)$}\label{opt: Master}  \\
    \text{s.t.: } & \ell^u(h) \leq u \label{ineq: Master}\\
    &(h_\wa)_{\wa \in \Routes} \in \edom{\Routes}[u] \cap \wir.\nonumber
\end{align}

Here, $\Routes$ contains all tuples of the form $\wa=(\hat{\wa},i)$ with $\hat{\wa}$ being a finite \stwalki{} and $i \in I$ some commodity. 
The decision variables of~\eqref{opt: Master} are the walk inflow rates $h_{\wa},\wa \in \Routes$ where $h_{\wa}$  
represents the walk inflow rate of commodity $i$ into the walk $\hat{\wa}$. 
The value~$\wttime_{\wa}(u,t)$ denotes for any time $t \in \hori$ and walk $\hat{\wa}$ 
the private costs for commodity $i$. 
In this regard, 
the objective   aggregates over all walks~$\wa$ and commodities~$i$ the 
experienced private costs of particles sent under $h$ into the different walks under the fixed traversal times and induced private costs of $u$.  
The key idea in the definition of~\eqref{opt: Master} is that feasibility of a walk inflow rate~$h$ is defined with respect to the  \emph{fixed} travel times induced by the flow~$u$ rather than with respect to the travel times induced by~$h$.
This modeling step requires the concept of \auto network loadings introduced in our previous work~\cite{GHS24FD}. 
Roughly speaking, $\ell^u_{\wa}(h_{\wa})$  describes how the particles 
sent under $h_{\wa}$ into the walk $\hat{\wa}$  would hypothetically propagate throughout
the network under the fixed travel times of~$u$. However, not every walk inflow rate leads to a well-defined edge flow when loaded according to~$u$ (cf.~\oref{exa: noarcflow}). Thus, the feasibility set of~\eqref{opt: Master} is exactly tailored to address this problem 
by intersecting 
the  set  of all possible walk inflow rates~$\wir$   further with the maximal set~$\edom{\Routes}[u]$ of  walk inflow rates leading to a well-defined edge flow.  
 Note that, otherwise,~\eqref{opt: Master}
would not be well-defined over the  infinite dimensional space of possible walk inflow rates.

\vspace{.5\topskip}
In \textbf{\oref{sec:uBasedNetworkLoadings}~(\nameref*{sec:uBasedNetworkLoadings})}, we 
derive several structural insights on \auto network loadings. We start by
considering in \Cref{sec:uBasedNetworkLoadings:Exis} the existence of \auto network loadings corresponding to a walk inflow rate into a walk~$\wa$. In  \ourref{lem: elluExistenceProperties}, we characterized the latter   for an underlying travel time function $\trav$ as follows: 
For all edges $\arc$ on the walk $\wa$, no flow of positive measure is sent into the walk  in such a way that these flow particles all arrive during a null set of times under $\trav$. 
In this paper, we show that for any edge inflow rate vector $\g$ that admits an edge outflow rate vector~$\g^-$, 
 the aforementioned property  is automatically  satisfied for null sets consisting only of arrival times at~$\arc$ such that the corresponding intermediate arrival times at this and all previous edges coincide with times of positive inflow under~$\g$ (\Cref{thm: arrLusin} and \Cref{cor: LuExAltChara}).
This is in particular applicable when we consider a flow $u$ fulfilling flow conservation as this guarantees the existence of edge outflow rates, as any edge outflow is the edge inflow of other edges. 
Thus, we end up with a tighter characterization of \auto network loadings for the \auto network loading induced by $u$. 
 In order to prove this, we derive in \Cref{lem: absconLus} a generalization of a result of M.A.\ Zarecki\u{\i} on the Lusin~$N^{-1}$ property of absolutely continuous monotone functions which may also be of independent interest.  
By exploiting the existence of an edge outflow vector $\g^-$, we show in \Cref{lem:  g>0Arr'>0} that for all 
walks and almost all departure times at which one arrives at all edges on the walk only if they have positive inflow under $\g$,  
the corresponding arrival time function has a positive derivative. 
Note that \Cref{lem:  g>0Arr'>0} is a  statement about general network loadings and, hence, may also be relevant outside the context of \auto network loadings.

Next, we establish in \oref{lem: aggCostsVSwalkCosts} that  the 
walk- and edge-based definitions of the total cost coincide. This will be important in \Cref{sec: SingleSinkImpl,sec:CharImplementability}, where it is necessary to switch between these two definitions.

In preparation of the aforementioned combinatorial characterization of implementable multi-source, single-\sink flows, 
we show in  \Cref{lem: DifferenceGeneral} that for any feasible solution~$\tilde{h}$ of the master problem with corresponding induced \auto 
  edge flow $\tilde{\g}$,  the difference  $u-\tilde{\g}$ is decomposable into flows on zero-cycles and $\dest$-cycles.   

Finally, we 
show  in \Cref{claim: ZeroDualityGapTildeH} the existence of a largest common flow of a walk and edge flow in the sense that this flow sends as much of the  walk flow as possible without exceeding the edge load of the edge flow. 
This result will play a crucial role for showing that~\eqref{opt: Master} fulfills strong duality in the single-\sink case.

\vspace{.5\topskip}
In \textbf{\oref{sec:Implementability}~(\nameref*{sec:Implementability})}, we come back to the question of implementability. 

In \oref{sec:CharImplementability}, we 
demonstrate that implementability is tightly connected to the master problem~\eqref{opt: Master}.
For this, we consider two types of strong duality of the mater problem \eqref{opt: Master}, depending on which set of dual functions 
we allow. While we allow for \Cref{thm: mainMSSS} in the multi-source, single-\sink case the set of all admissible tolls as dual functions, we consider for  multi-source, multi-\sink networks in \Cref{thm: mainMSMS} the only subset of admissible tolls that are uniformly bounded.  
Under the assumption of \eqref{opt: Master} admitting zero duality gap \wrt the respective set of tolls, we prove in \Cref{thm: mainMSMS,thm: mainMSSS}
that implementability of $u$ via such tolls is equivalent to the master problem \eqref{opt: Master} admitting an optimal solution $h^*$ inducing $u$. 

Proving that the latter condition is necessary for implementability is quite straight-forward and exploits the fact that the tuple of walk flow and tolls  implementing $u$ admits zero duality gap for~\eqref{opt: Master}.\footnote{Note that we already proved this in~\cite{GHS24SO} which appeared as one-page-abstract in the proceedings of the 26th ACM Conference on Economics and Computation (EC26)~\cite{EC25}}

For sufficiency, 
we start with a tuple of walk flow and tolls  admitting zero duality gap which exists due to the assumption of~\eqref{opt: Master} fulfilling strong duality.   
Roughly speaking, these tolls guarantee that for almost all points in time, 
all walks   in which flow can be send under~$u$ without resulting in an undefined network loading do not have smaller costs than the walks utilized in the optimal solution. In order to ensure that this holds for \emph{all} walks, we increase the tolls suitably on all edges not used under~$u$. 
We remark that the tighter characterization of existence of \auto network loadings in \Cref{thm: arrLusin}  is crucial in order to show that this adjustment is actually sufficient. Moreover, the need for this adjustment is the reason for us to distinguish between the two types of strong duality mentioned above: 
The adjustment in \Cref{thm: mainMSSS} is quite involved and exploits the fact that all commodities share the same \sink whereas the adjustment in \Cref{thm: mainMSMS} is simpler, yet crucially relies on dual tolls being uniformly  bounded.  
 
In \Cref{sec: SingleSinkImpl}, 
we prove 
 a combinatorial characterization of implementable flows $u$ for single-\sink network (\oref{thm: mainSingleSink}), stating that 
 $u$  does not contain  cycles starting at the destination with strictly positive private costs. 

In order to prove that the latter condition is necessary for implementability, 
we consider a pair of implementing tolls and walk inflows $(\prices^*,h^*)$ together with an arbitrary cycle starting at the destination which has only utilized edges under~$u$.  
We then show for every edge on this cycle  (by induction over their position) the following property: 
For any other walk 
over which flow is sent into this edge under $h^*$, the total costs experienced by this flow from that edge onward is zero. This, in particular, implies that the private costs of that edge are zero. 

For sufficiency, we show that the  combinatorial property implies the existence of an optimal solution for the master problem \eqref{opt: Master} 
with tight inequality which, in turn, shows the implementability of $u$ via \Cref{thm: mainMSSS}. 
 This is done by 
demonstrating that, under the assumption of $u$ fulfilling the combinatorial property, any optimal solution to~\eqref{opt: Master} induces an edge flow that only differs from~$u$ in flow on cycles of zero private costs. 
Then, the challenge is to consistently add these cycles to such an optimal solution in order to construct  an optimal solution to~\eqref{opt: Master} with tight inequality. 
Here, we rely heavily on the insights obtained in \cite{GHS24FD} regarding (pure) flow decompositions.

\vspace{.5\topskip}
Finally, in \textbf{\oref{sec:ExistenceOptSolutions}~(\nameref*{sec:ExistenceOptSolutions})}, we show that a whole class of optimization problems, including~\eqref{opt: Master} for multi-source, single-\sink networks and finitely supported~$u$, exhibits strong duality (\Cref{thm: ZeroDualityGapGeneral,thm: AlmostMasterZeroDual,thm: MasterZeroDual}). This result does not
follow from standard regularity conditions of infinite dimensional linear programs (see e.g.~\cite{anderson1983review,boct2008revisiting,flores2013strong}). 
In fact, it is known that strong duality 
 is not guaranteed in general for infinite-dimensional programs~\cite{RomeijnSB92} and we 
 even provide an example (\Cref{ex: DualGap}) of a single commodity flow with \emph{unbounded} support 
 whose corresponding master problem does not admit strong duality.  
 We think that our positive strong duality result might be of independent interest for the optimization community as well.

\subsection{Paper Organization}

After presenting the general model in \oref{sec: Model}, we start in \oref{sec:CostBalancingTolls} as a warm-up with the simpler case of a homogeneous populations of users sharing the same cost functions and \sink[.] 

For the rest of the paper, we then consider the much harder question of \emph{heterogeneous} users: In \oref{sec:uBasedNetworkLoadings}, we derive the aforementioned insights into \auto network loadings.
Then, in \oref{sec:Implementability}, we show come back to the question of implementability and provide the promised  characterizations. 
Finally, in
\oref{sec:ExistenceOptSolutions}, we show that the 
$u$-parameterized infinite dimensional master problem fulfills strong duality in the single-\sink case.

\section{Model}\label{sec: Model} 

We consider 
a directed graph $G=(\GV,\GA)$ with nodes $V$ and edges $\GA \subseteq V \times V$. 
We use $\edgesFrom{v}$ to refer to the set of edges leaving a node $v$ and $\edgesTo{v}$ for the set of edges entering $v$. Furthermore, 
we denote  for any pair $v,v' \in \GV$ by $\hat{\Routes}_{v,v'}$ the countable set of (finite) \stwalk[v][v']s in~$G$ and by $\hat{\Routes}=\bigcup_{v,v'\in V}\hat{\Routes}_{v,v'}$ the set of all finite walks in $G$. Here, a \stwalk[v][v'] $\hat{\wa}$ is a tuple 
of edges $\hat{\wa} = (\arc_1,\ldots,\arc_k) \in \hat{\Routes}$ with $\arc_j=(v_j,v_{j+1})\in \GA$ for all $j\in [k]:=\{1,\ldots,k\}$ for some $(v_j)_{j\in[k+1]}\in \GV^{k+1}$ with $v_1 = v$ and $v_{k+1} = v'$. We use $\hat\wa[j] \coloneqq e_j$ to refer to the $j$-th edge on walk~$\hat\wa$.
Moreover, we write $v\in \hat{\wa}$ and $\arc \in \hat{\wa}$ to say that there exists  some $j \in [k]$ and $\hat{v},v'\in \GV$ with $(\hat{v},v') = \hat{\wa}[j]$ and $v \in \{\hat{v},v'\}$, respectively $\arc=\hat{\wa}[j]$, and use $\abs{\hat{\wa}} \in \N_0$ for the length (=number of edges) of~$\hat{\wa}$.
Note, that we explicitly allow walks to contain cycles and include the same edge multiple times. 
Here, we call a walk $\hat c=(\arc_1,\ldots,\arc_m)$ a cycle if $\arc_1 \in \edgesFrom{v}$ and $\arc_m \in \edgesTo{v}$ for some node  $v \in \GV$. In case that $v= \dest$, we call $\hat c$ a $\dest$-cycle and similarly call  any walk $\hat{\wa}$ a $\dest$-walk if $\hat{\wa}[1]\in \edgesFrom{\dest}$.  
A walk $\hat{\wa}$ is called simple, if it does not visit a node twice except possibly the starting node, i.e.~for all $v\in \GV$ there exists at most one $\arc \in \hat{\wa}$ with $\arc \in\delta^+(v)$. We denote by $\SimpCyc$ the finite set of simple cycles,  by $\DestCyc$ the set of all finite (not necessarily simple) $\dest$-cycles and by $\hat{\Routes}^\dest$ the set of finite $\dest$-walks. 
Moreover, for a  walk $\hat{\wa}$ and $j\leq |\hat{\wa}|$, we denote by  $\hat{\wa}_{\geq j}$ and $\hat{\wa}_{>j}$ the sub-walk of $\hat{\wa}$ starting with $\hat{\wa}[j]$, respectively $\hat{\wa}[j+1]$. Analogously, we define $\hat{\wa}_{\leq j}$ and $\hat{\wa}_{<j}$.  
Furthermore, for two walks $\hat{\wa}^1 = (\arc_1,\ldots,\arc_{k_1}),\hat{\wa}^2=(\arc_1^2,\ldots,\arc^2_{k_2})$
with $\hat{\wa}^1$ ending in a node $v$ and $\hat{\wa}_2$ starting in it, 
we write $(\hat{\wa}^1,\hat{\wa}^2):=(\arc_1^1,\ldots,\arc^1_{k_1},\arc^2_1,\ldots,\arc^2_{k_2})$.

Next, we consider $\hori$ as our planning horizon during which flow particles can traverse the network. 
Since dynamic flows will be described by Lebesgue-integrable functions on~$\hori$, we equip $\hori$ with its Borel $\sigma$-algebra $\mathcal{B}(\hori)$. We denote by $\sigma$ the Lebesgue measure on $(\hori,\mathcal{B}(\hori))$ 
and by $L(\hori)$ and $L^\infty(\hori)$ the space of ($\sigma$-equivalence classes of)  $\sigma$-integrable and essentially bounded, respectively, real-valued functions over $\hori$  equipped with the standard norm induced topology and the partial order induced by $L_+(\hori)$ and $L_+^\infty(\hori)$, respectively, i.e.\ the subsets of nonnegative integrable functions. 
For any countable set~$M$, we denote by $\seql[1][M][L(\hori)]$ the set of 
vectors $(h_m)_{m \in M} \in L(\hori)^M$ whose sum $\sum_{m \in M}h_m \in L(\hori)$ is finite, i.e.
\begin{align*}
    \seql[1][M][L(\hori)] &:= \big\{ h \in L(\hori)^{M} \mid \norm{h} := \sum_{m \in M} \norm{h_m} < \infty \big\}.
\end{align*}
 This defines again a Banach space (cf.~\cite[Section 16.11]{guide2006infinite}) whose topological dual is 
 \begin{align*}
     \seql[\infty][M][L^\infty(\hori)]:= \big\{ f \in L^\infty(\hori)^M \mid \norm{f} := \sup_{m \in M} \norm{f_m}_\infty < \infty \big\}. 
\end{align*} 
We denote the bilinear form between this dual pair  by $\dup{f}{h} \coloneqq \sum_{m \in M}\int_\hori f_m\cdot h_m\di\sigma$ for $f\in \seql[\infty][M][L^\infty(\hori)],h\in\seql[1][M][L(\hori)]$. Here, we use $\int_\hori f\di\sigma$ to denote the integral of $f$ over~$\hori$ with respect to the Lebesgue measure~$\sigma$.\footnote{We use this notation instead of writing $\int_{\hori} f(t)\di t$ to stay consistent with the proofs of some of the more technical lemmas where we also have to consider integrals with respect to other measures.} 
Similarly, we denote for any two vectors of non-negative measurable functions  (not necessarily contained in $L^\infty(\hori)^M$ or  $L(\hori)^M$)    $\tilde{f},\tilde{h}:\hori \to \R^M_+$ the sum over their integrals via $\dup{\tilde f}{\tilde h} \coloneqq \sum_{m \in M}\int_\hori \tilde f_m\cdot \tilde h_m\di\sigma \in \R_+ \cup\{\infty\}$.   
Moreover, 
for any such non-negative measurable function $\tilde{h}:\hori \to \R^M_+$, we define via 
\begin{align*}
    \MeasFuncUInt[\tilde{h}]:= \big\{ \tilde{f}: \hori \to \R_+^M\,\big\vert\, \tilde{f} \text{ is measurable and } \dup{\tilde{f}}{\tilde{h}} < \infty \big\} 
\end{align*}
 the set containing all measurable non-negative functions $\tilde{f}:\hori \to \R_+^M$ 
 with bounded product. 
We say that a sequence $(h_n)_{ \in \N}$ converges weakly in $\seql[1][M][L(\hori)]$ to some $h \in \seql[1][M][L(\hori)]$ and write $h_n \wto h$, if $\dup{f}{h_n} \to \dup{f}{h}$ for any $f\in\seql[\infty][M][L^\infty(\hori)]$. 
Analogously to $\seql[1][M][L(\hori)]$, we define $\seql[1][M][L(\hori)^\GA]$ where we use $\norm{\g} := \sum_{\arc \in \GA}\norm{\g_\arc}$ for $\g \in L(\hori)^\GA$.

Finally, we have a finite set of commodities/populations~$I$ where each commodity comes with its own fixed network inflow rate $\inflow_i\in L_+(\hori)$ where $\inflow_i \neq 0$ specifying for (almost) every point in time at what rate particles of that commodity enter the network at the commodity specific source node~$\source_i$. The particles then aim to traverse the network by choosing an \stwalk[\source_i][\dest_i] where $\dest_i$ is the commodity specific \sink node of which we assume that it is connected to and different from $\source_i$.  
We denote by $\Routes \coloneqq \cup_{i \in I} \Routes_i \coloneqq \bigcup_{i \in I}\hat{\Routes}_{\source_i,\dest_i} \times\{i\}$ the set of walk-commodity pairs. By some abuse of notation we will also refer to elements $\wa=(\hat\wa,i) \in \Routes$ as walks and extend the above introduced terminology $\wa[j]:= \hat{\wa}[j]$, $v \in \wa :\Leftrightarrow v \in \hat{\wa}$, $\arc \in \wa :\Leftrightarrow \arc \in \hat{\wa}$, $\abs{\wa}:=\abs{\hat{\wa}}$ and $\wa_{\sim j}:=(\hat{\wa}_{\sim j},i)$ for ${\sim} \in \{<,\leq,\geq,>\}$.

\subsection{Dynamic Flows}
The main concept underlying dynamic flows are the traversal time functions:  

\myparagraph{Traversal time functions} Within our model any vector of edge inflow rates $\g\in L_+(\hori)^\GA$ induces    corresponding nonnegative  \emph{edge traversal time} functions
$\trav_\arc(\g,\cdot):\hori \to \R_+,\arc \in \GA$ 
 with 
$\trav_\arc(\g,t)$ denoting the time needed to traverse $\arc$ 
when entering the latter at time~$t$. 
To any such edge traversal time function, we also define two related functions: 
Firstly, the \emph{edge exit time} function $\exit
_\arc(\g,t):= t + \trav_\arc(\g,t)$ denoting the time a particle exits edge~$\arc$ when entering at~$t$. 
Secondly, the \emph{edge arrival time} function $\arr_{\wa,j}(\g,\cdot)$ denoting the time a particle arrives at the tail of the $j$-th edge of some walk~$\wa$ when entering~$\wa$ at time~$t$. More precisely, for an arbitrary walk $\wa$ we define $\arr_{\wa,1}(\g,\cdot):= \id$ and then, recursively, $\arr_{\wa,j}(\g,\cdot):= \exit_{\wa[{j-1}]}(\g,\cdot) \circ \arr_{\wa,j-1}(\g,\cdot)$ for $j\in\{2,\ldots,|\wa|\}$. Additionally, we define $\arr_{\wa,|\wa|+1}(\g,\cdot):= \exit_{\wa[|\wa|]}(\g,\cdot)\circ\arr_{\wa,|\wa|}(\g,\cdot)$
 denoting the arrival time at the \sink[.] 
 We assume that $\trav_\arc(\g,\cdot)$ is (locally) absolutely continuous\footnote{We omit from now on the term locally and call a function $\trav:\hori \to \R$ absolutely continuous  if it is absolutely continuous on every closed interval $[a,b]\subseteq \hori$ in the sense of \cite[Definition 5.3.1]{Bogachev2007I}. We remark that the main properties of absolutely continuous functions carry directly over to the locally absolutely continuous ones. All properties we require are gathered in \Cref{lem: PropAbsCon}.}  and adheres to the first-in first-out principle (FIFO), that is, $\exit_\arc(\g,\cdot)$ is a non-decreasing function. Note, that this also implies that both $\exit_\arc(\g,\cdot)$ and $\arr_{\wa,j}(\g,\cdot)$ are absolutely continuous as well (cf.\ \Cref{lem: PropAbsCon:Conca}).

With this, we can now formally describe dynamic flows. We will use two types of these flows:
Aggregated edge flows and walk flows: 

\myparagraph{Walk Flows}
For any countable collection of finite walks $\Routes'$, a 
 \emph{walk flow} or \emph{walk-inflow function} is a vector $h\in L_+(  \hori)^{{\Routes'}}$
with $h_{{\wa}}(t)$ representing the inflow rate at time $t\in \hori$ into the walk ${\wa}\in  {\Routes'}$. 
For $h \in L_+(\hori)^\Routes$, we often use $h^i:=(h_\wa)_{\wa \in \Routes_i}$ to denote the walk inflow rate of commodity $i\in I$.

\myparagraph{Aggregated Edge Flows} 
An (aggregated) \emph{edge flow} is a vector $\g \in L_+(\hori)^\GA$, where $\g_\arc(t)$ denotes the rate at which particles (of all commodities combined) enter edge~$\arc$ at time~$t$. 

We then connect the two types of flows by saying that an edge flow~$\g$ is \emph{induceable} by  a walk flow $h \in L_+(\hori)^{\Routes'}$  if: 
  \begin{align}\label{eq: DefEdgeFlow}
    \int_{\startint t]} \g_\arc\di\sigma = \sum_{\wa \in \Routes'}\sum_{j: \wa[j] = \arc}\int_{\arr_{\wa,j}(\g,\cdot)^{-1}({\startint t]})} h_\wa\di\sigma \text{  for all  $\arc \in \GA$ and $t \in \hori$.} 
\end{align}  
We denote by 
\begin{align*}
    \wir:= \Big\{ h\in L_+(  \hori)^\Routes \Big\vert\,\sum_{\wa\in \Routes_i} h_\wa=\inflow_i,i\in I,   \text{\eqref{eq: DefEdgeFlow} has a solution $\g$} \Big\} \subseteq \seql
\end{align*}
the set of \emph{admissible walk flows} (\wrt the commodity set $I$). 
We assume that we have a network loading operator (\wrt $\trav(\cdot,\cdot)$) $\Nl:\wir \to L_+(\hori)^\GA, h\mapsto \g$ sending each $h \in \wir$ to a solution $\g$  of \eqref{eq: DefEdgeFlow}. We then say that $h \in \wir$ induces $\g \in L_+(\hori)^\GA$ if $\Nl(h) = \g$  holds and call $\Nl(\wir) := \{\Nl(h) \in L_+(\hori)^\GA \mid h \in \wir\}$ the set of induced edge flows.

\myparagraph{\Auto Network Loadings}  
  If the travel times are independent of the edge inflow rates $\g$, that is, if there
exists a function $\tilde{\trav}:\hori \to \R^\GA$ such that $\trav(\g,t) = \tilde{\trav}(t)$ for all $\g \in L_+(\hori)^\GA$, we follow \cite{GHS24FD} in calling the corresponding network loading   \emph{\auto[}] and write~$\Nl[\tilde{\trav}(\cdot)]$.

 This type of network loading was  investigated by us in \cite{GHS24FD} and plays a key role in setting up the 
  master problem \eqref{opt: Master}. It allows us to describe for the given edge flow~$u$
how particles of a different walk flow~$h$  would hypothetically propagate throughout the network under the fixed travel times induced by~$u$. For this, we use the \auto network loading induced by $u$, that is, we consider the \auto network loading $\Nl[\trav^u(\cdot)]$
corresponding to the travel time function $\trav^u(\cdot):=\trav(u,\cdot)$ induced by $u$ under the initial  (in general non-\auto[)] travel time function~$\trav(\cdot,\cdot)$.

In \cite{GHS24FD}, we investigated the concept of \auto network loadings in a more general framework: For an arbitrary but fixed (flow-independent) absolutely continuous travel time function $\trav(\cdot):\hori \to \R_+^\GA, t \mapsto (\trav_\arc(t))_{\arc \in \GA}$ fulfilling FIFO, we analyzed the existence and properties of 
edge flows ${\g} \in L_+(\hori)^\GA$ which can be induced 
by a walk flow vector $h \in L_+(\hori)^{\Routes'}$ sending into all walks $\wa \in \Routes'$  flow $h_\wa$  under the \emph{fixed} traversal time functions $\trav$ for an arbitrary countable collection of walks $\Routes'$, i.e.~flows fulfilling for all $\arc \in \GA$: 
\begin{align}\label{eq: DefUEdgeFlow}
    \int_{\startint t]} \g_\arc\di\sigma =\sum_{\wa\in \Routes'} \sum_{j: \wa[j] = \arc}\int_{\arr_{\wa,j}^{-1}({\startint}t])}h_\wa\di\sigma \text{ for all } t \in \hori. 
\end{align}
Here, we denote as above by $\exit_\arc$ and $\arr_{\wa,j}$ the exit and arrival time function corresponding to $\trav$.

In the above situation \eqref{eq: DefUEdgeFlow}, 
we write $\Nl[\trav(\cdot)]_{\Routes'}(h) =  (\Nl[\trav(\cdot)]_{\Routes',\arc}(h))_{\arc \in \GA} := (\g_\arc)_{\arc \in \GA}$, say that $\Nl[\trav(\cdot)]_{\Routes'}(h)$ is well-defined, exists and   that the walk inflow rate $h$ induces the latter under $\trav$. 
We denote by $\edom{\Routes'}[\trav(\cdot)]$ the maximal subset of $L_+(\hori)^{\Routes'}$ for which $\Nl[\trav(\cdot)]_{\Routes'}(h)$ is well-defined and exists. 

Moreover,   we say 
that the walk inflow $h_\wa$ into a walk $\wa$ under the \emph{fixed} traversal time functions $\trav$ induces the \auto flow $\g^\wa_j \in L_+(\hori)$ on the $j$-th edge of $\wa$ if  $\g^\wa_j$ fulfills:
\begin{align}\label{eq: DefUEdgeFlowSingle}
    \int_{\startint t]} \g^\wa_j\di\sigma =  \int_{\arr_{\wa,j}^{-1}(\startint t])}h_\wa\di\sigma \text{ for all } t \in \hori. 
\end{align} 
In the above situation \eqref{eq: DefUEdgeFlowSingle}, we write $\Nl[\trav(\cdot)]_{\wa,j}(h_\wa):= \g^\wa_j$ for any $j \leq \abs{\wa}+1$ and say that the induced \auto edge flow $\Nl[\trav(\cdot)]_{\wa,j}(h_\wa)$ exists. Here, for $j = \abs{\wa} +1$, the latter is to be interpreted as the network outflow rate of the induced flow.  
We denote by $\edom{\wa,j}[\trav(\cdot)]$ the maximal subset of $L_+(\hori)$ for which $\Nl[\trav(\cdot)]_{\wa,j}(h_\wa)$ exists. 
If, for some edge~$\arc \in \GA$, the induced \auto edge flow $\Nl[\trav(\cdot)]_{\wa,j}(h_\wa)$ exists for all $j\leq \abs{\wa}$ with $\wa[j] = \arc$, we can define the total induced \auto flow of
$h_\wa$ under $\trav$ on this edge via $\Nl[\trav(\cdot)]_{\wa,\arc}(h_\wa) := \sum_{j:\wa[j]  =\arc}\Nl[\trav(\cdot)]_{\wa,j}(h_\wa)$ and we denote via $\edom{\wa,\arc}$ the maximal subset of $L_+(\hori)$ for which $\Nl[\trav(\cdot)]_{\wa,\arc}(h_\wa)$ exists.
The latter exists for all $\arc \in \GA$ iff $\Nl[\trav(\cdot)]_\wa(h_\wa) := \Nl[\trav(\cdot)]_{\{\wa\}}(h_\wa)$ exists in which case $\Nl[\trav(\cdot)]_{\wa}(h_\wa)= (\Nl[\trav(\cdot)]_{\wa,\arc}(h_\wa))_{\arc\in \GA}$.

Throughout this paper, we require several structural insights derived in \cite{GHS24FD}  
concerning the structure of \auto network loadings.
For the convenience of the reader, we collected the relevant definitions and statements in \Cref{sec:AppuBasedNetworkLoadings}.
Moreover, for the sake of readability, we will denote the \auto network loading operator~$\Nl[\trav(u,\cdot)]$ induced by~$u$ simply by~$\ell^u$ and $\edom{\Routes'}[\trav(u,\cdot)]$ via~$\edom{\Routes'}[u]$.  

We finish this section by giving an example for a well-studied flow propagation model that falls within our model, namely the \emph{Vickrey queuing model}. 
Here, let us remark that this model usually only uses $\R_+$ as planning horizon. However, any flow propagation model that is defined only on a subinterval $J$ of $\hori$ can be extended to a model on the whole $\hori$ by simply extending the flows by $0$ on $\hori \setminus J$ and the travel times constantly on $\hori \setminus J$, i.e.~$\trav(\g,t) = \trav(\g,t_0)$ for $t<  t_0$ and $\trav(\g,t) = \trav(\g,t_f)$ for $t > t_f$ where $t_0$ is the start and $t_f$ the end point of the interval $J$. 
To see this, just observe that if a pair $h,\g$ and travel times $\trav(\g,\cdot)$ defined on $J$ fulfills the equality stated in \eqref{eq: DefEdgeFlow}  for all intervals of the form $(t_0,t),t \in J$, then the above described extension of $h,\g$ and $\trav(\g,\cdot)$ fulfills \eqref{eq: DefEdgeFlow}. 
\begin{example}\label{ex:VickreyModel}
    In   the \emph{Vickrey queuing model}, each edge~$\arc \in \GA$ comes with a free flow travel time~$\tau_\arc > 0$ and a service rate~$\nu_\arc > 0$. The traversal time function~$\trav_\arc$ is then defined as the (unique) solution to a system of equations in terms of the corresponding edges flows $\g \in L_+(\hori)^\GA$:
         \begin{align}\label{eq: Vick}
              \trav_\arc(\g,t) = \tau_\arc + \frac{\q_\arc(\g,t)}{\nu_e} \quad\text{ and }\quad \q_\arc(\g, t) = \int_0^t \g_\arc\di\sigma-\int_{\exit_\arc(\g,\cdot)^{-1}([0,t+\tau_\arc])} \g_\arc \di\sigma
         \end{align}
    together with the conditions that the queue is always non-negative and  that the derivative of $t \mapsto \int_{\exit_\arc(\g,\cdot)^{-1}([0,t])} \g_\arc \di\sigma$ (i.e.\ the outflow rate of edge~$\arc$) is bounded by $\nu_\arc$ almost everywhere.\footnote{Note that, by \cite[Proposition~3.19e)]{GrafThesis}, this definition is equivalent to the more common definition of Vickrey flows in terms of in- and outflow rates.}
    Here, $\q_\arc(\g,t)$ denotes the flow volume in the queue of edge~$\arc$ at time~$t$ and, for~\eqref{eq: Vick}, we extend $\g$ by~$0$ outside of~$\hori$.
\end{example}

\subsection{Toll-Based Dynamic Equilibria and  Implementability}\label{sec:CostBalancingTolls}

In addition to the physical aspects of dynamic flows described in the previous section, there is also a behavioral aspect: We think of a dynamic flow as being made up of individual flow particles that each individually and selfishly want to minimize their overall cost which is the sum of private costs and tolls. A flow wherein every particle achieves this goal 
is a (tolled) dynamic equilibrium.

This is formalized by a function $\wttime:\Nl(\wir)\times \hori \to \R_+^\Routes$ where  $\wttime_{(\hat{\wa},i)}(\g,t)$ denotes the private cost a (hypothetical) particle belonging to commodity~$i$ and entering walk~$\wa = (\hat{\wa},i) \in \Routes_i$ at time~$t$ would experience under the edge flow $\g \in \Nl(\wir)$. 
A prime example for such private costs is the experienced travel time weighted with a commodity specific value of time parameter $\gamma_i>0$, i.e.
\begin{align}\label{eq: PC=WTT}
    \wttime_{(\hat{\wa},i)}(\g,t) := \vot_i\cdot \sum_{j \leq|{\wa}|} \trav_{{\wa}[j]}(\g,\arr_{{\wa},j}(\g,t)) = \vot_i\cdot \big(\arr_{\wa,\abs{\wa}+1}(\g,t)-t\big).
\end{align}
Furthermore, given  a vector-valued  function $\prices: \hori \to \R^\GA_+$ 
 associating with each edge $\arc \in \GA$ and every entry time $t\in \hori$ a nonnegative and finite toll $\prices_\arc(t)$, 
 we define the resulting total toll along walk $\wa=(\hat{\wa},i)$ for the entry time~$t$ under~$\g$ via $\Pf^{{\prices}}_{(\hat{\wa},i)}(\g,t) := \sum_{j \leq|{\wa}|}  \prices_{{\wa}[j]}(\arr_{{\wa},j}(\g,t))$. Note that the edge tolls~ $\prices_\arc(t)$ are anonymous in the sense that any particle entering edge~$\arc$ at time~$t$ must pay the same toll~$\prices_\arc(t)$, regardless of its identity (i.e.\ the \sink[,] VoT or chosen route).

\begin{definition}
    For any  toll function $\prices: \hori \to \R^\GA_+$,  a \emph{$\prices$-dynamic user equilibrium} ($\prices$-DUE) 
is a walk inflow rate function $h\in \wir$ with corresponding edge flow $\g = \Nl(h)$ which satisfies Wardrop's first principle, that is for all  $\wa \in \Routes_i,i\in I$ and almost all $t \in \hori$ it satisfies:
\begin{align*}
    h_\wa(t)>0 \implies \Psi_\wa(\g,t) + \Pf^\prices_\wa(\g,t) \leq \Psi_{\wa'}(\g,t) + \Pf^\prices_{\wa'}(\g,t)\text{ for all } \wa'\in \Routes_i.
\end{align*}
\end{definition}

Note that for $\prices = 0$, the notion of $\prices$-DUE coincides with the classical notion of a DUE, cf.~\cite{Friesz93,ZhuM00}.
\begin{definition}[Implementability]
A vector $u \in L_+(\hori)^\GA$ is \emph{implementable} via tolls $\prices \in \MeasFuncUInt$ and walk inflow rates $h \in \wir$, 
if  $h$ is a  $\prices$-DUE which induces $u$ (i.e.\ $\Nl(h) = u$).
We say that $u$ is \emph{implementable} if there exist tolls $\prices\in \MeasFuncUInt$ and walk inflow rates $h\in \wir$ such that $u$ is implementable via $(\prices,h)$.
 \end{definition}

As a motivating example for such tolls, let us first consider the case of homogeneous populations sharing a single destination, that is, we assume that all 
users share the same \sink[~]$\dest$ and the same private costs. Moreover, we assume that the private costs are separable over the edges:   
  \begin{assumption}\label{ass: PCSep}
  	The private costs are separable over the edges, that is, there exists for all $i\in I$ a  measurable function  $\psi^i : \hori \to \R^\GA_+$ such that for all $(\hat{\wa},i)\in \Routes_i$:
  	\begin{align*}
  		\wttime_{(\hat{\wa},i)}(u,\cdot) :=  \sum_{j \leq|\hat{\wa}|} \psi^i_{\hat{\wa}[j]}(\arr_{\hat{\wa},j}(u,\cdot))  .
  	\end{align*}
  \end{assumption}
  Note that this separability is only required for the private costs induced by $u$ but does not need to hold for other flows. 
  We require the above assumption in order to be able to associate private costs to subwalks of walks contained in $\Routes$. In this regard,  under \Cref{ass: PCSep}, we can extend the function $\wttime$ to all finite walks $\hat{\wa}$.

\begin{theorem}\label{thm:ImplementabilityHomogeneous}
	Consider a network with a single \sink and private costs fulfilling \Cref{ass: PCSep} and being the same for all commodities. Let
	$u \in \Nl(\wir)$ be an edge flow with finite support $[t_0,t_f]$ for some $t_f\in \R$ and induced by $h \in \wir$.  If 
	\begin{thmparts}
		\item there exist $M,\eps > 0$ such that we have $\eps \leq \trav_\arc(u,t) \leq M$ and $\psi_\arc(t) \leq M$  for all $\arc \in E$ and $t \in \hori$, and
		\item $u_\arc \equiv 0$ for all  $\arc \in \edgesFrom{\dest}$,
	\end{thmparts}
	then there exist bounded tolls $\prices$ such that $u$ is implementable via $\prices$ and $h$. 
	If, in addition, all traversal times~$\trav_\arc(u,\cdot)$ and edge costs $\psi_\arc$ are continuous and/or piecewise linear (with finitely many break points), the tolls can be chosen continuous and/or piecewise linear as well and, in the latter case, we can compute them in finite time.
\end{theorem} 
This can be shown by adapting the approach taken by Fotakis and Spirakis in \cite{FotakisS08} for static, atomic single-commodity flows to the dynamic setting and constructing so-called \emph{cost-balancing} tolls. These are tolls that ensure that for any node $v \in \GA$ all \stwalk[v]s consisting only of used edges have the same total cost while all other \stwalk[v] have at least that total cost. This then directly implies that any walk-decomposition of the given edge flow is a $\prices$-equilibrium.

For the static case, in which edges are equipped with edge-load dependent costs, such tolls can be constructed for any acyclic network (or, more generally, for any flow that only uses an acyclic subnetwork) as follows: For every node~$v$ define a node potential $\pi_v$ equal to the longest (\wrt edge costs) \stpath[v] and then set edge-tolls such that the sum of edge-costs and tolls equals the difference between the node potentials, i.e.\ $\prices_{vv'} \coloneqq \pi_v - c_{vv'} - \pi_{v'}$ where $c_{vv'}$ denotes the cost of edge~$vv'$. 

Intuitively, the same approach can be used for dynamic flows by considering the corresponding time-expanded network. Our assumption on the lower bound on the edge traversal times then ensure that this network is acyclic and can be chosen ``finite''. The formal proof of \oref{thm:ImplementabilityHomogeneous} (which can be found in \oref{sec:CostBalancingTollsProof}), however, does not even need this time expanded network but instead directly defines time-dependent node potentials and, based on those, edge tolls. Continuity (and piecewise linearity) of those tolls then follows directly from the continuity (and piecewise linearity) of the edge traversal times and private costs. Note that the assumption that the traversal times are piecewise linear holds, for example, for flows with piecewise constant edge inflow rates within the Vickrey queuing model or the linear edge delay model.
 
The following example shows why this same approach cannot work for heterogeneous population and also that we cannot expect to always find continuous tolls there, even in the special case where private costs represent weighted travel times:
\begin{example}\label{ex:NonContinuityForHeterogeneousUser}
    Consider a network consisting of only two parallel edges $\arc_1$ and $\arc_2$ with fixed (time and flow independent) traversal times of $1$ and $2$, respectively (cf.\ \oref{fig:NonContinuityForHeterogeneousUser}). Moreover, there are two commodities with different values of time $\vot_2 < \vot_1$ and private costs defined by~\eqref{eq: PC=WTT} (i.e., weighted travel times). We now want to induce an edge flow that distributes the flow equally over both edges. 

    \begin{figure}
        \centering
        \BigPicture[1]{%
        \begin{tikzpicture}[declare function={inflow(\x)=.5+(\x<.5)*1 + and(\x>=.5,\x<1.5)*(.5*cos(deg(pi*(x-.5)))+.5);}]
            \node[namedVertex] (s) at (0,0) {$s$};
            \node[namedVertex] (d) at (4,0) {$\dest$};

            \draw[edge, bend left] (s) to node[above]{$\trav_{\arc_1} \equiv 1$} (d);
            \draw[edge, bend right] (s) to node[below]{$\trav_{\arc_2} \equiv 2$} (d);
        
            \begin{scope}[xshift=.4cm,yshift=2.4cm]
                \begin{axis}[xmin=0,xmax=2,ymax=2, ymin=0, samples=500,width=4.5cm,height=2.8cm,anchor=west,
    				 axis y line*=left, axis lines=left, xtick={1},xticklabels={},ytick={1},yticklabels={}]
                    \addplot[blue,  thick,domain=0:2]  (\x,1) node[above,pos=.5]{$u_{\arc_1}$};
    			\end{axis}
            \end{scope}  
            
            \begin{scope}[xshift=.4cm,yshift=-2.4cm]
                \begin{axis}[xmin=0,xmax=2,ymax=2, ymin=0, samples=500,width=4.5cm,height=2.8cm,anchor=west,
    				 axis y line*=left, axis lines=left, xtick={1},xticklabels={},ytick={1},yticklabels={}]
                    \addplot[blue, thick,domain=0:2]  (\x,1) node[above,pos=.5]{$u_{\arc_2}$};
    			\end{axis}
            \end{scope}  

            \begin{scope}[yshift=1.5cm, xshift=6.5cm]
                \begin{axis}[xmin=0,xmax=2,ymax=2, ymin=0, samples=500,width=4.5cm,height=2.8cm,anchor=west,
    				 axis y line*=left, axis lines=left, xtick={1},xticklabels={$t$},ytick={1},yticklabels={}]
                    \addplot[blue, thick,domain=0:2]  {inflow(\x)} node[above,pos=.7]{$\inflow_1$};
    			\end{axis}
            \end{scope} 

            \begin{scope}[yshift=-1.5cm, xshift=6.5cm]
                \begin{axis}[xmin=0,xmax=2,ymax=2, ymin=0, samples=500,width=4.5cm,height=2.8cm,anchor=west,
    				 axis y line*=left, axis lines=left, xtick={1},xticklabels={$t$},ytick={1},yticklabels={}]
                    \addplot[blue, thick,domain=0:2]  {2-inflow(\x)} node[above,pos=.7]{$\inflow_2$};
    			\end{axis}
            \end{scope} 

            \begin{scope}[yshift=1.5cm, xshift=12cm]
                \begin{axis}[xmin=0,xmax=2,ymax=2, ymin=0, samples=500,width=4.5cm,height=2.8cm,anchor=west,
    				 axis y line*=left, axis lines=left, xtick={1},xticklabels={$t$},ytick={.6,1.6},yticklabels={$\vot_2$,$\vot_1$}]
                    \addplot[blue, thick,domain=0:1]  {1.6} ;
                    \addplot[blue, thick,domain=1:2]  {.6} node[above,pos=.5]{$\prices_{\arc_1}$};
    			\end{axis}
            \end{scope} 

            \begin{scope}[yshift=-1.5cm, xshift=12cm]
                \begin{axis}[xmin=0,xmax=2,ymax=2, ymin=0, samples=500,width=4.5cm,height=2.8cm,anchor=west,
    				 axis y line*=left, axis lines=left, xtick={1},xticklabels={$t$},ytick={0}]
                    \addplot[blue, thick,domain=0:2]  {0} node[above,pos=.75]{$\prices_{\arc_2}$};
    			\end{axis}
            \end{scope} 
        \end{tikzpicture}
        }
        \caption{The network discussed in \oref{ex:NonContinuityForHeterogeneousUser} together with an edge flow (left), network inflow rates of two commodities (center) and tolls implementing this flow (right).}
        \label{fig:NonContinuityForHeterogeneousUser}
    \end{figure}
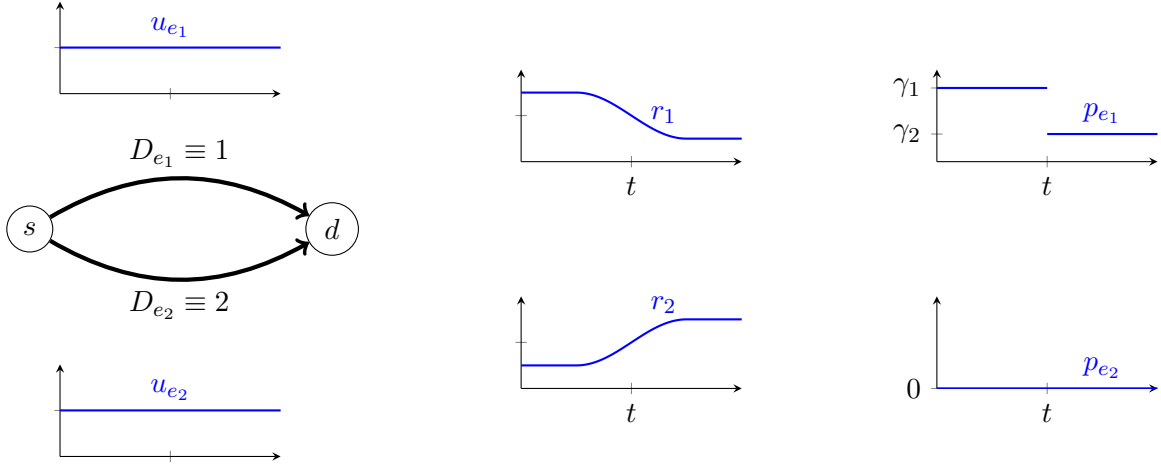
    
 We argue in the following that in this example a) the approach of Fotakis and Spirakis~(\cite{FotakisS08}) is not applicable, b) implementing tolls need to be discontinuous and c) choosing the correct walk-decomposition is essential for the implementability. 
\begin{enumerate}[label = \alph*)]
    \item The central idea of the Fotakis and Spirakis-approach for the homogeneous case is to use cost-balancing tolls. Since here, however, the two commodities have different value of time parameters, it is evident that there cannot be any tolls which are cost balancing for both commodities at the same time.
    
    \item Consider the time~$t$ such that commodity~1's network inflow rate accounts for more than half the combined network inflow rate directly before~$t$ and less than half directly after. Then, any toll implementing the given flow must be discontinuous at time~$t$ (as it must be cost balancing for commodity~1 before time~$t$ and cost-balancing for commodity~2 after time~$t$).  In \Cref{fig:NonContinuityForHeterogeneousUser} (right) we give an example for such tolls.
    
    \item The walk flow that always sends as much of commodity~1's particles as possible into edge~$\arc_1$ implements $u$ by the tolls shown in \oref{fig:NonContinuityForHeterogeneousUser} (right) whereas any walk flow wherein both edges are always used by both commodities cannot  implement $u$ by any tolls (as that would require cost-balancing tolls for both commodities at the same time).  
\end{enumerate}

The last point shows that, in contrast to the homogeneous case under cost-balancing tolls, it is possible in the heterogeneous case to have an edge flow~$u$ and tolls~$\prices$ such that some walk flows~$h$ that induce~$u$ also implement~$u$ (together with~$\prices$) while others do not and can, in fact, never implement $u$ by any tolls. 
This, in particular, implies that when we want to show that a given edge flow is implementable in the heterogeneous setting, we not only have to construct tolls but also a suitable walk-decomposition. Hence, in order to prove our result for heterogeneous populations, we need a more in-depth understanding on how walk flows induce edge flows. This will be the topic of the next \namecref{sec:uBasedNetworkLoadings}.
\end{example}

\section{Structural Results on \Auto Network Loadings}\label{sec:uBasedNetworkLoadings}
 In this section, we will derive several crucial structural properties of \auto network loadings. Here, we  consider the same general framework as in \cite{GHS24FD}, that is the  setting described in the paragraph about \auto network loadings in \Cref{sec: Model}: We are given an arbitrary but fixed (flow-independent) absolutely continuous travel time function $\trav:\hori \to \R_+^\GA, t \mapsto (\trav_\arc(t))_{\arc \in \GA}$ fulfilling FIFO with the corresponding \auto network loading operator $\Nl[\trav(\cdot)]$. This will be the only type of network loading utilized in this entire section and hence we can drop for the sake of readability the superscript and simply write $\ell := \Nl[\trav(\cdot)]$ and $\edom{\Routes'}:= \edom{\Routes'}[\trav(\cdot)]$.

\subsection{A Characterization of Existence}\label{sec:uBasedNetworkLoadings:Exis}

In \ourref{lem: elluExistenceProperties}, we characterized for general travel times $\trav:\hori\to \R_+^\GA$  the existence of the
 \auto network loading on the $j$-th edge of a walk for a corresponding inflow rate $h_\wa \in L_+(\hori)$ via the property:
     \begin{align}\label{eq: nlexists}
        h_\wa = 0 \text{ on } \arr_{\wa,j}^{-1}(\mathfrak T) \text{ for all null sets }\mathfrak T \subseteq \hori.  
    \end{align}
    When it comes to the question of implementability of a vector $u \in \Nl(\wir)$, we will consider the \auto network loading $\Nl[\trav(u,\cdot)]$ induced by $u$. In particular, compared to the general setting of an arbitrary travel times $\trav:\hori\to \R_+^\GA$ and corresponding \auto network loading $\Nl[\trav(\cdot)]$, we have the additional information that $u$ is a flow \wrt $\trav(u,\cdot)$ and, by flow conservation, that $u$ admits \wrt $\trav(u,\cdot)$ edge outflow rates $u^-_\arc$ for all edges $\arc \in \GA\setminus\cup_{i \in I}\edgesTo{\dest_i}$ (cf.~\oref{lem: outflow}). 
    Here, we  call $\eflow_\arc^- \in L_+(\hori)$ the edge outflow rate   corresponding to the edge inflow rate  $\eflow_\arc$ if it fulfills  
\begin{align}\label{eq: outflowCondition}
    \int_{\mathfrak T}\eflow_\arc^- \di\sigma =   \int_{\exit_\arc^{-1}(\mathfrak T)}\eflow_\arc \di\sigma \text{ for all }\mathfrak T \in \mathcal{B}(\hori).
\end{align}
    
    The existence of such edge outflow rates 
allows us to 
give a tighter characterization of
the existence of \auto network loadings (\wrt $u$) for this situation by 
showing that the condition in \eqref{eq: nlexists} is a priori satisfied for all null sets contained in a certain subset of $\hori$. 

Similar to the previous sections, we prove in the following a more general statement in the setting of \aauto network loading $\ell$ \wrt an arbitrary $\trav:\hori\to \R_+^\GA$ (being absolutely continuous and fulfilling FIFO). 
More precisely, we show that for an edge vector $\g \in L_+(\hori)^\GA$ that admits an edge outflow vector $\g^- \in L_+(\hori)^\GA$ \wrt $\trav$, there exists a representative for which the arrival time function at the $j$-th edge of a walk $\wa$ is differentiable with positive derivative for  every point in time 
at which one arrives at all edges $\arc \in \wa$ only if there is inflow under $\g_\arc$.   
This set of time points is denoted by  
$\Dwg$ and the aforementioned insights allows us to show further that the preimage under $\arr_{\wa,j}$ of any null set contained in $\arr_{\wa,j}(\Dwg)$  
  is again a null set, that is, \eqref{eq: nlexists} is a priori satisfied for these null sets.  
  This then ultimately allows us to give a tighter characterization of the existence of  \auto network loadings under the assumption of the existence of such an edge vector  $\g$.

A key ingredient in order to show this is the following  \namecref{lem: absconLus} which states
  that (locally) absolutely continuous, non-decreasing functions (like the arrival time functions) satisfy Lusin's $N^{-1}$ property on the image of all points where the function's derivative exists and is non-zero. This result is a (non-trivial) generalization of a result by Zarecki\u{\i} (cf.\ \cite[Exercise 5.8.54(ii)]{Bogachev2007I}) stating that this property holds for strictly increasing functions (on an interval) whenever the set of points where the function's derivative exists and equals zero is a null set.

Before we come to the proof, let us briefly sketch its main steps: 
It is easily verified that $\func$ is injective on $\Inter_{\func'>0}$ and, thus, admits a strictly increasing inverse function $\infunc$ \wrt this set. 
We then extend this inverse function to  the whole image set $\func(\R)$ 
in a differentiable way and call this function $\tilde{\infunc}$. 
Now, for any set $\mathfrak T$ as described in \Cref{lem: absconLus}, we can express its preimage $ \func^{-1}(\mathfrak T)$ as the image $\tilde{\infunc}(\mathfrak T)$. This, then, allows us to apply \cite[Lemma 7.10]{bruckner1997real}, giving an upper bound on the mass 
 of $\tilde{\infunc}(\mathfrak T)$ by integrating the derivative of $\tilde{\infunc}$ over the initial set $\mathfrak T$. Since the latter is a null set, this integral is zero, showing the desired statement.

\begin{theorem}\label{lem: absconLus}
    Let $\func:\Inter\to \R$ be a (locally) absolutely continuous and  non-decreasing function. Define $\Inter_{\func'>0}:=\{t \in \Inter\mid \func'(t) \text{ exists and } \func'(t) >0\}$.
    Then $\func$ fulfills Lusin's $N^{-1}$ property on $\func(\Inter_{\func'>0})$, that is, for every set $\mathfrak{T}\subseteq \func(\Inter_{\func'>0})$  with $ \sigma(\mathfrak{T})=0$, we have $\sigma(\func^{-1}(\mathfrak{T})) = 0$. 
\end{theorem}
\begin{proof}
We start with two observations which we will use multiple times throughout the proof: 
\begin{claim}\label{claim: t_1<t_2}
    For any  $t_1< t_2 \in \R$ with $A(t_1) = A(t_2)$, we have that 
$t_1,t_2 \notin \R_{A'>0}$.
\end{claim}
\begin{proofClaim}
    By $\func$ being non-decreasing, we have that 
$\func(t) =\func(t_1)$ for all $t \in [t_1,t_2]$. Thus, the left (resp.\ right) derivative at $t_2$ (resp.\ $t_1$) is equal to zero 
which shows the claim. 
\end{proofClaim}

\begin{claim}\label{claim: infsup}
    We have that $\inf \func(\R),\sup \func(\R)  \notin \func(\Inter_{\func'>0})$. 
\end{claim}
\begin{proofClaim}
    In case that the infimum (resp.~supremum) is not attained, it clearly is also not contained in $\func(\Inter_{\func'>0})$. 
    Hence, consider the case where the infimum or supremum is attained. 
    Then, by $\func$ being non-decreasing and continuous, there exist $x_1,x_2\in \R$, respectively, such that 
    $\func^{-1}(\inf \func(\R)) = (-\infty,x_1]$ or $\func^{-1}(\inf \func(\R)) =  [x_2,\infty)$. In particular, 
    by \Cref{claim: t_1<t_2}, it follows that $(-\infty,x_1] \cap \R_{A'>0} = \emptyset$ or  $[x_2,\infty) \cap \R_{A'>0} = \emptyset$ which, in turn, implies that 
    $\inf \func(\R), \sup \func(\R)  \notin \func(\Inter_{\func'>0})$, respectively.
\end{proofClaim}

By \Cref{claim: t_1<t_2}, we get in particular that $\func$ is injective (and strictly increasing) on $\Inter_{\func'>0}$ which, in turn, implies that $\func$  admits a strictly increasing inverse function $\infunc: \func(\Inter_{\func'>0}) \to \Inter$ with $\func(\infunc(y)) = y$  
    for all $y \in \Inter_{\func'>0}$. 
    Now we extend $\infunc$ to a non-decreasing function $\tilde{\infunc}:\func(\Inter)\to \Inter$ by setting 
    \begin{align*}
        \tilde{\infunc}(y) \coloneq \begin{cases}
             \sup\{\infunc(\tilde{y}) \mid  \Tilde{y} \leq y, \Tilde{y} \in \func(\Inter_{\func'>0}) \} &\text{ if } y \in \func(\Inter) \text{ and }y > \inf \func(\Inter), \\
             \tilde{\infunc}(y) := \inf \Inter_{\func'>0} &\text{ if } y = \inf \func(\R).
        \end{cases}
    \end{align*} 
    \begin{claim}
        $\tilde{\infunc}:\func(\Inter)\to \Inter$ is well-defined (real-valued), i.e.~$(-\infty,y]\cap \func(\Inter_{\func'>0}) $ is non-empty for all $y\in \func(\Inter)$  
        with $y > \inf_{t \in \Inter}\func(t)$ and $ \inf \Inter_{\func'>0} > -\infty$ if $y = \inf \func(\R)$ with 
    $\inf \func(\R) \in \func(\Inter)$. 
    
    In particular, we have $\infunc(y) = \tilde{\infunc}(y)$ for all $y \in \func(\Inter_{\func'>0})$. 
    \end{claim}
    \begin{proofClaim} We start by showing well-definedness for the first case. 
    The set $\func(\Inter)$ is an (open/half-open/closed) interval as $\func$ is non-decreasing and continuous. 
    Thus, if there exists $y > \inf_{t \in \Inter}\func(t)$ with $(-\infty,y]\cap \func(\Inter_{\func'>0}) = \emptyset$, then 
    $(-\infty,y]\cap \func(\Inter)$ is a non-empty and non-singleton interval that is contained in $\func(\Inter)\setminus \func(\Inter_{\func'>0})$. 
    
    Hence, in order to show the claim, it is sufficient to show that  $\sigma(\func(\Inter)\setminus \func(\Inter_{\func'>0})) = 0$ holds: In order to do this, we observe that 
     \begin{align}\label{eq: ClaimZarecki}
          \func(\Inter)\setminus \func(\Inter_{\func'>0}) \subseteq \func(\Inter_{\func' = 0}) \cup \func(\{t \in \Inter \mid \func'(t) \text{ does not exist}\})
     \end{align} 
     with  $\Inter_{\func'=0} := \{t \in \Inter \mid \func'(t) \text{ exists and } \func'(t) =0\}$. Here, we used that $\func$ is nondecreasing and, hence, its derivative either does not exist or does exist and is larger or equal to $0$, cf.~\Cref{lem: PropAbsCon:DerNonDec}.  
     For the first set on the right side of \Cref{eq: ClaimZarecki}, we get by \Cref{lem: PropAbsCon:Est} the estimation: 
     \begin{align*}
       \sigma(\func( \Inter_{\func'=0} )) \leq    \int_{\Inter_{\func'=0}} |\func'|\di\sigma  =\int_{\Inter_{\func'=0}} 0\di\sigma =  0.
     \end{align*}
    
     For the second set on the right side of \Cref{eq: ClaimZarecki}, 
     we exploit the fact that, as an absolutely continuous function, $\func$ fulfills Lusin's property (\Cref{lem: PropAbsCon:Lus}) and that the set of points where the derivative of $\func$ does not exist, is a null set (\Cref{lem: PropAbsCon:Der}), that is, 
     $\sigma(\func(\{t \in \R \mid \func'(t) \text{ doesn't exist}\}) ) = 0$. 
    
    Thus, we have shown that $\sigma(\func(\Inter)\setminus \func(\Inter_{\func'>0})) = 0$ which in turn shows the well-definedness in the first case. 

    For the second case of well-definedness (i.e.\ $y=\inf \func(\R) \in \func(\R)$), we 
    observe that $\func^{-1}(y) = (-\infty,x]$ holds for some $x \in \R$ by $\func$ being non-decreasing and continuous. 
    In particular, by \Cref{claim: t_1<t_2}, we have $(-\infty,x] \cap \Inter_{\func'>0} = \emptyset$ and hence 
    $\inf \Inter_{\func'>0} \geq x$. 
    
    Finally,  the claimed equality   of $\infunc$ and $\tilde{\infunc}$ 
    is trivial for all $y >  \inf \func(\R)$ by $\infunc$ being non-decreasing. 
    Thus, the only case remaining is $y = \inf \func(\R)$  and $y \in \func(\Inter_{\func'>0})$ which 
    can not occur due to  \Cref{claim: infsup}. 
        \end{proofClaim}

    For the next claim, recall that $\func(\Inter)$ is an (open/half-open/closed) interval and hence differentiability is defined as usual.  

    \begin{claim}
         $\tilde{\infunc}$ is differentiable at every point in $\func(\Inter_{\func'>0}) \subseteq\func(\Inter)$.
    \end{claim}
\begin{proofClaim}
    To show the claim, let $y \in \func(\Inter_{\func'>0})$  and  $t \in \Inter_{\func'>0}$ the unique point in $ \Inter_{\func'>0}$ with $\func(t) = y$ or, equivalently, $\infunc(y) = t$.
    Since $t \in \Inter_{\func'>0}$, we know that 
    \begin{align}\label{eq: f'>0}
        \lim_{t_n\to t}\frac{\func(t_n)-\func(t)}{t_n-t} = \func'(t) > 0
    \end{align}
    exists. 
    We show in the following that 
     \begin{align*}
        \lim_{n \to \infty}\frac{\tilde{\infunc}(y_n)-\tilde{\infunc}(y)}{y_n-y} =  \frac{1}{\func'(t)} 
    \end{align*}
    for any sequence $(y_n)_{n \in \N}$ with  $y_n\in \func(\Inter)\setminus\{y\}, n \in \N$ and $y_n \to y$. 

    Set $t_n:= \tilde{\infunc}(y_n)$. Since $y_n \to y$ and $y \in \func(\Inter_{\func'>0})$, we can assume \wlg that 
    $y_n \neq \min_{t \in \Inter}\func(t)$ in case that the latter minimum exists. This is due to the fact that in this case,  $y_n = \min_{t \in \Inter}\func(t)$ can only hold for finitely many $n \in \N$ (since $y> \inf\func(\R)$) and, hence, we can remove these members of the sequence. 
    
    We first argue that $\func(t_n) = y_n$. Indeed, by definition (and the above argument) we have $t_n = \tilde{\infunc}(y_n):= \sup\{\infunc(\Tilde{y}) \mid  \Tilde{y} \leq y_n, \Tilde{y} \in \func(\Inter_{\func'>0}) \}$, i.e.~there exists a sequence $(\tilde{y}_n^k)_{k \in \N} \subseteq \func(\Inter_{\func'>0}) \cap (-\infty,y_n]$ with $\infunc(\Tilde{y}_n^k) \mto{k}   \tilde{\infunc}(y_n) = t_n$. By continuity of $\func$ we get $\lim_{k \to \infty} \Tilde{y}_n^k= \lim_{k \to \infty} \func( \infunc(\Tilde{y}_n^k)) =\func(t_n)$. 
    It is, thus, sufficient to argue that $\tilde{y}_n^k \mto{k} {y}_n$. Assume for the sake of a contradiction that $\Tilde{y}_n^k\mto{k}\Bar{y}< y_n$. 
    Since $\sigma(\func(\Inter)\setminus \func(\Inter_{\func'>0})) = 0$ and $[\bar{y},y_n] \subseteq \func(\Inter)$, there exist $y'_1,y'_2 \in  \func(\Inter_{\func'>0})$ and $K \in \N$ with $\Bar{y} < y'_1 < y'_2 < y_n$ and $\Tilde{y}_n^k < y_1'$ for all $k\geq K$. Since $\infunc$ 
    is strictly increasing, we thus get $\infunc(\Tilde{y}_n^k) < \infunc(y'_1), k \geq K$ and, in particular, $\tilde{\infunc}(y_n) = \lim_{k \to \infty}\infunc(\Tilde{y}_n^k) \leq \infunc(y'_1) < \infunc(y'_2)$ in contradiction to the fact that $y'_2 \in  \func(\Inter_{\func'>0}) \cap (-\infty,y_n]$ and the definition of $\tilde{\infunc}$. 

    Hence, we have $\func(t_n) = y_n$ for all $n \in \N$. This implies that also $t_n \to t$ holds:
    To see this, we first argue that $(t_n)$ needs to be a bounded sequence. If there was a subsequence that would diverge to $-\infty$ (resp.\ $+\infty$), this would imply by $\func(t_n) = y_n$ and $y_n\to y$ that 
    $y\in \{\inf \func(\R),\sup\func(\R)\}$ which contradicts \Cref{claim: infsup}.
    Hence, it is enough to show that the existence 
    of  a subsequence  $(n_k)_{k \in \N}$ with $t_{n_k} \to t^*\in \R$ and  $t^*\neq t$ leads to a contradiction. 
    In this case, we have $\infunc(y_{n_k}) \to t^* \neq \infunc(y)$ and, by continuity of $\func$, it follows that 
    \begin{align*}
        y = \lim_{k \to \infty} y_{n_k} =  \lim_{k \to \infty}\func(\infunc(y_{n_k})) = \func( \lim_{k \to \infty}\infunc(y_{n_k})) = \func(t^*).
    \end{align*}
    Since $A(t) = y$, it follows that $\func$ has to be constant on the interval $[t^*,t]$ ($[t,t^*]$ in case $t < t^*$)
    which contradicts   $t \in \Inter_{\func'>0}$ by \Cref{claim: t_1<t_2}. Thus, we must have $t_n \to t$. 
    
    It now follows that 
     \begin{align*}
        \lim_{n \to \infty}\frac{\tilde{\infunc}(y_n)-\tilde{\infunc}(y)}{y_n-y} =  \lim_{n \to \infty}\frac{t_n-t}{\func(t_n)-\func(t)} =  \frac{1}{\func'(t)},
    \end{align*}
    where the last equality follows from the fact that $t_n \to t$ and the existence of the limit due to~\eqref{eq: f'>0}. 
     Thus, we have shown that $\tilde{\infunc}$ is differentiable at every point in $\func(\Inter_{\func'>0})$.
\end{proofClaim}
With this claim at hand, we are now in the position to prove the statement of the \namecref{lem: absconLus}. 
    Let $\mathfrak{T}\subseteq \func(\Inter_{\func'>0})$ with $\sigma(\mathfrak{T}) = 0$.  
    We then have $\func^{-1}(\mathfrak{T}) = {\infunc}(\mathfrak{T}) = \tilde{\infunc}(\mathfrak{T})$ from which the claim of the \namecref{lem: absconLus} follows immediately by an application of \cite[Lemma 7.10]{bruckner1997real} to $\tilde{\infunc}$ and $\mathfrak T$ (note that $\tilde{\infunc}$ is measurable by being non-decreasing) yielding the following estimation:  
     \begin{align*}
         \sigma(\func^{-1}(\mathfrak{T})) =\sigma(\tilde{\infunc}(\mathfrak{T}))  \leq \int_{\mathfrak{T}} |\tilde{\infunc}'|\di\sigma = 0.
     \end{align*}  
 Thus, the \namecref{lem: absconLus} is proven.
\end{proof}

With this property of absolutely continuous functions at hand,  we are now able to derive a tighter characterization of the existence of the \auto network loading. 

\begin{theorem}\label{thm: arrLusin}
 Consider an arbitrary walk $\wa$, $j \in[|\wa|+1]$, $h_{\wa} \in L_+(\hori)$ and an arbitrary subset $\mathfrak T^*\subseteq  \arr_{\wa,j}(\hori_{\arr_{\wa,j}'>0})$  where $\hori_{\arr_{\wa,j}'>0} := \{t \in \hori\mid \arr_{\wa,j}'(t) \text{ exists and }\arr_{\wa,j}'(t)>0\}$. 
   Then $\ell_{\wa,j}(h_{\wa})$ exists if and only if \eqref{eq: nlexists} holds for all null sets $\mathfrak T \subseteq \hori\setminus \mathfrak T^*$.
 \end{theorem}

\begin{proof} 
The ``only  if'' direction is a trivial consequence of
the characterization of \auto network loadings stated in \oref{lem: elluExistenceProperties} 
which requires \eqref{eq: nlexists} to be fulfilled for  all null sets $\mathfrak T \subseteq \hori$ rather than only the null sets contained in $\hori\setminus \mathfrak T^*$.

For the ``if'' direction, 
we argue that \eqref{eq: nlexists}   holds for every null set $\mathfrak T \subseteq \hori$. 
Indeed, we know by \Cref{lem: absconLus}  that the preimage of  $\mathfrak T\cap \mathfrak T^* \subseteq \arr_{\wa,j}(\hori_{\arr_{\wa,j}'>0})$ under $\arr_{\wa,j}$ is a null set. By assumption, also the preimage of  $\mathfrak T\cap \big(\hori \setminus\mathfrak T^*\big)$ is a null set and hence the claim follows.  
\end{proof}

The above \namecref{thm: arrLusin} motivates determining  subsets  $\mathfrak T^*$ of $\arr_{\wa,j}(\hori_{\arr_{\wa,j}'>0})$.   
It turns out that, 
given an arbitrary representative of an edge vector $\g \in L_+(\hori)^\GA$ that admits an edge outflow rate $\g^-$ \wrt the travel times $\trav$, we can determine such a subset via the following set of points in time 
 for which one arrives at all edges $\arc \in \wa$ only if there is inflow under $\g_\arc$:
\begin{align}\label{eq: defD_wa}
    \Dwg:= \{t \in \hori\mid \g_{\wa[j]}(\arr_{\wa,j}(t)) >0, j \leq |\wa|\}. 
\end{align} 

\begin{lemma}\label{lem:  g>0Arr'>0}
Let $\g\in L_+(\R)^\GA$ be an edge vector that admits an edge outflow vector $\g^-$ \wrt $\trav$. 
    For all representatives of $\g$, walks $\wa$ and $j\leq \abs{\wa}+1$, the following holds:
    \begin{align}\label{eq: lem:  g>0Arr'>0}
          \arr_{\wa,j}'(t) \text{ exists and }  \arr_{\wa,j}'(t) >0 \text{ for almost all }t \in \Dwg.
    \end{align}
    Moreover, $\arr_{\wa,j}^{-1}(\mathfrak T)\cap \Dwg$ is a null set for every null set $\mathfrak T \subseteq \arr_{\wa,j}(\Dwg)$. 
    In case \eqref{eq: lem:  g>0Arr'>0} is fulfilled for all $t \in \Dwg$, then  $\arr_{\wa,j}^{-1}(\mathfrak T)$ is already a null set for every null set $\mathfrak T \subseteq \arr_{\wa,j}(\Dwg)$, i.e.~$\arr_{\wa,j}$ fulfills Lusin's $N^{-1}$  property on $\arr_{\wa,j}(\Dwg)$. 
\end{lemma}
\begin{proof}
Consider an arbitrary representative of $\g$. 
We first remark that it is sufficient to show the statement of the lemma for $j = \abs{\wa}+1$. This is true as  the statement of the \namecref{lem:  g>0Arr'>0} for a walk $\wa$ and $j \leq \abs{\wa}$   is implied by 
the statement of the \namecref{lem:  g>0Arr'>0} 
for the walk $\wa_{< j}$ and index $j = \abs{\wa_{< j}}+1$. 
Thus, we can prove the lemma via the following induction over the length of walks: 
    
    \begin{proofbyinduction}
        \inductionclaim For all $k \in \N \cup\{0\}$ and all $\wa$ with $\abs{\wa} = k$, the statement of  the \namecref{lem:  g>0Arr'>0} holds for $j = k + 1$. 
        \basecases{$k \in\{0,1\}$}
        The case of $k = 0$ is trivial as $\arr_{\wa,1} := \id$ by definition. Hence, 
        consider an arbitrary walk $\wa$ with $\abs{\wa} = 1$. Let $\mathfrak T_{=0}$ be the set where   
        $\arr_{\wa,2}'(t)$ exists and is   equal to $0$.  Then \Cref{lem: PropAbsCon:Est} implies that $\sigma(\arr_{\wa,2}(\mathfrak T_{=0})) = 0$. By \Cref{lem: outflow},  the identity $\arr_{\wa,2} = \exit_{\wa[1]}$ and $\g^-$ existing, we hence get that $ 0 = \g_{\wa[1]}(t) = \g_{\wa[1]}(\arr_{\wa,1}(t))$ for almost every  $t \in \arr_{\wa,2}^{-1}(\arr_{\wa,2}(\mathfrak T_{=0})) \supseteq \mathfrak T_{=0}$. Thus, we can deduce that for for almost all  $t \in \Dwg$, the derivative
        $\arr_{\wa,2}'(t)$ does either not exist or fulfills $\arr_{\wa,2}'(t) \neq 0$. This implies the validity of \eqref{eq: lem:  g>0Arr'>0} for $j = 2$ as we have by \Cref{lem: PropAbsCon:Der,lem: PropAbsCon:DerNonDec} that  $\arr_{\wa,2}'(t)$ exists almost everywhere and $\arr_{\wa,2}'(t)\geq 0$ holds
        by $\arr_{\wa,2}$ being non-decreasing. 

        Now consider a null set $\mathfrak T \subseteq \arr_{\wa,2}(\Dwg)$ and denote by $\mathfrak T_{>0}$ the set where $\arr_{\wa,2}'$ exists and is larger than  $0$. In order to show that $\arr_{\wa,2}^{-1}(\mathfrak T) \cap  \Dwg$ is a null set, we argue in the following that 
        both sets $\arr_{\wa,2}^{-1}(\mathfrak T) \cap  (\Dwg \cap \mathfrak T_{>0})$ and   $\arr_{\wa,2}^{-1}(\mathfrak T) \cap (\Dwg\setminus\mathfrak T_{>0})$ are null sets.  
        This is clear for the latter set  as $(\Dwg\setminus\mathfrak T_{>0})$ is a null set by the first part of the \namecref{lem:  g>0Arr'>0}. To show that the former is a null set, we use the following inclusion: 
        \begin{align*}
        \arr_{\wa,2}^{-1}(\mathfrak T) \cap ( \Dwg \cap \mathfrak T_{>0}) \subseteq  \arr_{\wa,2}^{-1}\big(\mathfrak T \cap \arr_{\wa,2}(\mathfrak T_{>0})\big) \cap  \Dwg
    \end{align*} 
    and note that the right side is a null set due to  \Cref{lem: absconLus} and $\mathfrak T$ being a null set.

      		\inductionstep{$k-1 \to k$, $k \geq 2$}  Consider an arbitrary walk $\wa$ with $\abs{\wa} = k$. 
		We have $\arr_{\wa,k+1} = \arr_{\wa_{\geq k},2}\circ \arr_{\wa_{< k},k}$. 
		In the following, we aim to apply the chain rule for differentiation (\cite[Theorem 5.5]{Apostol_1974}) 
        stating that 
        \begin{align*}
            \arr_{\wa,k+1}'(t) =\big(\arr_{\wa_{\geq k},2}\circ \arr_{\wa_{< k},k}\big)'(t) = \big(\arr_{\wa_{\geq k},2}'\circ \arr_{\wa_{< k},k}\big)(t) \cdot \arr_{\wa_{< k},k}'(t) 
        \end{align*}
        whenever $\arr_{\wa_{< k},k}'(t)$ and $\arr_{\wa_{\geq k},2}'(\arr_{\wa_{< k},k}(t))$ exist.       
		By applying the induction hypothesis once for $\wa_{\geq k}$  and once for $\wa_{<k}$, we know that 
		$\arr_{\wa_{\geq k},2}'$  
		exists and is larger than $0$ for almost all $t \in \Dwg[\wa_{\geq k}]$ as well as that 
		$\arr_{\wa_{< k},k}'(t)$ exists and is larger than $0$  for almost all $t \in \Dwg[\wa_{<k}]$ by induction hypothesis. 
		Since  $\arr_{\wa_{< k},k}(\Dwg)\subseteq \Dwg[\wa_{\geq k}]$ and 
		 $\Dwg \subseteq \Dwg[\wa_{<k}]$, the chain rule implies the  validity of \eqref{eq: lem:  g>0Arr'>0} (for $j = k+1$), provided that 	$\arr_{\wa_{< k},k}^{-1}(\mathfrak T^2_{\leq}) \cap \Dwg$ is a null set 
		 where  $\mathfrak T^2_{\leq}$ denotes the null subset of $\arr_{\wa_{< k},k}(\Dwg)$  where $\arr_{\wa_{\geq k},2}'$ 
		 does not exists or is smaller or equal to $0$. 
		 We argue for this in the following: 
		 	By $\Dwg \subseteq \Dwg[\wa_{<k}]$, we have $\mathfrak T^2_{\leq} \subseteq\arr_{\wa_{< k},k}(\Dwg) \subseteq \arr_{\wa_{< k},k}(\Dwg[\wa_{<k}])$ which shows that the induction hypothesis is applicable (for $\wa_{< k}$). That is, we know that 
		 $\arr_{\wa_{< k},k}^{-1}(\mathfrak T^2_{\leq}) \cap \Dwg[\wa_{<k}] \supseteq \arr_{\wa_{< k},k}^{-1}(\mathfrak T^2_{\leq}) \cap \Dwg$ is a null set.

		The second statement of the \namecref{lem:  g>0Arr'>0} follows completely analogous to the base case.
    \end{proofbyinduction}
    Thus the induction is proven and the first  two statements of the lemma follow as explained earlier. 

    The third  statement of the lemma in which \eqref{eq: lem:  g>0Arr'>0} is assumed to be fulfilled for all $t \in \Dwg$ is a direct consequence of   \Cref{lem: absconLus}.  
\end{proof}

For an edge vector $\g$ as in \Cref{lem:  g>0Arr'>0}, the latter \namecref{lem:  g>0Arr'>0} implies in particular 
that  there also exists for any walk $\wa$ and $j\leq \abs{\wa}+1$ a representative of $\g$  
such that \eqref{eq: lem:  g>0Arr'>0} holds for all (not only almost all) $t \in \Dwg$. 
For such a representative, $\arr_{\wa,j}(\Dwg) \subseteq \arr_{\wa,j}(\hori_{\arr_{\wa,j}'>0})$ holds and, hence, 
$\arr_{\wa,j}(\Dwg)$ is a suitable candidate for the set $\mathfrak T^*$ in \Cref{thm: arrLusin}, that is, we get the following \namecref{cor: LuExAltChara}:
\begin{corollary}\label{cor: LuExAltChara}
Let $\wa \in \Routes$, $j \leq \abs{\wa}+1$ and $h_\wa \in L_+(\hori)$ be arbitrary. Consider an edge vector $\g \in L_+(\hori)^\GA$    that admits an edge outflow vector $\g^-$ \wrt $\trav$ and a representative of $\g$ that fulfills \eqref{eq: lem:  g>0Arr'>0} for all (not only almost all) $t \in \Dwg$. 
Then $\ell_{\wa,j}(h_\wa)$ exists if and only if \eqref{eq: nlexists} holds for all null sets $\mathfrak T \subseteq \hori \setminus \arr_{\wa,j}(\Dwg)$. 
\end{corollary}

\subsection{Adjoint Operator of \texorpdfstring{\boldmath$\ell_{\Routes'}$}{lW'}}\label{sec:uBasedNetworkLoadings:Adjoint}

For a walk flow~$h\in \edom{\Routes'}$ with corresponding edge flow~$\g$, we currently have two natural ways of computing the total toll-costs under some given tolls~$\prices$: Either by calculating the toll-costs on an edge-by-edge-basis yielding $\dup{\prices}{\g}$ or on a walk-basis resulting in $\sum_{\wa \in \Routes'}\dup{\Pf^{\prices}_\wa}{{h}_\wa}$ with $\Pf^{{\prices}}_{\wa}(t) := \sum_{j \leq|\wa|}  \prices_{{\wa}[j]}(\arr_{{\wa},j}(t))$. We show in this section that both give the same total toll-cost. In other words, the mapping $\prices \mapsto (\Pf^\prices_\wa)_{\wa \in \Routes'}$ is the adjoint operator to $\ell_{\Routes'}$. 
This insight  will play a key role in the proof of \Cref{thm: mainMSMS} and \ref{thm: mainSingleSink}.

In order to show this, we require two types of standard (Borel-)measures  which we introduce in the following: 
Firstly, for any measurable function $\g:\hori\to \R$, we denote by $\g \cdot \sigma$ the measure on $\mathcal{B}(\hori)$
given by $\g\cdot\sigma(\mathfrak T) := \int_{\mathfrak T} \g\di\sigma,\mathfrak T \in \mathcal{B}(\hori)$. Secondly, for any measurable function $A: \hori\to\hori$ and any measure $\mu$ on $\mathcal{B}(\hori)$, we denote by $\mu\circ A^{-1}$ the image measure of $\mu$ under $A$ which is defined by $\mu\circ A^{-1}(\mathfrak T) := \mu(A^{-1}(\mathfrak T))$. The latter is again a measure on $\mathcal{B}(\hori)$. 
Finally, we also introduce the notation of $\mu\leq \mu'$ for two measures, meaning that $\mu(\mathfrak T) \leq \mu'(\mathfrak T)$ for all $\mathfrak T \in \mathcal{B}(\hori)$.
We refer to~\cite{Bogachev2007I} for a comprehensive overview of measure theory. 
 
\begin{lemma} \label{lem: aggCostsVSwalkCosts} 
Consider  $h \in \edom{\Routes'}$ with  $\g:= \ell_{\Routes'}(h)$ as well as a measurable function $\prices:\hori\to \R^\GA_+$. 
Then,  
the equality  $\dup{\prices}{\g} = \sum_{\wa \in \Routes'}\dup{\Pf^{\prices}_\wa}{{h}_\wa}$ is valid (with possibly both expressions being equal to~$\infty$). 
\end{lemma}
\begin{proof}
Since for all $\mathfrak T \in \mathcal{B}(\hori)$ the equality 
\begin{align*}
    {\g}_\arc \cdot \sigma(\mathfrak T) &= \int_{\mathfrak T}    {\g}_\arc \di\sigma  = \sum_{\wa \in {\Routes'}} \sum_{j:\wa[j] = \arc} \int_{\arr_{\wa,j}^{-1}(\mathfrak T)} h_\wa \di\sigma =   \sum_{\wa \in {\Routes'}} \sum_{j: \wa[j] = \arc}({h}_\wa \cdot \sigma) \circ \arr_{\wa,j}^{-1}(\mathfrak T)
\end{align*}
is valid, the measures 
 ${\g}_\arc \cdot \sigma$ and  $\sum_{\wa \in {\Routes'}} \sum_{j: \wa[j] = \arc}({h}_\wa \cdot \sigma) \circ \arr_{\wa,j}^{-1}$ coincide. 
This allows us to calculate  (explanations follow)
        \begin{align}
      \dup{\prices}{\g}& \eqperdef \sum_{\arc \in \GA}\int_{\hori} \prices_\arc \cdot \g_\arc \di\sigma \nonumber=
         \sum_{\arc \in \GA}\int_{\hori} \prices_\arc   \di (\g_\arc \cdot\sigma) \nonumber\\
        &=   \sum_{\arc \in \GA}  \int_{\hori} \prices_\arc \di\Big(  \sum_{\wa \in {\Routes'}} \sum_{j:\wa[j] = \arc} ({h}_w\cdot\sigma) \circ 
        \arr_{\wa,j}^{-1}\Big) \nonumber\\
        &=  \sum_{\arc \in \GA} \sum_{\wa \in {\Routes'}} \sum_{j:\wa[j] = \arc} \int_{\hori} \prices_\arc \di(({h}_w\cdot\sigma) \circ \arr_{\wa,j}^{-1}) \label{eq: h4}\\
        &=  \sum_{\arc \in \GA} \sum_{\wa \in {\Routes'}} \sum_{j:\wa[j] = \arc}  \int_{\hori} \prices_\arc\circ \arr_{\wa,j} \di({h}_w\cdot\sigma ) \label{eq: h2}\\
        &= \sum_{\arc \in \GA} \sum_{\wa \in {\Routes'}} \sum_{j:\wa[j] = \arc} \int_{\hori}( \prices_\arc \circ \arr_{\wa,j}) \cdot {h}_w \di\sigma \nonumber\\
        &= \sum_{\wa \in {\Routes'}}  \sum_{j=1}^{|\wa|} \int_{\hori} (\prices_{\wa[j]} \circ \arr_{\wa,j}) \cdot {h}_w \di\sigma \label{eq: h3}\\
        &= \sum_{\wa \in {\Routes'}}  \int_{\hori} \big( \sum_{j=1}^{|\wa|} \prices_{\wa[j]} \circ \arr_{\wa,j}\big) \cdot {h}_w \di\sigma \label{eq: h5}\\
        &= \sum_{\wa \in {\Routes'}}  \int_{\hori} \Pf^{{\prices}}_{\wa} \cdot {h}_w \di\sigma \nonumber\\
        &\defpereq \sum_{\wa \in {\Routes'}}\dup{\Pf^{\prices}_\wa}{{h}_\wa}, \nonumber
    \end{align}
    where we used in~\eqref{eq: h4} the insights from~\cite[Problem 9.7]{schilling2017measures} and in~\eqref{eq: h2} the change of variables formula (\cite[Theorem 3.6.1]{Bogachev2007I}).  Note that we are allowed to change the sums in \eqref{eq: h3}, respectively the sum and integral in \eqref{eq: h5} as the integrand is   nonnegative by $\prices$ and $h$ being nonnegative. 
\end{proof}

\subsection{Difference of Induced Flows}
In this section,  we show  that two induced edge flows $\g$ and $\tilde{\g}$ 
with $\g$ dominating $\tilde{\g}$ differ only in flow on zero-cycles and $\dest$-cycles. 
This will play a key role for the combinatorial characterization of a flow $u$ as it 
allows us to show that the difference of $u$ ($=\g$) and any induced edge flow $\tilde{\g}=\ell_{\Routes'}(\tilde{h})$ 
of a feasible solution $\tilde{h}$ for the master problem is decomposable into flow on zero-cycles and $\dest$-cycles. 

In order to state this result, we  require some additional terminology from~\cite{GHS24FD}: We say that a vector $\eflow \in L_+(\hori)^\GA$ has (\auto[)] net outflow rate $\inflow_v \in L(\hori)$ at a node $v \in \GV$, if it fulfills for all $\mathfrak T \in \mathcal{B}(\hori)$:
\begin{align}\label{eq: FlowBalanceDerivative}
    \int_{\mathfrak T} {\inflow}_v\di\sigma =  
\sum_{\arc \in \delta^+(v)} \int_{\mathfrak T} \eflow_\arc \di \sigma -  \sum_{\arc \in \delta^-(v)} \int_{\exit_\arc^{-1}(\mathfrak T)} \eflow_\arc \di \sigma. 
\end{align}
Remark that $\inflow_v$ does not necessary exist in case that a non-trivial amount of particles all arrive during a null set of times at $v$. In case of existence and $\inflow_v = 0$, we say that $\eflow$ fulfills (\auto[)] flow conservation at $v$.

\begin{lemma}\label{lem: DifferenceGeneral}
Let    $h,\tilde{h} \in \edom{\Routes} \cap \wir$  be two walk inflow rates with corresponding aggregated edge flows $\g,\tilde{\g}$ with $\tilde{\g}\leq \g$ and $\g$ admitting an outflow rate $\inflow_\dest^{\g}$ at $\dest$. Then, $\g$ and $\tilde{\g}$ only differ in flow on zero-cycles and $\dest$-cycles, that is, there exist $\hat{h}_c^0 \in L_+(\hori), c \in \SimpCyc$ and $\hat{h}_c^\dest,c^\dest\in \DestCyc$ with  
\begin{align*}
   & {\g}_\arc-\Tilde{\g}_\arc = \sum_{c^\dest \in \DestCyc} \ell_{{c^\dest},\arc}(\hat{h}_{c^\dest}) + \sum_{c \in \SimpCyc}\ell_{c,\arc}(\hat{h}_c^0) = \sum_{{c^\dest} \in \DestCyc} \ell_{{c^\dest},\arc}(\hat{h}_{c^\dest}) +\sum_{c \in \SimpCyc:\arc\in c} \hat{h}_c^0 & \text{ for all } \arc \in \GA \text{ and}\\
&\hat{h}_c^0 (t) >0\implies \trav_{\arc}(t) = 0  & \mathllap{\text{ for all }c \in \SimpCyc, \arc \in c \text{ and almost all }t \in \hori. }
\end{align*}
   	Moreover,   the net outflow rate $\inflow_\dest^{\tilde{\g}}\in L_-(\hori)$ at $\dest$ of $\tilde{\g}$ exists and we have $ \sum_{c^\dest \in \DestCyc}  \hat{h}_{c^\dest} \leq -\inflow_\dest^{\tilde{\g}}$. 
\end{lemma}
\begin{proof}
    By \Cref{lem: flowconW'} and $h,\tilde{h} \in \wir$, the difference $\diffe:= \g-\tilde{\g}$ fulfills flow conservation at every node except $\dest$. Moreover, since $\g$ admits a net outflow rate at $\dest$, so does $\tilde{\g}$ and $\diffe$. Denote the outflow rate of $\diffe$ by $\inflow_\dest^\Delta$ and let $\inflow_\dest^{\Delta+}\in L_+(\hori),\inflow_\dest^{\Delta-}\in L_+(\hori)$ be its positive and negative part, i.e.~$\inflow_\dest^\Delta =  \inflow_\dest^{\Delta+} -\inflow_\dest^{\Delta-} $. 
    Moreover, remark that $\inflow_\dest^\Delta = \inflow_\dest^{\g}- \inflow_\dest^{\tilde{\g}}$ holds and therefore in particular 
    $\inflow_\dest^{\Delta+}\leq\inflow_\dest^\Delta \leq - \inflow_\dest^{\tilde{\g}}$ since $\inflow_\dest^{\g} \leq 0$ (cf.~\Cref{lem: flowconW'}). 

     We introduce a super \sink $\sdest$  and extend   $G=(\GV,\GA)$ to $\tilde G=(\GAS,\GVS)$ via   $\GAS:= \GA \cup \{(\dest,\sdest)\}$ as well as $\GVS:= \GV \cup \{\sdest\}$. The new edge has constant travel time of zero, i.e.\ $\trav_{(\dest,\sdest)}\equiv 0$. We extend $\diffe$ to a vector $\diffe \in L_+(\hori)^\GAS$  via $\diffe_{(\dest,\sdest)}  =   \inflow_\dest^{\Delta-}$. 
     Then, $\diffe$ fulfills flow conservation at all nodes except $\dest,\sdest$. In particular, 
    $\diffe$ is a $\dest$,$\sdest$-flow in the extended network with net outflow rate $ \inflow_\dest^{\Delta+}$ at $\dest$ and net outflow rate $-\inflow_\dest^{\Delta-}$ at $\sdest$. 

    By \oref{thm: FlowDecomp}, $\diffe$ admits a nonnegative flow decomposition $\hat{h}_{\tilde{\wa}},{\tilde{\wa}} \in \tilde{\Routes}_{\dest,\sdest},\hat{h}_c,c \in \tilde{\mathcal{C}}^{\mathrm{simp}}$ with $\hat{h}_c,c \in \tilde{\mathcal{C}}^{\mathrm{simp}}$ being zero-cycle inflow rates and 
    $\sum_{\tilde{\wa} \in \tilde{\Routes}_{\dest,\sdest}}\hat{h}_{\tilde{\wa}} =   \inflow_\dest^{\Delta+} \leq \inflow_\dest^{\tilde{\g}}$ (cf.~above for the last inequality). 
     Here, 
     $\tilde{\Routes}_{\dest,\sdest}$ denotes the set of finite walks from $\dest$ to $\sdest$   and $\tilde{\mathcal{C}}^{\mathrm{simp}}$ the set of simple cycles in the extended network. 
     Now any $\dest$,$\sdest$-walk in the extended network corresponds to a unique $\dest$-cycle in the original network and vice versa while the set of simple cycles in the extended network are exactly the simple cycles in the original network. Thus, setting $\hat{h}_{c^\dest}:=\hat{h}_{\tilde{\wa}^\dest}$ for all $c^\dest$ where $\tilde{\wa}^\dest$ is the corresponding $\dest$,$\sdest$-walk in the extended network yields the desired zero- and $\dest$-cycle inflow rates. 
\end{proof}

\subsection{Largest Common Flow} \label{sec:uBasedNetworkLoadings:LargestCommonFlow}

In this section, we show 
a crucial structural property of \auto network loadings that will allow us in \Cref{sec:ExistenceOptSolutions} to show that the master problem \eqref{opt: Master} fulfills strong duality. 
For an arbitrary countable collection of (not necessarily $\source_i$,$\dest_i$-)walks $\Routes'$, any  $h \in \edom{\Routes'}$  
with corresponding $\g= \ell_{\Routes'}(h)$ and any vector $\eflow \in L_+(\hori)^\GA$, we prove that there 
exists a largest common flow of $h$ and~$\eflow$ 
in the sense that this flow sends as much of the flow of $h$ as possible without exceeding the edge load of~$\eflow$.

\begin{lemma} \label{claim: ZeroDualityGapTildeH}  
Consider an arbitrary countable collection of walks $\Routes'$ containing each individual walk finitely often,  $h \in \edom{\Routes'}$  
with corresponding $\g= \ell_{\Routes'}(h)$ 
and  a vector $\eflow \in L_+(\hori)^\GA$. Then, there exists a walk flow 
    $\Tilde{h} \in \edom{\Routes'}$ with corresponding aggregated edge flow $\Tilde{\g}$ such that the  following properties are fulfilled:
\begin{thmparts}
    \item $\Tilde{\g}(t) \leq \eflow(t)$ for almost all $t \in \hori$. \label[thmpart]{claim: ZeroDualityGapTildeH:leqGhat}
    \item $\tilde{h} (t) \leq h(t)$ for almost all $t \in \hori$. \label[thmpart]{claim: ZeroDualityGapTildeH:leqHn}
    \item For almost all  $t\in \hori$ and all $\wa \in \Routes'$ we have the implication \label[thmpart]{claim: ZeroDualityGapTildeH:Walk}
        \[\tilde{h}_\wa(t)<h_\wa(t) \implies \exists j\leq|\wa| \text{ with } \Tilde{\g}_{\wa[j]}(\arr_{\wa,j}(t)) = \eflow_{\wa[j]}(\arr_{\wa,j}(t)).\] 
\end{thmparts} 
Moreover,  there are suitable representatives of $h,\Tilde{h},\g$ and $\Tilde{\g}$ fulfilling \ref{claim: ZeroDualityGapTildeH:leqGhat}, \ref{claim: ZeroDualityGapTildeH:leqHn} and \ref{claim: ZeroDualityGapTildeH:Walk} for all (not almost all) $t \in \hori$ as well as 
\begin{thmparts}[resume]
    \item For all $\wa \in \Routes'$  we have 
    \begin{align*}
            \bigcup_{\arc \in \GA}\bigcup_{j :\wa[j]= \arc} \arr_{\wa,j}^{-1}(\mathfrak T_\arc) \supseteq \{t \in \hori\mid \tilde{h}_\wa(t)<h_\wa(t) \}
    \end{align*}
 where 
 \begin{align}\label{eq: DefT_e}
     \mathfrak T_\arc:= \{t \in \hori:  \tilde{\g}_\arc(t)=\eflow_\arc(t) < \g_\arc(t) \}, \arc \in \GA.
 \end{align}   \label[thmpart]{claim: ZeroDualityGapTildeH:Union}
\end{thmparts}
\end{lemma}

\begin{proof}   
    We consider the following optimization problem
    \begin{align} 
        \max_{\wflow}\; & \sum_{\wa \in \Routes'}\int_\hori{\wflow_\wa}\di\sigma  \notag \\
        \text{s.t.: }  &\ell_{\Routes'}(\wflow) \leq \eflow \label{eq:ZeroDualityGapTildeH:leqGhat}\\
                    &\wflow \leq h \label{eq:ZeroDualityGapTildeH:leqHn}\\
                    &\wflow \in \edom{\Routes'} \nonumber 
    \end{align}

    It is straight forward to verify that this optimization problem  satisfies the assumption of \oref{thm: ExistenceOptSol}. Hence, it has an optimal solution $\wflow^*$. We set 
     $\Tilde{h}:= \wflow^*$  and choose  representatives of $\Tilde{h},h$ and $\Tilde{\g}$ that fulfill \ref{claim: ZeroDualityGapTildeH:leqGhat} and \ref{claim: ZeroDualityGapTildeH:leqHn} for all $t\in\hori$. Note that this is possible by the feasibility of~$\tilde h$ (i.e.\ constraints \eqref{eq:ZeroDualityGapTildeH:leqGhat} and \eqref{eq:ZeroDualityGapTildeH:leqHn}).
     
     We argue next that these representatives can be further chosen in such way that they also fulfill~\ref{claim: ZeroDualityGapTildeH:Walk} for all $t\in \hori$. Afterwards, we will conclude by showing that with respect to these representatives, we can additionally choose a suitable represent of $\g$ such that \ref{claim: ZeroDualityGapTildeH:Union} is satisfied as well.
    \begin{structuredproof}

        \proofitem{\ref{claim: ZeroDualityGapTildeH:Walk}} 
        If the required property is fulfilled for almost all $t \in \hori$, 
        then setting $h_\wa(t) = 0 = \Tilde{h}_\wa(t),\wa \in \Routes'$ on the null-set 
        where \ref{claim: ZeroDualityGapTildeH:Walk} is not fulfilled will yield representatives that still fulfill \ref{claim: ZeroDualityGapTildeH:leqHn}. 
        Hence, we show that  property~\ref{claim: ZeroDualityGapTildeH:Walk} holds for almost all $t\in\hori$. Assume for contradiction that property~\ref{claim: ZeroDualityGapTildeH:Walk} does not hold for almost all $t \in \hori$. That is, by the countability of $\Routes'$, there exists a walk~$\wa$ such that the set 
            \begin{align*}
                \mathfrak T \coloneqq \Set{t \in \hori | \tilde h_\wa(t)<h_\wa(t) , \tilde \g_{\wa[j]}(\arr_{\wa,j}(t)) < \eflow_{\wa[j]}(\arr_{\wa,j}(t)) \text{ for all } j \leq \abs{\wa}}
            \end{align*}
        has positive measure. Note that this set is indeed measurable due to the measurability of all occurring functions. 
        We now define a flow $\wflow$ by setting
            \[\wflow_{\wa'} \coloneqq \begin{cases}
                \tilde h_{\wa} + 1_{\mathfrak T}, &\text{if } \wa' = \wa \\
                \tilde h_{\wa'},                  &\text{else}
            \end{cases},\]
        where $1_{\mathfrak T}: \hori \to \set{0,1}$ denotes the characteristic function of the set~$\mathfrak T$. Since $\wflow$ only sends flow into walks at times where $h$ does so as well and $\ell_{\Routes'}(h)$ exists, the same holds for $\ell_{\Routes'}(\wflow)$ (cf.\ \oref{lem: elluExistenceProperties}).  
        Let us choose   arbitrary representatives of $\ell_{\wa'}(\wflow_\wa),\wa' \in \Routes'$.  
        We now observe that there must be some constant~$C \geq 1$ and a measurable subset $\mathfrak T^* \subseteq \mathfrak T$ of positive measure such that we have
            \begin{align}\label{eq:ZeroDualityGapTildeH:ChoiceOfTast}
                \frac{\ell_{\wa,j}(\wflow_\wa)(\arr_{\wa,j}(t))-\ell_{\wa,j}(\tilde h_\wa)(\arr_{\wa,j}(t))}{\eflow_{\wa[j]}(\arr_{\wa,j}(t))-\Tilde{\g}_{\wa[j]}(\arr_{\wa,j}(t))} \leq C \text{ for all } t \in \mathfrak T^* \text{ and } j \leq \abs{\wa}
            \end{align}
        as well as $h_\wa(t) \geq \tilde h_\wa(t) + \frac{1}{C}$ for all $t \in \mathfrak T^*$. We then define a walk inflow $\bar\wflow$ by
            \[\bar\wflow_{\wa'} \coloneqq \begin{cases}
                \tilde h_{\wa} + \tfrac{1}{C\cdot\abs{\wa}} \cdot 1_{\mathfrak T^*}, &\text{if } \wa' = \wa \\
                \tilde h_{\wa'},                  &\text{else}
            \end{cases}.\]
        The choice of $C$ then immediately implies $\bar\wflow \leq h$, i.e.\ that $\bar\wflow$ satisfies~\eqref{eq:ZeroDualityGapTildeH:leqHn}. 
        In particular, $\ell_{\wa}(\Bar{\wflow}_\wa),\wa \in \Routes'$ exist by $\ell_{\wa}(\Tilde{h}_\wa),\wa\in \Routes' $ existing and \oref{lem: elluExistenceProperties}.
         Furthermore, it is clear that $\sum_{\wa \in {\Routes'}}\ell_\wa(\bar{\wflow}_\wa)$ exists since this sum differs from $\sum_{\wa \in {\Routes'}}\ell_\wa(\Tilde{h}_\wa)$ in only a single  summand. 
        Hence, $\bar{\wflow} \in \edom{{\Routes'}}$. 
        In order to show that $\Bar{\wflow}$ also satisfies~\eqref{eq:ZeroDualityGapTildeH:leqGhat} we calculate for any $j \leq \abs{\wa}$:
            \begin{align*}
                \ell_{\wa,j}(\bar\wflow_\wa) 
                    &= \ell_{\wa,j}(\tilde h_\wa + \tfrac{1}{C\abs{\wa}}1_{\mathfrak T^*}) \symoverset{1}{=} \ell_{\wa,j}(\tilde h_\wa) + \tfrac{1}{C\abs{\wa}}\cdot\ell_{\wa,j}(1_{\mathfrak T^*}) \\
                    &\symoverset{3}{=} \ell_{\wa,j}(\tilde h_\wa) + \tfrac{1}{C\abs{\wa}}\cdot\ell_{\wa,j}(1_{\mathfrak T^*}\cdot 1_{\mathfrak T}) \\
                    &\symoverset{2}{=} \ell_{\wa,j}(\tilde h_\wa) + \tfrac{1}{C\abs{\wa}}\cdot 1_{\arr_{\wa,j}(\mathfrak T^*)} \cdot \ell_{\wa,j}(1_{\mathfrak T}) \\
                    &= \ell_{\wa,j}(\tilde h_\wa) + \tfrac{1}{C\abs{\wa}}\cdot 1_{\arr_{\wa,j}(\mathfrak T^*)} \cdot \ell_{\wa,j}(\wflow_\wa - \tilde h_\wa) \\
                    &\symoverset{1}{=} \ell_{\wa,j}(\tilde h_\wa) + \tfrac{1}{C\abs{\wa}}\cdot 1_{\arr_{\wa,j}(\mathfrak T^*)} \cdot \big(\ell_{\wa,j}(\wflow_\wa) - \ell_{\wa,j}(\tilde h_\wa)\big) \\
                    &\overset{\eqref{eq:ZeroDualityGapTildeH:ChoiceOfTast}}{\leq} \ell_{\wa,j}(\tilde h_\wa) + \tfrac{1}{C\abs{\wa}}\cdot 1_{\arr_{\wa,j}(\mathfrak T^*)} \cdot C\cdot \big(\eflow_{\wa[j]} - \tilde{\g}_{\wa[j]}\big) \\
                    &\leq \ell_{\wa,j}(\tilde h_\wa) + \tfrac{1}{\abs{\wa}}\cdot \big(\eflow_{\wa[j]} - \tilde{\g}_{\wa[j]}\big),
            \end{align*}
        where the equalities indicated by~\refsym{1} are due to the linearity of~$\ell_{\wa,j}$ (\oref{lem: elluExistenceProperties:Linearity}), 
        the one with~\refsym{3} is due to $\mathfrak T \subseteq \mathfrak T^*$, 
        the equality with~\refsym{2} holds since $\ell_{\wa,j}$ commutes with indicator functions (\oref{lem: elluindi}) and the inequality in the last line holds because we have $\eflow_{\wa[j]} \geq \tilde{\g}_{\wa[j]}$ (since $\tilde h$ satisfies~\eqref{eq:ZeroDualityGapTildeH:leqGhat}). 

        Let $\arc\in \GA$ be arbitrary. The above inequality is used to get the estimation~\refsym{4} in the following:
            \begin{align*}
                \sum_{\wa' \in {\Routes'}}\ell_{\wa',\arc}(\bar\wflow_{\wa'}) 
                    &= \ell_{\wa,\arc}(\bar\wflow_\wa) + \sum_{\wa' \in {\Routes'}\setminus\{\wa\}}\ell_{\wa',\arc}(\bar\wflow_{\wa'}) = \sum_{j:\wa[j]=e}\ell_{\wa,j}(\bar\wflow_\wa) + \sum_{\wa' \in {\Routes'}\setminus\{\wa\}}\ell_{\wa',\arc}(\tilde h_{\wa'}) \\
                    &\symoverset{4}{\leq} \sum_{j:\wa[j]=\arc}\Big( \ell_{\wa,j}(\tilde h_\wa) + \tfrac{1}{\abs{\wa}}\cdot \big(\eflow_{\wa[j]} - \tilde{\g}_{\wa[j]}\big)\Big) + \sum_{\wa' \in {\Routes'}\setminus\{\wa\}}\ell_{\wa',\arc}(\tilde h_{\wa'}) \\
                    &\symoverset{1}{\leq} \ell_{\wa,\arc}(\tilde h_\wa) + \big(\eflow_\arc - \tilde{g}_\arc\big) + \sum_{\wa' \in {\Routes'}\setminus\{\wa\}}\ell_{\wa',\arc}(\tilde h_{\wa'}) \\
                    &= \eflow_\arc - \tilde{g}_\arc + \sum_{\wa' \in {\Routes'}}\ell_{\wa',\arc}(\tilde h_{\wa'}) 
                        = \eflow_\arc,
            \end{align*}
        where we use for the inequality indicated with \refsym{1} that $\eflow_{\wa[j]} - \tilde{\g}_{\wa[j]} = \eflow_\arc - \sum_{\wa \in {\Routes'}}\ell_{\wa,\arc}(\tilde h_\wa) \geq 0$ because $\tilde h$ satisfies~\eqref{eq:ZeroDualityGapTildeH:leqGhat}. Hence, we have shown that $\bar\wflow$ satisfies \eqref{eq:ZeroDualityGapTildeH:leqGhat}. 
        
        At the same time, $\bar\wflow$ clearly has a larger objective value than~$\tilde h$, which is a contradiction to~$\tilde h$ being an optimal solution. Hence, $h$ must have already satisfied property~\ref{claim: ZeroDualityGapTildeH:Walk} for almost all~$t \in \hori$. Setting $h_{\wa'}(t) = \tilde h_{\wa'}(t) = 0$ for all other $t$ and $\wa' \in \Routes'$ then gives us the desired representatives fully satisfying property~\ref{claim: ZeroDualityGapTildeH:Walk}.
        
        \proofitem{\ref{claim: ZeroDualityGapTildeH:Union}} This now follows from~\ref{claim: ZeroDualityGapTildeH:Walk} as follows: Since $h-\tilde h$ is nonnegative and positive only where $h$ is positive, we can apply \oref{lem: Relations:h>0u>0:Pointwise} to this difference to see that for all $\wa \in {\Routes'},j\leq \abs{\wa}$ and almost all $t\in \hori$ we have
            \[\tilde h_\wa(t)<h_\wa(t)  \implies \tilde\g_{\wa[j]}(\arr_{\wa,j}(t))<\g_{\wa[j]}(\arr_{\wa,j}(t)) .\]
            Hence, by ${\Routes'}$ being countable,  every walk being finite and $\arr_{\wa,j}$ fulfilling Lusin's property (by being absolutely continuous, cf.~\Cref{lem: PropAbsCon:Lus}), we can choose  a suitable representative of $\g$ such that the latter property is fulfilled for all $t \in \hori$.
        For this representative of $\g$ and the prior chosen representatives of $h,\Tilde{h}$ and $\Tilde{\g}$, we have now by the definition in~\eqref{eq: DefT_e} that
            \[\mathfrak D_{\wa,j} \coloneqq \set{t \in \hori | \tilde h_\wa(t)<h_\wa(t) , \tilde\g_{\wa[j]}(\arr_{\wa,j}(t)) = \eflow_{\wa[j]}(\arr_{\wa,j}(t))} \subseteq \arr_{\wa,j}^{-1}(\mathfrak T_{\wa[j]})\]
        for all $j \leq \abs{\wa}, \wa \in {\Routes'}$. According to~\ref{claim: ZeroDualityGapTildeH:Walk}, the union of all those $\mathfrak D_{\wa,j},j\leq \abs{\wa},\wa \in {\Routes'}$ covers $\set{t \in \hori | \tilde h_\wa(t)<h_\wa(t) }$ and hence the required inclusion in~\ref{claim: ZeroDualityGapTildeH:Union} holds. 
        \qedhere 
    \end{structuredproof}
\end{proof}

\section{Implementability} \label{sec:Implementability} 
 In this section, we come back to the question of implementability of a vector $u \in \Nl(\wir)$. 
Remark that vectors $u \in L_+(\hori)^\GA \setminus \Nl(\wir)$ are trivially not implementable since those edge flows cannot be induced by any walk-inflow. 
Under the assumption of the master problem \eqref{opt: Master} admitting strong duality, we 
 provide a duality-based characterization of implementability in the first subsection and additionally a combinatorial characterization in the second subsection for single-\sink networks.  
 Finally, in the third subsection, we provide sufficient conditions for a class of optimization problems -- containing the master problem -- to admit strong duality. 
We require the following additional \Cref{ass: u} on $u$ and the flow model to hold throughout this section. 
Here, we denote from now on the \auto network loading corresponding to the travel times $\trav(u,\cdot)$ induced by $u$ via $\ell^u$ and write for brevity also  $\ell^u(h)$ instead of $\ell^u_{\Routes}(h)$ for any $h \in \wir$. Moreover, we denote via $\edom{\Routes}[u]$ the set of walk flows $h \in L_+(\hori)^\Routes$ admitting  \aauto network loading $\ell^u(h)$ \wrt $\trav(u,\cdot)$.

\begin{assumption}\label{ass: u} 
\hfill
\begin{thmparts}
    \item \label[thmpart]{ass: u: outflow}
For all $i \in I$ and  $\arc \in \delta^-(\dest_i)$, the edge outflow rate $u^-_\arc$ exists, that is, the function $\hori\to \R, t \mapsto \int_{\exit_\arc(u,\cdot)^{-1}(\startint t])} u_\arc \di\sigma$ is absolutely continuous.

 \item \label[thmpart]{ass: u: NLQuasiUniqueness} For any $h \in \wir$ we have $\ell^u(h) = u \implies \Nl(h) = u$. 

\end{thmparts}
\end{assumption}

 \Cref{ass: u: outflow} is rather technical and  requires that the outflow under $u$ from edges entering~a destination $\dest_i$ is describable via a function in $L(\hori)$. 
 By \oref{lem: outflow}, this assumption is equivalent to requiring $u_\arc = 0$ on $\exit_\arc(u,\cdot)^{-1}(\mathfrak T)$ for all null sets $\mathfrak{T} \subseteq \hori$. 
Remark that we do not need this assumption for nodes other than~the destinations $\dest_i,i\in I$ since there, by flow conservation,  any edge outflow is the edge inflow of other edges (and, hence, already described by functions in $L(\hori)$, cf.~again \oref{lem: outflow: FlowCon}).  

\Cref{ass: u: NLQuasiUniqueness} states that whenever a walk flow induces~$u$ under~the travel times induced by $u$, then $u$ is also the actual network loading of~$h$.  Equivalently, this requires that for any $h \in \wir$ with  $u$ being a solution of \eqref{eq: DefEdgeFlow}, we also have $\Nl(h) = u$. 
A sufficient condition for this to be fulfilled is the property that there is only one way to define the network loading, i.e.\ if for any $h\in \wir$ there exist a unique $\g \in L_+(\hori)^\GA$ such that  \eqref{eq: DefEdgeFlow} holds. 
Remark that Meunier and Wagner~\cite{MeunierW10} showed that this is the case 
under several assumptions regarding the travel time function $\trav$ which do not conflict with our assumptions here. 
In particular, Meunier and Wagner show that  the Vickrey queuing and linear edge delay model satisfy them. 
Since the reverse implication of \Cref{ass: u: NLQuasiUniqueness} is true in general, we will from now on simply say that $h$ induces $u$
if $\ell^u(h) = u$ ($\Leftrightarrow \Nl(h) = u$) holds.

\subsection{Duality-Based Characterization}\label{sec:CharImplementability}

In this section, under the assumption of the master problem \eqref{opt: Master} admitting   strong duality, we characterize  implementable edge flows $u \in \Nl(\wir)$ as those that can be induced by an optimal solution to their corresponding master problem \eqref{opt: Master}.   
Here, we say that \eqref{opt: Master} 
admits strong (Lagrangian) duality \wrt $\MeasFuncUInt$ ($L^\infty_+(\hori)^\GA$), if 
there exists an optimal solution $h^*$ for~\eqref{opt: Master} and $\prices^* \in \MeasFuncUInt$ ($\prices^* \in L^\infty_+(\hori)^\GA$) such that 
        \begin{align*}
        \inf_{h\in\edom{\Routes}[u]\cap \wir} \bigl(\dup{\wttime(u,\cdot)}{h}  + \dup{\prices^*}{\ell^u(h)} -  \dup{\prices^*}{u}   \bigr) =\dup{\wttime(u,\cdot)}{h^*} .
    \end{align*}  
    Note that any representative of a $\prices \in  L^\infty_+(\hori)^\GA$ is in particular an element of $\MeasFuncUInt$, that is, if \eqref{opt: Master} 
    admits strong (Lagrangian) duality \wrt $L^\infty_+(\hori)^\GA$, then in particular \wrt $\MeasFuncUInt$ as well. 
    
Under the assumption of \eqref{opt: Master} admitting strong duality \wrt $L_+^\infty(\hori)^\GA$, we  provide sufficient conditions (\Cref{thm: SuffConMSMS}) for a $h^* \in \wir$ to implement $u$ via \emph{bounded} tolls $\prices^* \in L_+^\infty(\hori)^\GA$ in the multi-source, multi-\sink case.\footnote{Here, we mean by $u$ being implementable via $\prices^* \in L_+^\infty(\hori)^\GA$ that any representative of $\prices^*$ implements $u$.}
Under the slightly less restrictive assumption of \eqref{opt: Master} admitting strong duality \wrt $\MeasFuncUInt$, 
 we  provide a sufficient condition  (\Cref{thm: SuffConMSSS}) for a $h^* \in \wir$ to implement $u$ via  tolls $\prices^* \in \MeasFuncUInt$ in the multi-source, single-\sink case. 
 These sufficient conditions are then used to obtain a characterization of implementability for the multi-source, multi-\sink case (\Cref{thm: mainMSMS}), respectively multi-source, single-\sink case (\Cref{thm: mainMSSS}). 
The necessity part of the aforementioned characterization of implementability is rather straight forward and is shown in the following \Cref{lem: NecessImpl}.  Note that it suffices to show this for $\prices^* \in \MeasFuncUInt$ as any representative of 
a  $\prices^* \in L_+^\infty(\hori)^\GA$ is contained in $\MeasFuncUInt$. 
 \begin{lemma}\label{lem: NecessImpl}
  	In a multi-source, multi-\sink network, if $u$ is implementable via $h^* \in \wir$ and $\prices^*\in \MeasFuncUInt$, then $h^*$ is an optimal solution for \eqref{opt: Master} with tight inequality \eqref{ineq: Master} 
  	and $(h^*,\prices^*)$ admit zero duality gap.  
\end{lemma}
\begin{proof}
     Let $u$ be implementable via tolls $\prices^* \in \MeasFuncUInt$ and $h^* \in \wir$, i.e.\ $h^*$ induces $u$ (under $\Nl$) and is a $\prices^*$-DUE. 
 
 It is clear that $h^*$ also induces $u$  under $\ell^u$, i.e.~$h^* \in \edom{\Routes}[u]$ and $\ell^u(h^*) = u$. In particular, \eqref{ineq: Master} is tight for $h^*$.  
 In order to show   that $h^*$ is optimal for \eqref{opt: Master}, 
 we require the following claim: 
 \begin{claim}
     For all $h \in \edom{\Routes}[u]\cap \wir$, the following inequality holds (with the right hand side possibly being equal to $\infty$): 
     \begin{align}\label{eq: ineqOpth}
         \sum_{\wa \in \Routes} \dup{\wttime_\wa(u,\cdot)}{h^*_\wa} + \dup{\prices^*}{u}\leq \sum_{\wa \in \Routes} \dup{\wttime_\wa(u,\cdot)}{{h}_\wa} + \dup{\prices^*}{{\ell^u(h)}}.
     \end{align}
 \end{claim}
 \begin{proofClaim}
 We start by observing that  \Cref{lem: aggCostsVSwalkCosts} implies the inequality is equivalent to the following:
    \begin{align*}
         \sum_{\wa \in \Routes} \dup{\wttime_\wa(u,\cdot) + \Pf^{\prices^*}_\wa(u,\cdot)}{h^*_\wa}  \leq \sum_{\wa \in \Routes} \dup{\wttime_\wa(u,\cdot)+\Pf^{\prices^*}_\wa(u,\cdot)}{{h}_\wa}.
     \end{align*}
 Assume for the sake of a contradiction that there 
 exists $h \in \edom{\Routes}[u] \cap \wir$ for which this inequality does not hold. 
Then, there has to exist an $i \in I$ with
    \begin{align*}
         \sum_{\wa \in \Routes_i} \dup{\wttime_\wa(u,\cdot) + \Pf^{\prices^*}_\wa(u,\cdot)}{h^*_\wa} > \sum_{\wa \in \Routes_i} \dup{\wttime_\wa(u,\cdot)+\Pf^{\prices^*}_\wa(u,\cdot)}{{h}_\wa}.
     \end{align*}
     Since, furthermore, $\sum_{\wa \in \Routes_i}h^*_\wa = \sum_{\wa \in \Routes_i}{h}_\wa$, there 
     have to exist two walks $\wa,\tilde{\wa} \in \Routes_i$ and a set $\mathfrak T \in \mathcal{B}(\hori)$ with $\sigma(\mathfrak T)>0$ such that $\wttime_\wa(u,t) + \Pf^{\prices^*}_\wa(u,t)>\wttime_{\tilde{\wa}}(u,t) + \Pf^{\prices^*}_{\tilde{\wa}}(u,t)$ and $h^*_\wa(t) >  {h}_\wa(t)\geq 0$ for a.e.\ $t \in \mathfrak T$. 
This, however, contradicts the fact that $h^*$ is a $\prices^*$-DUE. 
 \end{proofClaim}
By
subtracting the constant term $\dup{\prices^*}{u}$ from both sides in \eqref{eq: ineqOpth}, we get
\begin{align*} 
         \sum_{\wa \in \Routes} \dup{\wttime_\wa(u,\cdot)}{h^*_\wa} \leq \sum_{\wa \in \Routes} \dup{\wttime_\wa(u,\cdot)}{{h}_\wa} + \dup{\prices^*}{\ell^u({h})-u}  \text{ for all } h \in \edom{\Routes}[u]\cap \wir,
     \end{align*}
     and hence 
\begin{align*} 
         \sum_{\wa \in \Routes} \dup{\wttime_\wa(u,\cdot)}{h^*_\wa} \leq \inf_{h \in \edom{\Routes}[u]\cap \wir } \sum_{\wa \in \Routes} \dup{\wttime_\wa(u,\cdot)}{{h}_\wa} + \dup{\prices^*}{\ell^u({h})-u},
     \end{align*}
    showing that \eqref{opt: Master} fulfills strong duality \wrt $(h^*,\prices^*)$.  
    In particular, $h^*$ is  optimal for \eqref{opt: Master} and hence the proof is finished. 
\end{proof}

 In the following, we provide sufficient conditions for a $h^* \in \wir$ to implement $u$ via \emph{bounded} tolls in the multi-source, multi-\sink case. 
  \begin{lemma}\label{thm: SuffConMSMS}
  	Consider a multi-source, multi-\sink network with \Cref{ass: u} being fulfilled and \eqref{opt: Master} admitting strong duality \wrt $L_+^\infty(\hori)^\GA$. Then $u$ is  implementable via $h^* \in \wir$ and \emph{bounded} tolls, if $h^* \in \wir$ fulfills the following conditions: 
  	\begin{thmparts}
  		\item $h^*$ has a finite walk support, i.e.~there exists a finite set of walks ${\Routes}^{h^*}\subseteq \Routes$ such that $h^*_\wa = 0,\wa \notin {\Routes}^{h^*}$, \label[thmpart]{thm: SuffConMSMS: FiniteWalkSupp}  
  		\item The experienced private costs under $h^*$ are bounded, i.e.~there exists $C^{h^*} \in \R$ such that $h^*_\wa(t)>0 \Rightarrow \wttime_\wa(u,t)\leq C^{h^*}$ holds for almost all $t \in \hori$ and all $\wa \in \Routes$.\label[thmpart]{thm: SuffConMSMS: BoundedExpTravel}    
  		\item $h^*$ is an optimal solution for \eqref{opt: Master} with tight inequality \eqref{ineq: Master}. \label[thmpart]{thm: SuffConMSMS: OptimalInducingU} 
  	\end{thmparts}
  \end{lemma}
  The above in combination with \Cref{lem: NecessImpl}  immediately  yields the following characterization: 
  \begin{theorem}\label{thm: mainMSMS}
  	In  a multi-source, multi-\sink network with \Cref{ass: u} being fulfilled, \eqref{opt: Master} admitting strong duality \wrt $L_+^\infty(\hori)^\GA$ and \Cref{thm: SuffConMSMS: FiniteWalkSupp,thm: SuffConMSMS: BoundedExpTravel} being valid for any $h^*\in \wir$ inducing $u$, the following statements are equivalent: 
  	\begin{thmparts}
  		\item $u$ is implementable via bounded tolls. 
  		\item There exists an  optimal solution for \eqref{opt: Master} with tight inequality \eqref{ineq: Master}.  
  	\end{thmparts}
  \end{theorem}
  We remark that for finitely supported $u$, certain natural models satisfy the requirements stated in 
 \Cref{thm: SuffConMSMS: FiniteWalkSupp,thm: SuffConMSMS: BoundedExpTravel} for any $h^*$ inducing $u$ automatically: In case that  
  the travel times induced by $u$ are lower bounded by a constant strictly larger than $0$, \Cref{thm: SuffConMSMS: FiniteWalkSupp} holds.
  Similarly, if the private costs represent the weighted travel times (i.e., are defined by \eqref{eq: PC=WTT}), then $h^*$ also has bounded experienced travel time, i.e.~fulfills \Cref{thm: SuffConMSMS: BoundedExpTravel}. 
  Both of these statements are shown in \Cref{lem: AssUZeroDG}. 

  Before we come to the proof of \Cref{thm: SuffConMSMS}, let us give a brief proof sketch: 
  We start by exploiting strong duality of~\eqref{opt: Master} (\wrt $L_+^\infty(\hori)^\GA$) to get dual tolls which, 
  roughly speaking, guarantee that for all commodities $i\in I$ and almost all points in time, 
  all \stwalki{s}  in which flow can be send under $\trav(u,\cdot)$ without resulting in an undefined network loading do not have smaller total costs than the walks utilized under $h$. In order to ensure that this holds for \emph{all} \stwalki s, we increase the tolls on all edges not used under~$u$ in such a way that all walks using such an edge have total cost at least as high as any used walk. 
  Here, we exploit that, by the fulfillment of \Cref{thm: SuffConMSMS: FiniteWalkSupp,thm: SuffConMSMS: BoundedExpTravel}, the total costs under $\prices^*$ among all particles   are upper bounded. Hence, adding this upper bound to the tolls on all edges not used under $u$ yields suitable tolls.

  \begin{proof}[Proof of \Cref{thm: SuffConMSMS}]
  	Let $(\hat{\prices},h^*)$ admit  zero duality gap with tight inequality \eqref{ineq: Master} where  $\hat{\prices} \in L_+^\infty(\hori)^\GA$ exists by the assumption of \eqref{opt: Master} admitting  strong duality \wrt $L_+^\infty(\hori)^\GA$. That is, we have  
  	\begin{align}\label{eq: ZeroDGhat2}
  		\inf_{{h} \in \edom{\Routes}[u]\cap \wir}\sum_{\wa \in \Routes} \dup{\wttime_\wa(u,\cdot)}{{h}_\wa} + \dup{{\hat{\prices}}}{ \ell^u(h) -u} =  \sum_{\wa \in \Routes} \dup{\wttime_\wa(u,\cdot)}{h^*_\wa} .  
  	\end{align}  
  	By adding to both sides the constant term $\dup{{\hat{\prices}}}{u}\in \R$ and using  \Cref{lem: aggCostsVSwalkCosts}, this is equivalent to 
  	\begin{align}\label{eq: ZeroDGhatwalk2}
  		\sum_{\wa \in \Routes} \dup{\wttime_\wa(u,\cdot) + \Pf^{{\hat{\prices}}}_\wa(u,\cdot)}{ {h}_\wa} \geq \sum_{\wa \in \Routes} \dup{\wttime_\wa(u,\cdot) + \Pf^{{\hat{\prices}}}_\wa(u,\cdot)}{h^*_\wa} && \forall {h} \in \edom{\Routes}[u]\cap \wir.  
  	\end{align}  

  	By assumption,   the set of walks  $\Routes^{h^*}$ on which $h^*$ is supported is finite. 
  	Thus, there exists an $m \in \N$ large enough such that $\abs{\wa} \leq m$ for all $\wa \in \Routes^{h^*}$. In particular, we get the bound $\Pf^{\hat{\prices}}_\wa(u,t) \leq m \cdot \norm{\hat{\prices}}_\infty$ for all $\wa \in \Routes^{h^*}$ and $t \in \hori$. 
  	This together with the constant~$C^{h^*}$ from the edment of \Cref{thm: SuffConMSMS: BoundedExpTravel} implies that for $\tilde{C}:= m \cdot \norm{\hat{\prices}}_\infty + C^{h^*}$  the following implication holds:   
  	\begin{align}\label{eq: h^*>0ImpliesBoundedCosts}
  		h^*_\wa(t) > 0 \implies \wttime_\wa(u,t) + \Pf^{{\hat{\prices}}}_\wa(u,t) \leq \tilde{C} \text{ for almost all } t \in \hori \text{ and all } \wa \in \Routes.
  	\end{align}
  	For the remainder of the proof, choose an arbitrary representative of $u$ 
  	that fulfills 
  		\begin{align}\label{eq: RepOfU}
  		\arr_{\wa,j}(u,\cdot)'(t) \text{ exists and }  \arr_{\wa,j}(u,\cdot)'(t) >0 \text{ for all }t \in \Dwg[\wa][u] \text{ and all } \wa \in \Routes
  	\end{align}
  	where 
  	 \begin{align*}
  	\Dwg[\wa][u] = \Set{t \in \hori \mid u_{{{\wa}}[j]}(\arr_{{{\wa}},j}(u,t)) >0 , j\leq \abs{{{\wa}}}  }. 
  	\end{align*}
 	This is possible by \Cref{lem:  g>0Arr'>0}, the countability of $\Routes$ and finiteness of all walks in $\Routes$. 
 	 Note that \Cref{lem:  g>0Arr'>0} is applicable as  \Cref{ass: u: outflow} implies that $u$ admits an edge outflow rate for every edge. This is true as for every $\arc \in \GA$ that does not enter a \sink[,] $u_\arc$ admits an edge outflow rate by flow conservation, cf.~\Cref{lem: outflow: FlowCon}. 
  	Moreover, choose a representative of 
  	 $h^*$ that fulfills the  implication in \eqref{eq: h^*>0ImpliesBoundedCosts} for all $t \in \hori$ and 
  	 	additionally 
  	 \[h^*_\wa(t) > 0 \implies {u}_{\wa,j}(\arr_{\wa,j}(u,t)) > 0\]
  	 for all $ t\in \hori, \wa \in \Routes, j \leq |\wa|$ which is possible by \Cref{lem: Relations:h>0u>0:Pointwise}.  
  	 Let us also choose 
  	  an arbitrary representative of $\hat{\prices}$ that is bounded by $\norm{\hat{\prices}}_\infty$ for all $t \in \hori$. We define for all $\arc \in \GA$ 
  	\begin{align*}
  		\prices^*_\arc(t) \coloneqq \begin{cases}
  			\hat{\prices}_\arc(t) + \tilde{C}, &\text{if } t \in \mathfrak T_\arc\\ 
  			\hat{\prices}_\arc(t),  &\text{else}  
  		\end{cases}
  		\quad \text{ where }\quad \mathfrak T_\arc \coloneqq \{t \in \hori\mid u_\arc(t) = 0\} .
  	\end{align*}
  	It is clear that $\prices^*$ is a representative of an element of $L_+^\infty(\R)$ as well as still fulfills the equality~\eqref{eq: ZeroDGhat2} and subsequently~\eqref{eq: ZeroDGhatwalk2} since we only increased $\hat{\prices}_\arc(t)$ if ${u}_\arc(t)= 0$. 
  	Moreover, by the choice of the representatives $h^*$ and $u$, the implication in \eqref{eq: h^*>0ImpliesBoundedCosts} 
  	is also still valid for all $t \in \hori$ for $\prices^*$ instead of $\hat{\prices}$. 
  	
  	With these insights, we can now show that $h^*$ fulfills the Wardrop conditions of a $\prices^*$-DUE:
  	\begin{claim}
  		$h^*$ fulfills for almost all $t \in \hori$, all  $i \in I$ and $\wa \in \Routes_i$  
  		\begin{align*}
  			h^*_\wa(t)>0 \implies \wttime_\wa(u,t) + \Pf^{\prices^*}_\wa(u,t) \leq \wttime_{\wa'}(u,t) + \Pf^{\prices^*}_{\wa'}(u,t)\text{ for all } \wa'\in \Routes_i.
  		\end{align*}
  	\end{claim}
  	
  	\begin{proofClaim}   
  		Assume for the sake of a contradiction that  the claim was wrong.
  		Then, 
  		by the countability of $\Routes_i\times \Routes_i, i \in I$, 
  		there exists a commodity~$i \in I$, walks $\wa,\wa' \in \Routes_i$ and a set $\mathfrak D_{\wa>} \in\mathcal{B}(\hori)$ with $\leb(\mathfrak D_{\wa>})> 0$ such that 
  		\begin{align}\label{eq: DefDwa>}
  			h^*_\wa(t)>0 \quad  \text{ and } \quad   \wttime_\wa(u,t) + \Pf^{\prices^*}_\wa(u,t) > \wttime_{\wa'}(u,t)+ \Pf^{\prices^*}_{\wa'}(u,t) \; \text{ for all }t \in \mathfrak D_{\wa>}.
  		\end{align}

  		Define the walk inflow rate function ${h} \in L_+(\hori)^\Routes$ by shifting all inflow into~$\wa$ during~$\mathfrak D_{\wa>}$ to~$\wa'$:
  		\[{h}_{\Tilde{\wa}} \coloneqq \begin{cases}
  			h^*_{\wa'} + h^*_{\wa} \cdot \Indi_{\mathfrak D_{\wa>}}   &\text{if } \Tilde{\wa}=\wa' ,\\
  			h^*_{\wa} - h^*_{\wa} \cdot \Indi_{\mathfrak D_{\wa>}}     &\text{if } \Tilde{\wa}=\wa ,\\
  			h^*_{\tilde{\wa}} &\text{else.}
  		\end{cases}\]
  		We argue in the following that ${h} \in \edom{\Routes}[u]\cap \wir$. 
  		It is clear that ${h}$ is measurable and ${h} \in \wir$. 
  		We require the following claim: 
  		\begin{subclaim}
  			$\ell^u_{\tilde{\wa}}(h_{\tilde{\wa}})$ exists for all $\tilde{\wa} \in \Routes$. 
  		\end{subclaim}
  		\begin{proofClaim}
  			The existence of  $\ell^u_{\tilde{\wa}}({h}_{\tilde{\wa}})$ for $\tilde{\wa}\neq \wa'$ follows immediately by the existence of $\ell^u_{\tilde{\wa}}(h^*_{\tilde\wa})$. Remark that $\ell^u_{\tilde{\wa}}(h^*_{\tilde{\wa}}),\tilde{\wa} \in \Routes$ exist  by $\ell^u(h^*)$ existing, cf.~\Cref{lem: elluExistenceProperties:ExistenceInducedFlow}. 
  			
  			For the existence of $\ell^u_{\wa'}({h}_{\wa'})$, we aim to apply the ``tighter'' characterization of existence stated in  \Cref{cor: LuExAltChara}. 
  			The latter is applicable by the choice of our representative of $u$ fulfilling \eqref{eq: RepOfU}.
  			Hence, consider  an arbitrary $j^*\leq |\wa'|$ and a Borel-measurable null set $\mathfrak T^0\subseteq \hori\setminus \arr_{\wa',j^*}(u,\cdot)(\Dwg[\wa'][u])$.
  			We have to show that ${h}_{\wa'} = 0$ almost everywhere on 
  			$ \arr_{\wa',j^*}(u,\cdot)^{-1}(\mathfrak T^0)$. 
  			By $\ell^u_{\wa'}(h^*)$ existing, it is clear that $h^*_{\wa'}  = 0$ 
  			on the latter set. Since ${h}_{\wa'} = h^*_{\wa'}$ outside of $\mathfrak D_{\wa>}$, 
  			it is, therefore, sufficient to show that $\arr_{\wa',j^*}(u,\cdot)^{-1}(\mathfrak T^0) \cap \mathfrak D_{\wa>} = \emptyset$ holds. This is an immediate consequence of  $ \mathfrak D_{\wa>} \subseteq \Dwg[\wa'][u]$ for which we argue in the following: 
  			 For any $t \in \mathfrak D_{\wa>} \setminus \Dwg[\wa'][u]$ we have  
  			\begin{align*}
  				\Pf^{\prices^*}_{\wa'}(u,t) \symoverset{1}{\geq} \tilde{C} \symoverset{2}{\geq} \Psi_\wa(u,t) + \Pf^{\prices^*}_\wa(u,t) \overset{\eqref{eq: DefDwa>}}{>} \wttime_{\wa'}(u,t)+ \Pf^{\prices^*}_{\wa'}(u,t) 
  			\end{align*}
  			where \refsym{1} holds by the definition of $\prices^*$ while \refsym{2} holds by the implication in \eqref{eq: h^*>0ImpliesBoundedCosts} being valid for all $t \in \mathfrak{D}_{\wa>}$. 
  			Since $\wttime_{\wa'}(u,t) \geq 0$, this  shows that such a $t$ can not exist and, hence, the claimed inclusion holds which finishes the proof of this claim. 
  		\end{proofClaim} 
  		In order to derive from this that $\ell^u(h)$ exists and, hence, $h \in \edom{\Routes}[u]$  holds, it remains by \Cref{lem: elluExistenceProperties:ExistenceInducedFlow}  to observe    that  $(\ell^u_{\tilde{\wa}}({h}_{\tilde{\wa}}))_{\tilde{\wa} \in \Routes} \in  \seql[1][\Routes][L_+(\hori)^\GA]$ holds: This is true as   $(\ell^u_{\tilde{\wa}}({h}_{\tilde{\wa}}))_{\tilde{\wa} \in \Routes}$ and $(\ell^u_{\tilde{\wa}}(h^*_{\tilde{\wa}}))_{\tilde{\wa} \in \Routes}$ only differ in the two entries corresponding to $\wa,\wa'$ and $(\ell^u_{\tilde{\wa}}(h^*_{\tilde{\wa}}))_{\tilde{\wa} \in \Routes} \in  \seql[1][\Routes][L_+(\hori)^\GA]$. The latter holds by the existence of $\ell^u(h^*)$ and \Cref{lem: elluExistenceProperties:ExistenceInducedFlow}. Thus, we have  $h \in \edom{\Routes}[u]$.

  		Next, we observe that 
  		\begin{align*}
  			\sum_{{\tilde{\wa}} \in \Routes} \dup{\wttime_{\tilde{\wa}}(u,\cdot) + \Pf^{\prices^*}_{\tilde{\wa}}(u,\cdot)}{h^*_{\tilde{\wa}}}  < \infty 
  		\end{align*}
  		 by the validity of \eqref{eq: h^*>0ImpliesBoundedCosts} for $\prices^*$ and $h \in \edom{\Routes}[u]\subseteq \seql$ (cf.~\Cref{lem: elluContinuity:Subset}). This, in combination with the definition of $\mathfrak D_{\wa>}$, allows us to get the strict inequality \refsym{1} in the following: 
  		 \begin{align*}
  		 	\sum_{{\tilde{\wa}} \in \Routes} &\dup{\wttime_{\tilde{\wa}}(u,\cdot) + \Pf^{\prices^*}_{\tilde{\wa}}(u,\cdot)}{h^*_{\tilde{\wa}}}  
  		 	=\sum_{\tilde{\wa} \in \Routes\setminus \{\wa',\wa\} } \dup{\wttime_{\tilde{\wa}}(u,\cdot) +
  		 	 \Pf^{\prices^*}_{\tilde{\wa}}(u,\cdot)}{h_{\tilde{\wa}}} \\ 
  		 	 + \,&\dup{\wttime_{{\wa}}(u,\cdot) + \Pf^{\prices^*}_{{\wa}}(u,\cdot)}{h_{{\wa}}+h^*_{\wa}\cdot \Indi_{\mathfrak{ D}_{\wa>}}}  
  		 	 +\dup{\wttime_{{\wa'}}(u,\cdot) +\Pf^{\prices^*}_{{\wa'}}(u,\cdot)}{h_{{\wa'}}-h^*_{\wa}\cdot \Indi_{\mathfrak{ D}_{\wa>}}}\\
  		 	 \symoverset{2}{=}&\sum_{\tilde{\wa} \in \Routes} \dup{\wttime_{\tilde{\wa}}(u,\cdot) +
  		 	 	\Pf^{\prices^*}_{\tilde{\wa}}(u,\cdot)}{h_{\tilde{\wa}}} + \dup{\wttime_{\wa}(u,\cdot) + \Pf^{\prices^*}_{{\wa}}(u,\cdot) - \wttime_{\wa'}(u,\cdot) - \Pf^{\prices^*}_{{\wa}'}(u,\cdot)  }{h^*_{\wa }\cdot \Indi_{\mathfrak D_{\wa>}}}\\
  		 	 \symoverset{1}{>}&\sum_{\tilde{\wa} \in \Routes } \dup{\wttime_{\tilde{\wa}}(u,\cdot) +
  		 	 	\Pf^{\prices^*}_{\tilde{\wa}}(u,\cdot)}{h_{\tilde{\wa}}} .
  		 \end{align*}
  		 	Here, we used for \refsym{2} that for $\tilde{\wa} \in \{\wa',\wa\}$ we have
  		 \begin{align*}
  		 	\dup{\wttime_{\tilde{\wa}}(u,\cdot) + \Pf^{\prices^*}_{\tilde{\wa}}(u,\cdot)}{h^*_{\wa}\cdot \Indi_{\mathfrak{ D}_{\wa>}}} \overset{\eqref{eq: DefDwa>}}&{\leq}\dup{\wttime_{\wa}(u,\cdot) + \Pf^{\prices^*}_{{\wa}}(u,\cdot)}{h^*_{\wa}\cdot \Indi_{\mathfrak{ D}_{\wa>}}}\\ &\leq 	\sum_{{\tilde{\wa}} \in \Routes} \dup{\wttime_{\tilde{\wa}}(u,\cdot) + \Pf^{\prices^*}_{\tilde{\wa}}(u,\cdot)}{h^*_{\tilde{\wa}}}  < \infty.
  		 \end{align*}
  		   Hence, we arrive at the desired contradiction since we already argued that the opposite inequality (namely~\eqref{eq: ZeroDGhatwalk2}) 
  		also holds for $\prices^*$. 
  	\end{proofClaim}
  	Thus, we have shown that $h^*$ fulfills the $\prices^*$-DUE Wardrop condition. Moreover, $h^*$ also induces $u$ under $\Nl$ by \Cref{ass: u} and $\ell^u(h^*) = u$ being valid. 
  	Thus, the proof is finished.\qedhere 
  \end{proof}

  As already mentioned, for finite supported $u$ and certain models, 
  the requirements in 
 \Cref{thm: SuffConMSMS: FiniteWalkSupp,thm: SuffConMSMS: BoundedExpTravel} are always satisfied:
  
  \begin{lemma}\label{lem: AssUZeroDG}
  Assume that $u$ has finite support, i.e.~there exist $t_0,\tEnd\in \R$ such that $u= 0$ on $\hori \setminus [t_0,\tEnd]$. 
  Then, the following statements are valid:
  \begin{thmparts}
      \item  If the travel times induced by $u$ are lower bounded by a positive constant, i.e.\ if there exists $\varepsilon>0$ with $\trav_\arc(u,t)\geq \varepsilon$ for all $t \in \hori$ and $\arc \in \GA$, then \Cref{thm: SuffConMSMS: FiniteWalkSupp} holds for all $h$  with  $\ell^u(h)\leq u$. \label[thmpart]{lem: AssUZeroDG: 1}
      \item If  the private costs are defined by \eqref{eq: PC=WTT}, then there exists $C \in \R$ such that all $h$ with  $\ell^u(h)\leq u$ fulfill \Cref{thm: SuffConMSMS: BoundedExpTravel} \wrt $C$.\label[thmpart]{lem: AssUZeroDG: 2}
  \end{thmparts} 
  \end{lemma}

  \begin{proof}
  For both statements, we make the following observation: Consider an arbitrary $i \in I$, $\wa \in \Routes_i$ and $h$ with $\ell^u(h) \leq u$. 
  	Then, for all  $j\leq\abs{\wa}$, we have $\ell^u_{\wa,j}(h_w)\leq \ell^u_{\wa[j]}(h) \leq u_{\wa[j]} = 0=\ell_{\wa,j}(0)$ a.e.~on  
  	$\hori\setminus[t_0,\tEnd]=\arr_{\wa,j}(u,\cdot)\big(\arr_{\wa,j}(u,\cdot)^{-1}(\hori\setminus [t_0,\tEnd]) \big)$. 
      Hence, by $\ell_{\wa,j}$ being order preserving (\Cref{lem: ellOrderPreservingSharpened}), we get that 
  	$h_\wa =0$ on $\arr_{\wa,j}(u,\cdot)^{-1}(\hori\setminus [t_0,\tEnd])$.
  For $j = 1$, this yields 
  	$h_\wa = 0$ on $\hori\setminus [t_0,\tEnd]$ while for $j = \abs{\wa}$, we get 
  	$h_\wa =0$ on $\arr_{\wa,\abs{\wa}}(u,\cdot)^{-1}(\hori\setminus [t_0,\tEnd])$. 
  	Hence, we can conclude that for almost all $t\in \hori$ we have
  	\begin{align}
  		h_\wa(t) > 0 &\implies t \notin \Big(\big(\hori\setminus [t_0,\tEnd]\big) \cup  \arr_{\wa,\abs{\wa}}(u,\cdot)^{-1}\big(\hori\setminus [t_0,\tEnd]\big) \Big) \nonumber\\
  		&\implies t \in  [t_0,\tEnd] \land   \arr_{\wa,\abs{\wa}}(u,t)\in  [t_0,\tEnd]   \nonumber\\
  		&\implies t \in  [t_0,\tEnd] \land   \arr_{\wa,\abs{\wa}+1}(u,t)\in \left[t_0,\max_{\arc \in \GA}\exit_\arc(u,\tEnd)\right].  \label{eq: lem: AssUZeroDG}
  	\end{align}
    \begin{structuredproof}
  \proofitem{\ref{lem: AssUZeroDG: 1}}
  Since $\trav_\arc(u,t)\geq \varepsilon$ for all $\arc \in \GA$ and $t \in \hori$, 
  we have for any $t \in[t_0,t_f]$ that
  \begin{align*}
      \arr_{\wa,\abs{\wa}+1}(u,t) \geq \arr_{\wa,\abs{\wa}+1}(u,t_0) \geq  t_0 + \varepsilon \cdot \abs{\wa}.
  \end{align*}
  Now using \eqref{eq: lem: AssUZeroDG} implies for any $t \in \hori$ with $h_\wa(t)>0$ that $t \in [t_0,t_f]$ 
  and hence with the above that 
  $\abs{\wa} \leq \frac{1}{\varepsilon} \cdot\big(\max_{\arc \in \GA}\exit_\arc(u,\tEnd) - t_0\big)$. 
  
      \proofitem{\ref{lem: AssUZeroDG: 2}}  
      From \eqref{eq: lem: AssUZeroDG}, we immediately get: 
      \begin{align*}
          \wttime_\wa(u,t) \overset{\eqref{eq: PC=WTT}}{=} \gamma_i\cdot \left( \arr_{\wa,\abs{\wa}+1}(u,t) - t\right) \leq \gamma_i\cdot \Bigl(\max_{\arc \in \GA}\exit_\arc(u,\tEnd) - t_0  \Bigr),
      \end{align*}
  	which finishes the proof. \qedhere
  \end{structuredproof}
  \end{proof}

  Next, we provide sufficient conditions for a $h^*\in\wir$ to implement $u$ in the case of multi-source, single-\sink networks. 
  For this, we require as in the homogeneous case, that the private costs induced by $u$ are separable over the edges: 
  \begin{rstassumption}{\ref{ass: PCSep}}
  	The private costs are separable over the edges, that is, there exists for all $i\in I$ a  measurable function  $\psi^i : \hori \to \R^\GA_+$ such that for all $(\hat{\wa},i)\in \Routes_i$:
  	\begin{align*}
  		\wttime_{(\hat{\wa},i)}(u,\cdot) :=  \sum_{j \leq|\hat{\wa}|} \psi^i_{\hat{\wa}[j]}(\arr_{\hat{\wa},j}(u,\cdot))  .
  	\end{align*}
  \end{rstassumption}
  As for the homogeneous user case, 
  we require the above assumption in order to associate private costs to subwalks of walks contained in $\Routes$ and hence extend the function $\wttime$ to all finite walks $\hat{\wa}$.  
  In the following,  we call \eqref{opt: Master} \wellposed if there exists a feasible solution with a finite objective value.

  \begin{lemma}\label{thm: SuffConMSSS}
  	Consider a multi-source, single-\sink network with \Cref{ass: u,ass: PCSep} being fulfilled, \eqref{opt: Master} being \wellposed and admitting strong duality \wrt $\MeasFuncUInt$. 
  	Then, $u$ is implementable via $h^* \in \wir$, if $h^*$ is an optimal solution for \eqref{opt: Master} with tight inequality \eqref{ineq: Master}.  
  \end{lemma}
    Note that \eqref{opt: Master} being \wellposed is implied in \Cref{thm: SuffConMSMS} by the requirement of bounded experienced private costs for any $u$-inducing $h\in \wir$. 
  As an immediate consequence of the above and \Cref{lem: NecessImpl}, we get the following characterization: 
  \begin{theorem}\label{thm: mainMSSS}
  	In a  multi-source, single-\sink network with \Cref{ass: u,ass: PCSep} being fulfilled, \eqref{opt: Master}  being \wellposed and 
  	admitting strong duality \wrt $\MeasFuncUInt$, the following statements are equivalent:
  	\begin{thmparts}
  		\item $u$ is implementable. 
  		\item There exists an optimal solution for \eqref{opt: Master} with tight inequality \eqref{ineq: Master}. \label[thmpart]{thm: mainMSSS: Con}
  	\end{thmparts}
  \end{theorem}

  The proof of \Cref{thm: SuffConMSSS} works similar to the proof of \Cref{thm: SuffConMSMS}. The main difference is the way we need to adjust the tolls for unused edges. Without $h^*$ fulfilling \Cref{thm: SuffConMSMS: FiniteWalkSupp,thm: SuffConMSMS: BoundedExpTravel}, and, without the optimal dual solution $\prices^* \in \MeasFuncUInt$ being necessarily uniformly bounded (i.e.~$\prices^*\notin L_+^\infty(\hori)^\GA$), the  total costs under $\prices^*$ among all particles in $h^*$ are not necessary  upper bounded anymore. 
  In order to adjust the tolls suitably, we instead define 
  a node potential for each node $v \in \GV$ and commodity~$i$ which describe the minimal costs 
  among all \stwalk[v][\dest] walks only containing used edges under $u$. Adding these node potentials to the tolls on unused edges that start at the corresponding node  yields suitable tolls.

  \begin{proof}[\textbf{Proof of \Cref{thm: SuffConMSSS}}]   
  	Let us consider an optimal $h^* \in \edom{\Routes}[u]\cap \wir$ for~\eqref{opt: Master}
  	with $\ell^u(h^*) = u$.   
  	For the remainder of the proof,  choose a representative of $u$ such that 
  	\begin{align*}
  		\arr_{\wa,j}(u,\cdot)'(t) \text{ exists and }  \arr_{\wa,j}(u,\cdot)'(t) >0 \text{ for all }t \in \Dwg[\wa][u] \text{ and all } \wa \in \Routes
  	\end{align*}
  	where 
  	\begin{align*}
  		\Dwg[\wa][u] = \Set{t \in \hori \mid u_{{{\wa}}[j]}(\arr_{{{\wa}},j}(u,t)) >0 , j\leq \abs{{{\wa}}}  }. 
  	\end{align*}
  	This is possible by \Cref{lem:  g>0Arr'>0}, the countability of $\Routes$ and finiteness of all walks in $\Routes$. 
  	Note that \Cref{lem:  g>0Arr'>0} is applicable as  \Cref{ass: u: outflow} implies that $u$ admits an edge outflow rate for every edge. This is true as for every $\arc \in \GA$ that does not enter a \sink[,] $u_\arc$ admits an edge outflow rate by flow conservation, cf.~\Cref{lem: outflow: FlowCon}. 
  	Additionally, choose a representatives of $h^*$  
 that fulfills 
 \[h^*_\wa(t) > 0 \implies {u}_{\wa,j}(\arr_{\wa,j}(u,t)) > 0\]
 for all $ t\in \hori, \wa \in \Routes, j \leq |\wa|$ (which exists by \Cref{lem: Relations:h>0u>0:Pointwise}).

  	By assumption, we have strong duality for~\eqref{opt: Master} \wrt $\MeasFuncUInt$. Hence, there exists $\hat{\prices}\in \MeasFuncUInt$ fulfilling 
  	\begin{align}\label{eq: ZeroDGhat}
  		\inf_{{h} \in \edom{\Routes}[u]\cap \wir}\sum_{\wa \in \Routes} \dup{\wttime_\wa(u,\cdot)}{{h}_\wa} + \dup{\hat{\prices}}{ \ell^u(h) -u} =  \sum_{\wa \in \Routes} \dup{\wttime_\wa(u,\cdot)}{h^*_\wa} .  
  	\end{align}  
  	By adding to both sides the constant term $\dup{\hat{\prices}}{u}\in \R$ and using  \Cref{lem: aggCostsVSwalkCosts}, this is equivalent to 
  	\begin{align}\label{eq: ZeroDGhatwalk}
  		\sum_{\wa \in \Routes} \dup{\wttime_\wa(u,\cdot) + \Pf^{\hat{\prices}}_\wa(u,\cdot)}{ {h}_\wa} \geq \sum_{\wa \in \Routes} \dup{\wttime_\wa(u,\cdot) + \Pf^{\hat{\prices}}_\wa(u,\cdot)}{h^*_\wa} && \forall {h} \in \edom{\Routes}[u]\cap \wir . 
  	\end{align}  
  	Note that the right hand side of the above equation is smaller than infinity by the assumption that \eqref{opt: Master} is \wellposed[.]  
  	
  	Using \eqref{eq: ZeroDGhatwalk}, we will show  several properties of the functions $\phi_i^v: \hori \to \Rnn, i \in I,v \in \GV$ where, for any point in time $t$, $\phi_i^v(t)$ denotes the minimal total cost for a particle of commodity $i$  along any \stwalk[v][\dest] consisting only of used edges under~$u$, i.e.
  	\begin{align*}
  		\phi_i^v(t) \coloneqq \begin{cases} 
  			\inf\Big\{\Indi_{\hori \setminus\Dwg[\wa][u]}(t) \cdot \infty + \big( \wttime_{{{\wa}}}(u,t) + \Pf^{\hat{\prices}}_{{{\wa}}}(u,t)\big) \,\Big|\,  {{\wa}} \in {\Routes}_{v,\dest}^i\Big\}, &\text{if } t \in \bigcup_{{{\wa}} \in {\Routes}_{v,\dest}}\Dwg[\wa][u] \\
  			0, &\text{else.}
  		\end{cases}
  	\end{align*}
  	Here,  ${\Routes}_{v,\dest}^i:= \hat{\Routes}_{v,\dest}\times \{i\}$ denotes the set of all   commodity $i$-typed (finite) \stwalk[v][\dest]s. 
  	Remark that we allow for the trivial walk ${{\wa}}=((),i)$ to be contained in ${\Routes}_{v,\dest}^i$ in case of $v = \dest$ 
  	with $\Dwg[((),i)][u]:= \hori$. Also note that $\Dwg[\wa][u]$ for any walk ${{\wa}}$ is measurable by $u$ being measurable and, subsequently, the countable union $\bigcup\{\Dwg[\wa][u] \mid {{\wa}} \in {\Routes}_{v,\dest}^i\}$ is measurable as well.  
This implies that $\phi_i^v$ is the infimum of countably many measurable functions  which are all bounded by $0$ from below and take values in the  extended reals $\bar{\R}$. Here, we mean measurable with $\Bar{\R}$ being equipped with  the standard Borel $\sigma$-algebra $\mathcal{B}(\bar{\R})$ (cf.~\cite[Remark 2.1.6]{Bogachev2007I}).
Since $\phi_i^v$ is  real-valued (due to the definition), it follows by~\cite[Theorem 2.1.5 and Remark 2.1.6]{Bogachev2007I}
that $\phi_i^v$ is a measurable function from $\hori$ to $\R$.

  	We will now show that the total cost for commodity $i$ along some \stwalk[v][\dest] is equal to $\phi_i^v$ whenever this walk is a suffix of an \stwalk[\source_i][\dest] that is used under~$(h^*_\wa)_{\wa \in \Routes_i}$:
  	\begin{claim}\label{claim: PhiEqualsHatWalks}
  		For all $v \in \GV$, $\wa \in \Routes_i$, $i\in I$ and $j\leq |\wa|$ with $\wa[j]\in\delta^+(v)$ we have
  		\begin{align*}
  			\phi_i^v= \wttime_{{\wa}_{\geq j}}(u,\cdot) + \Pf^{\hat{\prices}}_{{\wa}_{\geq j}}(u,\cdot) \quad \quad \quad  \text{ a.e.\ on }   \hat{\mathfrak T}_{\wa,j}:= \arr_{\wa,j}(u,\cdot)(\{t \in \hori\mid h^*_\wa(t)>0\}).
  			\end{align*}   
  		In particular, we have for almost all $t \in \hori$, all $i \in I$ and all $\wa \in \Routes_i$ the implication  
  		\begin{align*}
  			h_\wa^*(t)>0 \implies\wttime_{{\wa}}(u,t) + \Pf^{\hat{\prices}}_{{\wa}}(u,t) = \phi_i^{\source_i}(t).
  		\end{align*}
  	\end{claim}
  	
  	\begin{proofClaim} 
  		Let $v \in \GV$ and $i \in I$ be arbitrary. 
  		The choice of our representatives of $h^*$ and $u$ ensures that for any $\wa=(\hat{\wa},i) \in \Routes_i$ and $j\leq|\wa|$ with $\wa[j] \in \delta^+(v)$ we get $\hat{\mathfrak T}_{\wa,j} \subseteq \Dwg[\wa_{\geq j}][u]$. 
  		In particular, since $\wa_{\geq j} \in {\Routes}_{v,\dest}^i$, we know by definition that for all $t\in \hat{\mathfrak T}_{\wa,j}$:
  		\begin{align*}
  			\phi_i^v(t) =  \inf\Big\{\Indi_{\hori \setminus\mathfrak D_{{{\wa}'}}}(t) \cdot \infty + \big( \wttime_{{{\wa}'}}(u,t) + \Pf^{\hat{\prices}}_{{{\wa}'}}(u,t)\big) \,\Big|\,  {{\wa}'} \in {\Routes}_{v,\dest}^i\Big\} \leq  \wttime_{{{\wa}_{\geq j}}}(u,t) + \Pf^{\hat{\prices}}_{{{\wa_{\geq j}}}}(u,t).
  		\end{align*}

  		Hence, in order to prove the \namecref{claim: PhiEqualsHatWalks}, it is sufficient to derive a contradiction from the following assumption: 
  		Assume  that   there  exists 
  		a $v,\dest$-walk $\wa'=(\hat{\wa}',i)$ for which there exists
  		$\mathfrak T \subseteq  \hat{\mathfrak T}_{\wa,j} \cap \Dwg[\wa'][u] $ with positive measure $\leb(\mathfrak T)> 0$ such that
  		\begin{align}\label{eq: DefViolatingSet}
  			\wttime_{{{\wa}'}}(u,t) + \Pf^{\hat{\prices}}_{{{\wa}'}}(u,t) < \wttime_{{\wa}_{\geq j}}(u,t) + \Pf^{\hat{\prices}}_{{\wa}_{\geq j}}(u,t) \quad \quad  \text{ for all } t \in \mathfrak T. 
  		\end{align}

  		Let us define $\wa^*:=(\hat\wa_{< j},\hat\wa',i) \in \Routes_i$ which equals the walk $\wa$ until it reaches the node $v$ and afterwards coincides with the walk $\wa'$. By the prior inequality \eqref{eq: DefViolatingSet}, we   have  
  		\begin{align}\label{eq: PropD}
  			\wttime_{{\wa^*}}(u,t) + \Pf^{\hat{\prices}}_{{\wa}^*}(u,t)<\wttime_{{\wa}}(u,t) + \Pf^{\hat{\prices}}_{\wa}(u,t) \text{ for all } t\in\arr_{\wa,j}(u,\cdot)^{-1}(\mathfrak T) . 
  		\end{align}
  		In particular, the above holds for all $t \in  \mathfrak D :=\arr_{\wa,j}(u,\cdot)^{-1}(\mathfrak T) \cap \{t \in \hori\mid h^*_\wa(t)>0\}$.  This set is not a null set ($\leb(\mathfrak D)>0$)  
  		as $\arr_{\wa,j}(u,\cdot)$ has Lusin's property (cf.~\Cref{lem: PropAbsCon:Lus}), $\leb(\mathfrak T)>0$ and 
  		$\arr_{\wa,j}(u,\cdot)(\mathfrak D)=\mathfrak T$. For the last equality, 
  		note that by definition, $\mathfrak T \subseteq  \hat{\mathfrak T}_{\wa,j} \cap \Dwg[\wa'][u]$ and $ \hat{\mathfrak T}_{\wa,j}  = \arr_{\wa,j}(u,\cdot)(\{t \in \hori\mid h^*_\wa(t)>0\}$ holds.

  		 Let us define new walk inflow rates ${h}$ by shifting all inflow into~$\wa$ during~$\mathfrak D$ to~$\wa^*$:
  		\[{h}_{\tilde{\wa}} \coloneqq \begin{cases}
  			h^*_{\wa^*} + h^*_{\wa}\cdot \Indi_{\mathfrak D}, &\text{if } \tilde{\wa} = \wa^*, \\
  			h^*_{\wa} - h^*_{\wa}\cdot \Indi_{{\mathfrak D}}, &\text{if } \tilde{\wa} = \wa, \\
  			h^*_{\tilde{\wa}},                                 &\text{else.}
  		\end{cases}\]
  		We argue in the following that ${h}$ is contained in $\edom{\Routes}[u]\cap \wir$:
  		By $\mathfrak D$ being measurable, it is clear that $\Tilde{h}$ is likewise. 
  		It is, then, also clear that ${h} \in  \wir$.

  		Next, we show the following subclaim:
  		\begin{subclaim}
  			$\ell^u_{\tilde{\wa}}({h}_{\tilde{\wa}})$ exists for all ${\tilde{\wa}}\in \Routes$.   
  		\end{subclaim}
  		\begin{proofClaim} 
  			The existence of  $\ell^u_{\tilde{\wa}}({h}_{\tilde{\wa}})$ for $\tilde{\wa}\neq \wa^*$ follows immediately by the existence of $\ell^u_{\tilde{\wa}}(h^*_{\tilde\wa})$. Remark that $\ell^u_{\tilde{\wa}}(h^*_{\tilde{\wa}}),\tilde{\wa} \in \Routes$ exist  by $\ell^u(h^*)$ existing, cf.~\Cref{lem: elluExistenceProperties:ExistenceInducedFlow}. 
  			For the existence for $\wa^*$, we verify in the following the sufficient condition of the ``tighter'' characterization of existence in \Cref{cor: LuExAltChara}: 
  			Consider an arbitrary $\tilde{j} \leq \abs{\wa^*}$ and null set $\mathfrak T' \subseteq \hori \setminus \arr_{\wa^*,\tilde{j}}(u,\cdot)(\Dwg[\wa^*][u])$.   
  			
  			Since we have $h^*_{\wa^*} = 0$ on $\arr_{\wa^*,\tilde{j}}(u,\cdot)^{-1}(\mathfrak T')$, by the existence of $\ell^u_{\wa^*}(h^*_{\wa^*})$, and $h_{\wa^*} = h^*_{\wa^*}$ outside of~$\mathfrak D$, it suffices to show that $\arr_{\wa^*,\tilde{j}}(u,\cdot)^{-1}(\mathfrak T')\cap \mathfrak D = \emptyset$. To show this, it is in turn sufficient to show 
            that   $\mathfrak D \subseteq \Dwg[\wa^*][u]$ holds as  $\mathfrak T' \subseteq \hori \setminus \arr_{\wa^*,\tilde{j}}(u,\cdot)(\Dwg[\wa^*][u])$: 
            By the choice of the representatives of $h^*$ and $u$
            together with $\mathfrak D \subseteq \{t \in \hori\mid h^*_\wa(t)>0\}$, we get that $\mathfrak D \subseteq \Dwg[\wa][u]$ and, in particular, $\mathfrak D \subseteq \Dwg[\wa^*_{<j}][u]$. 
            For the second  part of the walk~$\wa^*$ (i.e.~$\wa'$), 
            we observe that the inclusions $\mathfrak D \subseteq \arr_{\wa,j}(u,\cdot)^{-1}(\mathfrak T)$ and $\mathfrak T \subseteq \Dwg[\wa'][u]$ ensure that we have $ \arr_{\wa,j}(u,\cdot)(\mathfrak D) \subseteq \Dwg[\wa'][u]$. Together, this shows that  $\mathfrak D \subseteq \Dwg[\wa^*][u]$. 
  		\end{proofClaim}
  		
  		In order to derive from this that ${h} \in \edom{\Routes}[u]$, it remains by \Cref{lem: elluExistenceProperties:ExistenceInducedFlow} to observe   that the sum $\sum_{{\tilde{\wa}} \in \Routes}\ell^u_{\tilde{\wa}}({h}_{\tilde{\wa}})$ exists as well: This is due to the latter sum for $h^*$ existing as well as $(\ell^u_{\tilde{\wa}}({h}_{\tilde{\wa}}))_{\tilde{\wa} \in \Routes}$ and $(\ell^u_{\tilde{\wa}}(h^*_{\tilde{\wa}}))_{\tilde{\wa} \in \Routes}$ only differing in the two entries corresponding to $\wa,\wa^*$. 
  		
  		With this at hand, we get 
   		 \begin{align*}
 	\sum_{{\tilde{\wa}} \in \Routes} &\dup{\wttime_{\tilde{\wa}}(u,\cdot) + \Pf^{\hat{\prices}}_{\tilde{\wa}}(u,\cdot)}{h^*_{\tilde{\wa}}}  
 	=\sum_{\tilde{\wa} \in \Routes\setminus \{\wa^*,\wa\} } \dup{\wttime_{\tilde{\wa}}(u,\cdot) +
 		\Pf^{\hat{\prices}}_{\tilde{\wa}}(u,\cdot)}{h_{\tilde{\wa}}} \\ 
 	+&\,\dup{\wttime_{{\wa}}(u,\cdot) + \Pf^{\hat{\prices}}_{{\wa}}(u,\cdot)}{h_{{\wa}}+h^*_{\wa}\cdot \Indi_{\mathfrak{ D}}}  
 	+\dup{\wttime_{{\wa^*}}(u,\cdot) +\Pf^{\hat{\prices}}_{{\wa^*}}(u,\cdot)}{h_{{\wa^*}}-h^*_{\wa}\cdot \Indi_{\mathfrak{ D}}}\\
 	\symoverset{1}{=}&\sum_{\tilde{\wa} \in \Routes} \dup{\wttime_{\tilde{\wa}}(u,\cdot) +
 		\Pf^{\hat{\prices}}_{\tilde{\wa}}(u,\cdot)}{h_{\tilde{\wa}}} + \dup{\wttime_{\wa}(u,\cdot) + \Pf^{\hat{\prices}}_{{\wa}}(u,\cdot) - \wttime_{\wa^*}(u,\cdot) - \Pf^{\hat{\prices}}_{\wa^*}(u,\cdot)  }{h^*_{\wa }\cdot \Indi_{\mathfrak D_{\wa>}}}\\
 	\symoverset{2}{>}&\sum_{\tilde{\wa} \in \Routes } \dup{\wttime_{\tilde{\wa}}(u,\cdot) +
 		\Pf^{\hat{\prices}}_{\tilde{\wa}}(u,\cdot)}{h_{\tilde{\wa}}} .
 \end{align*} 		
  		Here, we used for \refsym{1} that for $\tilde{\wa} \in \{\wa^*,\wa\}$ we have
  		\begin{align*}
  			 \dup{\wttime_{\tilde{\wa}}(u,\cdot) + \Pf^{\hat{\prices}}_{\tilde{\wa}}(u,\cdot)}{h^*_{\wa}\cdot \Indi_{\mathfrak{ D}}} \overset{\eqref{eq: DefViolatingSet}}&{\leq}\dup{\wttime_{\wa}(u,\cdot) + \Pf^{\hat{\prices}}_{{\wa}}(u,\cdot)}{h^*_{\wa}\cdot \Indi_{\mathfrak{ D}}}\\ &\leq 	\sum_{{\tilde{\wa}} \in \Routes} \dup{\wttime_{\tilde{\wa}}(u,\cdot) + \Pf^{\hat{\prices}}_{\tilde{\wa}}(u,\cdot)}{h^*_{\tilde{\wa}}}  < \infty.
  		\end{align*}
  		For \refsym{2}, we used \eqref{eq: PropD} with $\mathfrak D \subseteq \arr_{\wa,j}(u,\cdot)^{-1}(\mathfrak T)$,  $h_\wa^*(t)>0$ for a.e.~$t \in \mathfrak D$ by  the definition of $\mathfrak D$,  $\leb(\mathfrak D)>0$ and the finiteness of $\sum_{{\tilde{\wa}} \in \Routes} \dup{\wttime_{\tilde{\wa}}(u,\cdot) + \Pf^{\hat{\prices}}_{\tilde{\wa}}(u,\cdot)}{h^*_{\tilde{\wa}}}$ (cf.~above).
  		 This now 
  		contradicts~\eqref{eq: ZeroDGhatwalk} and, hence, finishes the proof  the claim. 
  	\end{proofClaim}

  	With this insight, we define in the following  tolls ${\prices^*}$ by changing $\hat{\prices}$: 
  	For any $\arc  \in \GA$ we denote for the remainder of the proof by
  	\begin{align}\label{eq: DefTimesU=0}
  		\mathfrak T_\arc \coloneqq \{t \in \hori\mid u_\arc(t) = 0\}
  	\end{align}
  	the set of times at which $u$ does not send flow into $\arc$. Note that these sets are measurable by $u$ being measurable. With respect to these sets, we then define for any $\arc =(v,v') \in \GA$
  	\begin{align}\label{eq: Defp^*}
  		\prices^*_\arc(t) \coloneqq \begin{cases}
  			1+ \sum_{i\in I} \phi_i^v(t) + \hat{\prices}_\arc(t), &\text{if } t \in \mathfrak T_\arc\\
  			\hat{\prices}_\arc(t),  &\text{else.}
  		\end{cases}
  	\end{align}
  	It is clear that ${\prices^*}\in \MeasFuncUInt$ and $\prices^*$ still fulfills the equality~\eqref{eq: ZeroDGhat} and, subsequently,~\eqref{eq: ZeroDGhatwalk} since we only increased $\hat{\prices}_\arc(t)$ if ${u}_\arc(t)= 0$. 
  	Note that $\prices^*$ is measurable by $\phi_i^v,i \in I,v\in \GV$ and $\mathfrak T_\arc$ being measurable. 
  	
  	We will now show that under the adjusted tolls, the total cost for commodity $i$ for any walk $\wa'$ starting at $\source_i$ 
  	and containing at least one unused edge under $u$ is at least~$\phi^{\source_i}_i$, that is, the costs along $\wa'$ are at least as high as the total cost for commodity $i$ of any used walk under~$(h^*_\wa)_{\wa \in \Routes_i}$ (cf.\ \Cref{claim: PhiEqualsHatWalks}). 
  	Hence, when it comes to verifying that $h^*$ is a $\prices^*$-DUE, we won't have to consider these walks $\wa'$. 
  	
  	\begin{claim} \label{claim: ineqphi}
  		For almost all $t \in \hori$, all $i \in I$, all walks $\wa'\in\Routes_i$ and all $j\leq |\wa'|$, we have the implication 
  		$t \in \Dwg[\wa'][u,j] \implies   \wttime_{\wa'}(u,t) + \Pf^{{\prices^*}}_{\wa'}(u,t) \geq \phi^{s_i}_i(t)$ where
  		\begin{align*}
  			\Dwg[\wa'][u,j] := \{t \in \hori \mid u_{\wa'[j']}(\arr_{\wa,j'}(u,t)) >0 , j'<j,  u_{\wa'[j]}(\arr_{\wa',j'}(u,t)) =0  \}.
  		\end{align*}
  	\end{claim}
  	\begin{proofClaim}
  		
  		Since the set $I$ is finite, the set $\Routes_i$ is countable, and the set of indices $j\leq \abs{\wa'}$ is finite for any $\wa' \in \Routes_i$,  it is enough to derive a contradiction from the following assumption: 
  		Assume that there exists $i \in I$, a walk $\wa' \in \Routes_i$ and $j\leq |\wa'|$ with the measurable set 
  		\[
  		\Dwg[\wa'<][u,j]:= \Dwg[\wa'][u,j] \cap \{t \in \hori \mid\wttime_{\wa'}(u,t) + \Pf^{{\prices^*}}_{\wa'}(u,t) < \phi^{\source_i}_i(t)\}
  		\]
  		having positive measure $\leb(\Dwg[\wa'<][u,j])>0$.
  		
  		We differ between the cases $j = 1$ and $j>1$:
  		
  		\begin{structuredproof}
  			\proofitem{Case $j =1$} Since $\wa'[j]=\wa'[1]\in\delta^+(\source_i)$ and $\Dwg[\wa'][u,1] = \{t \in \hori \mid u_{\wa[1]}(t)=0\} = \mathfrak T_{\wa'[1]}$ (cf.~\eqref{eq: DefTimesU=0}),
  			we have for all $t \in \Dwg[\wa'][u,1]$ that
  			$\wttime_{\wa'}(u,t) + \Pf^{{\prices^*}}_{\wa'}(u,t) \geq \prices^*_{\wa'[1]}(t) \geq \phi_i^{\source_i}(t)$ where the last inequality holds by the definition of $\prices^*$ in \eqref{eq: Defp^*} and $t \in \Dwg[\wa'][u,1] = \mathfrak T_{\wa'[1]}$. 
  			This shows  that $\mathfrak D_{\wa'<}^1= \emptyset$ which is a contradiction to $\leb(\mathfrak D_{\wa'<}^1)>0$. 
  		\end{structuredproof}
  		
  		For $j \in\{2,\ldots, |\wa'|\}$,  we first give a brief sketch of proof: 
  		The overall idea is to construct an \stwalk[\source_i][\dest] $\wa^* \in \Routes_i$ that lower bounds the costs of $\wa'$ and only contains used edges under $u$. 
  		Since, by definition, $\phi^{\source_i}_i$ is upper bounded by the costs of such a walk, we get the desired contradiction. 
  		Now in order to construct this walk $\wa^*$, 
  		we start by considering the edge $\wa'[j-1]$. 
  		By definition of $\Dwg[\wa'<][u,j]$, this edge $\wa'[j-1]$ carries flow $u_{\wa'[j-1]}>0$ during the arrival times corresponding to the entry times $\Dwg[\wa'<][u,j]$ when traveling along the walk $\wa'$. 
  		Hence, by $h^*$ inducing $u$, we can find a walk $\wa \in \Routes$ (called $\wa^m$ below) for which $h^*$ sends flow into this walk such that it induces (a nontrivial part of) the aforementioned flow on $\wa'[j-1]$. 
  		This, in turn, implies that $\phi^{v_j}_i$  is upper bounded  by the costs of the suffix of the walk $\wa$ that starts at the node~$v_j$ between the edges~$\wa'[j-1]$ and~$\wa'[j]$. From this, we can find a walk $\bar{\wa} \in \Routes_{v_j,\dest}^i$ whose costs (almost) match the costs of $\phi_i^{v_j}$ and, hence, has lower costs than the costs of entering the edge $\wa'[j]$ when traveling along $\wa'$ and starting during $\Dwg[\wa'][u,j]$. Hence, setting $\wa^*:=(\wa'_{< j},\bar{\wa},i)$ yields the desired walk. 
  		
  		\begin{structuredproof}
  			\proofitem{Case $j > 1$} 
  			By \Cref{lem:  g>0Arr'>0}, 
  			$\mathfrak T_{\wa'<}^j:= \arr_{\wa',j-1}(u,\cdot)(\Dwg[\wa'<][u,j])$ is not a null set as $\Dwg[\wa'<][u,j] \subseteq \Dwg[\wa'_{\leq j-1}][u]$ and by assumption $\Dwg[\wa'<][u,j]$ is not a null set. 
  			Since  $u_{\wa'[j-1]}(t) >0$ for all  $t \in \mathfrak T_{\wa'<}^j$,  we can apply \Cref{lem: Relations:h>0u>0:u>0ExistsCountableM} to $h^*$ and $u$ and get a countable set $M$ with  $\wa^m,j_m,\mathfrak D_m, m \in M$ with the stated properties. 
  			By the countability of $M$, there has to exist an $m \in M$ with $\wa^m \in \Routes$  such that $\arr_{\wa^m,j_m}(u,\cdot)(\mathfrak D_m)$ and 
  			$\mathfrak T_{\wa'<}^j$ have an intersection with positive measure. 
  			Denote the set of arrival times of the time points in the latter set at the end of the $j_m$-th edge of $\wa^m$
  			by
  			\begin{align}\label{eq: Desc: T^*}
  				\mathfrak T^* := \exit_{\wa^m[j_m]}(u,\cdot)(\arr_{\wa^m,j_m}(u,\cdot)(\mathfrak D_m) \cap \mathfrak T_{\wa'<}^j) = \arr_{\wa^m_{\geq j_m},2}(u,\cdot)(\arr_{\wa^m,j_m}(u,\cdot)(\mathfrak D_m) \cap \mathfrak T_{\wa'<}^j)
  			\end{align}
  			We can assume \wlg that $h^*_{\wa^m}(t)>0$ for all $t \in \mathfrak D_m$. Then, by 
  			the choice of our representatives, we have 
  			$\mathfrak T^* \subseteq \Dwg[\wa^m_{>j_m}][u]$. This implies two things: Firstly, since $\arr_{\wa^m,j_m}(u,\cdot)(\mathfrak D_m)$ and 
  			$\mathfrak T_{\wa'<}^j$ have an intersection with positive measure,  
  			\Cref{lem:  g>0Arr'>0} in combination with the equality in \eqref{eq: Desc: T^*} 
  			 implies that $\mathfrak T^*$  is not a null set.  
  			Secondly, from the definition of~$\phi^{v_j}_i(t)$, we get $\phi^{v_j}_i(t) \leq \wttime_{\wa^m_{>j_m}}(u,t) + \Pf^{\hat{\prices}}_{\wa^m_{>j_m}}(u,t)$ 
  			for all $t \in \mathfrak T^*$ where $v_j$ is the end note of edge $\wa^m[j_m] = \wa'[j-1]$. 
  			
  			By the countability of the set of finite $v_j,\dest$-walks, 
  			there exists a \stwalk[v_j][\dest] $\bar{\wa}\in \Routes_{v_j,\dest}^i$ and subset $\hat{\mathfrak T}^*\subseteq \mathfrak T^*$ of positive measure  with   
  			$1 + \phi^{v_j}_i \geq  \Indi_{\hori\setminus\Dwg[\bar{\wa}][u]}\cdot \infty + (\wttime_{\bar{\wa}}(u,\cdot) + \Pf^{\hat{\prices}}_{\bar{\wa}}(u,\cdot))$ on 
  			$  \hat{\mathfrak T}^*$. Due to $\phi^{v_j}_i$ being finitely valued,  it must be that
  			$\hat{\mathfrak T}^* \subseteq \Dwg[\bar{\wa}][u]$ and hence 
  			$1 + \phi^{v_j}_i \geq \wttime_{\bar{\wa}}(u,\cdot) + \Pf^{\hat{\prices}}_{\bar{\wa}}(u,\cdot)$ on 
  			$  \hat{\mathfrak T}^*$.

  			Define $\wa^*:=(\wa'_{< j},\bar{\wa},i) \in \Routes_i$ and  $\mathfrak D^* := \arr_{\wa',j}(u,\cdot)^{-1}(\hat{\mathfrak T}^*) \cap \Dwg[\wa'<][u,j]$. 
  			Remark that $\wa^*$ is indeed an \stwalk[\source_i][\dest] as $\wa'$ is an  \stwalk[\source_i][\dest] and $\bar{\wa}$ is a \stwalk[v_j][\dest]. 
  			\begin{subclaim}\label{subclaim: CorrectTimes}
  				We have $\mathfrak D^*\subseteq \Dwg[\wa^*][u]$ as well as $\leb(\mathfrak D^*)>0$. In particular, 
  				by definition of~$\phi_i^{\source_i}$, the inequality  $\wttime_{\wa^*}(u,t)+ \Pf^{\hat{\prices}}_{\wa^*}(u,t) \geq \phi_i^{\source_i}(t)$ holds for  all $t \in \mathfrak D^*$.  
  			\end{subclaim}
  			\begin{proofClaim}
  				We first argue for the claimed inclusion. For this, observe that $\mathfrak D^* \subseteq \Dwg[\wa'<][u,j]\subseteq \Dwg[\wa'_{<j}][u]$ and, hence, 
  				$\mathfrak D^* \subseteq \arr_{\wa',j}(u,\cdot)^{-1}(\hat{\mathfrak T}^*)$ implies $\mathfrak D^*\subseteq \Dwg[\wa^*][u]$ by $\hat{\mathfrak T}^* \subseteq   \Dwg[\bar{\wa}][u]$.  
  				
  				For $\leb(\mathfrak D^*)>0$, we show in the following that $\arr_{\wa',j}(u,\cdot)(\mathfrak D^*) = \hat{\mathfrak T}^*$ holds. 
  				This shows the desired inequality as $\arr_{\wa',j}(u,\cdot)$ fulfills Lusin's property (\Cref{lem: PropAbsCon:Lus}) and $\leb(\hat{\mathfrak T}^*)>0$. 
  				
  				In order to show that  $\arr_{\wa',j}(u,\cdot)(\mathfrak D^*) = \hat{\mathfrak T}^*$ holds, it is sufficient to prove that the inclusion
  				$\hat{\mathfrak T}^*  \subseteq \arr_{\wa',j}(u,\cdot)(\Dwg[\wa'<][u,j])$
  				is valid: 
  				\begin{align*}
  					\hat{\mathfrak T}^* \subseteq \mathfrak T^* &\eqperdef \exit_{\wa^m[j_m]}(u,\cdot)(\arr_{\wa^m,j_m}(u,\cdot)(\mathfrak D_m) \cap \mathfrak T_{\wa'<}^j) =  \exit_{\wa'[j-1]}(u,\cdot)(\arr_{\wa^m,j_m}(u,\cdot)(\mathfrak D_m) \cap \mathfrak T_{\wa'<}^j) \\
  					&\subseteq \exit_{\wa'[j-1]}(u,\cdot) (\mathfrak T_{\wa'<}^j) \eqperdef \exit_{\wa'[j-1]}(u,\cdot) (\arr_{\wa',j-1}(u,\cdot)(\Dwg[\wa'<][u,j]) ) \eqperdef \arr_{\wa',j}(u,\cdot)(\Dwg[\wa'<][u,j]).\qedhere
  				\end{align*}    
  			\end{proofClaim}
  			The desired contradiction then follows almost immediately by the next subclaim: 
  			\begin{subclaim}\label{subclaim: inequality}
  				The inequality  $\wttime_{\wa'}(u,t) + \Pf^{{\prices^*}}_{\wa'}(u,t) \geq  \wttime_{\wa^*}(u,t) + \Pf^{\hat{\prices}}_{\wa^*}(u,t)$ holds for all $t \in \mathfrak D^*$.
  			\end{subclaim}
  			\begin{proofClaim} 
  				We make two observations: Firstly,
  				\begin{align*}
  					\wttime_{\wa'_{< j}}(u,t) + \Pf^{{\prices^*}}_{\wa'_{< j}}(u,t) =  \wttime_{\wa^*_{< j}}(u,t) + \Pf^{{\prices^*}}_{\wa^*_{<j}}(u,t)= \wttime_{\wa^*_{< j}}(u,t) + \Pf^{\hat{\prices}}_{\wa^*_{<j}}(u,t) 
  				\end{align*}
  				  for all $t \in \Dwg[\wa'][u,j]$ where the first equality holds since $\wa'_{<j} = \wa^*_{<j}$ and the second one due to 
  				$\Dwg[\wa'][u,j]\subseteq \Dwg[\wa'_{<j}][u]$ and the fact that $\hat{\prices}_\arc(t) = \prices^*_\arc(t)$ in case of $u_\arc(t) >0$. 
  				
  				Secondly, we have $\hat{\mathfrak T}^*=\arr_{\wa',j}(u,\cdot)(\mathfrak D^*) \subseteq \arr_{\wa',j}(u,\cdot)(\Dwg[\wa'<][u,j])  \subseteq \mathfrak T_{\wa'[j]}$ (where the first equality was shown in the proof of \Cref{subclaim: CorrectTimes}), implying the inequality \refsym{1} in the following for all $t \in \hat{\mathfrak T}^*$:
  				\begin{align*}
  					\wttime_{\wa'_{\geq j}}(u,t) + \Pf^{{\prices^*}}_{\wa'_{\geq j}}(u,t) \geq \prices^*_{\wa'[j]}(t) \symoverset{1}{\geq} 1+\phi_i^{v_j}(t) &\geq  \wttime_{\bar{\wa}}(u,t) + \Pf^{\hat{\prices}}_{\bar{\wa}}(u,t) \\\overset{\wa^*_{\geq j}=\bar{\wa}}&{=} \wttime_{\wa^*_{\geq j}}(u,t) + \Pf^{\hat{\prices}}_{\wa^*_{\geq j}}(u,t).
  				\end{align*}
  				Since  $\arr_{\wa',j}(u,\cdot)(\mathfrak D^*) = \hat{\mathfrak T}^*$ (cf.~above), 
  				putting both observations together implies 
  				the claim.
  			\end{proofClaim}
  			
  			This now yields the desired contradiction as we get by \Cref{subclaim: CorrectTimes} and \Cref{subclaim: inequality} that 
  			\begin{align*}
  				\wttime_{\wa'}(u,t) + \Pf^{{\prices}}_{\wa'}(u,t) \geq   \wttime_{\wa^*}(u,t)+ \Pf^{\hat{\prices}}_{\wa^*}(u,t) \geq \phi_i^{\source_i}(t) 
  			\end{align*}
  			holds true for all $t \in \mathfrak D^*$.
  			However, $\mathfrak D^*$ is also a subset of~$\Dwg[\wa'<][u,j]$ and, hence, can only contain times where the opposite (and strict) inequality to the above one holds. Thus, $\mathfrak D^*$ must be empty which is a contradiction to $\leb(\mathfrak D^*)>0$. \qedhere
  		\end{structuredproof}
  		\end{proofClaim}

  	With these insights, we can now show that $h^*$ fulfills the Wardrop conditions of a $\prices^*$-DUE:
  	\begin{claim}\label{claim: hateqcon}
  		$h^*$ fulfills for almost all $t \in \hori$, all  $i \in I$ and $\wa \in \Routes_i$  
  		\begin{equation}\label{eq : DUEhat}
  			\begin{aligned}
  				h^*_\wa(t)>0 \implies \wttime_\wa(u,t) + \Pf^{\prices^*}_\wa(u,t) \leq \wttime_{\wa'}(u,t) + \Pf^{\prices^*}_{\wa'}(u,t)\text{ for all } \wa'\in \Routes_i.
  			\end{aligned}
  		\end{equation}
  	\end{claim}
  	
  	\begin{proofClaim}   
  		Assume for the sake of a contradiction that  the claim was wrong.
  		Then, 
  		by the countability of $\Routes_i\times \Routes_i, i \in I$, 
  		there exists a commodity~$i \in I$, walks $\wa,\wa' \in \Routes_i$ and a set $\mathfrak D_{\wa>} \in\mathcal{B}(\hori)$ with $\leb(\mathfrak D_{\wa>})> 0$ such that 
  		\begin{align}\label{eq: DefDwa>2}
  		 h^*_\wa(t)>0 \text{ and } \wttime_\wa(u,t) + \Pf^{\prices^*}_\wa(u,t) > \wttime_{\wa'}(u,t)+ \Pf^{\prices^*}_{\wa'}(u,t) \text{ for almost all } t \in \mathfrak D_{\wa>}. 
  		\end{align}
  		We can assume \wlg that $h^*_\wa(t)>0$ for all $t \in \mathfrak D_{\wa>}$.
  		
  		Define the walk inflow rate function ${h} \in L_+(\hori)^\Routes$ by shifting all inflow into~$\wa$ during~$\mathfrak D_{\wa>}$ to~$\wa'$:
  		\[{h}_{\Tilde{\wa}} \coloneqq \begin{cases}
  			h^*_{\wa'} + h^*_{\wa} \cdot \Indi_{\mathfrak D_{\wa>}}   &\text{if } \Tilde{\wa}=\wa' ,\\
  			h^*_{\wa} - h^*_{\wa} \cdot \Indi_{\mathfrak D_{\wa>}}     &\text{if } \Tilde{\wa}=\wa ,\\
  			h^*_{\tilde{\wa}} &\text{else.}
  		\end{cases}\]
  		We argue in the following that ${h} \in \edom{\Routes}[u]\cap \wir$. 
  		It is clear that ${h}$ is measurable and ${h} \in \wir$. 
  		We require the following claim: 
  		\begin{subclaim}
  			$\ell^u_{\tilde{\wa}}(h_{\tilde{\wa}})$ exists for all $\tilde{\wa} \in \Routes$.
  		\end{subclaim}
  		\begin{proofClaim}
  			The existence of  $\ell^u_{\tilde{\wa}}({h}_{\tilde{\wa}})$ for $\tilde{\wa}\neq \wa'$ follows immediately by the existence of $\ell^u_{\tilde{\wa}}(h^*_{\tilde\wa})$. Remark that $\ell^u_{\tilde{\wa}}(h^*_{\tilde{\wa}}),\tilde{\wa} \in \Routes$ exist  by $\ell^u(h^*)$ existing and  $h^* \in L_+(\hori)^\Routes$, cf.~\Cref{lem: elluExistenceProperties:ExistenceInducedFlow}. 
  			
  			Hence, consider the case of $\wa'$, an arbitrary $j^*\leq |\wa'|$ and a null set $\mathfrak T^0$. By \Cref{lem: elluExistenceProperties}, it is sufficient 
  			to show that ${h}_{\wa'} = 0$ almost everywhere on $\mathfrak D_{j^*} := \arr_{\wa',j^*}(u,\cdot)^{-1}(\mathfrak T^0)$. 
  			By $\ell^u_{\wa'}(h^*)$ existing, it is clear that $h^*_{\wa'}  = 0$ on $\mathfrak D_{j^*}$ and since ${h}_{\wa'} = h^*_{\wa'}$ outside of $\mathfrak D_{\wa>}$, 
  			it is sufficient to show that $\mathfrak D_{j^*}^\cap:= \mathfrak D_{j^*} \cap \mathfrak D_{\wa>}$ is a null set. In order to show this, we argue first that   the following  holds for almost all $t \in \mathfrak D_{j^*}^\cap$:
  			\begin{align}\label{eq:claim: hateqcon}
  				\wttime_{\wa'}(u,t)+ \Pf^{\prices^*}_{\wa'}(u,t) \symoverset{2}{\geq} \phi_i^{\source_i}(t)  \symoverset{1}{=} \wttime_\wa(u,t) + \Pf^{\hat{\prices}}_{\wa}(u,t) \symoverset{3}{=} \wttime_\wa(u,t) + \Pf^{\prices^*}_{\wa}(u,t) . 
  			\end{align}
  			For the above, we start by observing that $h^*_\wa(t) >0$ for every  $t \in \mathfrak D_{j^*}^\cap \subseteq \mathfrak D_{\wa>}$. The equality \refsym{1} then follows  from \Cref{claim: PhiEqualsHatWalks}.
  			The equality \refsym{3} holds since we only changed $\hat{\prices}_\arc(\tilde{t})$ when $u_\arc(\tilde{t}) = 0$ and our choice of the representatives of $h^*$ and $u$ ensures that $u_{\wa[j]}(\arr_{\wa,j}(u,t))>0$ for all $t \in \mathfrak D_{j^*}^\cap $ by
  			$h^*_\wa(t) >0$.  For the inequality \refsym{2}, we distinguish between two cases: 
  			First, if $t \in \Dwg[\wa'][u,j]$ for any $j \leq \abs{\wa'}$, the inequality holds by \Cref{claim: ineqphi}.
  			For the remaining case of $t \in \mathfrak D_{j^*}^\cap\setminus\bigcup_{j\leq|\wa'|}\Dwg[\wa'][u,j]=  \mathfrak D_{j^*}^\cap \cap \Dwg[\wa'][u]$, we get $\wttime_{\wa'}(u,t)+ \Pf^{\prices^*}_{\wa'}(u,t)=\wttime_{\wa'}(u,t)+ \Pf^{\hat{\prices}}_{\wa'}(u,t)$ since, again, we did not changes the tolls along $\wa'$ when starting at a $t \in \Dwg[\wa'][u]$. The inequality \refsym{2} then follows by the definition of~$\phi^{\source_i}_i$.
  			
  			As the costs of $\wa'$ are strictly lower than those of~$\wa$ for almost all $t \in \mathfrak D_{\wa>}$ (by the definition of that set), \eqref{eq:claim: hateqcon} implies that $\mathfrak D_{j^*}^\cap \subseteq \mathfrak D_{\wa>}$ is indeed a null set. Thus, the claim is proven.
  		\end{proofClaim}
  		In order to derive from this that $\ell^u(h)$ exists and hence $h \in \edom{\Routes}[u]$ holds, it remains by \Cref{lem: elluExistenceProperties:ExistenceInducedFlow} to observe   that also the sum $\sum_{{\tilde{\wa}} \in \Routes}\ell^u_{\tilde{\wa}}({h}_{\tilde{\wa}})$ exists: This is due to the latter sum for $h^*$ existing as well as $(\ell^u_{\tilde{\wa}}({h}_{\tilde{\wa}}))_{\tilde{\wa} \in \Routes}$ and $(\ell^u_{\tilde{\wa}}(h^*_{\tilde{\wa}}))_{\tilde{\wa} \in \Routes}$ only differing in the two entries corresponding to $\wa,\wa'$.

  				Next, we observe that 
  		\begin{align*}
  			\sum_{{\tilde{\wa}} \in \Routes} &\dup{\wttime_{\tilde{\wa}}(u,\cdot) + \Pf^{\prices^*}_{\tilde{\wa}}(u,\cdot)}{h^*_{\tilde{\wa}}}  
  			=\sum_{\tilde{\wa} \in \Routes\setminus \{\wa',\wa\} } \dup{\wttime_{\tilde{\wa}}(u,\cdot) +
  				\Pf^{\prices^*}_{\tilde{\wa}}(u,\cdot)}{h_{\tilde{\wa}}} \\ 
  			+ &\dup{\wttime_{{\wa}}(u,\cdot) + \Pf^{\prices^*}_{{\wa}}(u,\cdot)}{h_{{\wa}}+h^*_{\wa}\cdot \Indi_{\mathfrak{ D}_{\wa>}}}  
  			+\dup{\wttime_{{\wa'}}(u,\cdot) +\Pf^{\prices^*}_{{\wa'}}(u,\cdot)}{h_{{\wa'}}-h^*_{\wa}\cdot \Indi_{\mathfrak{ D}_{\wa>}}}\\
  			\symoverset{2}{=}&\sum_{\tilde{\wa} \in \Routes} \dup{\wttime_{\tilde{\wa}}(u,\cdot) +
  				\Pf^{\prices^*}_{\tilde{\wa}}(u,\cdot)}{h_{\tilde{\wa}}} + \dup{\wttime_{\wa}(u,\cdot) + \Pf^{\prices^*}_{{\wa}}(u,\cdot) - \wttime_{\wa'}(u,\cdot) - \Pf^{\prices^*}_{{\wa}'}(u,\cdot)  }{h^*_{\wa }\cdot \Indi_{\mathfrak D_{\wa>}}}\\
  			\symoverset{1}{>}&\sum_{\tilde{\wa} \in \Routes } \dup{\wttime_{\tilde{\wa}}(u,\cdot) +
  				\Pf^{\prices^*}_{\tilde{\wa}}(u,\cdot)}{h_{\tilde{\wa}}} .
  		\end{align*}
  		Here, we used for \refsym{2} that for $\tilde{\wa} \in \{\wa',\wa\}$ we have
  		\begin{align*}
  			\dup{\wttime_{\tilde{\wa}}(u,\cdot) + \Pf^{\prices^*}_{\tilde{\wa}}(u,\cdot)}{h^*_{\wa}\cdot \Indi_{\mathfrak{ D}_{\wa>}}} \overset{\eqref{eq: DefDwa>2}}&{\leq}\dup{\wttime_{\wa}(u,\cdot) + \Pf^{\prices^*}_{{\wa}}(u,\cdot)}{h^*_{\wa}\cdot \Indi_{\mathfrak{ D}_{\wa>}}}\\ &\leq 	\sum_{{\tilde{\wa}} \in \Routes} \dup{\wttime_{\tilde{\wa}}(u,\cdot) + \Pf^{\prices^*}_{\tilde{\wa}}(u,\cdot)}{h^*_{\tilde{\wa}}}  < \infty.
  		\end{align*}
  		For \refsym{1}, we used the definition of $\mathfrak D_{\wa>}$ (i.e.~\eqref{eq: DefDwa>2} together with $\leb(\mathfrak D_{\wa>})>0$) and the finiteness of  $	\sum_{{\tilde{\wa}} \in \Routes} \dup{\wttime_{\tilde{\wa}}(u,\cdot) + \Pf^{\prices^*}_{\tilde{\wa}}(u,\cdot)}{h^*_{\tilde{\wa}}} $. 
  		Hence, we arrive at the desired contradiction since we already argued that~\eqref{eq: ZeroDGhatwalk} 
  		also holds for $\prices^*$. 
  	\end{proofClaim} 
  	Thus, we have shown that $h^*$ fulfills the $\prices^*$-DUE Wardrop condition. Since $h^*$ also induces $u$ under $\Nl$ by \Cref{ass: u} and $\ell^u(h^*) = u$ being valid, this concludes the proof.\qedhere 
  \end{proof}

\subsection{Combinatorial Characterization}\label{sec: SingleSinkImpl}
In this section, we consider the \emph{multi-source, single-sink} case $\dest_i = \dest,\, i \in I$.  
For this setting, we provide, in addition to the duality-based characterization in \Cref{thm: mainMSSS}, a \emph{combinatorial characterization} of implementable vectors~$u$.  
For this result, we require, besides the separability of the private cost functions (\Cref{ass: PCSep}), two additional assumptions:  
first, any $\dest$-cycle that admits zero costs for one commodity must do so for all commodities; and second, any cycle not containing the destination and having zero traversal time must also have zero private costs for all commodities.
\begin{assumption}\label{ass: PCSep+ZeroCycles}\hfill 
\begin{thmparts}
    \item \Cref{ass: PCSep} holds, i.e.: The private costs are separable over the edges, that is, for all $i \in I$ there exists a  measurable function  $\psi^i : \hori \to \R^\GA_+$ such that for all $\wa =(\hat{\wa},i) \in \Routes_i$
    \begin{align*}
    \wttime_{(\hat{\wa},i)}(u,t) =  \sum_{j \leq|\hat{\wa}|} \psi^i_{\hat{\wa}[j]}(\arr_{\hat{\wa},j}(u,t))  .
    \end{align*}\label[thmpart]{ass: PCSep+ZeroCycles: Sep}
    \item For all $\dest$-cycles $c\in \DestCyc$, $i\in I$ and  all $t \in \hori$ we have the implication 
    \begin{align*}
        \wttime_{(c,i)}(u,t)= 0 \implies  \wttime_{(c,i')}(u,t)= 0 \text{ for all } i' \in I. 
    \end{align*} 
    \label[thmpart]{ass: PCSep+ZeroCycles: DestCycles}

    \item For all cycles $c \in \SimpCyc\setminus\DestCyc$, $i \in I$ and almost all $t \in \hori$,  the following holds:    
    \begin{align*}
        \Big( u_{c[j']}(t)>0  \,\land\, \forall j'\leq \abs{c} : \trav_{c[j']}(u,t) = 0\Big)  \implies \wttime_{(c,i)}(u,t) = 0.  
    \end{align*}
    \label[thmpart]{ass: PCSep+ZeroCycles: ZeroCycles}

\end{thmparts}
  
\end{assumption}
Under the above assumption, we will simply write $\psi_\arc(t)=0$ as a shorthand for $\psi^i_\arc(u,t)=0$ for all  $i \in I$.
Similarly, for any $\hat{\wa}\in \hat{\Routes}$, we will write $\wttime_{\hat{\wa}}(u,t) = 0$ meaning $\wttime_{(\hat{\wa},i)}(u,t) = 0$ for all $i\in I$. 
We remark that the second  part of the above assumption is in particular fulfilled if 
$\psi^i$ is a commodity specific weighted variant of a common private cost function $\psi:\R \to \R^\GA_+$, i.e.
\begin{align}\label{eq: WeightedVarian}
    \psi^i=\phi_i \cdot \psi \text{ for some measurable weight functions }\phi_i:\R \to \R_+\setminus\{0\} \text{ for all } i \in I.
\end{align} 
In particular,  the entire \Cref{ass: PCSep+ZeroCycles} is fulfilled if the private costs represent weighted travel times, i.e.~if \eqref{eq: PC=WTT} holds.

 \begin{theorem}\label{thm: mainSingleSink}
 	For a multi-source, single-\sink network with \Cref{ass: u,ass: PCSep+ZeroCycles} being fulfilled, \eqref{opt: Master} being \wellposed and admitting strong duality \wrt $\MeasFuncUInt$,  the following statements are equivalent:  
 	\begin{thmparts}
 		\item $u$ is implementable. \label[thmpart]{thm: mainSingleSink: impl}

 		\item Every used edge that is reachable from $\dest$ via used edges  has zero private costs, i.e.\ for an arbitrary representative of $u$ we have the following implication for almost all $t \in \hori$ and all walks $\wa \in \hat{\Routes}^\dest$  and $j \leq |\wa|$: 
 		\begin{align*}
 			u_{\wa[j']}(\arr_{\wa,j'}(u,t)) > 0 \text{ for all } j'\leq j  \implies \psi_{\wa[j]}(\arr_{\wa,j}(u,t)) = 0.
 		\end{align*}\label[thmpart]{thm: mainSingleSink: ReachableEdgesHaveZeroD}
 		\item $u$ only sends flow along a $\dest$-cycle if the latter has zero private costs, i.e.\ for an arbitrary representative of $u$ we have the following implication for almost all $t \in \hori$ and  all $\dest$-cycles $c\in \DestCyc$: 
 		\begin{align*}
 			u_{c[j]}(\arr_{c,j}(u,t)) >0  \text{ for all } j \leq \abs{c}   \implies \wttime_{c}(u,t) = 0 .
 		\end{align*}
 		\label[thmpart]{thm: mainSingleSink: UsedDCyclesHaveZeroD}
 		\item There exists $h^* \in \wir$ optimal for \eqref{opt: Master} with tight inequality \eqref{ineq: Master}. \label[thmpart]{thm: mainSingleSink: Optimalh}
 	\end{thmparts}
 \end{theorem}
 We remark that if one representative of $u$ fulfills \ref{thm: mainSingleSink: ReachableEdgesHaveZeroD} (resp.~\ref{thm: mainSingleSink: UsedDCyclesHaveZeroD}), then all representatives of $u$ do so. Moreover, 
 it is sufficient to check the condition in \ref{thm: mainSingleSink: UsedDCyclesHaveZeroD} only for $\dest$-cycles visiting $\dest$ once. 
 Both of these claims are proven in \Cref{lem: SimplificationsFlowOnlyOnZeroDestCycles} after the proof of \Cref{thm: mainSingleSink}.  
 
 Before we come to the actual proof of the latter \namecref{thm: mainSingleSink}, let us briefly sketch the main proof ideas. 
 We will prove  \Cref{thm: mainSingleSink}   by showing the chain of implications 
 \begin{align}\label{eq: ChainOfImplications}
 	\text{\ref{thm: mainSingleSink: Optimalh} }
 	\implies \text{ \ref{thm: mainSingleSink: impl} }
 	\implies \text{ \ref{thm: mainSingleSink: ReachableEdgesHaveZeroD} }
 	\implies \text{ \ref{thm: mainSingleSink: UsedDCyclesHaveZeroD} }
 	\implies \text{ \ref{thm: mainSingleSink: Optimalh}}.
 \end{align}

 The implication \ref{thm: mainSingleSink: Optimalh}$\Rightarrow$\ref{thm: mainSingleSink: impl} holds by \Cref{thm: SuffConMSSS}.

 Regarding the implication \ref{thm: mainSingleSink: impl}$\Rightarrow$\ref{thm: mainSingleSink: ReachableEdgesHaveZeroD}, we start with a pair $(h^*,\prices)$ implementing $u$. 
 We then take any walk starting at the destination~$\dest$ and consisting only of used edges under~$u$ and show for each of those edges, by induction over their position, that they all have a total cost of zero (under~$\prices$). In particular, they need to have zero private costs.  The proof-idea for the induction is illustrated in \Cref{fig:thm:main:part1:induction}.

 The implication \ref{thm: mainSingleSink: ReachableEdgesHaveZeroD}$\Rightarrow$\ref{thm: mainSingleSink: UsedDCyclesHaveZeroD} 
 is straight forward.
 
 Finally, we argue for \ref{thm: mainSingleSink: UsedDCyclesHaveZeroD}$\Rightarrow$\ref{thm: mainSingleSink: Optimalh} as follows: By  
 \Cref{lem: DifferenceGeneral}, 
 the aggregated edge flow $\tilde{\g}$ of any optimal solution $\Tilde{h}$ to \eqref{opt: Master} differs from $u$ only in flow on zero-cycles or $\dest$-cycles. 
 The task then becomes to show that we can add those cycles to the walk flow $\tilde{h}$ where 
 we will rely heavily on  
 the insights gathered in \Cref{thm: PureFlowDecompIntuitive} and \Cref{cor: PureFlowDecomp} about pure  flow decompositions.   
 The so constructed walk inflow rate~$h$ then induces $u$. Since, compared to the optimal solution $\tilde{h}$,  we only added cycles of travel time zero or $\dest$-cycles which by assumption have zero private costs, the optimality of~$h$ follows.  
  \begin{proof}[Proof of \Cref{thm: mainSingleSink}]   
  	We show the chain of implications \eqref{eq: ChainOfImplications} in this order. As already remarked above, 
  	the implication \ref{thm: mainSingleSink: Optimalh}$\Rightarrow$\ref{thm: mainSingleSink: impl} was shown in \Cref{thm: SuffConMSSS}.
  	
  	\begin{structuredproof}
  		
  		\proofitem{\ref{thm: mainSingleSink: impl}$\Rightarrow$\ref{thm: mainSingleSink: ReachableEdgesHaveZeroD}}
  		Let $u$ be implementable via nonnegative tolls $\prices:\hori\to\R^\GA_+$ and $h^* \in \wir$, i.e.\ $h^*$ induces $u$ and is a $\prices$-DUE. Choose an arbitrary representative of $u$. 
  		We aim to show the following claim: 
  		\begin{claim}\label{claim:thm:main:part1}
  			For all walks $\wa^\dest =(\arc_1,\ldots,\arc_m) \in \hat{\Routes
  			}^\dest$ and all sets of times $\mathfrak T \in \mathcal{B}(\hori)$
  			such that $u_{\arc_j}(\arr_{\wa^\dest,j}(u,t)) >0$ for all  $j \leq m$ and almost all $t \in \mathfrak T$, we have for all $j\leq m$:
  			\begin{align}\label{eq:thm:main:part1}
  				\psi_{\arc_{j}}(\arr_{\wa^\dest,j}(u,t)) + \prices_{\arc_{j}}(\arr_{\wa^\dest,j}(u,t)) = 0 \text{ for almost all } t\in \mathfrak T.
  			\end{align}
  		\end{claim}
  		\begin{proofClaim}
  			
  			Consider an arbitrary $\wa^\dest$ and $\mathfrak T$ as described.
  			We prove the claim via an induction over $j$ with the following (stronger) induction claim:
  			\begin{proofbyinduction}
  				\inductionclaim For every $j\leq m$, $i \in I$ and every walk $\wa \in \Routes_i$ with $\wa[j_\wa] = \arc_j$ for some $j_\wa \leq |\wa|$, the implication 
  				\begin{align}\label{eq: ImplicationSubclaim:thm:main:part1}
  					h^*_\wa(t) > 0 \implies \psi_{\wa[j']}(\arr_{\wa,j'}(u,t)) + \prices_{\wa[j']}(\arr_{\wa,j'}(u,t)) = 0 \text{ for all } j'\in\{j_\wa,\ldots,|\wa|\} 
  				\end{align}
  				holds
  				for almost every  $t \in \arr_{\wa,j_\wa}(u,\cdot)^{-1}\big(\arr_{\wa^\dest,j}(u,\cdot)(\mathfrak T)\big)$. 
  			\end{proofbyinduction} 
  				
  			We first show that if the induction claim holds for some fixed $j \leq m$, then so does \eqref{eq:thm:main:part1} (for the same~$j$). Hence, the induction claim (for all $j \leq m$) implies the statement of \Cref{claim:thm:main:part1}. Moreover, in the proof of the induction step, we will be able to use not only the induction claim itself but also~\eqref{eq:thm:main:part1} for smaller~$j$.
  			
  			Assume for the sake of a contradiction that, for some fixed $j \leq m$, the induction claim holds while \eqref{eq:thm:main:part1} does not. Then, there exists $i \in I$,   $\mathfrak D^>_j \subseteq \mathfrak T, \leb(\mathfrak D^>_j )> 0$ with $\psi^i_{\arc_j}(\arr_{\wa^\dest,j}(u,t))+ \prices_{\arc_j}(\arr_{\wa^\dest,j}(u,t)) >  0$ for all $t \in \mathfrak D^>_j$. 
  			Define $\mathfrak T_j \coloneq  \arr_{\wa^\dest,j}(u,\cdot)(\mathfrak D^>_j)$. 
  			Since $u_{\arc_j}(\arr_{\wa^\dest,j}(u,t))  >0$ for almost all $ t\in \mathfrak T\supseteq \mathfrak D_j^>$ by assumption, 
  			we also have  $u_{\arc_j}(t)  >0$ for almost all $ t\in \mathfrak T_j$ by $\arr_{\wa^\dest,j}(u,\cdot)$ having Lusin's property. 
  			Then, \Cref{lem: Relations:h>0u>0:u>0ExistsCountableM} guarantees the existence of a walk $\wa \in \Routes, {j_\wa} \leq |\wa|$ with $\wa[{j_\wa}] = \arc_j$ and   ${\mathfrak D^\wa}\subseteq \arr_{\wa,{j_\wa}}(u,\cdot)^{-1}(\mathfrak T_j) $ with $\leb( {\mathfrak D^\wa}) > 0$ and 
  			$h^*_\wa(t)> 0$ for almost every $t \in \mathfrak D^\wa$. 
  			But this implies, by the induction claim, that for almost all $t \in {\mathfrak D^\wa}$, we have  
  			\begin{align*}
  				0 = \psi_{\wa[{{j_\wa}}]}(\arr_{\wa,{j_\wa}}(u,t)) + \prices_{\wa[j_\wa]}(\arr_{\wa,j_\wa}(u,t))= \psi_{\arc_j}(\arr_{\wa,{j_\wa}}(u,t)) +   \prices_{\arc_j}(\arr_{\wa,j_\wa}(u,t)) .
  			\end{align*}
  			Yet, $\arr_{\wa,j_\wa}(u,t) \in \mathfrak T_j$ and by definition of $\mathfrak T_j$ and $\mathfrak D_j^>$, we have 
  			$\psi^i_{\arc_j}(t')+ \prices_{\arc_j}(t') >  0$ for all $t' \in \mathfrak T_j$. 
  			This yields a contradiction as   $\mathfrak D^\wa \neq \emptyset$ by $\leb(\mathfrak D^\wa)> 0$.
  			
  			\begin{proofbyinduction}  				
  				\basecase{$j=1$}
  				We start by showing that the desired statement holds with $\psi^i$ instead of $\psi$ on the right-hand side of \eqref{eq: ImplicationSubclaim:thm:main:part1}. 
  				
  				For this, assume for the sake of a contradiction that there exists $i \in I$ and $\wa=(\hat{\wa},i) \in \Routes_i$, $j^\star \in\{{j_\wa},\ldots,|\wa|\}$ and $\mathfrak D \subseteq\arr_{\wa,j_\wa}(u,\cdot)^{-1}\big(\arr_{\wa^\dest,1}(u,\cdot)(\mathfrak T)\big) = \arr_{\wa,{j_\wa}}(u,\cdot)^{-1}(\mathfrak T), \leb(\mathfrak D)>0$ with $h^*_\wa(t) > 0$ and $\psi^i_{\wa[j^\star]}(\arr_{\wa,j^\star}(u,t)) +  \prices_{\wa[j^\star]}(\arr_{\wa,j^\star}(u,t))> 0$ for almost all $t \in \mathfrak D$.  
  				Define $\tilde{\wa}\coloneq  \wa_{<j_\wa} \in \Routes$ where the latter is indeed an $s,\dest$-walk due to $\wa[j_w] =\wa^\dest[1] \in \delta^+(\dest)$ by $\wa^\dest \in \hat{\Routes}^\dest$. By the nonnegativity of $\psi^i(\cdot)$ and $\prices$, we get 
  				\begin{align*}
  					\wttime_\wa(u,t) +\Pf^{\prices}_\wa(u,t) &\geq  \wttime_{\tilde{\wa}}(u,t)+\Pf^{\prices}_{\tilde{\wa}}(u,t) + \psi^i_{\wa[j^\star]}(\arr_{\wa,j^\star}(u,t)) +  \prices_{\wa[j^\star]}(\arr_{\wa,j^\star}(u,t)) \\
  					&> \wttime_{\tilde{\wa}}(u,t)+\Pf^{\prices}_{\tilde{\wa}}(u,t) 
  				\end{align*}
  				for almost all $t \in \mathfrak D$  which contradicts that $h^*$ is a $\prices$-DUE. 
  				
  				Hence, we can conclude that the desired statement holds with $\psi^i$ instead of $\psi$ on the right-hand side of \eqref{eq: ImplicationSubclaim:thm:main:part1}. 
  				
  				From this, we get for   an arbitrary representative $h^*$,  $i \in I$ and  walk $\wa=(\hat{\wa},i) \in \Routes$ with $\wa[j_\wa] = \arc_j=\arc_1$ for some $j_\wa \leq |\wa|$ that $\wttime_{\wa_{\geq j_\wa}}(u,\arr_{\wa,j_\wa}(u,t)) + \Pf^\prices_{\wa_{\geq j_\wa}}(u,\arr_{\wa,j_\wa}(u,t))= 0$ holds for almost all  $t \in \arr_{\wa,j_\wa}(u,\cdot)^{-1}(\mathfrak T)\cap \{t \in \hori \mid h^*_\wa(t)>0\}$. 
  				Note that $\wa_{\geq j_\wa}=(\hat{\wa}_{\geq j_\wa},i)$ is a commodity-typed walk corresponding to commodity $i$. Subsequently, 
  				$\wttime_{\wa_{\geq j_\wa}}(u,\arr_{\wa,j_\wa}(u,t)) $ refers to the private costs of commodity $i$. However, 
  				since $\wa_{\geq j_\wa} $ is a $\dest$-cycle by $\wa[j_\wa] =\wa^\dest[1] \in \edgesFrom{\dest}$ and $\wa \in \Routes_i$ being an \stwalk[\source_i][\dest], we get by  \Cref{ass: PCSep+ZeroCycles: DestCycles} that the private costs for all commodities on the walk is zero and subsequently: 
  				$\wttime_{\hat{\wa}_{\geq j_\wa}}(u,\arr_{\wa,j_\wa}(u,t)) +\Pf^\prices_{\wa_{\geq j_\wa}}(u,\arr_{\wa,j_\wa}(u,t))= 0$ holds for almost all  $t \in \arr_{\wa,j_\wa}(u,\cdot)^{-1}(\mathfrak T) \cap \{t \in \hori \mid h^*_\wa(t)>0\}$.  From this,  the desired implication follows immediately by the nonnegativity of $\psi$ and $\prices$.
  				
  				
  				\begin{figure}
  					\centering
  					\BigPicture[1]{
  						\begin{tikzpicture}
  							\coordinate(s1)at(0,2.5);
  							\coordinate(s2)at(0,0);
  							\coordinate(v0)at(10,2.5);
  							\coordinate(v1)at(5,2.5);
  							\coordinate(v2)at(5,0);
  							\coordinate(v3)at(5,-2.5);
  							\coordinate(v4)at(9,-3.5);
  							\coordinate(v5)at(11,-3.5);
  							\coordinate(d)at(10,0);
  							
  							\node[namedVertexF](d)at(d){$\dest$};
  							
  							\draw[colA,->,line width=2,rounded corners=6](s1)to node[above]{$\wa^{j-1}_{<\rho(j-1)}$}($(v1)+(.1,0)$)to node[right]{$\wa^{j-1}[\rho(j-1)]$}($(v2)+(.1,0)$)to node[above]{$\wa^{j-1}_{>\rho(j-1)}$}(d);
  							\draw[colB,->,line width=2,rounded corners=6](s2)to node[above]{$\wa^j_{<\rho(j)}$}($(v2)+(.1,0)$)to node[right]{$\wa^j[\rho(j)]$}($(v3)+(.1,0)$)to[out=0,in=-120]node[above, yshift=0.1cm]{$\wa^j_{>\rho(j)}$}(d);
  							
  							\node[namedVertexF](s1)at(s1){$\source_1$};
  							\node[namedVertexF](s2)at(s2){$\source_2$};
  							\node[vertexF](v0)at(v0){};
  							\node[vertexF](v1)at(v1){};
  							\node[vertexF](v2)at(v2){};
  							\node[vertexF](v3)at(v3){};
  							\node[vertexF](v4)at(v4){};
  							\node[vertexF](v5)at(v5){};
  							
  							\draw[edge](d)--node[right]{$\arc_1$}(v0);
  							\draw[edge,dashed](v0)to[out=110,in=70] node[above]{$\wa^d_{>1,<j-1}$}(v1);
  							\draw[edge](v1)--node[left]{$\arc_{j-1}$}(v2);
  							\draw[edge](v2)--node[left]{$\arc_{j}$}(v3);
  							\draw[edge](v4)--node[below]{$\arc_m$}(v5);
  							\draw[edge,dashed](v3)to[out=-70,in=180] node[below]{$\wa^d_{>j,<m}$}(v4);
  							
  							\node[right=.1 of d]{$\mathfrak T$};
  							\node[left=.1 of s2]{$\mathfrak D \subseteq \mathfrak D_j$};
  							\node[left=.1 of s1]{$\mathfrak D_{j-1}$};
  							\node[right=.1 of v1]{$\mathfrak T_j \supseteq \mathfrak T_{j-1}$};
  						\end{tikzpicture}
  					}
  					\caption{A depiction for the situation considered in the induction step in the proof of \Cref{claim:thm:main:part1}: The walk $\wa^d=(\arc_1,\dots,\arc_m)$ starts at the destination~$\dest$
                    and for some set of times $\mathfrak T$ any of its edges $j$ is used by $u$ during $\arr_{\wa^d,j}(\mathfrak T)$. 
                    Our induction hypothesis states that for the $(j-1)$-th edge (and all previous ones) any (used) \stwalk[\source_i][\dest] passing through this edge at these time (like the red walk~$\wa^{j-1}$ here) has zero total cost for its suffix starting with this edge. The induction step is then to show that the same also holds for any such walk passing through edge~$\arc_j$, e.g.\ the blue walk~$\wa^j$ here. For this, we use the fact that the given flow is a $\prices$-DUE and, hence, switching from using $\wa^j$ to using $\wa^{j-1}$ after the node between edges $\arc_{j-1}$ and $\arc_j$ cannot decrease the overall cost. Hence, the respective suffix of walk~$\wa^j$ also has at most total cost zero. \newline 
                    In the figure, the sets $\mathfrak D,\mathfrak D_j,\mathfrak D_{j-1},\mathfrak T_{j-1},\mathfrak T_{j}$ denote sets of departure or arrival times used during the proof of the induction step and are placed next to their respective nodes.}
  					\label{fig:thm:main:part1:induction}
  				\end{figure}
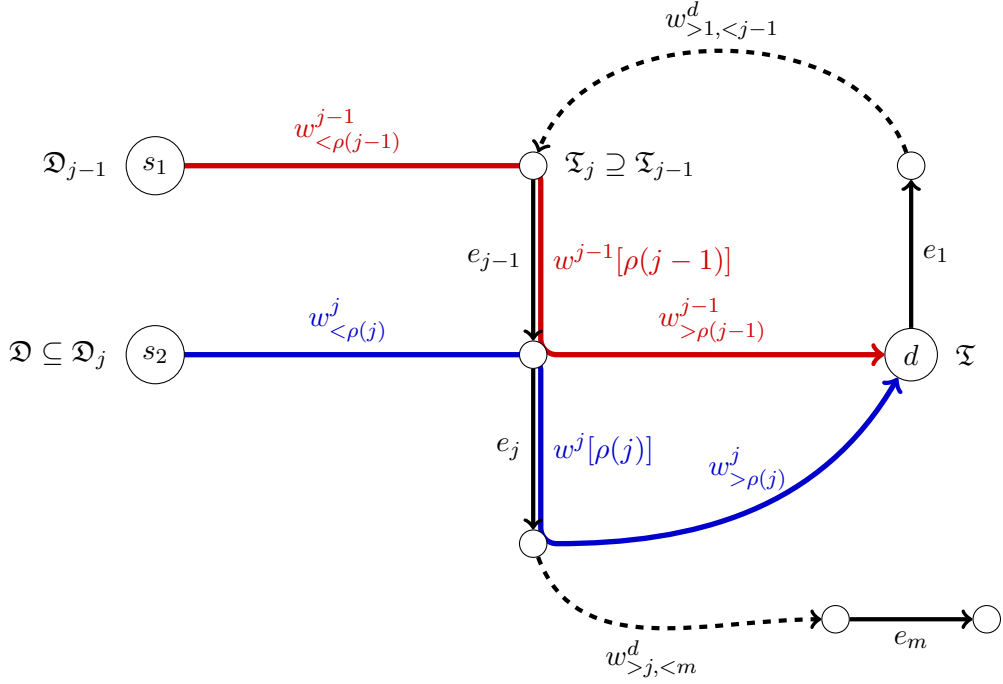
  				
  				\inductionstep{$j-1 \to j$ for $j\geq 2$} 
  				Assume we have shown the induction claim up to $j-1$. 
  				
  				We start again by showing that the desired statement holds with $\psi^i$ instead of $\psi$ on the right-hand side of \eqref{eq: ImplicationSubclaim:thm:main:part1}.   
  				For this, assume for the sake of a contradiction that there exists $i\in I$, $\wa^j=(\hat{\wa}^j,i) \in \Routes_i$, $\ain(j) \leq|\wa^j|$ with $\wa^j[{\ain(j)}] = \arc_j$, $j^\star \in\{{{\ain(j)}},\ldots,|{\wa^j}|\}$ and $\mathfrak D_j \subseteq \arr_{\wa,\ain(j)}(u,\cdot)^{-1}\big(\arr_{\wa^\dest,j}(u,\cdot)(\mathfrak T)\big) , \leb(\mathfrak D_j )>0$ with\footnote{For the sake if readability, we write here $\ain(j)$ instead of $j_{\wa^j}$.} 
  				\begin{align}\label{eq:thm:main:part1induction1}
  					h^*_{\wa^j}(t) > 0 \text{ and } \psi^i_{{\wa^j}[{j^\star}]}(\arr_{\wa^j,j^\star}(u,t))  + \prices_{\wa^j[j^\star]}(\arr_{\wa^j,j^\star}(u,t))  > 0 \text{ for almost all } t \in \mathfrak D_j.
  				\end{align} 
  				We choose a representative of $h^*$ such that the inequalities in \eqref{eq:thm:main:part1induction1} hold for all $t \in \mathfrak D_j$. 
  				
  				By definition of $\mathfrak T$, we have 
  				$u_{\arc_{j-1}}(\arr_{\wa^\dest,j-1}(u,t))  >0$ for almost all 
  				$ t\in \mathfrak T$. In particular, 
  				by $\arr_{\wa^\dest,j-1}(u,\cdot)$ having Lusin's property (cf.~\Cref{lem: PropAbsCon:Lus}), 
  				we have 
  				$u_{\arc_{j-1}}(t)  >0$ for almost all 
  				$t \in \arr_{\wa^\dest,j-1}(u,\cdot)(\mathfrak T)$. Now consider the subset  
  				\begin{align*}
  					\mathfrak T_j\coloneq \exit_{\wa^\dest[j-1]}(u,\cdot)^{-1}\big(\arr_{{\wa^j,{\ain(j)}}}(u,\cdot)(\mathfrak D_j)\big) \cap \arr_{\wa^\dest,j-1}(u,\cdot)(\mathfrak T).
  				\end{align*}
  				\begin{subclaim}\label{subclaim: SetEqual}
  					The equality $\exit_{\wa^\dest[j-1]}(u,\cdot)(\mathfrak T_j) = \arr_{{\wa^j,{\ain(j)}}}(u,\cdot)(\mathfrak D_j)$ and, subsequently,
  					$\leb(\mathfrak T_j)>0$ holds. 
  				\end{subclaim}
  				\begin{proofClaim}
  					For the claimed equality, note that  $\subseteq$ is an immediate consequence of the definition of $\mathfrak T_j \subseteq \exit_{\wa^\dest[j-1]}(u,\cdot)^{-1}\big(\arr_{{\wa^j,{\ain(j)}}}(u,\cdot)(\mathfrak D_j)\big)$. 
  					For the reverse inclusion $\supseteq$, observe that by definition of $\mathfrak D_j$, we have $\arr_{{\wa^j,{\ain(j)}}}(u,\cdot)(\mathfrak D_j) \subseteq \arr_{\wa^\dest,j}(u,\cdot)(\mathfrak T)=\exit_{\wa^\dest[j-1]}(u,\cdot)\big(\arr_{\wa^\dest,j-1}(u,\cdot)(\mathfrak T)\big)$ from which the claimed inclusion follows by definition of $\mathfrak T_j$. 
  					
  					Now observe that $\arr_{{\wa^j,{\ain(j)}}}(u,\cdot)(\mathfrak D_j)$ is not a null set 
  					due to $h^*_{\wa^j}(t) > 0,t \in \mathfrak D_j $ and $\leb(\mathfrak D_j)>0$, \Cref{lem: elluExistenceProperties:ExistenceInducedFlowOnJthEdge} and the existence of $\ell^u(h^*) =u$. Hence, by $\exit_{\wa^\dest[j-1]}(u,\cdot)$ fulfilling Lusin's property (cf.~\Cref{lem: PropAbsCon:Lus}), it follows that   $\leb(\mathfrak T_j)>0$  has to hold.
  				\end{proofClaim}
  				
  				As $u_{\arc_{j-1}}(t)  >0$ for almost all 
  				$t \in\mathfrak T_j \subseteq \arr_{\wa^\dest,j-1}(u,\cdot)(\mathfrak T)$ by definition of $\mathfrak T$, \Cref{lem: Relations:h>0u>0:u>0ExistsCountableM} implies the existence of 
  				an $i' \in I$, $\wa^{j-1}=(\hat{\wa}^{j-1},i') \in \Routes$, ${\ain(j-1)}\leq|\wa^{j-1}|$ with $\wa^{j-1} [{\ain(j-1)}] = \arc_{j-1}$ and a 
  				non-null set $\mathfrak D_{j-1} \subseteq \arr_{\wa^{j-1},\ain(j-1)}(u,\cdot)^{-1}(\mathfrak T_j)$ such that 
  				$h^*_{\wa^{j-1}}(t)>0 $ for a.e.~$t \in \mathfrak D_{j-1}$. 
  				By induction hypothesis, we have for almost all $t \in \mathfrak D_{j-1}$:
  				\begin{align}\label{eq: DefDj-1}
			\psi_{\wa^{j-1}[j']}(\arr_{\wa,j'}(u,t)) + \prices_{\wa^{j-1}[j']}(\arr_{\wa,j'}(u,t)) = 0 \quad \quad \text{ for all } j'\in\Set{\ain(j-1),\ldots,\abs{\wa^{j-1}}}.
  				\end{align}
   
  				By choosing $\mathfrak D_{j-1}$ suitably, we may assume \wlg that this holds for all $t \in \mathfrak D_{j-1}$. 
  				Define $ \mathfrak T_{j-1} \coloneq  \arr_{\wa^{j-1},\ain(j-1)}(u,\cdot) (\mathfrak D_{j-1})\subseteq \mathfrak T_j$ which is  not a null set by \Cref{lem: elluExistenceProperties}, the existence of $\ell^u(h^*)$,  $h^*_{\wa^{j-1}}(t)>0$ for a.e~$t\in \mathfrak{D}_{j-1}$ and $\leb(\mathfrak{D}_{j-1})$. 
  				Then, by \eqref{eq: DefDj-1}, we have 
  				\begin{align}\label{eq:thm:main:part1induction2}
  					\wttime_{\hat{\wa}^{j-1}_{\geq \ain(j-1)}}(u,t) + \Pf^{\prices}_{\hat{\wa}^{j-1}_{\geq \ain(j-1)}}(u,t) = 0 \text{ for all } t \in \mathfrak T_{j-1}.
  				\end{align}
  				Furthermore, by  $\leb(\mathfrak T_{j-1}) > 0$ we have that the
  				set $\mathfrak D \coloneq  \mathfrak D_j \cap  \arr_{\wa^j,\ain(j)}(u,\cdot)^{-1} (\exit_{\wa^\dest[j-1]}(u,\cdot)(\mathfrak T_{j-1}))$ is not a null set since 
  				$\arr_{\wa^j,\ain(j)}(u,\cdot)$ and $\exit_{\wa^\dest[j-1]}(u,\cdot)$ both have Lusin's property and we have 
  				$\arr_{\wa^j,\ain(j)}(u,\cdot)(\mathfrak D) = \exit_{\wa^\dest[j-1]}(u,\cdot)(\mathfrak T_{j-1})$. 
  				Note for the latter equality that $\arr_{\wa^j,\ain(j)}(u,\cdot) (\mathfrak D_j) =\exit_{\wa^\dest[j-1]}(u,\cdot)(\mathfrak T_j) \supseteq \exit_{\wa^\dest[j-1]}(u,\cdot)(\mathfrak T_{j-1})$ by \Cref{subclaim: SetEqual}. 
                
  				Define $\tilde{\wa}\coloneq (\hat{\wa}^{j}_{<{\ain(j)}}, \hat{\wa}^{j-1}_{>\ain(j-1)},i) \in \Routes$ 
  				where the latter is indeed an \stwalk[\source_i][\dest] as $\wa^j$ is an \stwalk[\source_i][\dest], 
  				$\hat{\wa}^{j-1}$ is an 
  				\stwalk[\source_{i'}][\dest], 
  				 and the end node of 
  				$\hat{\wa}^j_{<\rho(j)}$ 
  				is the tail of 
  				$\hat{\wa}^j[\rho(j)] = \arc_j$ which, in turn, is the head of 
  				$\arc_{j-1} = 
  				 \hat{\wa}^{j-1}[\rho(j-1)]$, that is, the starting node of 
  				  $\hat{\wa}^{j-1}_{>\rho(j-1)}$. 
  				Now observe for any $t \in \mathfrak D$  
  				\begin{align*}
  					&\wttime_{\tilde{\wa}}(u,t) + \Pf^{\prices}_{\tilde{\wa}}(u,t) \\
  					&\quad\quad= \sum_{j'<\ain(j)}  \psi^i_{\tilde{\wa}[j']}(\arr_{\tilde{\wa},j'}(u,t)) + \prices_{\tilde{\wa}[j']}(\arr_{\tilde{\wa},j'}(u,t))\\
  					&\quad\quad\quad\quad+\sum_{j'\geq \ain(j)}\psi^i_{\tilde{\wa}[j']}(\arr_{\tilde{\wa},j'}(u,t)) + \prices_{\tilde{\wa}[j']}(\arr_{\tilde{\wa},j'}(u,t))  \\
  					&\quad\quad=\wttime_{{\wa}^j_{<\ain(j)}}(u,t) + \Pf^{\prices}_{{\wa}^j_{<\ain(j)}}(u,t) \\
  					&\quad\quad\quad\quad+\wttime_{(\hat{\wa}^{j-1}_{> \ain(j-1)},i)}(u,\arr_{\tilde{\wa},\ain(j)}(u,t)) + \Pf^{\prices}_{\hat{\wa}^{j-1}_{> \ain(j-1)}}(u,\arr_{\tilde{\wa},\ain(j)}(u,t))  \\
  					&\quad\quad\symoverset{1}{=}\wttime_{{\wa}^j_{<\ain(j)}}(u,t) + \Pf^{\prices}_{{\wa}^j_{<\ain(j)}}(u,t) + 0 \\
  					&\quad\quad\symoverset{2}{<} \wttime_{{\wa}^j_{<\ain(j)}}(u,t) + \Pf^{\prices}_{{\wa}^j_{<\ain(j)}}(u,t) + \psi^i_{{\wa^j}[{j^\star}]}(\arr_{\wa^j,j^\star}(u,t))  + \prices_{\wa^j[j^\star]}(\arr_{\wa^j,j^\star}(u,t))\\
  					&\quad\quad\leq \wttime_{{\wa}^j}(u,t) + \Pf^{\prices}_{{\wa}^j}(u,t) 
  				\end{align*}
  				where we used $\arr_{\tilde{\wa},\ain(j)}(u,t) = \arr_{\wa^j,\ain(j)}(u,t) \in \exit_{\wa^\dest[j-1]}(u,\cdot)(\mathfrak T_{j-1})$ and \eqref{eq:thm:main:part1induction2} at~\refsym{1} and \eqref{eq:thm:main:part1induction1} together with $\mathfrak D \subseteq \mathfrak D_{j}$ for the strict inequality \refsym{2}. 
  				
  				Since, by definition,  $\mathfrak D$ has positive measure and we also have $h^*_{\wa^j}(t)>0$ for almost all $t \in \mathfrak D \subseteq \mathfrak D_j$, this gives us the desired  contradiction as the above demonstrates that $h^*$ is not a $\prices$-DUE. 
  				
  				Hence, we can conclude that the desired statement holds with $\psi^i$ instead of $\psi$ on the right-hand side of \eqref{eq: ImplicationSubclaim:thm:main:part1}. 
  				
  				From this, we get for   an arbitrary representative of $h^*$,  $i \in I$ and  walk $\wa=(\hat{\wa},i) \in \Routes$ with $\wa[j_\wa] = \arc_j$ for some $j_\wa \leq |\wa|$ that 
  				\begin{equation}\label{eq: DefD_j}
  						\begin{aligned} 
  						\wttime_{\wa_{\geq j_\wa}}(u,&\arr_{\wa,j_\wa}(u,t)) +\Pf^\prices_{\wa_{\geq j_\wa}}(u,\arr_{\wa,j_\wa}(u,t)) = 0 \text{ for a.e.\ }t \in \mathfrak D_j \text{ where }\\ \mathfrak D_j\coloneq &\arr_{\wa,j_\wa}(u,\cdot)^{-1}\big(\arr_{\wa^\dest,j}(u,\cdot)(\mathfrak T)\big)\cap \Set{t \in \hori \mid h^*_\wa(t)>0}.
  					\end{aligned}
  				\end{equation}
  				   We choose a suitable representative of $h^*$ such that the latter property holds for all $t \in \mathfrak D_j$. 
  				Note that $\wa_{\geq j_\wa}=(\hat{\wa}_{\geq j_\wa},i)$ is a commodity-typed walk corresponding to commodity $i$. Subsequently, 
  				$\wttime_{\wa_{\geq j_\wa}}(u,\arr_{\wa,j_\wa}(u,t)) $ refers to the private costs of commodity $i$. We will argue in the following that the private costs of all commodities need to be zero by applying \Cref{ass: PCSep+ZeroCycles: DestCycles} for  $c^\dest\coloneq (\wa^\dest_{< j},\hat{\wa}_{\geq j_\wa})$. Remark that the latter is a $\dest$-cycle since $\wa^\dest[1] \in \edgesFrom{\dest}$,     $\wa \in \Routes_i$ is an \stwalk[\source_i][\dest] and $\wa^\dest[j] = \arc_j = \wa[j_\wa] $.  
  				By induction hypothesis together with the argument at the start of the induction proof, we know that \eqref{eq:thm:main:part1} holds for all $j'<j$. 
  				In particular, we know that $\wttime_{\wa^\dest_{< j}}(u,t) = 0$ for almost all $t \in \mathfrak T$ and subsequently also for almost all $t \in \mathfrak D\coloneq  \arr_{\wa^\dest,j}(u,\cdot)^{-1}\big(\arr_{\wa,j_\wa}(u,\cdot)(\mathfrak D_j)\big) \cap\mathfrak T$. 
  				Since the statement we want to prove is about a property required to only hold almost everywhere on $\mathfrak T$, it is enough to show the statement for a set that differs from $\mathfrak T$ in a null set. Hence, 
  				we can assume \wlg that $\wttime_{\wa^\dest_{< j}}(u,t) = 0$ for all $t \in \mathfrak D$. 
  				Now observe that $\arr_{\wa^\dest,j}(u,\cdot)(\mathfrak D)= \arr_{\wa,j_\wa}(u,\cdot)(\mathfrak D_j)$ by 
  				$\arr_{\wa,j_\wa}(u,\cdot)(\mathfrak D_j) \subseteq \arr_{\wa^\dest,j}(u,\cdot)(\mathfrak T)$ due to the definition of $\mathfrak D_j$. 
  				Hence, we get that  for  all $t \in \mathfrak D_j$, 
  				there exists $\tilde{t}\in \mathfrak D$ with $\arr_{\wa^\dest,j}(u,\tilde{t}) = \arr_{\wa,j_\wa}(u,t)$ and 
  				\begin{align*}
  					\wttime_{(c^\dest,i)}(u,\tilde{t}) = \wttime_{(\wa^\dest_{< j},i)}(u,\tilde{t}) +  \wttime_{\wa_{\geq j_\wa}}(u, \arr_{\wa^\dest,j}(u,\tilde{t}))= \wttime_{(\wa^\dest_{< j},i)}(u,\tilde{t}) +  \wttime_{\wa_{\geq j_\wa}}(u, \arr_{\wa,j_\wa}(u,t)) =0,
  				\end{align*}
  				 where we used  $\tilde{t}\in \mathfrak D$ for $\wttime_{(\wa^\dest_{< j},i)}(u,\tilde{t})= 0$ and \eqref{eq: DefD_j} with $t \in \mathfrak D_j$ for $\wttime_{\wa_{\geq j_\wa}}(u, \arr_{\wa,j_\wa}(u,t)) = 0$.
  				Now \Cref{ass: PCSep+ZeroCycles: DestCycles} implies that 
  				$\wttime_{c ^\dest}(u,\tilde{t}) = 0$ holds, implying in particular that we have 
  				$\wttime_{\hat{\wa}_{\geq j_\wa}}(u, \arr_{\wa,j_\wa}(u,t)) = 0$.   
  				This, in combination with $\Pf^\prices_{\wa_{\geq j_\wa}}(u,\arr_{\wa,j_\wa}(u,t)) = 0$ for every $t \in \mathfrak D_j$ (by \eqref{eq: DefD_j}), 
  				shows the desired implication in \eqref{eq: ImplicationSubclaim:thm:main:part1}, which finishes the proof of the induction.
  			\end{proofbyinduction} 
  			As argued at the start of the induction, this now directly implies that \Cref{claim:thm:main:part1} holds as well.
  		\end{proofClaim}

  		In order to finally deduce \ref{thm: mainSingleSink: ReachableEdgesHaveZeroD}
  		from this, we argue as follows: 
  		Let $\mathfrak T^*$ be the set where the stated property is not fulfilled. 
  		We argue in the following that $\mathfrak T^*$ is a null set. 
  		Define for any $\wa \in \hat{\Routes}^\dest$ the set $\mathfrak T^*_\wa$ as the set where the stated property in \Cref{thm: mainSingleSink: ReachableEdgesHaveZeroD} is not fulfilled \wrt $\wa$ and its last edge, i.e.
  		\begin{align*}
  			\mathfrak T^*_\wa\coloneq  \Set{t \in \hori |  u_{\wa[j']}(\arr_{\wa,j'}(u,t))>0,j'\leq |\wa| \text{ but } \psi_{\wa[|\wa|]}(u,\arr_{\wa,\abs{\wa}}(u,t)) \neq 0  }. 
  		\end{align*}
  		Note that the latter is measurable due to the measurability of $u$ and $\arr$.  By $\hat{\Routes}^\dest$ being countable and  
  		$\mathfrak T^* = \bigcup_{\wa \in \hat{\Routes}^\dest} \mathfrak T^*_\wa $, the latter is also measurable and  it is sufficient to show that  $\mathfrak T^*_\wa$ is a null set for every $\wa \in \hat{\Routes}^\dest$. 
  		This is however a direct consequence of~\Cref{claim:thm:main:part1} for $  \wa^\dest= \wa$ and $\mathfrak T= \mathfrak T^*_\wa$.  
  		Thus, 
  		the proof of 
  		$\ref{thm: mainSingleSink: impl}\Rightarrow~\ref{thm: mainSingleSink: ReachableEdgesHaveZeroD}$ is complete.

  		\proofitem{\ref{thm: mainSingleSink: ReachableEdgesHaveZeroD}$\Rightarrow$\ref{thm: mainSingleSink: UsedDCyclesHaveZeroD}}
  		Let $\mathfrak T^*,\mathfrak T^*_\wa, \wa \in \hat{\Routes}^\dest$ be as above. By \ref{thm: mainSingleSink: ReachableEdgesHaveZeroD}, $\mathfrak T^*$ is a null set. 
  		We define similarly  
  		$\mathfrak T'$ to be the set, where the stated property in~\ref{thm: mainSingleSink: UsedDCyclesHaveZeroD} is not fulfilled. 
  		We argue in the following that $\mathfrak T'$ is a null set.  
  		Define for any $\dest$-cycle $c$, the set $\mathfrak T'_c$ as the set where the stated property in~\ref{thm: mainSingleSink: UsedDCyclesHaveZeroD} is not fulfilled \wrt $c$, i.e.
  		\begin{align*}
  			\mathfrak T'_c \coloneq  \Big\{ t \in \hori \mid u_{c[j']}(\arr_{c,j'}(u,t)) >0, j' \leq |c|  \text{ but } \wttime_{c}(u,t) \neq 0   \Big\}
  		\end{align*}
  		Then,   
  		$\mathfrak T'$ is the union of $\mathfrak T'_c$ over all $\dest$-cycles $c$. Now 
  		for every $t \in \mathfrak T'_c$, denote by $j(t)$ the first index~$j$ such that $\psi_{c[j]}(\arr_{c,j}(u,t))\neq0$. 
  		Then  
  		$t \in \mathfrak T^*_{c_{\leq j(t)}}$ where we remark that $c_{\leq j(t)} \in \hat{\Routes}^\dest$ is a $\dest$-walk. This shows that $\mathfrak T'\subseteq \mathfrak T^*$ and since the latter is a null set by assumption, the claim follows.   
  		
  		\proofitem{\ref{thm: mainSingleSink: UsedDCyclesHaveZeroD}$\Rightarrow$\ref{thm: mainSingleSink: Optimalh}}
  		Fix an arbitrary representative of $u$. 
  		Let $\check{h} \in \wir$ be an optimal solution  for \eqref{opt: Master} with corresponding aggregated edge flow $\check{\g}= \ell^u(\check{h})$. Note that $\check{h}$ exists due to the assumption that the master problem admits strong duality. 
  		Since $u$ is also induced by some $h' \in \wir$ and admits an outflow rate at $\dest$ due to \Cref{ass: u}, 
  		\Cref{lem: DifferenceGeneral} is applicable, showing that 
  		$u$ and $\check{\g}$ only differ in flow on cycles of travel time zero and on $\dest$-cycles, i.e.~there exist $\hat{h}_{c}^0 \in L_+(\hori), c \in \SimpCyc$ and $\hat{h}_{c^\dest} \in L_+(\hori),c^\dest\in\cup \DestCyc$ with 
  		\begin{align*}
  			u -\ell^u(\check{h}) = \sum_{c^\dest \in \DestCyc} \ell_{{c^\dest}}(\hat{h}_{c^\dest}) + \sum_{c \in \SimpCyc}\ell_{c}(\hat{h}^0_c)
  		\end{align*}
  		and for all $c\in \SimpCyc$ and almost all $t\in \hori$ we have that  $\hat{h}_c^0 (t) >0\implies \trav_{\arc}(t) = 0$ for all $\arc \in c$.
  		
  		The zero-cycles have zero private costs by \Cref{ass: PCSep+ZeroCycles: ZeroCycles}, while the $\dest$-cycles do so by the assumption of \Cref{thm: mainSingleSink: UsedDCyclesHaveZeroD}. Thus, the task becomes now to add these cycles to the optimal solution~$\check{h}$. We start by adding the $\dest$-cycles. 
  		\begin{claim}
  			There exists $\tilde{h}\in \wir$ with corresponding  aggregated edge flow $\tilde{\g}$ that is optimal   for \eqref{opt: Master} with $u$ and $\tilde{\g}$ only differing in flow on zero-cycles.
  		\end{claim}
  		\begin{proofClaim}
  			\Cref{lem: DifferenceGeneral} states that $\sum_{c^\dest\in \DestCyc}\hat{h}_{c^\dest} \leq -\check{\inflow}_\dest$ where $\check{\inflow}_\dest$ is the net outflow rate at $\dest$ of $\check{\g}$. 
  			By \Cref{lem: flowconW'}, we have $\sum_{\wa \in \Routes}\ell^u_{\wa,\abs{\wa}+1}(\check{h}_\wa) = -\check{\inflow}_\dest$ where we remark that the existence of $\ell^u_{\wa,\abs{\wa}+1}(\check{h}_\wa)$ follows by the existence of $\check{\inflow}_\dest$. 
  			Thus, 
  			 $\sum_{c^\dest\in \DestCyc}\hat{h}_{c^\dest} \leq -\check{\inflow}_\dest= \sum_{\wa \in \Routes}\ell^u_{\wa,\abs{\wa}+1}(\check{h}_\wa)$ holds and 
  			we can distribute the inflows $\hat{h}_c,c \in\DestCyc$ among the inflows $\ell^u_{\wa,\abs{\wa}+1}(\check{h}_\wa),\wa \in \Routes$, leading 
  			to inflows  $\eflow_{\wa,c^\dest} \in L_+(\hori), \wa \in \Routes,c^\dest\in \DestCyc$ with $\sum_{c^\dest\in\DestCyc}\eflow_{\wa,c^\dest} \leq \ell^u_{\wa,\abs{\wa}+1}(\check{h}_\wa)$ for all $\wa \in \Routes$ and 
  			$\hat{h}_{c^\dest} = \sum_{\wa\in \Routes}\eflow_{\wa,c^\dest}$. 
  			By \Cref{lem: ellOrderPreserving,lem: 1to1:h-f} and $\sum_{c^\dest\in\DestCyc}\eflow_{\wa,c^\dest} \leq \ell^u_{\wa,\abs{\wa}+1}(\check{h}_\wa)$, it follows that there exist $\wflow_{\wa,c^\dest} \in \edom{\wa}[u], \wa \in \Routes,c^\dest\in \DestCyc$ with 
  			$\ell^u_{\wa,\abs{\wa}+1}(\wflow_{\wa,c^\dest}) = \eflow_{\wa,c^\dest}$ and
  			$\sum_{c^\dest\in\DestCyc}\wflow_{\wa,c^\dest} \leq  \check{h}_\wa$. 
  			The latter inequality immediately implies that the following definition yields nonnegative walk inflow rates:
  			\begin{align*}
  				\tilde{h}_\wa \coloneq  \check{h}_\wa - \sum_{c^\dest\in\DestCyc}\wflow_{\wa,c^\dest}+ \sum_{\substack{(\wa',c^\dest)\in \Routes\times \DestCyc:
  						\\(\wa',c^\dest) = \wa}} \wflow_{\wa',c^\dest} \text{ for all }\wa \in \Routes,
  			\end{align*}
  			where $(\wa',c^\dest)\coloneq  (\hat{\wa}',c^\dest,i)$ for any $\wa'=(\hat{\wa}',i) \in \Routes$. 
  			Moreover, we have $\tilde{h}\in \wir$ as  for any $i \in I$  we have
  			\begin{align*}
  				\sum_{\wa \in \Routes_i}\tilde{h}_\wa &=  \sum_{\wa \in \Routes_i}\Big(    \check{h}_\wa - \sum_{c^\dest\in\DestCyc}\wflow_{\wa,c^\dest}+ \sum_{\substack{(\wa',c^\dest)\in \Routes\times \DestCyc:
  						\\(\wa',c^\dest) = \wa}} \wflow_{\wa',c^\dest}\Big)  \\
  				&= \inflow_i - \sum_{(\wa,c^\dest)\in\Routes_i\times\DestCyc}\wflow_{\wa,c^\dest} +  \sum_{\wa \in \Routes_i} \sum_{\substack{(\wa',c^\dest)\in \Routes\times \DestCyc:
  						\\(\wa',c^\dest) = \wa}} \wflow_{\wa',c^\dest} = \inflow_i
  			\end{align*}
  			
  			where we used in the last equality that  $(\wa',c^\dest) \in \Routes_i$ if and only if $\wa' \in \Routes_i$.  
  			
  			Next, we consider the induced edge flow of $\tilde{h}$: \stepcounter{equation}
  			\begin{align*}
  				\ell^u(\tilde{h}) \overset{\text{\Crefshort{lem: elluExistenceProperties}}}&{=} \ell^u(\check{h}) - \sum_{\wa \in \Routes}\sum_{c^\dest\in\DestCyc}\ell^u_\wa(\wflow_{\wa,c^\dest}) +
  				\sum_{\wa \in \Routes} \sum_{\substack{(\wa',c^\dest)\in \Routes\times \DestCyc:
  						\\(\wa',c^\dest) = \wa}} \ell^u_\wa(\wflow_{\wa',c^\dest} ) \\
  				&= \ell^u(\check{h}) - \sum_{\wa \in \Routes}\sum_{c^\dest\in\DestCyc}\ell^u_\wa(\wflow_{\wa,c^\dest}) +
  				\sum_{\wa \in \Routes}  \sum_{c^\dest\in \DestCyc} \ell^u_{(\wa,c^\dest)}(\wflow_{\wa,c^\dest} ) \\
  				\overset{\text{\Crefshort{lem: elluPropagation}}}&{=}  \ell^u(\check{h}) +   \sum_{\wa \in \Routes}  \sum_{c^\dest\in \DestCyc} \ell^u_{c^\dest}(\ell^u_{\wa,\abs{\wa}+1}(\wflow_{\wa,c^\dest})) \\
  				&= \ell^u(\check{h}) +  \sum_{c^\dest\in \DestCyc}  \sum_{\wa \in \Routes}   \ell^u_{c^\dest}( \eflow_{\wa,c^\dest})\tag{\theequation}\label{eq: Claim9Equality}\\
  				\overset{\text{\Crefshort{lem: elluContinuity:SingleWalk}}}&{=} \ell^u(\check{h}) +  \sum_{c^\dest\in \DestCyc}   \ell^u_{c^\dest}\big( \sum_{\wa \in \Routes} \eflow_{\wa,c^\dest}\big) \\
  				&= \ell^u(\check{h}) +  \sum_{c^\dest\in \DestCyc}  \ell^u_{c^\dest}(\hat{h}_{c^\dest })
  			\end{align*}
  			Hence, $\ell^u(\tilde{h})$ and $u$ only differ in flow on zero-cycles and $\tilde{h}$ is in particular feasible for \eqref{opt: Master}. 
  			
  			It thus remains to observe that $\tilde{h}$ has  the same objective value for the master problem as $\check{h}$, implying particularly that $\tilde{h}$ is optimal for \eqref{opt: Master}. 
  			For this, note that the above chain of equalities remains valid up to the equality \eqref{eq: Claim9Equality}, if we consider the commodity-specific flows, that is, exchanging $\ell^u(\tilde{h}), \ell^u(\check{h}) $ with $\ell^u_{\Routes_i}(\tilde{h}^i), \ell^u_{\Routes_i}(\check{h}^i) $ as well as $\Routes$ via $\Routes_i$ for arbitrary $i \in I$. 
  			With this insight and \Cref{lem: aggCostsVSwalkCosts}, we calculate: 
  			\begin{equation}\label{eq: CostsAddDCycles}
  				\begin{aligned}
  					\dup{\wttime(u,\cdot)}{\tilde{h}}  &= \sum_{i \in I} \sum_{\wa \in \Routes_i} \dup{\wttime_\wa(u,\cdot)}{\tilde{h}_\wa} \Croverset{lem: aggCostsVSwalkCosts}{=} \sum_{i \in I} \dup{\psi^i(\cdot)}{\ell^u_{\Routes_i}(\tilde{h}^i)} \\
  					&= \sum_{i \in I} \dup{\psi^i(\cdot)}{\ell^u_{\Routes_i}(\check{h}^i)} + \sum_{c^\dest\in\DestCyc}\sum_{\wa \in \Routes_i} \dup{\psi^i(\cdot)}{ \ell^u_{c^\dest}( \eflow_{\wa,c^\dest})}\\
  					\overset{\text{\Crefshort{lem: aggCostsVSwalkCosts}}}&{=}   \dup{\wttime(u,\cdot)}{\check{h}}  +\sum_{i \in I} \sum_{c^\dest\in\DestCyc} \sum_{\wa \in \Routes_i}  \dup{ \wttime_{(c^\dest,i)}(u,\cdot)}{\eflow_{\wa,c^\dest}}\\
  					&= \dup{\wttime(u,\cdot)}{\check{h}}
  				\end{aligned}
  			\end{equation}
  			For the last equality, note that $ \ell^u_{c^\dest}( \eflow_{\wa,c^\dest}) \leq \ell^u_{c^\dest}(\hat{h}_{c^\dest }) \leq u$ and thus, by assumption of used $\dest$-cycles having zero private costs (\ref{thm: mainSingleSink: UsedDCyclesHaveZeroD}),  we have the implication $\eflow_{\wa,c^\dest}(t)>0 \Rightarrow  \wttime_{(c^\dest,i)}(u,t) = 0$ for almost all $t \in \hori$ by \Cref{lem: Relations:h>0u>0:Pointwise}. 
  		\end{proofClaim}

  		Hence, it remains to add the zero-cycles to the above constructed walk inflow rates $\tilde{h}$ resulting in the desired flow $h$. To this end,  
  		 we will define the desired walk inflow rate $h$ recursively over the commodities. Before we do so, let us sketch the recursion idea briefly:  
  		For the first commodity, we set the entire edge flow on zero-cycles induced by  $\hat{h}_c,c\in \SimpCyc$ as the edge flow that has yet to be added to $\tilde{g}$.  
  		For each commodity $i$, we then consider the edge flow on zero cycles that has yet to be added to $\tilde{g}$ and consider the flow that results by adding these zero-cycle edge flows to commodity~$i$'s flow $\tilde{g}^i$. 
  		To this flow, we then apply \Cref{cor: PureFlowDecomp} to get 
  		a maximally pure $\source_i$,$\dest$-walk decomposition 
  		and set 
  		the inflow rates $h_\wa,\wa \in \Routes_i$ of commodity~$i$ 
  		to the $\source_i$,$\dest$-walk inflow rates of this decomposition. 
  		The remaining edge flow -- which can only be induced via zero cycle inflow rates -- is then defined 
  		as the edge flow that has yet to be added for the next commodity. 
  		For the so constructed walk inflow rates $h$, we then show that they induce $u$ by using the 
  		maximality of the pure $\source_i$,$\dest$-flow decompositions.
  		Formally, this is done as follows: 
  		
  		Fix  arbitrary representatives of $\tilde{h}$ and $\hat{h}_c,c \in \SimpCyc$ and assume 
  		for simplicity  that $I$ is of the form $I = \{1,\ldots,\abs{I}\}$.  
  		Set $\hat{h}^0_c\coloneq \hat{h}_c,c \in \SimpCyc$.  
  		Consider $i \in I$ and assume that 
  		$\hat{h}^{i'}$ has been constructed for all $i'\in\{0,\ldots,i-1\}$. 
  		Set $\eflow^i\coloneq  \tilde{\g}^i + \sum_{c \in \SimpCyc}\ell^u_c(\hat{h}_c^{i-1})$ and denote by 
  		$\check{h}_{\wa}^i,\wa \in \hat{\Routes},\check{h}^i_c,c\in \SimpCyc$ the maximally pure $\source_i$,$\dest$-flow decomposition constructed 
  		in \Cref{cor: PureFlowDecomp} \wrt $\eflow^i$ and $\tilde{h}_{\wa},\wa \in \Routes_i,\hat{h}_c^{i-1},c \in \SimpCyc$ and source, \sink pair $\source_i,\dest$.
  		We then set $h_{(\hat\wa,i)} = \check{h}^i_{\hat\wa},\hat\wa \in \hat{\Routes}_{\source_i,\dest}$ and $\hat{h}_c^i \coloneq  \check{h}^i_c,c \in \SimpCyc$. 
  		Denote by $C_{\n1^i},\mathfrak T_{C_{\n1^i}}^i,\n1^i\in \capn1^i$ for all $i \in \{0,\ldots, \abs{I}\}$ 
  		the sets defined in \Cref{def: ConnectedComp} for $\hat{h}_c^i,c \in \SimpCyc$. For the sake of readability, assume \wlg that the index sets $N^i,i\in I$ are disjoint.
  		We claim the following: 
  		\begin{claim} \label{claim: AddCyclesCommodities}
  			Set $\eflow^{\abs{I}+1} \coloneq  \sum_{c \in \SimpCyc}\ell^u_c(\hat{h}_c^{\abs{I}})$ and $\tilde{\g}^{\abs{I}+1}\coloneq 0$. Then, 
  			\begin{thmparts}
  				\item for all $ i \in [\abs{I}+1]$, we have the equality $u = \sum_{i'< i}\g^{i'} + \eflow^{i } + \sum_{i'>i }\tilde{\g}^{i'}$.
  				\label[thmpart]{claim: AddCyclesCommodities: enum1}
  				\item for all $ i \in \{1,\ldots\abs{I}\}$ and  all $\n1^{\abs{I}} \in \capn1^{\abs{I}}$, there exists $\n1^{i-1}\in \capn1^{i-1}$ such that  $C_{\n1^{\abs{I}}} 
  				= C_{\n1^{i-1}}$ and $\mathfrak T_{C_{\n1^{\abs{I}}}}^{\abs{I}} \subseteq \mathfrak T_{C_{\n1^{i-1}}}^{i-1}$. 
  				\label[thmpart]{claim: AddCyclesCommodities: enum2}
  			\end{thmparts}
  		\end{claim}
  		\begin{proofClaim}
  			We prove both parts separately: 
  			\begin{structuredproof}
  				\proofitem{\ref{claim: AddCyclesCommodities: enum1}}
  				We prove the claim via induction over $i\in [\abs{I}+1]$. 
  				\begin{proofbyinduction}
  					\basecase{$i = 1$} This is true as we have 
  					\begin{align*}
  						\eflow^1 + \sum_{i'> 1}  \tilde{\g}^{i'} = \tilde{\g}^1 + \sum_{c \in \SimpCyc}\ell^u_c(\hat{h}_c^{0}) + \sum_{i'> 1}  \tilde{\g}^{i'}  = \tilde\g + \sum_{c \in \SimpCyc}\ell^u_c(\hat{h}_c) = u. 
  					\end{align*}
  					
  					\inductionstep{$i-1 \to i$} Let $i \geq 2$ be arbitrary and assume we have shown the claim for all $i'<i$. 
  					We calculate: 
  					\begin{align*}
  						\sum_{i'< i}\g^{i'} + \eflow^i + \sum_{i'> i}  \tilde{\g}^{i'} &
  						= \sum_{i'< i-1}\g^{i'} +  \g^{i-1}  + \eflow^i  + \sum_{i'> i-1}  \tilde{\g}^{i'}  -   \tilde{\g}^{i}  \\
  						&\symoverset{1}{=} u - \eflow^{i-1} + \eflow^i  +  \g^{i-1} -   \tilde{\g}^{i}\\
  						&\symoverset{2}{=} u - \big(\g^{i-1} +\sum_{c \in \SimpCyc}\ell^u_c(\hat{h}_c^{i-1}) \big) + \eflow^i  + \g^{i-1} -   \tilde{\g}^{i}  \\
  						&\eqperdef u + \eflow^i  -  \big( \tilde{\g}^{i} +\sum_{c \in \SimpCyc}\ell^u_c(\hat{h}_c^{i-1}) \big)  = u.
  					\end{align*}
  					Here, we used for \refsym{1} the induction hypothesis and for \refsym{2} that $h_\wa,\wa\in \Routes_{i-1},\hat{h}_c^{i-1}, c \in \SimpCyc$ is a flow decomposition of $\eflow^{i-1}$. 
  				\end{proofbyinduction}
  				\proofitem{\ref{claim: AddCyclesCommodities: enum2}} 
  				By construction and \Cref{cor: PureFlowDecomp} we have for all $i \in [\abs{I}]$ and any $\n1^{i}\in \capn1^i$, that there exists $\n1^{i-1} \in \capn1^{i-1}$  such that  $C_{\n1^{i}}
  				= C_{\n1^{i-1}}$ and $\mathfrak T_{C_{\n1^{i}}}^i \subseteq \mathfrak T_{C_{\n1^{i-1}}}^{i-1}$. This implies part~\ref{claim: AddCyclesCommodities: enum2} of the \namecref{claim: AddCyclesCommodities}. \qedhere
  			\end{structuredproof}
  		\end{proofClaim}
  		 
  		It is clear that $\g^i$ only differs from $\tilde{\g}^i$ by flow on zero cycles for any $i \in I$. Hence, an analogue argument as in \eqref{eq: CostsAddDCycles} together with \Cref{ass: PCSep+ZeroCycles: ZeroCycles} show that $h$ is again an optimal 
  		solution to \eqref{opt: Master}. Thus, it remains to show that $\g = u$ holds. By \Cref{claim: AddCyclesCommodities: enum1}, we have 
  		$u= \g + \sum_{c \in \SimpCyc}\check{h}^{\abs{I}}_c$ and hence we have to show that 
  		 $\check{h}^{\abs{I}}_c = 0,c \in \SimpCyc$. 
  		The latter is in turn equivalent to  
  		$\capn1^{\abs{I}} = \emptyset$ (cf.~\Cref{def: ConnectedComp}). 
  		
  		Hence, assume for the sake of a contradiction that there exists $\n1^{|I|}\in\capn1^{\abs{I}}$ with $\leb(\mathfrak T_{C_{\n1^{\abs{I}}}}^{\abs{I}})>0$. 
  		Denote by $\n1^i,i \in I$ the indices corresponding to $\n1^{\abs{I}}$ in the sense of \Cref{claim: AddCyclesCommodities: enum2}. 
  		We argue now that neither \Cref{thm: PureFlowDecomp: Dest} nor \Cref{thm: PureFlowDecomp: NotDest} are  satisfied for $u$, $h_\wa,\wa \in \Routes$,  $\check{h}^{\abs{I}}_c,c \in \SimpCyc$ and almost all  $t \in \mathfrak T_{C_{\n1^{\abs{I}}}}^{\abs{I}}$. 
  		Note that by \Cref{claim: AddCyclesCommodities: enum1}, 
  		$h_\wa,\wa \in \Routes$,  $\check{h}^{\abs{I}}_c,c \in \SimpCyc$ 
  		is a flow decomposition of~$u$. 
  		This then yields a contradiction as $\leb(\mathfrak T_{C_{\n1^{\abs{I}}}}^{\abs{I}})>0$ 
  		but $u$ has a flow decomposition purely into \stwalk s since $u$ is the edge flow corresponding to some $h' \in \wir$ by assumption.  
  		
  		Assume for the sake of a contradiction that \Cref{thm: PureFlowDecomp: Dest} was  satisfied for a subset  $\mathfrak T \subseteq \mathfrak T_{C_{\n1^{\abs{I}}}}^{\abs{I}}$ with positive measure, i.e.~$\inflow^u_\dest(t)<0$ for a.e.~$t\in \mathfrak T$ and $\dest \in \GV_{C_{\n1^{\abs{I}}}}$. Here, 
  		we  denote by $\inflow^u_\dest$ the outflow rate of~$u$ at~$\dest$ (which exists by \Cref{ass: u: outflow}) and similarly by $\inflow_\dest^i$ the outflow rate of~$\g^i$ at~$\dest$ (which equals the outflow rate of  $\tilde{g}_i$ since these flows only differ in zero-cycle flows).
  		Note that these outflow rates exist by $\g^i\leq u, i \in I$ and we have $\inflow^u_\dest = \sum_{i \in I}\inflow_\dest^i$. 
  		The latter implies in particular that there has to exist $i \in I$ and a further subset $\mathfrak T' \subseteq \mathfrak T$ with positive measure 
  		with $\inflow^i_\dest (t)<0$ for a.e.~$t \in \mathfrak T'$.
  		Yet, $\mathfrak T' \subseteq \mathfrak{T}_{C_{\n1^{\abs{I}}}}^{\abs{I}} \subseteq   \mathfrak T_{C_{\n1^i}}^{i}$ 
  		and $C_{\n1^{\abs{I}}} = C_{\n1^i}$ (by choice of $\n1^i$), contradicting the property of the maximally pure $\source_i$,$\dest$-flow decomposition $\check{h}_{\wa}^i,\wa \in \hat{\Routes},\check{h}^i_c,c\in \SimpCyc$ of $\eflow^i$ constructed in \Cref{cor: PureFlowDecomp}.
  		
  		Analogously, assume for the sake of a contradiction that \Cref{thm: PureFlowDecomp: NotDest} was  satisfied for a subset  $\mathfrak T \subseteq \mathfrak T_{C_{\n1^{\abs{I}}}}^{\abs{I}}$ with positive measure 
  		\wrt an $\arc =(v,v') \in \edgesFrom{C_{\n1^{\abs{I}}}}$ with  $u_\arc (t)>0$ for a.e.~$t\in \mathfrak T$. 
  		Now we have $\sum_{i \in I}\g^i_\arc +  \sum_{c \in \SimpCyc:\arc \in c} \hat{h}_c^{\abs{I}} = u_\arc$ 
  		by \Cref{claim: AddCyclesCommodities: enum1} and \Cref{lem: FLowOnZeroTrav}. 
  		Moreover,  for all $c \in \SimpCyc$ 
  		with $\arc \in c$ we  have $ \hat{h}_c^{\abs{I}} = 0$ by $\arc \notin \GA_{C_{\n1^{\abs{I}}}}$. 
  		Combining these two observations 
  		lets us deduce that 
  		$\sum_{i \in I}\g^i_\arc(t)>0$ for a.e.~$t \in \mathfrak T$ has to hold. 
  		Hence, 
  		there has to exist $i \in I$ and a further subset $\mathfrak T' \subseteq \mathfrak T $ with positive measure 
  		such that $\g^{i}_\arc(t)>0$ for a.e.~$t\in \mathfrak T'$. As in the other case, this contradicts  the property of the maximally pure $\source_i$,$\dest$-flow decomposition $\check{h}_{\wa}^i,\wa \in \hat{\Routes},\check{h}^i_c,c\in \SimpCyc$ of $\eflow^i$ constructed in \Cref{cor: PureFlowDecomp}. 
  		
  		Hence, we can conclude that $\g = u$ and since $h$ is optimal for \eqref{opt: Master}, the proof of $\ref{thm: mainSingleSink: UsedDCyclesHaveZeroD}\Rightarrow~\ref{thm: mainSingleSink: Optimalh}$ is finished which in turn finishes the entire proof. \qedhere
  	\end{structuredproof}
  \end{proof}
  In the following \namecref{lem: SimplificationsFlowOnlyOnZeroDestCycles}, we prove that 
  it suffices to check the conditions \ref{thm: mainSingleSink: ReachableEdgesHaveZeroD} and \ref{thm: mainSingleSink: UsedDCyclesHaveZeroD} of \Cref{thm: mainSingleSink} only for a single representative of $u$ and only those $\dest$-cycles that visit $\dest$ exactly once (for \ref{thm: mainSingleSink: UsedDCyclesHaveZeroD}):
  
  \begin{lemma}\label{lem: SimplificationsFlowOnlyOnZeroDestCycles}
  	The following statements are valid: 
  	\begin{thmparts}
  		\item If one representative of $u$ fulfills \Cref{thm: mainSingleSink: UsedDCyclesHaveZeroD}, then already all representatives fulfill this condition.  \label[thmpart]{lem: SimplificationsFlowOnlyOnZeroDestCycles: Repre}
  		
  		\item If one representative of $u$ fulfills \Cref{thm: mainSingleSink: ReachableEdgesHaveZeroD}, then already all representatives fulfill this condition.  \label[thmpart]{lem: SimplificationsFlowOnlyOnZeroDestCycles: Walks}
  		
  		\item If $u$ fulfills \Cref{thm: mainSingleSink: UsedDCyclesHaveZeroD} for $\dest$-cycles $c \in \DestCyc$ that only visit $\dest$ once, i.e.~$c[j] \notin \delta^-(\dest), j < \abs{c}$, then it  already holds for all $\dest$-cycles. \label[thmpart]{lem: SimplificationsFlowOnlyOnZeroDestCycles: Cycles}
  	\end{thmparts}
  \end{lemma}
  \begin{proof}
  	\begin{structuredproof}
  		\proofitem{\ref{lem: SimplificationsFlowOnlyOnZeroDestCycles: Repre}} 
  		Let $u^1,u^2$ be two representatives of $u$ with $u^1$ fulfilling \Cref{thm: mainSingleSink: UsedDCyclesHaveZeroD}. 
  		Let $\mathfrak T^*_i,i = 1,2$ be the set of times where the latter is not fulfilled for $u^i,i = 1,2$. 
  		By assumption $\mathfrak T^*_1$ is a null set. We have to show that $\mathfrak T^*_2$ is likewise. We first observe that we can represent the latter as follows: 
  		\begin{align*}
  			&\mathfrak T^*_2 = \bigcup_{c \in \DestCyc} \bigcup_{j \leq |c|} \mathfrak T_{c,j} \text{ where } \mathfrak T_{c,j}:= \left\{t \in \hori \,\middle|\, \begin{array}{c}
  				u^2_{c[j']}(\arr_{c,j}(u,t))>0, j'\leq \abs{c}    \\
  				\psi_{c[j']}(\arr_{c,j}(u,t)) =0, j'<j,  \psi_{c[j]}(\arr_{c,j}(u,t)) \neq0
  			\end{array} \right\}. 
  		\end{align*}
  		Now for any $c \in \DestCyc$ and $j\leq c$, we have 
  		\begin{align*}
  			\mathfrak T_{c,j} &\subseteq \mathfrak T^*_1 \cup  \bigcup_{\tilde{j}\leq \abs{c}}
  			\left\{t \in \hori \,\middle|\, \begin{array}{c}
  				u^2_{c[j']}(\arr_{c,j'}(u,t)) = u^1_{c[j']}(\arr_{c,j'}(u,t))>0, j'< \tilde{j}    \\
  				u^2_{c[\tilde{j}]}(\arr_{c,\tilde{j}}(u,t)) > 0 \geq  u^1_{c[\tilde{j}]}(\arr_{c,\tilde{j}}(u,t))
  			\end{array} \right\} \\
  			&\subseteq \mathfrak T^*_1 \cup  \bigcup_{\tilde{j}\leq \abs{c}}
  			\left\{t \in \Dwg[c_{<\tilde{j}}][u^2] \,\middle|\, 
  			u^2_{c[\tilde{j}]}(\arr_{c,\tilde{j}}(u,t)) > 0 \geq  u^1_{c[\tilde{j}]}(\arr_{c,\tilde{j}}(u,t))\right\} \\
  			&\subseteq \mathfrak T^*_1 \cup  \bigcup_{\tilde{j}\leq \abs{c}} \arr_{c,\tilde{j}}(u,\cdot)^{-1}\Big(\left\lbrace t \in\arr_{c,\tilde{j}}(u,\cdot)(\Dwg[c_{<\tilde{j}}][u^2] )   \mid    u^2_{c[\tilde{j}]}(t)>0 \geq u^1_{c[\tilde{j}]}(t)  \right\rbrace \Big) \cap  \Dwg[c_{<\tilde{j}}][u^2] 
  		\end{align*}
  		where $\Dwg[c_{<\tilde{j}}][u^2] $ is the set defined in \eqref{eq: defD_wa} \wrt $u^2$. 
  		Now observe that $\mathfrak T^*_1$ is a null set by assumption and \Cref{lem:  g>0Arr'>0} implies that the other sets appearing in the last line  
  		are null sets likewise. 
  		For the latter, note that $\{t \in \hori \mid   u^2_{c[\tilde{j}]}(t)>0 \geq u^1_{c[\tilde{j}]}(t)\}$ is a null set because $u^1,u^2$ are both representatives of $u$. 
  		Hence,  $ \mathfrak T_{c,j}$ is a null set. Subsequently, $\mathfrak T^*_2$ is likewise a null set as a countable union of null sets. 
  		
  		\proofitem{\ref{lem: SimplificationsFlowOnlyOnZeroDestCycles: Walks}} The proof works completely analogous to the proof of \ref{lem: SimplificationsFlowOnlyOnZeroDestCycles: Repre}. 
  		
  		\proofitem{\ref{lem: SimplificationsFlowOnlyOnZeroDestCycles: Cycles}}
  		
  		Let us fix an arbitrary representative of $u$. 
  		Define for any cycle $c$ the set 
  		\begin{align*}
  			\mathfrak D_{c>}:=
  			\Dwg[c][u]
  			 \cap \{t \in \hori \mid \exists j \leq |c|:  \psi_{c[j]}(\arr_{c,j}(t)) \neq 0 \}.
  		\end{align*} 
  		where $\Dwg[c][u] $ is the set defined in \eqref{eq: defD_wa} \wrt $u$. 
  		Let us denote by $\DestCyc_1$ the set of all  $\dest$-cycles visiting $\dest$ only once and by $\DestCyc$ the set of all $\dest$-cycles. We want to show that all $ \mathfrak D_{c>},{c \in \DestCyc_1}$ being null sets implies that 
  		$\bigcup_{c \in \DestCyc}\mathfrak D_{c>}$ is also a null set. For this, it is sufficient to show that $\mathfrak D_{c>}$ is a null set for any $c \in \DestCyc$. 
  		Let $c \in \DestCyc$ be an arbitrary $\dest$-cycle in the following. Then we can write $c$ as   $c=(c_1,\ldots,c_k)$ with $c_l \in \DestCyc_1$ being a $\dest$-cycle 
  		which only visits $\dest$ once 
  		for every $l\leq k$. 
  		Let $j_l$ be the index of the edge of $c$ that corresponds to the first edge of $c_l$. 
  		We observe that 
  		\begin{align*}
  			\mathfrak D_{c>} = \bigcup_{l \leq k} \big \{ t \in \mathfrak D_{c>}  \mid \arr_{c,j_l}(u,t) \in  \mathfrak D_{c_l>} \big\} = 
  			\bigcup_{l \leq k} \arr_{c,j_l}(u,\cdot)^{-1}(\mathfrak D_{c_l>}) \cap  \mathfrak D_{c>}
  		\end{align*}
  		as well as that $\mathfrak D_{c_l>} \subseteq \arr_{c,j_l}(u,\cdot)(\mathfrak D_{c>})$. Since, by assumption, $\mathfrak D_{c_l>}$ is a null set, \Cref{lem:  g>0Arr'>0} implies that all of the sets on the right of the above equation are null sets and hence so is $\mathfrak D_{c>}$. \qedhere
  	\end{structuredproof}
  \end{proof}

\subsubsection{The Role of Assumption~\ref*{ass: PCSep+ZeroCycles}\ref*{ass: PCSep+ZeroCycles: ZeroCycles}}

In this section, we investigate the question whether the assumption that flow-carrying zero-cycles incur zero private costs (\Cref{ass: PCSep+ZeroCycles: ZeroCycles}) can be incorporated into the \namecref{thm: mainSingleSink} as an additional condition. While this will clearly result in a sufficient condition for implementability, it is in fact not necessary (cf.~\Cref{exa: CounterZeroCyclesNecessary}). 
 Only for the special case of private costs being commodity-independent, this results in a characterization of implementability. 
 In particular, the latter shows that \Cref{ass: PCSep+ZeroCycles: ZeroCycles} can also not be simply   dropped in \Cref{thm: mainSingleSink}.

 \begin{lemma}\label{lem: mainSingleSinkZeroCycles}
 	For a multi-source, single-\sink network in which \Cref{ass: u}, \Cref{ass: PCSep+ZeroCycles: Sep,ass: PCSep+ZeroCycles: DestCycles} are satisfied and \eqref{opt: Master} is \wellposed and admits strong duality \wrt $\MeasFuncUInt$,  consider the following statements:  
 	\begin{thmparts}
 		\item $u$ is implementable. \label[thmpart]{lem: mainSingleSinkZeroCycles: impl} 
 		
 		\item Every used edge that lies on a used zero-cycle or $\dest$-walk has zero private costs, i.e.\   an arbitrary representative of $u$ satisfies:  
 		\begin{itemize}

 			\item for all cycles $c \in \SimpCyc$, all $j \leq \abs{c}$ and almost all $t \in \hori$, the following holds:
 			\begin{align*}
 				\Big( u_{c[j']}(t)>0 \,\land\,\forall   j'\leq \abs{c}: \trav_{c[j']}(u,t) = 0\Big)  \implies \psi_{c[j]}(t) = 0.
 			\end{align*}
 			\item for all walks $\wa\in \hat{\Routes}^\dest$, all $j \leq |\wa|$ and almost all $t \in \hori$, the following holds: 
 			\begin{align*}
 				\Big(\forall j'\leq \abs{\wa} :  u_{\wa[j']}(\arr_{\wa,j'}(u,t))>0  \Big)\implies \psi_{\wa[j]}(t) = 0.
 			\end{align*}
 		\end{itemize} \label[thmpart]{lem: mainSingleSinkZeroCycles: ReachableEdgesHaveZeroD}
 		
 		\item $u$ only sends flow along a zero- or $\dest$-cycle if it has zero private costs, i.e.\ an arbitrary representative of $u$ satisfies  
 		\begin{itemize}
 			\item for all $c \in \SimpCyc$ and almost all $t \in \hori$ the following implication holds:
 			\begin{align*}
 				\Big( u_{c[j']}(t)>0  \,\land\, \forall  j'\leq \abs{c}:\trav_{c[j']}(u,t) = 0\Big)  \implies \wttime_{c}(u,t) = 0 .
 			\end{align*}   
 			\item for all $c \in \DestCyc$ and almost all $t \in \hori$  the following implication holds: 
 			\begin{align*}
 				\Big(\forall j'\leq \abs{c}:  u_{c[j']}(\arr_{c,j'}(u,t))>0 \Big)  \implies \wttime_{c}(u,t) = 0 .
 			\end{align*}   
 		\end{itemize}  \label[thmpart]{lem: mainSingleSinkZeroCycles: UsedDCyclesHaveZeroD} 
 		
 		\item There exists $h^* \in \wir$ optimal for \eqref{opt: Master} with tight inequality \eqref{ineq: Master}. \label[thmpart]{lem: mainSingleSinkZeroCycles: Optimalh}
 	\end{thmparts}
 	Then, the following implications hold: \begin{align} \label{eq: lem: mainSingleSinkZeroCycles: Implications}
 		\text{ \ref{lem: mainSingleSinkZeroCycles: ReachableEdgesHaveZeroD} }
 		\implies \text{ \ref{lem: mainSingleSinkZeroCycles: UsedDCyclesHaveZeroD} }
 		\implies \text{ \ref{lem: mainSingleSinkZeroCycles: Optimalh} } 
 		\iff \text{ \ref{lem: mainSingleSinkZeroCycles: impl} }.
 	\end{align}
 	In case that $\psi^i=\psi$ for all $i \in I$, also $\text{  \ref{lem: mainSingleSinkZeroCycles: Optimalh} } \Rightarrow \text{ \ref{lem: mainSingleSinkZeroCycles: ReachableEdgesHaveZeroD} }$ holds and subsequently all of the above statements are equivalent. 
 \end{lemma}
 \begin{proof}
 	It is clear by the same arguments as in \Cref{thm: mainSingleSink} that the condition regarding zero-cycles in  \ref{lem: mainSingleSinkZeroCycles: ReachableEdgesHaveZeroD}   implies the corresponding one in \ref{lem: mainSingleSinkZeroCycles: UsedDCyclesHaveZeroD}. 
    Moreover, the implication $\text{ \ref{lem: mainSingleSinkZeroCycles: UsedDCyclesHaveZeroD} }
 		\Rightarrow \text{ \ref{lem: mainSingleSinkZeroCycles: Optimalh} }$ holds by  \Cref{thm: mainSingleSink}. 
        Finally, we know by \Cref{thm: mainMSSS} that $\text{ \ref{lem: mainSingleSinkZeroCycles: impl} } \Leftrightarrow \text{ \ref{lem: mainSingleSinkZeroCycles: Optimalh} }$ is valid and, hence, \eqref{eq: lem: mainSingleSinkZeroCycles: Implications} is shown. 
 	
 	Now assume that $\psi^i=\psi$ for all $i \in I$ holds. We argue in the following that then also $\text{ \ref{lem: mainSingleSinkZeroCycles: impl} } \Rightarrow \text{ \ref{lem: mainSingleSinkZeroCycles: ReachableEdgesHaveZeroD} }$ holds. 
 	Assume for the sake of a contradiction that there exists a cycle $c \in \SimpCyc$ and $j \leq \abs{c}$ such that the implication in \ref{lem: mainSingleSinkZeroCycles: ReachableEdgesHaveZeroD} is not fulfilled. Note that it is sufficient to derive a contradiction by this assumption as the second condition then holds again by \Cref{thm: mainSingleSink}. We argue in the following that \ref{lem: mainSingleSinkZeroCycles: Optimalh} can not hold. 
 	
 	By \Cref{lem: aggCostsVSwalkCosts} and $\psi^i=\psi$ for all $i \in I$, we can rewrite the objective of the master problem via 
 	$\dup{\psi}{\ell^u(h)}$ for any feasible $h$. Now let $\mathfrak T \in \mathcal{B}(\hori),\sigma(\mathfrak T)>0$ be a set where for all $t \in \mathfrak T$ the implication in \ref{lem: mainSingleSinkZeroCycles: ReachableEdgesHaveZeroD} is not fulfilled \wrt $c \in \SimpCyc$ and $j \leq \abs{c}$. 
 	Fix an arbitrary representative of $u$. 
 	Then, $\hat{h}_c(t) := \min_{j'\leq \abs{c}}u_{c[j']}(t)>0$ for $t \in\mathfrak T$ and $ \hat{h}_c(t) :=0 $ else defines a zero-cycle inflow rate with $\ell^u_c(\hat{h}_c)\leq u$ (cf.~\Cref{lem: FLowOnZeroTrav}). 
 	Now the difference $u- \ell^u_c(\hat{h}_c)$ is nonnegative and has the same net outflow rates as $u$ at all nodes since $ \hat{h}_c$ is a zero-cycle inflow rate. In particular, by a super-source argument together with  the flow decomposition theorem (\Cref{thm: FlowDecomp}) (cf.~also the proof of \Cref{thm: AlmostMasterZeroDual}), we can find a walk inflow rate $\tilde{h} \in \edom{\Routes}[u] \cap \wir$ inducing $u- \ell^u_c(\hat{h}_c)$. 
 	Then, the difference of the objective values of any $h$ inducing $u$ and $\tilde{h}$ is given by $\dup{\psi}{\ell^u_c(\hat{h}_c)}$ which is strictly larger than zero as
 	\begin{align*}
 		\dup{\psi}{\ell^u_c(\hat{h}_c)} \geq \dup{\psi_{c[j]}}{\ell^u_{c,j}(\hat{h}_c)} \geq \int_{\mathfrak T}\psi_{c[j]} \cdot \ell^u_{c,j}(\hat{h}_c) \di \sigma \symoverset{1}{=}  \int_{\mathfrak T}\psi_{c[j]} \cdot  \hat{h}_c  \di \sigma >0
 	\end{align*}
 	where we used \Cref{lem: FLowOnZeroTrav} in \refsym{1}  and for the strict inequality that $\psi_{c[j]}(t),\hat{h}_c(t) >0$ on $\mathfrak T$ and $\sigma(\mathfrak T)>0$. Hence, \ref{lem: mainSingleSinkZeroCycles: Optimalh} can not hold which yields the desired contradiction as outlined before. 
 \end{proof}

  Now let us demonstrate with the following example that the implication $\text{ \ref{lem: mainSingleSinkZeroCycles: impl} } \Rightarrow \text{ \ref{lem: mainSingleSinkZeroCycles: ReachableEdgesHaveZeroD} }$ does not need to hold in general, even in the single-source case and even if the private costs are a commodity specific weighted variant of a common private cost function, i.e.~if \eqref{eq: WeightedVarian} holds. 
  \begin{example}\label{exa: CounterZeroCyclesNecessary}
  	Consider the network depicted in \Cref{fig: CounterZeroCyclesNecessary}. 
  	There are two commodities with an inflow rate of $\inflow_i = 1_{[0,1]},i =1,2$ and 
  	common source $\source$ and \sink $\dest$. The private costs for each edge are given by a commodity-specific weighted variant $\gamma_i\cdot \psi_\arc$ of a common (constant) cost function $\psi_\arc$ with $\gamma_1 = 1$ and $\gamma_2 = 100$. 
  	All edges have constant travel time of zero and 
  	we consider the edge flow $u$ given by $u_\arc=1_{[0,1]},\arc \in \GA$. 
  	The tuple on the edges correspond to the common cost function $\psi$ and tolls $\prices$ that implement $u$ via $h_{\wa_1}=h_{\wa_2} = 1_{[0,1]}$ and $h_\wa = 0,\wa \notin \{\wa_1,\wa_2\}$ for $\wa_1 =((\arc_1,\arc_4,\arc_6),1)$ and $\wa_2 =((\arc_2,\arc_3,\arc_5),2)$. 
  	The fact that these tolls make $h$ into a $\prices$-DUE is immediately verified by considering the total costs along the four possible simple walks (as every proper walk containing the cycle has strictly larger total costs than one of the simple walks). 
  	Clearly, $u$ sends flow along a cycle of travel time zero, yet is implementable. 
  \end{example}
  
  \begin{figure}[h]
  	\centering
  	\BigPicture[0]{
  		\def\animationdatapath{tikz/}
  		\renewcommand{\flowcolorlist}{fcolAr,fcolBr,fcolCr,fcolDr,fcolEr,fcolFr}
  		\begin{adjustbox}{max width=\textwidth}
  			\input{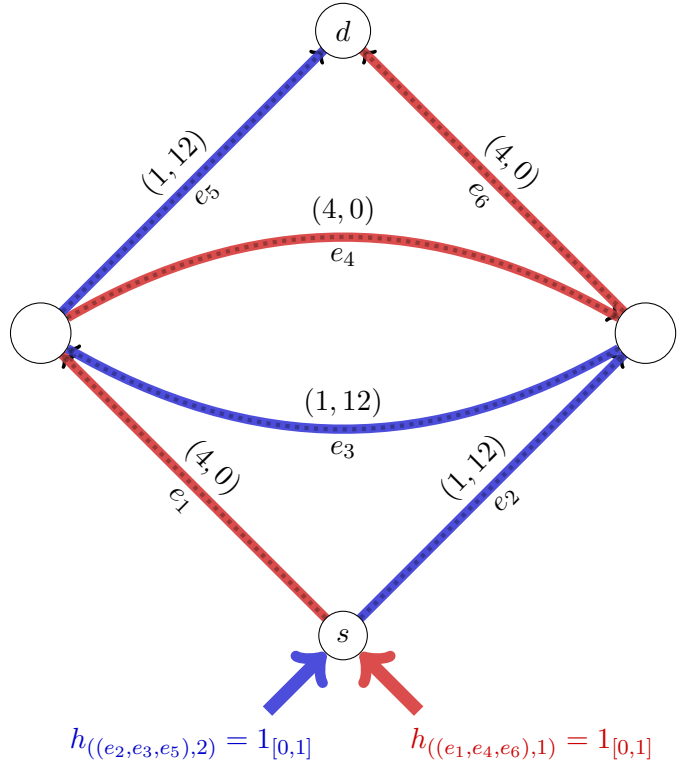}
  		\end{adjustbox}
  	}
  	\caption{The network and flow from \Cref{exa: CounterZeroCyclesNecessary}. 
  		All edge travel times are zero. The tuple on the edges correspond to the common (constant) cost function   and tolls $(\psi_\arc,\prices_\arc)$ that implement $u=(1_{[0,1]})_{\arc \in \GA}$.  
  		Clearly, $u$ sends flow along a cycle of travel time zero, yet is implementable. }
  	\label{fig: CounterZeroCyclesNecessary}
  \end{figure}

 \subsubsection{Necessary and/or Sufficient Combinatorial Conditions for Multi-Destination Networks}
In this section, we investigate whether 
 or not a similar characterization of implementability as in \Cref{thm: mainSingleSink} holds for the general multi-source, multi-\sink case. 
 \Cref{thm: mainSingleSink: ReachableEdgesHaveZeroD} and \Cref{thm: mainSingleSink: UsedDCyclesHaveZeroD} require that 
 flow is send  into $\dest$-walks/cycles only if they have zero private costs (for all commodities) and  
 we have two natural ways of adapting these statements to the multi-\sink case: Either we require the properties to hold for the commodity-aggregated flow~$u$ at all \sink[s] or require it to hold for each commodity $i$ separately for the commodity-specific flow~$u^i$ at the respective sink $\dest_i$. In the following lemmata and examples we show that some of the four resulting statements yield necessary conditions, while others yield sufficient ones.

  \begin{lemma}\label{lem: NoFlowOnOutgoingEdges} 
  	Consider the situation of \Cref{thm: mainMSMS} together with a flow $u$ that does not send flow along  zero-cycles nor sends flow into edges leaving any destination $\dest_i$, i.e.~$u_\arc = 0,\arc\in \edgesFrom{\dest_i}, i\in I$. Then, $u$ is implementable via bounded tolls.
  \end{lemma}
  \begin{proof}
  	Consider an optimal solution $\tilde{h }$ to \eqref{opt: Master} (which exists by strong duality) with its aggregated edge flow ${\tilde\g}:= \ell^u(\tilde{h}) \leq u$. Remark that the last inequality holds by feasibility of $\tilde{h}$ for \eqref{opt: Master}. We argue in the following that $u=\tilde{\g}$ holds, implying by \Cref{thm: mainMSMS} that $u$ is implementable via bounded tolls. 
  	
  	By \Cref{lem: flowconW'}, $u- \tilde{\g}$ fulfills flow conservation at all nodes $v \notin \{\dest_i\}_{i\in I}$. 
  	We will prove in the \namecref{claim: NoFlowOnOutgoingEdges} below (\Cref{claim: NoFlowOnOutgoingEdges: 3}) that 
  	$\tilde{\g}_\arc = u_\arc$ holds for all $\arc \in \edgesTo{\dest_i},i\in I$. This, together with $\tilde{\g}_\arc \leq u_\arc = 0$ for all $\arc \in \edgesFrom{\dest_i},i\in I$ by assumption, implies that $\tilde{\g}$ and $u$ also have the same node balance at any destination. In particular, we get that the difference $u-\tilde{\g}$ fulfills flow conservation at all nodes. 
  	Thus,  \Cref{lem: ZeroCycleDecomposition} implies that the difference is composed of zero-cycle flows and since $u$ does not contain such zero-cycle flows by assumption, $u = \tilde{\g}$ has to hold. 
  	It thus remains to show that $\tilde{\g}_\arc = u_\arc$ holds for all $\arc \in \edgesTo{\dest_i},i\in I$. We show this in the following \namecref{claim: NoFlowOnOutgoingEdges} in three steps: 
  	\begin{claim}\label{claim: NoFlowOnOutgoingEdges}
  		For any $v \in \{\dest_i\}_{i \in I}$, the following statements are valid:
  		\begin{thmparts}
  			\item \label[thmpart]{claim: NoFlowOnOutgoingEdges: 1}  The following equality holds:
  			\begin{align*}
  				\sum_{\arc \in \edgesTo{v}}  \sum_{i \in I:\dest_i = v} \int_\hori \tilde{\g}^i_\arc \di\sigma =   \sum_{\arc \in \edgesTo{v}} \int_\hori u_\arc \di\sigma . 
  			\end{align*}
  			\item\label[thmpart]{claim: NoFlowOnOutgoingEdges: 2}  $\sum_{i \in I:\dest_i = v}\tilde{\g}^i_\arc = \tilde{\g}_\arc$ holds for all $\arc \in \edgesTo{v}$.
  			\item \label[thmpart]{claim: NoFlowOnOutgoingEdges: 3}  $\tilde{\g}_\arc  = u_\arc$ holds for all $\arc \in \edgesTo{v}$. 
  		\end{thmparts}
  	\end{claim}
  	\begin{proofClaim}
  		Consider an arbitrary $v \in \{\dest_i\}_{i \in I}$. We first note that $\tilde\g \leq u$ implies that we have  
  		\begin{align}\label{eq: claim: NoFlowOnOutgoingEdges}
  			\tilde{\g}^i_\arc \leq u_\arc = 0 \text{ for all } i\in I \text{ with } \dest_i = v \text{ and } \arc \in \edgesFrom{\dest_i}. 
  		\end{align}
  		\begin{structuredproof}
  			\proofitem{\ref{claim: NoFlowOnOutgoingEdges: 1}}
  			\Cref{lem: flowconW'} and  \eqref{eq: claim: NoFlowOnOutgoingEdges} let us deduce: 
  			\begin{align*}
  				-  \sum_{\arc \in \edgesTo{v}}  \sum_{i \in I:\dest_i = v} \int_\hori \tilde{\g}^i_\arc \di\sigma \overset{\eqref{eq: claim: NoFlowOnOutgoingEdges}}&{=}   \sum_{i \in I:\dest_i = v}  \sum_{\arc \in \edgesFrom{\dest_i}} \int_\hori \tilde{\g}^i_\arc\di\sigma -         \sum_{\arc \in \edgesTo{\dest_i}} \int_\hori \tilde{\g}^i_\arc \di\sigma    \\\overset{\text{\Crefshort{lem: flowconW'}}}&{=} -\sum_{i \in I:\dest_i = v} \int_\hori \inflow_i \di\sigma \overset{\text{\Crefshort{lem: flowconW'}}}{=} -   \sum_{\arc \in \edgesTo{v}} \int_\hori u_\arc \di\sigma. 
  			\end{align*} 
  			\proofitem{\ref{claim: NoFlowOnOutgoingEdges: 2}} 
  			This is a direct consequence of \eqref{eq: claim: NoFlowOnOutgoingEdges}  since flow from a   commodity with $\dest_i \neq v$ can not visit $v$ as it  would need to leave $v$ again.  
  			
  			\proofitem{\ref{claim: NoFlowOnOutgoingEdges: 3}}  This is a direct consequence of \ref{claim: NoFlowOnOutgoingEdges: 1} and \ref{claim: NoFlowOnOutgoingEdges: 2} together with  $\tilde{g}\leq u$. \qedhere
  		\end{structuredproof}
  	\end{proofClaim}
  \end{proof}
  
  A simple example  shows that the sufficient condition in \Cref{lem: NoFlowOnOutgoingEdges} is not also necessary for implementability:

  \begin{figure}
  	\centering
  	\BigPicture[3]{
  		\begin{adjustbox}{max width=\textwidth}
  			\input{tikz/CE_ImplementableImpliesNoCycles}
  		\end{adjustbox}
  	}
  	\caption{The network and flow from \Cref{exa: NoFlowOnOutgoingEdges}. 
  		}
  	\label{fig: NoFlowOnOutgoingEdges}
  \end{figure}

  \begin{example}\label{exa: NoFlowOnOutgoingEdges}
  	We consider the network depicted in \Cref{fig: NoFlowOnOutgoingEdges} with flow-independent, constant travel times of~$1$ on all edges. There are $2$ commodities with identical VoT parameters $\gamma_i = 1,i \in I$, the   source, \sink pairs $\source_i,\dest_i, i \in I$ and the commodities' network inflow rates being  equal to $r_1:= 1_{[0,1]}$ and $r_2 := 1_{[1,2]}$. 
  	
  	Consider the walk flow $h_{\wa_i} = \inflow_i, i \in I$ for  $ \wa_1:=((\arc_1),1)$ and $\wa_2:=((\arc_2),2)$. 
  	The corresponding induced edge flow is given by $u_{\arc_1}= 1_{[0,1]}$ and $u_{\arc_2}= 1_{[1,2]}$. 
  	The latter is clearly implementable (with arbitrary tolls), yet, does not satisfy  the sufficient condition in \Cref{lem: NoFlowOnOutgoingEdges}, demonstrating that it is not necessary. 
  \end{example}
  
  Note that \Cref{exa: NoFlowOnOutgoingEdges} also shows that the following adapted version of \Cref{thm: mainSingleSink: UsedDCyclesHaveZeroD}, which is a weaker condition than the one stated in \Cref{lem: NoFlowOnOutgoingEdges}, is also not a necessary condition: 
  \begin{itemize}
  	\item For all destinations $\dest_i,i\in I$, there is no $\dest_i$-cycle carrying flow under $u$.
  \end{itemize}

  Next, we investigate an adaptation of \Cref{thm: mainSingleSink} in which we require the properties for each commodity separately. 
  
  \begin{lemma}\label{lem: NecesCommoditySepara}
  	Let $h\in \wir$ be a walk inflow rate  that induces $u$ for the multi-source, multi-\sink case with $u^i=\ell^u_{\Routes_i}(h^i),i \in I$ being the commodity-specific edge flows.
  	Under \Cref{ass: u: outflow} and \Cref{ass: PCSep+ZeroCycles: Sep}, the chain of implications $\ref{lem: NecesCommoditySepara:Optimal} \Rightarrow\ref{lem: NecesCommoditySepara:ReachableEdgesHaveZeroD} \Rightarrow \ref{lem: NecesCommoditySepara:UsedDCyclesHaveZeroD}$ holds for the following statements and all $i\in I$:  
  	\begin{thmparts} 
  		\item $u$ is implementable via $h$. \label[thmpart]{lem: NecesCommoditySepara:Optimal}
  		\item Every used edge under $u^i$ that is reachable from $\dest_i$ via used edges under $u^i$, has zero private costs, i.e.~we have the following implication for almost all $t \in \hori$ and all walks $\wa \in \hat{\Routes}^{\dest_i}$  and $j \leq |\wa|$: 
  		\begin{align*}
  			u^i_{\wa[j']}(\arr_{\wa,j'}(u,t)) > 0 \text{ for all } j'\leq j  \implies \psi^i_{\wa[j]}(u,\arr_{\wa,j}(u,t)) = 0.
  		\end{align*}\label[thmpart]{lem: NecesCommoditySepara:ReachableEdgesHaveZeroD}
  		\item  $u^i$ only sends flow along a $\dest_i$-cycle if the latter has zero private costs, i.e.~we have the following implication 
  		for an arbitrary representative of $u^i$, for almost all $ t\in \hori$ and  all $\dest_i$-cycles $c$: 
  		\begin{align*}
  			u^i_{c[j]}(\arr_{c,j}(u,t)) >0  \text{ for all } j \leq \abs{c}   \implies \psi^i_{c[j]}(u,\arr_{c,j}(u,t)) = 0 \text{ for all } j \leq \abs{c}.
  		\end{align*}  \label[thmpart]{lem: NecesCommoditySepara:UsedDCyclesHaveZeroD}
  	\end{thmparts}
  \end{lemma}
  \begin{proof} 
  	The implications $\ref{lem: NecesCommoditySepara:Optimal} \Rightarrow \ref{lem: NecesCommoditySepara:ReachableEdgesHaveZeroD} \Rightarrow \ref{lem: NecesCommoditySepara:UsedDCyclesHaveZeroD}$  follow completely analogous to the corresponding implications in \Cref{thm: mainSingleSink}. 
  	Remark that for these implications, only \Cref{ass: u: outflow} and \Cref{ass: PCSep+ZeroCycles: Sep,ass: PCSep+ZeroCycles: DestCycles} were used in the proof of \Cref{thm: mainSingleSink}. \Cref{ass: PCSep+ZeroCycles: DestCycles} however, was only necessary in the proof the induction hypothesis for \Cref{claim:thm:main:part1} as the used edges along the $\dest$-walk were not required to be induced by the same commodity. 
  \end{proof}

  With the following example we show that, in general, 
  \ref{lem: NecesCommoditySepara:ReachableEdgesHaveZeroD}   in \Cref{lem: NecesCommoditySepara} holding for all $i\in I$ does not imply~\ref{lem: NecesCommoditySepara:Optimal}. 
  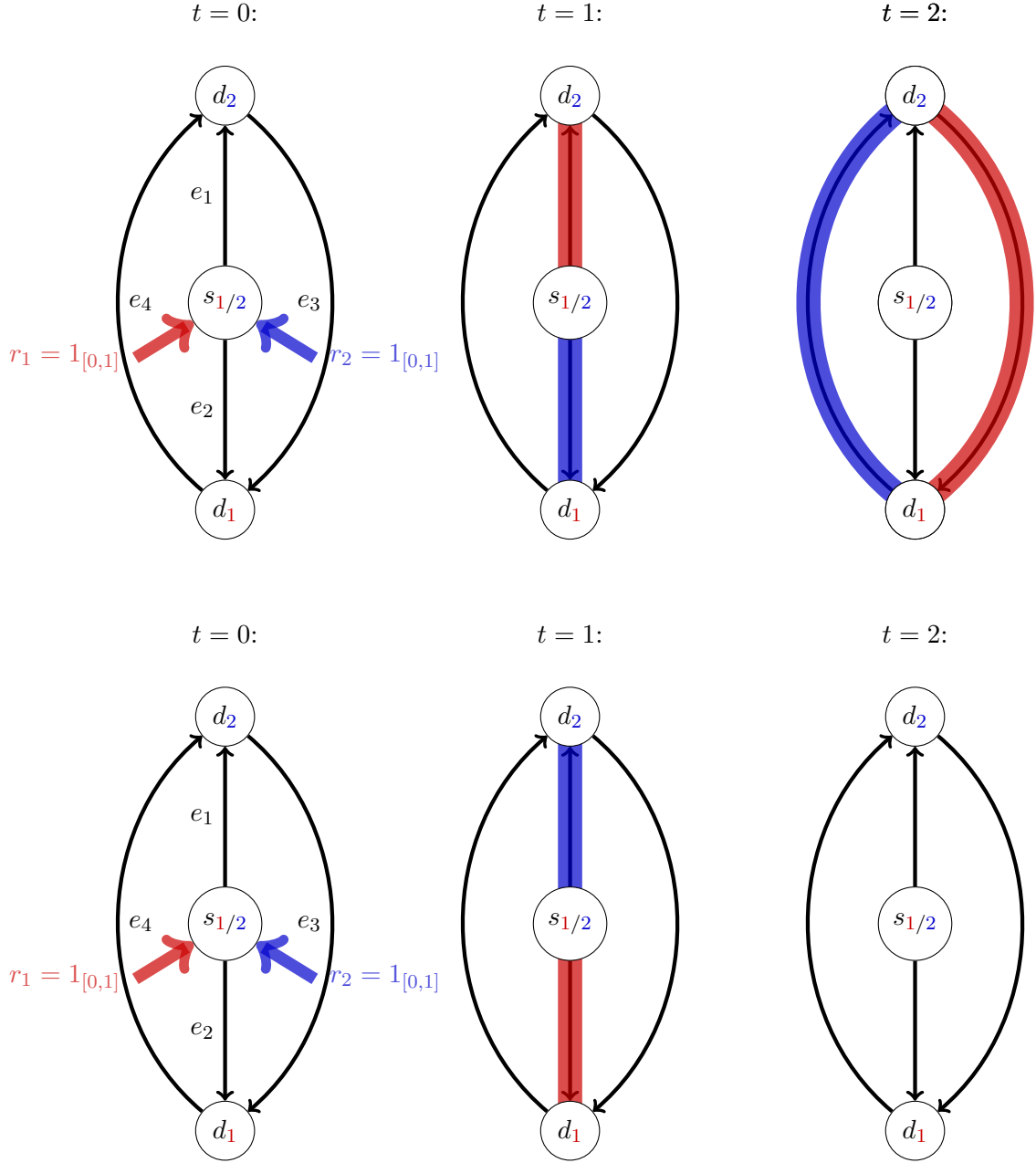
\begin{figure}
  	\centering
  	\BigPicture[3]{
  		\begin{adjustbox}{max width=\textwidth}
  			\definecolor{colA}{rgb}{.8,0,0}%
\definecolor{colB}{rgb}{0,0,.8}%
\definecolor{colC}{rgb}{0,.8,0}%
\definecolor{colD}{rgb}{.8,.8,0}%
\definecolor{colE}{rgb}{.8,0,.8}%
\definecolor{colF}{rgb}{0,.8,.8}%
\colorlet{fcolAr}{colA!70!white}%
\colorlet{fcolBr}{colB!70!white}%
\colorlet{fcolCr}{colC!70!white}%
\colorlet{fcolDr}{colD!70!white}%
\colorlet{fcolEr}{colE!70!white}%
\colorlet{fcolFr}{colF!70!white}%
\tikzstyle{fcolA}=[color=colA,opacity=.7]%
\tikzstyle{fcolB}=[color=colB,opacity=.7]%
\tikzstyle{fcolC}=[color=colC,opacity=.7]%
\tikzstyle{fcolD}=[color=colD,opacity=.7]%
\tikzstyle{fcolE}=[color=colE,opacity=.7]%
\tikzstyle{fcolF}=[color=colF,opacity=.7]%
 \makeatletter%
\ifx\c@timestep\undefined%
 \newcounter{timestep}%
\fi%
 \makeatother%
\newcommand{\drawNetworkSkeleton}[1][]{%
    \coordinate(t2)at(0,0);
    \coordinate(s1s2)at(0,-3);
    \coordinate(t1)at(0,-6);
    
    \node[namedVertexW](temp-t2)at(t2){$\dest_{{\color{colB}2}}$};
    \node[namedVertexW](temp-s1s2)at(s1s2){$\source_{\color{colA}1/{\color{colB}2}}$};
    \node[namedVertexW](temp-t1)at(t1){$\dest_{\color{colA}1}$};

    \draw[edge](temp-s1s2)--node[left]{\ifthenelse{\equal{#1}{desc}}{$e_1$}{}}(temp-t2);
    \draw[edge](temp-s1s2)--node[left]{\ifthenelse{\equal{#1}{desc}}{$e_2$}{}}(temp-t1);
    \draw[edge](temp-t1)to[bend left=50]node[right]{\ifthenelse{\equal{#1}{desc}}{$e_4$}{}}(temp-t2);
    \draw[edge](temp-t2)to[bend left=50]node[left]{\ifthenelse{\equal{#1}{desc}}{$e_3$}{}}(temp-t1);
}%
\newcommand{\drawNetwork}[1][]{
    \ifthenelse{\equal{#1}{}}{
        \stepcounter{timestep}
    }{
        \ifthenelse{\equal{#1}{desc}}{
            \setcounter{timestep}{0}
        }{}
    }
    \node at($(t2)+(0,1.2)$){$t=\arabic{timestep}$:};

    \node[namedVertexF](t2)at(t2){$\dest_{{\color{colB}2}}$};
    \node[namedVertexF](s1s2)at(s1s2){$\source_{{\color{colA}1}/{\color{colB}2}}$};
    \node[namedVertexF](t1)at(t1){$\dest_{\color{colA}1}$};

    \ifthenelse{\equal{#1}{desc} \OR \value{timestep}<1}{
        \draw[line width=5,<-,fcolA](s1s2) -- +(-1.3,-.8)node[anchor=east,color=colA]{$r_1=\CharF[[0,1]]$};
        \draw[line width=5,<-,fcolB](s1s2) -- +(1.3,-.8)node[anchor=west,color=colB]{$r_2=\CharF[[0,1]]$};
    }{}
}%
\begin{tikzpicture}
    \newcommand{\horshiftamount}{5cm}
    

    \begin{scope}[yshift=0cm]
    \setcounter{horshift}{0}

    \begin{scope}[xshift=\arabic{horshift}*\horshiftamount]
        \drawNetworkSkeleton[desc]
    
        \drawNetwork[desc]
    \end{scope}

    \stepcounter{horshift}
    \begin{scope}[xshift=\arabic{horshift}*\horshiftamount]
        \drawNetworkSkeleton

        \draw[line width=10pt,fcolA](s1s2) -- (t2);
        \draw[line width=10pt,fcolB](s1s2) -- (t1);
    
        \drawNetwork
    \end{scope}

    \stepcounter{horshift}
    \begin{scope}[xshift=\arabic{horshift}*\horshiftamount]
        \drawNetworkSkeleton
    
        \drawNetwork

        \draw[line width=10pt,fcolA,shorten <=-1,shorten >=-1](t2) to[bend left=50] (t1);
        \draw[line width=10pt,fcolB,shorten <=-1,shorten >=-1](t1) to[bend left=50] (t2);

        \drawNetwork[notime]
    \end{scope}

    \end{scope}


    \begin{scope}[yshift=-9cm]
    \setcounter{horshift}{0}

    \begin{scope}[xshift=\arabic{horshift}*\horshiftamount]
        \drawNetworkSkeleton[desc]
    
        \drawNetwork[desc]
    \end{scope}

    \stepcounter{horshift}
    \begin{scope}[xshift=\arabic{horshift}*\horshiftamount]
        \drawNetworkSkeleton

        \draw[line width=10pt,fcolA](s1s2) -- (t1);
        \draw[line width=10pt,fcolB](s1s2) -- (t2);
    
        \drawNetwork
    \end{scope}

    \stepcounter{horshift}
    \begin{scope}[xshift=\arabic{horshift}*\horshiftamount]
        \drawNetworkSkeleton
    
        \drawNetwork
    \end{scope}

    \end{scope}
\end{tikzpicture}
  		\end{adjustbox}
  	}
  	\caption{The network and flows from \Cref{exa: CounterCommoditySepara}. All edges have a flow-independent constant travel time of~$1$. The two-commodity flow in the top row has no flow of a commodity~$i$ travelling along a $\dest_i$-cycle, which, by \Cref{lem: NecesCommoditySepara}, is a necessary condition for being implementable. Nevertheless, this flow is not optimal for \eqref{opt: Master} (as proven by the flow in the bottom row) and, hence, not implementable.}
  	\label{fig: CounterCommoditySepara}
  \end{figure}
  
  \begin{example}\label{exa: CounterCommoditySepara}
  	We consider the network depicted in \Cref{fig: CounterCommoditySepara} with flow-independent, constant travel times of~$1$ on all edges. There are $2$ commodities with identical VoT parameters $\gamma_i = 1,i \in I$, the   source, sink pairs $\source_i,\dest_i, i \in I$ and both commodity network inflow rates being  equal to $1_{[0,1]}$.  
  	
  	The edge flow $u$ given by  $u_{\arc_i} = 1_{[0,1]}, i =1,2$ and $u_{\arc_i} = 1_{[1,2]},i=3,4$ is induceable only by the
  	walk flow $h_{\wa} = 1_{\wa_1}(\wa) \cdot 1_{[0,1]} + 1_{\wa_2}(\wa) \cdot 1_{[0,1]},\wa \in \Routes$ for $\wa_1:=((\arc_2,\arc_4),1),\wa_2:=((\arc_1,\arc_3),2)$. 
  	Moreover, since no flow of commodity~$i$ leaves its respective destination~$\dest_i$, the resulting commodity split of $u$ does fulfill the conditions in \Cref{lem: NecesCommoditySepara:ReachableEdgesHaveZeroD} and \Cref{lem: NecesCommoditySepara:UsedDCyclesHaveZeroD}. 
  	Yet, $h$ is not optimal for the master problem \eqref{opt: Master} and hence does not implement $u$ by \Cref{lem: NecessImpl}. This is because $\tilde{h}_\wa=  1_{\wa_3}(\wa) \cdot 1_{[0,1]} + 1_{\wa_4}(\wa) \cdot 1_{[0,1]},\wa \in \Routes$ for $\wa_3=((\arc_1),1),\wa_4=((\arc_2),2)$ has a strictly better objective value and is also feasible for \eqref{opt: Master} as the corresponding edge flow $\tilde{\g}$  is given by  $\tilde{\g}_{\arc_i} = 1_{[0,1]}, i =1,2$ and $\tilde{\g}_{\arc_i} = 0,i=3,4$ and, hence, fulfills $\tilde{\g}\leq u$.  
  \end{example}

\subsection{Sufficient Conditions for Strong Duality}\label{sec:ExistenceOptSolutions} 
 In this \namecref{sec:ExistenceOptSolutions}, we provide sufficient conditions for the master problem \eqref{opt: Master}  admitting strong duality \wrt $L_+^\infty(\hori)$ (hence, in particular \wrt $\MeasFuncUInt$) in the multi-source, single-\sink case (i.e.~$\dest_i = \dest,i\in I$). In fact we will prove this statement  for general optimization problems of the form

{\allowdisplaybreaks[0]
    \begin{align} 
         \inf_{\wflow}\;    &\objfunc(\wflow)  \tag{P} \label{opt: General} \\
        \text{s.t.: } &\ell_{\Routes'}(\wflow) \leq \eflow  \label{eq: ExOptSolLeq}\\
                    &\wflow \in \ofeas. \nonumber
    \end{align}}
    
Here, $\Routes'$ is an arbitrary countable collection of walks which contains each walk only finitely often. 
The objective function $\objfunc$ maps the domain $\ofeas$ to the extended real numbers $\R\cup\{\infty\}$. 
The function $\ell_{\Routes'}$ denotes as in \Cref{sec:uBasedNetworkLoadings} the \auto network loading \wrt an arbitrary fixed travel time function $\trav:\hori \to \R_+^\GA$ (i.e.~$\ell_{\Routes'} := \Nl[\trav(\cdot)]_{\Routes'}$).   
The constraint vector  $\eflow$ is an arbitrary  element in $L_+(\hori)^\GA$ and $\ofeas $ is a subset of $\edom{\Routes'}$. 
We assume that $\ofeas$ contains at least one element $\wflow$ fulfilling~\eqref{eq: ExOptSolLeq}, i.e.~the set of feasible solutions is non-empty.

Note that we obtain \eqref{opt: Master} as a special case of~\eqref{opt: General}  by choosing $\eflow := u$, $\Routes' := \Routes$, $\ofeas := \edom{\Routes}[u] \cap \wir$ and $\objfunc(\wflow):= \sum_{\wa \in \Routes}\dup{\wttime_\wa(u,\cdot)}{\wflow_\wa}$. 
Analogously to the case of \eqref{opt: Master},  we say that the optimization problem~\eqref{opt: General} fulfills strong (Lagrangian) duality \wrt $\MeasFuncUInt[\eflow]$ ($L^\infty_+(\hori)^\GA$), if 
there exists an optimal solution $\wflow^*$ for~\eqref{opt: General} and $\prices \in \MeasFuncUInt[\eflow]$ ($\prices \in L^\infty_+(\hori)^\GA$) such that 
        \begin{align*}
        \inf_{\wflow\in \ofeas} \bigl(\objfunc(\wflow)  + \dup{\prices}{\confunc(\wflow)}\bigr) = \objfunc(\wflow^*) ,
    \end{align*}  
    where $\confunc(\wflow) :=  \ell_{\Routes'}(\wflow)-\eflow $ denotes the constraint mapping of~\eqref{opt: General}. 
    Here, we set $\dup{\prices}{\ell_{\Routes'}(\wflow)-\eflow} := \dup{\prices}{\ell_{\Routes'}(\wflow)} - \dup{\prices}{\eflow} \in \R \cup \{\infty\}$ with the convention that $\infty - a := \infty$ for all $a \in \R$. Remark that  $\dup{\prices}{\ell_{\Routes'}(\wflow)} \in \R \cup \{\infty\}$ is well-defined for all $h \in \wir$ as $\prices$ and $\ell_{\Routes'}(\wflow)$ are non-negative.

    For our result, we need the objective function~$\objfunc$ to satisfy the following assumption, requiring that  the potential decrease in the objective value of an infeasible $\wflow \in \ofeas$ compared to the optimal value is at most proportionally to the  violation of the constraint.

\begin{assumption}\label{ass: StrongDualObj}
    There exists a constant $\const \in\R_+$ such that 
\begin{align*}
    \mu - \objfunc(\wflow) \leq \const \cdot \norm{\abs{\confunc(\wflow)}_+} \text{ for all } \wflow \in \ofeas,
\end{align*}
where $\abs{\confunc(\wflow)}_+ :=\max\{0,\confunc(\wflow)\}$ and $\mu$ denotes the infimum of \eqref{opt: General}. 
\end{assumption}
For the following \namecref{thm: ZeroDualityGapGeneral}, we remark that $\edom{\Routes'}\subseteq \seql[1][\Routes'][L_+(\hori)]$ holds according to \Cref{lem: elluContinuity}. In particular, we can equip $\edom{\Routes'}$ with the subspace topology inherited from~$\seql[1][\Routes'][L_+(\hori)]$. 

\begin{theorem}\label{thm: ZeroDualityGapGeneral}
     Assume that $\ofeas$ is convex and sequentially weakly closed in  $\edom{\Routes'}$  while the objective function $\objfunc$  is sequentially weakly lower semi-continuous on the feasible set $\{\wflow \in \ofeas\mid \ell_{\Routes'}(\wflow)\leq \eflow\}$,\footnote{i.e.\ $\limsup_{n \to\infty} \objfunc(h^n) \leq \objfunc(h)$ for any weakly converging sequence $h^n\wto h$ in $\seql[1][\Routes'][L_+(\hori)]$ with all $h^n$ and $h$ contained in the feasible set.}
     real-valued and linear on $\ofeas$ and 
     fulfills \oref{ass: StrongDualObj}. 
     Then \eqref{opt: General} fulfills strong duality \wrt $L_+^\infty(\hori)^\GA$. 
\end{theorem}
\begin{proof}
With the constraint mapping $\confunc$, we can rewrite \eqref{opt: General} as 
    \begin{align*} 
        \inf_{\wflow}\; &  \objfunc(\wflow)   \\
        \text{s.t.: }&\confunc(\wflow) \in -L_+(\hori)^\GA\\
                    &\wflow \in \ofeas \subseteq \seql[1][\Routes'][L_+(\hori)]
    \end{align*}
    For problems of this form, it was shown 
    in \cite[Theorem~3.2]{flores2013strong} via an application of the Hahn-Banach separation theorem  that strong duality (\wrt $L_+^\infty(\hori)^\GA$) is equivalent to the property that the closure of the conic hull of $\mathcal{E}:=(\objfunc,\confunc)(\ofeas) - (\mu,0) + (\R_+ \times L_+(\hori)^\GA)$ has no points contained in $(-\infty,0) \times \{0\}$. 
    Here, $\mu$ denotes again the value of the optimal solution of~\eqref{opt: General} which exists by \oref{thm: ExistenceOptSol}. 
    For the applicability of \cite[Theorem~3.2]{flores2013strong}, remark that $\seql[1][\Routes'][L_+(\hori)]$ is a Banach space and hence in particular a Hausdorff topological vector space. Furthermore, 
 note that $\mathcal{E}$ is convex as $\objfunc,\confunc$ are linear functions on $\ofeas$  and $ \ofeas$ is convex by assumption. 

Since $\mathcal{E} \subseteq \R \times L(\hori)^\GA$ is the subset of a normed space, the closure and the sequential closure of the conic hull of $\mathcal{E}$ coincide (cf.~\cite[Lemma 3.3]{guide2006infinite}). Hence, it is enough to 
consider sequences $(\lambda_n)\subseteq \R_+$, $(\wflow_n)\subseteq \ofeas$, $(\beta_n) \subseteq \R_+$ and $(y_n) \subseteq L_+(\hori)^\GA$ with $\lambda_n \cdot (\objfunc(\wflow_n) - \mu+\beta_n,\confunc(\wflow_n) + y_n) \to (l,0) \in \cl{\mathrm{cone}(\mathcal{E})} $ and  show that $l\geq 0$. 

This now follows directly from \oref{ass: StrongDualObj}: 
Indeed, we have $\lim_{n \to \infty } \lambda_n \norm{\abs{\confunc(\wflow_n)}_+} = 0$ because of 
$\lambda_n\norm{\confunc(\wflow_n) +y_n} \to 0$ and $y_n \geq 0$. 
Hence, by \oref{ass: StrongDualObj}, we get 
\begin{align*}
   l &= \lim_{n \to \infty}\lambda_n \cdot (\objfunc(\wflow^n) - \mu+\beta_n) \geq \liminf_{n \to \infty} \lambda_n \cdot (\objfunc(\wflow^n) - \mu) +   \liminf_{n \to \infty} \lambda_n \cdot \beta_n\\
   &\geq \liminf_{n \to \infty}\lambda_n\cdot( -\const \cdot \norm{\abs{\confunc(\wflow^n)}_+}) + 0 =  0,
\end{align*}
 which finishes the proof. 
\end{proof}

From now on, we consider the case of single-destination networks. 
As the next step, we use the above \Cref{thm: ZeroDualityGapGeneral} to  show that problems of the form 
 \begin{align}
    \inf_{h} \; & \dup{\decwttime }{h }\tag{${\tilde{\mathrm{P}}}$}\label{opt: AlmostMaster}  \\
    \text{s.t.: } &  \ell_\Routes(h) \leq \eflow \nonumber\\
    &h \in \edom{\Routes} \cap \wir,\nonumber
\end{align} 
with $\decwttime \in \seql[\infty][\Routes][L_+^\infty(\hori)]$ and $\eflow \in \ell_{\Routes}(\edom{\Routes}\cap \wir) $ admit strong duality \wrt $L^\infty_+(\hori)^\GA$. 
Remark that the master problem \eqref{opt: Master} does not always belong to this class of problems as $\wttime(u,\cdot) \in\seql[\infty][\Routes][L_+^\infty(\hori)]$ is not required. However, we will show in \Cref{thm: MasterZeroDual} that the strong duality of \eqref{opt: Master} follows from the strong duality of a related problem contained in the class of problems described via \eqref{opt: AlmostMaster}. 

Remark that the requirement $\wttime(u,\cdot) \in\seql[\infty][\Routes][L_+^\infty(\hori)]$ would, in particular, rule out the possibility to consider (weighted) travel times as private costs, i.e.~\eqref{eq: PC=WTT} holding. 
In this case, we have  $\wttime(u,\cdot) \notin\seql[\infty][\Routes][L_+^\infty(\hori)]$ whenever there exists a sequence of walks whose travel time is unbounded for a non-null set of times, e.g.~when there exists a cycle whose travel time is lower bounded by some $\varepsilon >0$.

\begin{theorem}\label{thm: AlmostMasterZeroDual}
Consider the problem  \eqref{opt: AlmostMaster} 
  in the multi-source, single-destination case   for some $\decwttime \in \seql[\infty][\Routes][L^\infty_+(\hori)]$ and  $\eflow \in \ell_{\Routes}(\edom{\Routes}\cap \wir) $. Then,  \eqref{opt: AlmostMaster}  admits strong duality \wrt $L_+^\infty(\hori)^\GA$. 
\end{theorem}

\begin{proof}
    We verify the conditions stated in \oref{thm: ZeroDualityGapGeneral}. 
    It is clear that $\edom{\Routes} \cap \wir$ is convex. Regarding its sequential weak closedness in $\edom{\Routes}$, consider a weakly converging sequence $h^n \wto h$ in $\edom{\Routes}$ with $h^n,n \in \N$ being contained in $\wir$. 
    Then, we have for arbitrary $\mathfrak T \in \mathcal{B}(\hori)$ and $i \in I$ by the weak convergence 
    \begin{align*}
         \int_{\mathfrak T}\inflow_i \di\sigma =   \sum_{\wa \in \Routes_i}\int_{\mathfrak T} h^n_\wa \di\sigma = \dup{\mathbf{1}_{\Routes_i,\mathfrak T}}{h^n} \to \dup{\mathbf{1}_{\Routes_i,\mathfrak T}}{h} =\sum_{\wa \in \Routes_i}\int_{\mathfrak T} h_\wa \di\sigma,
    \end{align*}
    where $\mathbf{1}_{\Routes_i,\mathfrak T } \in \seql[\infty][\Routes][L_+^\infty(\hori)]$ with $\mathbf{1}_{\Routes_i,\mathfrak T,\wa}(t):= 1$ if $\wa \in \Routes_i$ and $t \in \mathfrak T$ and $\mathbf{1}_{\Routes_i,\mathfrak T,\wa}(t):= 0$ else. 
    As $\mathfrak T$ was arbitrary, it follows that $\sum_{\wa \in \Routes_i} h_\wa = \inflow_i, i \in I$ and hence $h \in \wir$. 

    Regarding the properties required for $\objfunc$, let us start by observing that 
     $\objfunc$ is weakly continuous and real-valued as it is an element of the topological dual of $\seql\supseteq \edom{\Routes}$ (cf.~\Cref{lem: elluContinuity}). 

    Thus, it only remains to show that $\objfunc$ satisfies \Cref{ass: StrongDualObj}. We will do so for the constant $\const := \norm{\decwttime}$:  
    Take any arbitrary $h \in \ofeas = \edom{\Routes}\cap \wir$  with  corresponding aggregated edge flow $\g$ and apply \Cref{claim: ZeroDualityGapTildeH} to it and (an arbitrary representative of) the bounding flow $\eflow$ to get a largest common flow $\tilde{h} \in\edom{\Routes}$ and corresponding edge flow~$\tilde{\g}$ fulfilling the properties stated in \Cref{claim: ZeroDualityGapTildeH}. With the help of this flow we can now show the following chain of inequalities:
    \begin{align}\label{eq: ZeroDualityGap0}
        \sum_{\arc \in \GA}\norm{|\g_\arc-\eflow_\arc|_+} \geq  \norm{\sum_{\wa \in \Routes}h_\wa^* - \sum_{\wa \in \Routes}\tilde{h}_\wa} \geq 
        \frac{1}{\norm{\decwttime}} \cdot (\objfunc(h^*)-\objfunc(h)),
    \end{align}
    where $h^*$ is an optimal solution of~\eqref{opt: Master}, which exists by \oref{thm: ExistenceOptSol}.  This then directly implies that \Cref{ass: StrongDualObj} holds.

\begin{structuredproof}
    \proofitem{First inequality in~\eqref{eq: ZeroDualityGap0}} 
    We calculate (explanations follow): 
    \begin{align}
        \sum_{\arc \in \GA}\norm{|\g_\arc-\eflow_\arc|_+} &= \sum_{\arc \in \GA} \int_{\hori} |\g_\arc- \eflow_\arc|_+ \di\sigma \geq   \sum_{\arc \in \GA} \int_{\mathfrak T_\arc} |\g_\arc- \eflow_\arc|_+ \di\sigma = \sum_{\arc \in \GA} \int_{\mathfrak T_\arc} \g_\arc- \tilde{g}_\arc\di\sigma \label{eq: ZeroDualityGap1} \\
          & =\sum_{\arc \in \GA} \int_{\mathfrak T_\arc} \sum_{\wa \in \Routes} \sum_{j:\wa[j] =\arc}\ell_{\wa,j}(h_\wa) -  \sum_{\wa \in \Routes} \sum_{j:\wa[j] =\arc}\ell_{\wa,j}(\tilde{h}_\wa)  \di\sigma  \nonumber    \\
        & =\sum_{\arc \in \GA} \int_{\mathfrak T_\arc} \sum_{\wa \in \Routes} \sum_{j:\wa[j] =\arc}\ell_{\wa,j}(h_\wa) - \ell_{\wa,j}(\tilde{h}_\wa)  \di\sigma  \label{eq: ZeroDualityGap7}\\
        &= \sum_{\arc \in \GA} \sum_{\wa \in \Routes} \sum_{j:\wa[j] =\arc} \int_{\mathfrak T_\arc}  \ell_{\wa,j}(h_\wa) - \ell_{\wa,j}(\tilde{h}_\wa) \di\sigma  \nonumber \\
        &= \sum_{\arc \in \GA}\sum_{\wa \in \Routes} \sum_{j:\wa[j] =\arc} \int_{\mathfrak T_\arc}  \ell_{\wa,j}(h_\wa-\tilde{h}_\wa)  \di\sigma  \label{eq: ZeroDualityGap2} \\
        &= \sum_{\arc \in \GA} \sum_{\wa \in \Routes} \sum_{j:\wa[j] =\arc} \int_{\arr_{\wa,j}^{-1}(\mathfrak T_\arc)}  h_\wa - \tilde{h}_\wa \di\sigma \nonumber  \\
        &= \sum_{\wa \in \Routes} \sum_{\arc \in \GA} \sum_{j:\wa[j] =\arc} \int_{\arr_{\wa,j}^{-1}(\mathfrak T_\arc)}  h_\wa - \tilde{h}_\wa \di\sigma    \label{eq: ZeroDualityGap8} \\
        &\geq \sum_{\wa \in \Routes}   \int_{ \bigcup_{\arc \in \GA } \bigcup_{j:\wa[j] = \arc} \arr_{\wa,j}^{-1}(\mathfrak T_\arc)}  h_\wa - \tilde{h}_\wa \di\sigma\label{eq: ZeroDualityGap5}  \\ 
        &= \sum_{\wa \in \Routes}   \int_{\hori}  h_\wa - \tilde{h}_\wa \di\sigma  \label{eq: ZeroDualityGap3} \\
        &= \norm{\sum_{\wa \in \Routes}h_\wa - \sum_{\wa \in \Routes}\tilde{h}_\wa}  \label{eq: ZeroDualityGap6} \\
        &= \norm{\sum_{\wa \in \Routes}h_\wa^* - \sum_{\wa \in \Routes}\tilde{h}_\wa}\label{eq: ZeroDualityGap4}  .
    \end{align}
    In the first line~\eqref{eq: ZeroDualityGap1}, we used $|\g_\arc-\eflow_\arc|_+ \geq 0$ as well as the definition of $\mathfrak T_\arc$ (see~\eqref{eq: DefT_e}).  
    In~\eqref{eq: ZeroDualityGap7}, we used the absolute convergence of both series. 
    In~\eqref{eq: ZeroDualityGap2}, we used the linearity of $\ell_{\wa,j}$, shown in \oref{lem: elluExistenceProperties}. 
    For \eqref{eq: ZeroDualityGap8}, we again used absolute convergence of the series. 
    The inequality~\eqref{eq: ZeroDualityGap5} and equality~\eqref{eq: ZeroDualityGap6} is valid as $h_\wa - \tilde{h}_\wa \geq 0$ (\oref{claim: ZeroDualityGapTildeH:leqHn}). 
    For the equality in~\eqref{eq: ZeroDualityGap3}, we utilized \oref{claim: ZeroDualityGapTildeH:Union} and $h \geq \Tilde{h}$. Finally, in~\eqref{eq: ZeroDualityGap4}, we used that $\sum_{\wa \in \Routes}h_\wa = \sum_{i \in I}\inflow_i=\sum_{\wa \in \Routes}h^*_\wa$ due to $h,h^*\in \wir$. 

    \proofitem{Second inequality in~\eqref{eq: ZeroDualityGap0}} 
    By \oref{lem: flowconW'}, $\tilde{\g}$ and $\eflow$ fulfill flow conservation at all $v\neq \dest,\source_i,i\in I$ 
    and their net outflow rate at any $v \in \{\source_i:{i \in I}\}$ is given by  $\inflow_v^{\tilde{\g}}:=\sum_{i \in I: s_i = v}\sum_{\wa \in \Routes_i}\tilde{h}_\wa$  and $\inflow_v^{\eflow}:=\sum_{i \in I: s_i = v} \inflow_i = \sum_{i \in I: s_i = v}\sum_{\wa \in \Routes_i}{h}_\wa$, respectively. 
    
    We introduce a super source~$\ssource$  and extend   $G=(\GV,\GA)$ to $\tilde G=(\GAS,\GVS)$ via   $\GAS:= \GA \cup \{(\ssource,\source_i)\}_{i \in I}$ as well as $\GVS:= \GV \cup \{\ssource\}$. The new edges have constant travel time of zero, i.e.\ $\trav_\arc\equiv 0$ for $\arc \in \GAS\setminus\GA$. We extend $\tilde{\g}$ to a vector $\tilde{\g} \in L_+(\hori)^\GAS$  via $\tilde{\g}_{(\ssource,v)}  =   \inflow_v^{\tilde{\g}}$ for all $v \in \{\source_i:{i \in I}\}$ and $\eflow$ analogously. 

    In this extended network, both $\tilde{\g}$ and $\eflow$ fulfill flow conservation at all $v \neq \ssource,\dest$ and, thus, their difference $\eflow -\tilde{\g}$ does likewise. 
    Furthermore,  this difference has a net outflow from $\ssource$ given by $\sum_{i \in I} \inflow_i - \sum_{\wa \in \Routes} \Tilde{h}_\wa$ which is nonnegative by \oref{claim: ZeroDualityGapTildeH:leqHn}. Moreover, 
      \Cref{claim: ZeroDualityGapTildeH:leqGhat,claim: ZeroDualityGapTildeH:leqHn} show that $\eflow-\tilde{g}\geq 0$ (in the extended network).   
    Thus, by \oref{thm: FlowDecomp}, this difference admits a nonnegative flow decomposition $\tilde{\wflow}_{\tilde{\wa}},{\tilde{\wa}} \in \tilde{\Routes}_{\ssource,\dest},\tilde{\wflow}_c,c \in \tilde{\mathcal{C}}^{\mathrm{simp}}$ with 
     $\tilde{\Routes}_{\ssource,\dest}$ denoting the set of finite walks from $\ssource$ to $\dest$   and $\tilde{\mathcal{C}}^{\mathrm{simp}}$ the set of simple cycles in the extended network. Moreover, we know that 
    $\sum_{{\tilde{\wa}} \in \tilde{\Routes}_{\ssource,\dest}} \tilde{\wflow}_{\tilde{\wa}}$ is equal to the net outflow of $\ssource$, i.e.
    \begin{align}\label{eq: InflowRateEqualDifference}
        \sum_{{\tilde{\wa}} \in \tilde{\Routes}_{\ssource,\dest}} \tilde{\wflow}_{\tilde{\wa}} = \sum_{i \in I} \inflow_i - \sum_{\wa \in \Routes} \Tilde{h}_\wa.
    \end{align}

    There is a bijection between  the set of finite walks from $\ssource$ to $\dest$ in the extended network   $\tilde{\Routes}_{\tilde\source,\dest}$ and the set of all $\source_i,\dest$-walks $\hat{\Routes}_{\source_i,\dest}$ combined over all $i \in I$ in the original network via 
    $\iota:    \bigcup_{i \in I}\hat{\Routes}_{\source_i,\dest}\xrightarrow{\sim} \tilde{\Routes}_{\ssource,\dest}$ with $\iota(\hat{\wa}) = ((\ssource,\source_i),\hat{\wa})$ for all $\hat{\wa} \in \hat{\Routes}_{\source_i,\dest}$ and $i \in I$. 
    Hence, we can transform $ \tilde{\wflow}_{\tilde{\wa}},{\tilde{\wa}} \in \tilde{\Routes}_{\ssource,\dest}$ into walk inflow rates 
    $\wflow_{\wa},\wa \in  \bigcup_{i \in I}\hat{\Routes}_{\source_i,\dest}$. 
    Since the travel times on the artificial edges $(\ssource,\source_i)$ in the extended network are always equal to $0$, 
    the induced flow of $\wflow$ in the original network is equal to the induced flow of $\tilde{\wflow}$ in the extended network on the original edges, that is, 
    \begin{align}\label{eq: NLEqualExtended}
        \ell_{\bigcup_{i \in I}\hat{\Routes}_{\source_i,\dest}}(\wflow) = (\tilde\ell_{\tilde{\Routes}_{\ssource,\dest},\arc}(\tilde{\wflow}))_{\arc \in \GA} = \eflow - \tilde{\g} - \ell_{\tilde{\mathcal{C}}^{\mathrm{simp}}}((\tilde{\wflow}_c)_{c \in \tilde{\mathcal{C}}^{\mathrm{simp}}}) \leq \eflow - \tilde{\g}  \in L_+(\hori)^\GA 
    \end{align} 
    where $\tilde{\ell}$ denotes the (\auto[)] network loading operator in the extended network.

    Since $\Tilde{h} \leq h$ and $\sum_{\wa \in \Routes}h_\wa = \sum_{i \in I} \inflow_i$ by $h \in \wir$, 
    \eqref{eq: InflowRateEqualDifference} implies that 
    we can  
    add  $\wflow_\wa,\wa \in  \bigcup_{i \in I}\hat{\Routes}_{\source_i,\dest}$ suitably to $\tilde{h}_\wa,\wa \in \Routes$ such that  we arrive at a walk inflow rate vector  
     $\Tilde{h}^\uparrow$ that is contained in $\edom{\Routes}\cap \wir$. Moreover, the inequality in \eqref{eq: NLEqualExtended} implies that 
     $\ell_{\Routes}(\Tilde{h}^\uparrow) \leq \eflow$ 
     and subsequently  $\Tilde{h}^\uparrow$ is feasible for~\eqref{opt: Master}. 
      By the optimality of $h^*$, we get $\objfunc(h^*)\leq \objfunc(\tilde{h}^\uparrow)$ and, hence, the following estimation: 
    \begin{align*}
       \objfunc(h^*)\leq \objfunc(\tilde{h}^\uparrow) &=  \dup{\decwttime}{\tilde{h}^\uparrow} = \dup{\decwttime}{\tilde{h} -\tilde{h} + \tilde{h}^\uparrow} 
       = \dup{\decwttime}{\tilde{h}} + \dup{\decwttime}{\tilde{h}^\uparrow -\tilde{h}} \\
       &\symoverset{1}{\leq} \dup{\decwttime}{ {h}} + \dup{\decwttime}{\tilde{h}^\uparrow -\tilde{h}} \\
       &\symoverset{2}{\leq} \objfunc({h}) +\norm{\decwttime}\cdot\norm{\sum_{\wa \in \Routes}\tilde{h}^\uparrow_\wa - \sum_{\wa \in \Routes}\Tilde{h}_\wa} \\
        &\symoverset{3}{=}\objfunc({h})  +\norm{\decwttime}\cdot\norm{\sum_{\wa \in \Routes}h^*_\wa - \sum_{\wa \in \Routes}\Tilde{h}_\wa}.
    \end{align*}
    Here, \refsym{1} holds by non-negativity of $\decwttime$ and $h \geq \tilde{h}$. 
    For \refsym{2}, we used the Cauchy-Schwarz inequality.  
    Equality \refsym{3} is valid due to $\tilde{h}^\uparrow,h^* \in \wir$. 
    
    From this, the claimed inequality follows immediately by subtracting $\objfunc(h)$ and dividing by $\norm{\decwttime}$. \qedhere
\end{structuredproof}
\end{proof}

With \Cref{thm: AlmostMasterZeroDual} at hand, we are now in the position to prove the main result of this section: We show that the master problem \eqref{opt: Master} admits strong duality \wrt $L^\infty_+(\hori)^\GA$ if all feasible solutions   experience uniformly bounded private costs (i.e.~\Cref{thm: SuffConMSMS: BoundedExpTravel} being valid \wrt a common constant $C$).  Remark that, by \Cref{lem: AssUZeroDG} and \Cref{thm: ExistenceOptSol}, this is in particular fulfilled if the private costs represent weighted travel times (i.e.~\eqref{eq: PC=WTT} holding) and $u$ has bounded support.  

\begin{theorem}\label{thm: MasterZeroDual} 
   In the multi-source, single-destination case,  the master problem~\eqref{opt: Master} fulfills strong duality \wrt $L^\infty_+(\hori)^\GA$ if 
   all feasible solutions $h$ admit uniformly bounded experienced private costs, i.e.~if all feasible $h$ fulfill \Cref{thm: SuffConMSMS: BoundedExpTravel} \wrt a common constant $C$.
\end{theorem}
As an immediate consequence of the above, \Cref{thm: mainSingleSink} and \Cref{lem: AssUZeroDG}, we get the following statement: 
\begin{corollary}\label{cor: CharImplFiniteSupp}
    For a multi-source, single-\sink network with \Cref{ass: u} being fulfilled, private costs representing weighted travel times and  $u$ having finite support,  the following statements are equivalent:  
    \begin{thmparts}
        \item $u$ is implementable.
        \item Every used edge that is reachable from $\dest$ via used edges  has zero traversal time.
        \item  $u$ only sends flow along a $\dest$-cycle if the latter has zero traversal time.
        \item There exists $h^* \in \wir$ optimal for \eqref{opt: Master} with tight inequality \eqref{ineq: Master}.  
    \end{thmparts}
\end{corollary}

\begin{proof}[Proof of \Cref{thm: MasterZeroDual}]
We start by constructing a related problem to the master problem that is of the form \eqref{opt: AlmostMaster}. For this, 
      set the travel times $\trav$ to the travel times induced by $u$, i.e.~$\trav(\cdot):= \trav(u,\cdot)$. Then, $u \in \ell^u(\edom{\Routes}[u] \cap \wir)$ as $u \in \Nl(\wir)$ and thus we can set $\eflow = u$. 
    Let $C$ be the constant of \Cref{thm: SuffConMSMS: BoundedExpTravel} and set $\decwttime := \min\{C,\Psi(u,\cdot)\} \in \seql[\infty][\Routes][L_+^\infty(\hori)]$. 
    It is straight forward to verify that the optimal solutions of  the master problem are also optimal for the newly constructed problem. 
    Let $h^*$ be an optimal solution of the master problem. 
    By \Cref{thm: AlmostMasterZeroDual}, the corresponding problem \eqref{opt: AlmostMaster} fulfills strong duality \wrt $L^\infty_+(\hori)^\GA$. That is, there exists $\prices^* \in L^\infty_+(\hori)^\GA$   such that  
    \begin{align}\label{eq: AlmostMasterHasZeroDG}
        \inf_{h \in \edom{\Routes}[u]\cap \wir} \dup{\decwttime}{h} + \dup{\prices^*}{\ell^u(h) - u} = \dup{\decwttime}{h^*}. 
    \end{align}
By \Cref{thm: SuffConMSMS: BoundedExpTravel} holding, we know that the implication  
\begin{align}\label{eq: hImpliesCostsEqual}
    h^*_\wa(t) >0 \implies \decwttime_\wa(t) = \wttime_\wa(u,t) \text{ holds for almost all }t \in \hori \text{ and all } \wa \in \Routes. 
\end{align}
     Using $\wttime(u,\cdot) \geq \decwttime$, and the above equations \eqref{eq: AlmostMasterHasZeroDG}  and \eqref{eq: hImpliesCostsEqual}  now leads to: 
     \begin{align*}
        \inf_{h \in \edom{\Routes}[u]\cap \wir} \dup{\wttime(u,\cdot)}{h} + \dup{\prices^*}{\ell^u(h) - u}\overset{\wttime(u,\cdot) \geq \decwttime}&{\geq}   \inf_{h \in \edom{\Routes}[u]\cap \wir} \dup{\decwttime}{h} + \dup{\prices^*}{\ell^u(h) - u} \\
        \overset{\eqref{eq: AlmostMasterHasZeroDG}}&{=}\dup{\decwttime}{h^*} \overset{\eqref{eq: hImpliesCostsEqual}}{=} \dup{\wttime(u,\cdot)}{h^*} .
     \end{align*}
     By $\ell^u(h^*) \leq  u$  and $\prices^* \in L_+^\infty(\hori)^\GA$, we have  $\dup{\prices^*}{\ell^u(h^*) - u} \leq 0$, showing that weak duality also holds:  
    \begin{align*}
        \inf_{h \in \edom{\Routes}[u]\cap \wir} \dup{\wttime(u,\cdot)}{h} + \dup{\prices^*}{\ell^u(h) - u} \leq \dup{\wttime(u,\cdot)}{h^*}
    \end{align*}
      Hence, $(h^*,\prices^*)$ demonstrates that \eqref{opt: Master} fulfills strong duality \wrt $L_+^\infty(\hori)^\GA$. 
 \end{proof}
 
 Let us conclude this section by remarking that the requirement of bounded experienced costs in \Cref{thm: MasterZeroDual} 
 can not be dropped. In the following example, we consider a 
  single-commodity network with private costs representing the travel times and a flow $u$ without finite support. 
  The corresponding master problem  does not admit strong duality \wrt $\MeasFuncUInt$ as the flow $u$ is
  not implementable. 
 
\begin{example}[No Strong Duality without Finite Support]\label{ex: DualGap}  
    Consider the network and flow depicted in \Cref{fig: CounterNonFinSupp}. 
    There is a single commodity with   network inflow rate $\inflow_1= 1_{[0,1]}$ and private costs given by the travel times, i.e.~\eqref{eq: PC=WTT} holds with VoT $\gamma_1 = 1$.  
    The edges have a constant travel time of $1$ for all $\g \in L_+(\hori)^\GA$ and $t \in \hori$. 
    Consider the edge flow $u$  induced by the walk inflow rates $h_{((s,v),(v,\dest))} = 0$ and $h_{\wa^l} = \frac{1}{2^l} \cdot 1_{[0,1]}, l \in \N$ where $\wa^l=((\source,v),(v,v),\ldots,(v,v),(v,\dest))$ is the walk containing $l$ times the edge $(v,v)$. For these walk inflow rates $h$, 
    the resulting flow $u$ is given by $u_{(\source,v)} = 1_{[0,1]}$, $u_{(v,v)} = \sum_{l = 0}^\infty\frac{1}{2^l} \cdot 1_{[l+1,l+2]}$ 
    and $u_{(v,\dest)} = \sum_{l = 1}^\infty\frac{1}{2^l}\cdot 1_{[l+1,l+2]} $. In particular, 
    since $\sum_{l = 0}^\infty\frac{1}{2^l} \cdot 1_{[l+1,l+2]}  =2$, 
    $u \in L_+(\hori)^\GA$ is integrable but does not have a finite support. Moreover, $u$ is not implementable as for almost all points in time $t \in [0,1]$, we have $\sup\{\wttime_\wa(u,t) \mid \wa \in \Routes: h_\wa(t)>0\} = \infty$ while $\wttime_\wa(u,t) < \infty,\wa \in \Routes$. 

    \begin{figure}
        \centering
        \BigPicture[1]{%
        \begin{tikzpicture}
            \clip (-3.5,3.5) rectangle (13,-1);
            
            \node[namedVertex] (s) at (0,0) {$\source$};
            \node[namedVertex] (v) at (6,0) {$v$};
            \node[namedVertex] (d) at (12,0) {$\dest$};

            \draw[edge] (s) to 
            node[above]{$u_{(\source,v)} = \CharF[[0,1]]$} (v);
            \draw[edge] (v) to 
            node[above]{$u_{(\source,v)} = \sum_{l = 1}^\infty\frac{1}{2^l}\CharF[[l+1,l+2]]$} (d);
            \draw[edge] (v) to[out=35,in=145,looseness=25] 
            node[above]{$u_{(\source,v)} = \sum_{l = 0}^\infty\frac{1}{2^l}\CharF[[l+1,l+2]]$} (v);

            \draw[line width=5,<-,fcolA](s) -- +(-1,0)node[anchor=east,color=colA]{$r_1=\CharF[[0,1]]$};
        \end{tikzpicture}
        }
        \caption{The network and flow $u$ considered in \Cref{ex: DualGap}, describing a situation in which the corresponding master problem does not admit strong duality.}
        \label{fig: CounterNonFinSupp}
    \end{figure}
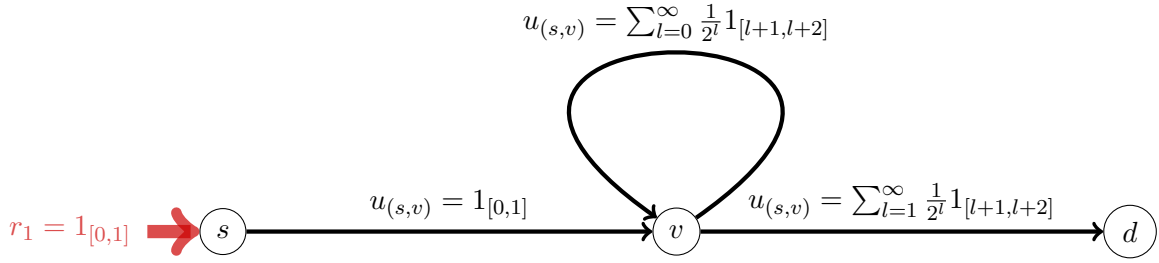

    It is easy to verify that the master problem \eqref{opt: Master} only has one feasible solution, namely $h$ itself. 
    Hence, $h$ is optimal with tight inequality \eqref{ineq: Master}. 
	In particular, \eqref{opt: Master} is \wellposed as the objective value of $h$ is bounded by 
  	\begin{align*}
  		\dup{\wttime(u,\cdot)}{h} \Croverset{lem: aggCostsVSwalkCosts}{=} \dup{\trav(u,\cdot)}{u} = 1 + \sum_{l =0}^\infty\frac{l}{2^l} + \sum_{l =1}^\infty\frac{l}{2^l} = 1 + 2+ 1 = 4. 
  	\end{align*}
  	Hence,  the condition characterizing implementability in \Cref{thm: mainMSSS: Con} is fulfilled, yet $u$ is not implementable. Since all other 
  	requirements of	\Cref{thm: mainMSSS} are fulfilled except 
  	\eqref{opt: Master} admitting strong duality, the validity of the latter is ruled out. 
\end{example}

\section{Conclusions}
In this paper, we studied dynamic edge flows that can be implemented as tolled dynamic equilibria assuming heterogeneous user populations with 
different source, destination-pairs and private costs. Based on a novel  infinite dimensional optimization problem~\eqref{opt: Master}, we derived necessary and sufficient conditions for the implementability of a dynamic edge flow $u$. 
We further derived a combinatorial characterization for the single \sink case showing that any edge flow  is implementable if and only if   there is no flow-carrying cycle containing the \sink with non-zero private costs. 
Our characterizations rely on the master problem~\eqref{opt: Master} admitting strong duality and for the single-\sink case, we derive a sufficient condition for this which is, in particular, fulfilled for the important special case where private costs represent weighted travel times and the edge flow is finitely supported.
Note that our implementability characterizations  subsume in particular the characterizations derived in~\cite{ColeDR03,Fleischer04,Karakostas04,Yang04} for the static case by setting travel times to zero.

Our work leads to several new questions for future research. 
In particular, under which conditions the aforementioned strong duality holds for the for the multi-\sink case is still open.
For applicability in actual traffic networks, one might be interested in simple toll functions such as piece-wise linear or piece-wise constant functions. For the homogeneous case we provided sufficient conditions for the existence of such tolls, but for the case of heterogeneous populations this is still an open question. Closely related to this problem is the computational complexity of finding tolls for a given flow. As we derive our tolls from the duals of the infinite dimensional optimization problem~\eqref{opt: Master}, the computation of such tolls remains unaddressed. However, it seems reasonable that for piece-wise constant or piece-wise linear inflow functions and the Vickrey queuing model, problem~\eqref{opt: Master} and its dual can be reformulated as a finite dimensional optimization problem, possibly leading to a computationally tractable formulation.

Finally, it is also interesting to know which classes of flows are contained in the set of implementable flows and, in particular, whether (or under what additional conditions) socially optimal flows are included in it.  In the static flow case, social optima are always implementable since they are minimal. In the dynamic case, however, the corresponding minimality
concept (aka \umini edge flows) is more involved as it relies on the \auto network loading. Hence, in contrast to the static case, implementability of social optima does not follow directly from the characterization via \eqref{opt: Master} in the dynamic case. In a follow-up paper~\cite{GHS24SO}, we show that for single-\sink instances, system-optimal flows are implementable (under suitable assumptions on the flow model), while they are not always implementable for multi-\sink instances -- even for well-behaved flow models like the Vickrey model!

\section*{Acknowledgements}
\hfill\\
This research has been funded by the Deutsche Forschungsgemeinschaft
(DFG) in the project 543678993 (Aggregative gemischt-ganzzahlige
Gleichgewichtsprobleme: Existenz, Approximation und Algorithmen).
We acknowledge the support of the DFG. 

\appendix

{\section{Proof of Theorem~\ref*{thm:ImplementabilityHomogeneous}}\label{sec:CostBalancingTollsProof}

\NewDocumentCommand{\costAlt}{O{\arc}}{\psi_{#1}}
\NewDocumentCommand{\CostAlt}{O{\wa}}{\wttime_{#1}}

As discussed in \Cref{sec:CostBalancingTolls}, tolls for homogeneous users can be constructed in the form of cost balancing tolls.  Note that, since all commodities have the same \sink[,] we use $\dest$ to denote this common \sink[.] Similarly, we use~$\psi_\arc$ to denote the common private edge cost function shared by all commodities.

	\begin{proof}[Proof of \Cref{thm:ImplementabilityHomogeneous}]
		It is enough to show the existence of such tolls for the case  that the network does not have edges leaving the destination as well as that the edge costs $\psi_\arc$ are $0$ after $t_f +1$. 
		Given a network that does not fulfill this, we can instead construct tolls~$\prices$ 
		in an artificial network where we delete the edges leaving the destination and adjust the costs in a piecewise-linear and continuous way on the interval  $[t_f,\infty)$ such that the costs are equal to $0$ after $t_f+1$. 
		This transformation, therefore, preserves continuity and piecewise-linearity of $\psi_\arc$ in case the latter exhibits these properties. 
		
		Then, any implementing tolls~$\prices$ in the artificial network can be extended to implementing tolls in the original network by setting $\prices_\arc \coloneq 0$ for all $\arc \in \edgesFrom{\dest}$. This works, because any \stwalk[v] using such an edge $\arc \in \edgesFrom{\dest}$ also contains as a subwalk a (shorter) \stwalk[v] in the artificial network. Hence, \stwalk[v]s using edges leaving the \sink cannot create any shortcuts when comparing the original network to the artificial one.
		Additionally, we have for all $\wa$ with $h_\wa \neq 0$ that  $\arr_{\wa,j}(u,t) \leq t_f$ for almost all $t \in \hori$ and $j\leq \abs{\wa}$ and, subsequently, setting the edge costs back to the original ones only changes the costs of non-utilized walks (and makes those more expensive). 
		These two observations then immediately imply that $\prices$ also implement $u$ via $h$ in the original network.

		Hence, assume from now on that the network does not have edges leaving the destination as well as that the edge costs $\psi_\arc$ are $0$ after $t_f +1$.
		Let us define time dependent node labels $\pi_v$ by
		\begin{align*}
			\pi_v: [t_0,\infty) \to \Rnn, t \mapsto \pi_v(t) &\coloneqq \sup\set{\wttime_\wa(t) | \wa \in\hat\Routes_{v,\dest}}. 
		\end{align*} 
		
		\begin{claim}\label{claim: NodeLabels}
	The node labels are real-valued. Moreover, the node labels are continuous and/or piecewise linear if all edge traversal times and costs are continuous and/or piecewise linear.
		\end{claim}
		\begin{proofClaim}
			We start by choosing some constant $K \geq 0$ such that we have $\arr_{\wa,\abs{\wa}+1}(t) - t < K$ for all times $t \in \R$ and all simple \stwalk[v]s $\wa$ and all $v \in V$. This is possible since there are only finitely many such walks and the edge traversal times are upper bounded by $M$ by assumption. We now define the set
			\[\bar\Routes_{v,d} \coloneqq \set{\wa \text{ a \stwalk[v]} | \arr_{w,\abs{w}+1}(t_0) \leq t_f+1+K}.\]
			This set has finite cardinality, since we have $\arr_{w,\abs{w}+1}(t) \geq t + \abs{w}\cdot\eps$ for any walk~$\wa$. Hence, 
			\[\bar\pi_v: [t_0,\infty) \to \Rnn, t \mapsto \max\set{\CostAlt(u,t) | \wa \in \bar\Routes_{v,d}}\]
			is well-defined. 
			Moreover, it follows immediately that $\bar\pi_v$ is continuous and/or piecewise linear whenever the traversal times $\trav_\arc(u,\cdot)$ and edge costs $\psi_\arc$ have the corresponding properties, since $\bar\pi_v$ is obtained from finitely many of them by compositions, sums, and maxima. 
			
			It remains to show that we have $\pi_v = \bar\pi_v$.   
			Since we clearly have $\pi_v \geq \bar\pi_v$, it suffices to show that we also have $\CostAlt(u,t) \leq \bar\pi_v(t)$ for all \stwalk[v]s~$\wa$ and times $t \in [t_0,\infty)$. If $t \geq t_f+1$, then all  edge costs are zero anyway and, therefore, we have $\CostAlt(u,t) = 0 = \bar\pi_v(t)$. Otherwise, we remove all cycles from $\wa$ which are reached after time $t_f+1$ to obtain a \stwalk[v]~$\wa'$ with $\CostAlt[\wa'](u,t) = \CostAlt(u,t)$. Moreover, we have $\arr_{\wa',\abs{\wa'}+1}(t) \leq t_f+1+K$ (since the part of $\wa'$ reached after time $t_f+1$ must be a simple path and, therefore, has a traversal time of at most~$K$). Thus, we have $\wa' \in \bar\Routes_{v,d}$ and, therefore, $\CostAlt[\wa](u,t) = \CostAlt[\wa'](u,t) \leq \bar\pi_v(t)$.		 
		\end{proofClaim}

		With this, we can now define tolls by setting 
		\begin{align*}
			\prices_\arc(t) \coloneq \pi_{v}(t) -\psi_\arc(t) - \pi_{v'}(\exit_\arc(u,t)) 
		\end{align*}
	  for all $\arc=(v,v') \in \GA$ and $t \in [t_0,\infty)$ and arbitrarily (but piecewise linearly and continuously) extending them to $(-\infty,t_0)$.
	  
	  	\begin{claim}\label{claim:PricesWellDefinedAndCostBalancing}
	  	These tolls are non-negative and satisfy $\pi_v(t)=\wttime_{\wa}(u,t)+\Pf^\prices_\wa(u,t)$ for all nodes $v \in V$, \stwalk[v]s~$\wa$ and times $t \in [t_0,\infty)$. They are, in addition, continuous and/or piecewise linear if all traversal times and edge costs are continuous and/or piecewise linear.
	  \end{claim}
	  	\begin{proofClaim}
	  	Take any edge $\arc=(v,v')$ and time $t \in [t_0,\infty)$ and let $\wa' \in \bar\Routes_{v',d}$ be a \stwalk[v'] with $\pi_{v'}(\exit_\arc(t)) = \wttime_{\wa'}(u,\exit_\arc(u,t))$. Then, $\wa \coloneqq (\arc,\wa')$ is a \stwalk[v] and we have 
	  	\begin{align*}
	  		\prices_\arc(t) = \pi_v(t) - \psi_\arc(t) - \pi_{v'}(\exit_\arc(u,t)) = \pi_v(t) - \psi_\arc(t) - \wttime_{\wa'}(u,\exit_\arc(u,t)) = \pi_v(t) - \wttime_{(\arc,\wa')}(u,t) \geq 0.
	  	\end{align*}
	  	 
	  	Hence, $\prices_\arc$ is non-negative. Moreover, $\prices_\arc$ is continuous and/or piecewise linear if  $\costAlt$ and $\trav_\arc$ are likewise (since this implies the same property for $\pi_v$ by \Cref{claim: NodeLabels}). 
	  	Finally, we show that $\pi_v(t) = \CostAlt[\wa](u,t)+\Pf^\prices_\wa(u,t)$ holds for all \stwalk[v]s $\wa$, all nodes $v \in V$ and all times $t \in \hori$ via induction on $\abs{\wa}$: 
	  	\begin{proofbyinduction}
	  		\basecase{$\abs{\wa}=0$} This is only possible if $v=\dest$. Since we assume that there are no outgoing edges from $\dest$, the claim trivially holds here as the empty walk is the only \stwalk[\dest] then, and we have $\pi_\dest(t)=0=0+0=\CostAlt[()](u,t)+\Pf^\prices_{()}(u,t)$.
	  		
	  		\inductionstep{$\abs{\wa}>0$} Let $\wa$ be of the form $(\arc,\wa')$ for some edge $\arc=(v,v')$ and some \stwalk[v'] $\wa'$. Then, we have
	  		\begin{align*}
	  			\CostAlt[\wa](u,t)+\Pf^\prices_\wa(u,t) 
	  			&=\costAlt(t)+\prices_\arc(t)+\CostAlt[\wa'](u,\exit_\arc(t))+\Pf^\prices_{\wa'}(u,\exit_\arc(t)) \\
	  			&=\costAlt(t) + \pi_v(t) - \costAlt(t) - \pi_{v'}(\exit_\arc(t)) + \CostAlt[\wa'](u,\exit_\arc(t))+\Pf^\prices_{\wa'}(u,\exit_\arc(t)) \\
	  			&=\pi_v(t) + \Big(\CostAlt[\wa'](u,\exit_\arc(t)) + \Pf^\prices_{\wa'}(u,\exit_\arc(t)) - \pi_{v'}(\exit_\arc(t))\Big) \\
	  			&=\pi_v(t)
	  		\end{align*}
	  		where the last step holds by induction.  
	  	\end{proofbyinduction}
	  	This finishes the induction and, thus, the proof of \Cref{claim:PricesWellDefinedAndCostBalancing}.
	  \end{proofClaim}
	  
	  From \Cref{claim:PricesWellDefinedAndCostBalancing} for $v = \source_i$, it now follows immediately that $h$ and $\prices$ implement~$u$. 
	  If, additionally, all traversal times and edge costs are continuous and/or piecewise linear, \Cref{claim:PricesWellDefinedAndCostBalancing} also guarantees that the same is true for the tolls~$\prices_\arc$. Moreover, it is clear that we can compute all relevant $\trav_\arc$, $\arr_{\wa,j}$, $\costAlt$, $\CostAlt$, $\bar\pi_v$ and, hence, $\prices_\arc$ in finite time then.
	\end{proof}

\renewcommand{\wflow}{h}
 \section{Properties of \Auto Network Loadings}\label{sec:AppuBasedNetworkLoadings}

In this \namecref{sec:AppuBasedNetworkLoadings} we collect several structural results on \auto network loadings  which were derived in~\cite{GHS24FD} and which we use throughout this paper. As in \Cref{sec:uBasedNetworkLoadings}, we 
 consider the same general framework as in \cite{GHS24FD}, that is,  the  setting is as described in the paragraph about \auto network loadings in \Cref{sec: Model}: We are given an arbitrary but fixed (flow-independent) absolutely continuous travel time function $\trav:\hori \to \R_+^\GA, t \mapsto (\trav_\arc(t))_{\arc \in \GA}$ fulfilling FIFO with the corresponding \auto network loading operator $\Nl[\trav(\cdot)]$. This will be the only type of network loading utilized in this entire section and hence we can drop for the sake of readability the superscript and simply write $\ell := \Nl[\trav(\cdot)]$ and similarly $\edom{\wa,j}\coloneq \edom{\wa,j}[\trav(\cdot)]$, $\edom{\wa,\arc}\coloneq \edom{\wa,\arc}[\trav(\cdot)]$ and $\edom{\Routes'}\coloneq \edom{\Routes'}[\trav(\cdot)]$ for any countable collection of walks $\Routes'$.

\subsection{Existence of \Auto Network Loadings}
 
\begin{theorem}[\ourref{lem: elluExistenceProperties}]\label{lem: elluExistenceProperties}
   Consider an arbitrary countable collection of walks $\Routes'$, $h \in L_+(\hori)^{\Routes'}$,  $\wa\in \Routes'$, $j \in[|\wa|+1]$ and $\arc \in \GA$. Then, the following holds:
    \begin{thmparts}
        \item $h_\wa \in \edom{\wa,j}$   if and only if $h_\wa$ satisfies~\eqref{eq: nlexists}.  In this case $\ell_{\wa,j}(h_\wa)$ is uniquely determined. \label[thmpart]{lem: elluExistenceProperties:ExistenceInducedFlowOnJthEdge}

        \item 
         $h_\wa \in \edom{\wa,\arc}[]$ if and only if  $h_\wa \in \edom{\wa,j}$ for all $j\leq \abs{\wa}$ with $\wa[j] = \arc$. In this case, $\ell_{\wa,\arc}(h_\wa)$ is uniquely determined by $\ell_{\wa,\arc}(h_\wa) = \sum_{j:\wa[j]=\arc}\ell_{\wa,j}(h_\wa)$.
          \label[thmpart]{lem: elluExistenceProperties:ExistenceInducedFlowOnEdgeE}

        \item  $h_\wa \in \edom{\wa}[]$ if and only if  $h_\wa \in \edom{\wa,j}$ for all $j\leq \abs{\wa}$. In this case, $\ell_{\wa}(h_\wa)$ is uniquely determined by $\ell_{\wa,\arc}(h_\wa) = \sum_{j:\wa[j]=\arc}\ell_{\wa,j}(h_\wa)$ for all $\arc \in \GA$.\label[thmpart]{lem: elluExistenceProperties:ExistenceInducedFlowOnAllEdges}

        \item The maximal domains $\edom{\wa,j}$ and $\edom{\wa,\arc}$ of $\ell_{\wa,j}$ and $\ell_{\wa,\arc}$  are sequentially weakly closed convex cones of $L_+(\hori)$, that is, if $h_\wa^n\wto h_\wa$ and $\ell_{\wa,j}(h_\wa^n),n\in \N$ exist, then so does $\ell_{\wa,j}(h_\wa)$ and, analogously, for $\ell_{\wa,\arc}$.

        Moreover, both functions are linear on their respective domains.
        \label[thmpart]{lem: elluExistenceProperties:Linearity}

        \item  If the aggregated edge flow $\ell_{\Routes'}(h) \in L_+(\hori)^\GA$ exists, then it is uniquely determined. On its maximal domain $\edom{\Routes'} \subseteq L_+(\hori)^{\Routes'}$, the function  $\ell_{\Routes'}$ is  linear. \label[thmpart]{lem: elluExistenceProperties:AggLinearity}

        \item The aggregated edge flow $\ell_{\Routes'}(h) \in L_+(\hori)^\GA$ exists if and only if $\ell_{\wa}(h_\wa)$ exists for all $\wa \in \Routes'$ and $(\ell_{\wa}(h_\wa))_{\wa\in \Routes'} \in \seql[1][\Routes'][L_+(\hori)^\GA]$ holds.   In this case, $\Nl[]_{\Routes'}(h)$ is uniquely determined by $\Nl[]_{\Routes'}(h) = \sum_{\wa \in \Routes'}\Nl[]_\wa(h_\wa)$.
        \label[thmpart]{lem: elluExistenceProperties:ExistenceInducedFlow}
    \end{thmparts} 
\end{theorem}

\begin{lemma}[\ourref{lem: 1to1:h-f}]\label{lem: 1to1:h-f}
	Consider an arbitrary walk  $\wa$, $j\in [|\wa|+1]$ and $\g^{\wa,j} \in L_+(\hori)$. 
	Then, there exists a  ${h}_{\wa,j}\in L_+(\hori)$ with $\ell_{\wa,j}({h}_{\wa,j}) = \g^{\wa,j}$ if and only if 
	$\g^{\wa,j} = 0$ on ${\startint}\arr_{\wa,j}(-\infty)]$ where $\arr_{\wa,j}(-\infty)\coloneq \lim_{t\to -\infty}\arr_{\wa,j}(t)$. In this case, ${h}_{\wa,j}\in L_+(\hori)$ is uniquely determined. 
\end{lemma}

\subsection{Optimization Problems Involving \Auto Network Loadings}
We consider general optimization problems of the following form: 
 
     \begin{align} 
        \max_{\wflow}\;    &\objfunc(\wflow)  \tag{$\mathrm{P}^{\mathrm{App}}$} \label{opt: GeneralApp} \\
        \text{s.t.: } &\ell_{\Routes'}(h) \leq \eflow  \label{eq: ExOptSolLeqApp}\\
                    &\wflow \in \ofeas \nonumber
    \end{align}
Here, $\Routes'$ is an arbitrary countable collection of walks which may contain individual walks multiple but at most finitely many times.
The constraint vector  $\eflow$ is an arbitrary  element in $L_+(\hori)^\GA$. 
The objective~$\objfunc$ is some real-valued function on $\ofeas$,  which, in turn, is some subset of $\edom{\Routes'}$ containing at least one $\wflow$ fulfilling~\eqref{eq: ExOptSolLeqApp}, i.e.~the set of feasible solutions is non-empty. 
Here, $\edom{\Routes'}$ 
denotes the maximal domain of $\ell_{\Routes'}$, i.e.~the set of inflow rates $h \in L_+(\hori)^{\Routes'}$ for which $\ell_{\Routes'}(h)$ is well-defined and exists (cf.~\Cref{lem: elluExistenceProperties:AggLinearity}).   
In this regard, \eqref{opt: GeneralApp} is well-defined as $\ofeas \subseteq \edom{\Routes'}$ ensures that $\ell_{\Routes'}(h)$ is well-defined. 
 \begin{lemma}[\ourref{lem: elluContinuity}]\label{lem: elluContinuity}
 Consider an arbitrary walk  $\wa$, $j\in [|\wa|+1]$, $\arc \in \GA$ and a countable collection of walks~$\Routes'$. 
     Then, the following statements are true.
    \begin{thmparts}
        \item The mappings $\ell_{\wa,j}$ and $\ell_{\wa,\arc}$ are strong-strong and sequentially weak-weak continuous from their maximal domains to $L_+(\hori)$. \label[thmpart]{lem: elluContinuity:SingleWalk}

        \item For all $h \in \edom{\Routes'}$, we have $\norm{\ell_{\Routes'}(h)} = \sum_{\wa \in \Routes'}\abs{\wa}\cdot \norm{h_\wa}$. In particular, $\edom{\Routes'}\subseteq \seql[1][{\Routes'}][L_+(\hori)]$.
        \label[thmpart]{lem: elluContinuity:Subset}
        
        \item The mapping $\ell_{\Routes'}: \edom{\Routes'} \to L_+(\hori)^\GA$ is sequentially weakly lower semi-continuous in the following sense: 
        For any   sequence $(h^n)_{n\in \N}$ with $h^n\in \edom{\Routes'}, n\in \N$ and   $h^n_\wa \wto h_\wa,\wa \in \Routes'$ for some $h \in L_+(\hori)^{\Routes'}$ and all $\g \in L_+(\hori)^\GA$ for which there exists $N \in \N$ such that for all $n \geq N$ the inequality $\ell_{\Routes'}(h^n) \leq \g$ holds, we have $h \in \edom{\Routes'}$ and 
        $\ell_{\Routes'}(h) \leq \g$.  \label[thmpart]{lem: elluContinuity:Sum} 
    \end{thmparts}
\end{lemma}

\begin{theorem}[\ourref{thm: ExistenceOptSol}]\label{thm: ExistenceOptSol} 
    Assume that $\objfunc$ is sequentially weakly upper semi-continuous on the non-empty feasible set $\{h \in \ofeas\mid \ell_{\Routes'}(h)\leq \eflow\}$, i.e.\ $\limsup_{n \to\infty} \objfunc(h^n) \leq \objfunc(h)$ for any weakly converging sequence $h^n\wto h$ in $\seql[1][\Routes'][L_+(\hori)]$ with all $h^n$ and $h$ contained in the feasible set.
    Moreover, assume that  $\ofeas$ is sequentially weakly closed in~$\edom{\Routes'}$ \wrt the subspace topology induced by $\seql[1][\Routes'][L_+(\hori)]$.  Then, 
    the optimization problem~\eqref{opt: GeneralApp} has an optimal solution. 
\end{theorem}

\renewcommand{\eflow}{\g}

\subsection{\Auto Node Balances and  \boldmath{$\source$,$\dest$}-Flows}

The (\auto[)] node balance at node $v \in \GV$ for an arbitrary vector $\eflow \in L(\hori)^\GA$ is given by the measure  $\op_v\eflow \coloneqq \sum_{\arc \in \delta^+(v)}   \eflow_\arc \cdot \sigma -  \sum_{\arc \in \delta^-(v)}  ( \eflow_\arc \cdot\sigma) \circ \exit_\arc^{-1}$ which describes for an arbitrary $\mathfrak T \in \mathcal{B}(\hori)$ the difference between the cumulative inflow into $v$   and the cumulative outflow from $v$ during $\mathfrak T$, i.e.
\begin{align}\label{eq: FlowBalance}
    \op_v \eflow(\mathfrak T) =  
\sum_{\arc \in \delta^+(v)} \int_{\mathfrak T} \eflow_\arc \di \sigma -  \sum_{\arc \in \delta^-(v)} \int_{\exit_\arc^{-1}(\mathfrak T)} \eflow_\arc \di \sigma. 
\end{align}
 If the Radon-Nikodym derivative ${\inflow}_v$ of $\op_v \eflow$ exists, i.e.~the function satisfying for all $\mathfrak T \in \mathcal{B}(\hori)$
\begin{align*}
    \int_{\mathfrak T} {\inflow}_v\di\sigma =  
\sum_{\arc \in \delta^+(v)} \int_{\mathfrak T} \eflow_\arc \di \sigma -  \sum_{\arc \in \delta^-(v)} \int_{\exit_\arc^{-1}(\mathfrak T)} \eflow_\arc \di \sigma,
\end{align*}
 we say that $\eflow$ has the (\auto[)] net (node)   outflow rate ${\inflow}_v$ at $v$, or equivalently, the (\auto[)] net inflow rate $-{\inflow}_v$. 
 If the latter is equal to zero almost everywhere, we 
say that $\eflow$ fulfills (\auto[)] flow conservation at $v$. 
A vector $\eflow \in L_+(\hori)^\GA$ who has a net outflow rate $\inflow_s \in L_+(\hori)$ at $s$, fulfills flow conservation at 
all $v \neq \source,\dest$ and has a nonpositive \auto node balance at $\dest$ is called \emph{\auto $\source$,$\dest$-flow}. 
Here, we say that the (\auto[)] node balance $\op_v\eflow$ is nonpositive if \eqref{eq: FlowBalance} is nonpositive for any $\mathfrak  T \in \mathcal{B}(\hori)$.

  \begin{lemma}[\ourref{lem: flowconW'}]\label{lem: flowconW'}
Consider an arbitrary countable collection of walks $\Routes'$, a corresponding walk inflow rate vector $h \in \edom{\Routes'}$ with   $\g:=\ell_{\Routes'}(h)$  and a node $v \in \GV$. 
Then we have 
\begin{align*} 
    \op_v \g = \sum_{\wa \in \Routes'_{v+}} h_\wa \cdot \sigma   - \sum_{\wa \in \Routes'_{v-}} (h_\wa \cdot\sigma)\circ\arr_{\wa,|\wa|+1}^{-1},
\end{align*}
where $\Routes'_{v+}$ denotes the set of walks in $\Routes'$ starting at $v$ while $\Routes'_{v-}$ denotes the set of walks in $\Routes'$ ending in  $v$. 

If, additionally, $\ell_{\wa,\abs{\wa}+1}(h_\wa)$ exist for all $\wa \in \Routes'_{v-}$, we even  have
\begin{align*}
    \op_v \g = \sum_{\wa \in \Routes'_{v+}} h_\wa \cdot \sigma - \sum_{\wa \in \Routes'_{v-}} \ell_{\wa,\abs{\wa}+1}(h_\wa) \cdot\sigma.
\end{align*}
       \end{lemma}

\subsection{Properties of \Auto \boldmath{$\source$,$\dest$}-Flows}

\begin{lemma}[\ourref{lem: Relations:h>0u>0}]\label{lem: Relations:h>0u>0} 
Let $\Routes'$ be an arbitrary countable collection of walks and $h \in \edom{\Routes'}$ with $\g: = \ell_{\Routes'}(h)$.  The following statements are true: 
    \begin{thmparts}
        \item  For all $\mathfrak T \in \mathcal{B}(\hori)$, $\wa \in \Routes$ and $j \leq \abs{\wa}$ the following implication holds: ${h}_\wa(t)>0$ for a.e.~$t \in \mathfrak T$ $\implies $ ${\g}_{\wa[j]}(t)>0$ for a.e.~$t \in \arr_{\wa,j}(\mathfrak T)$. \label[thmpart]{lem: Relations:h>0u>0:Setwise}

        \item For any edge $\arc \in \GA$, any representative of $h$ and for  all $\mathfrak T \in \mathcal{B}(\hori)$ with ${\g}_\arc(t)>0$ for a.e.~$t \in \mathfrak T$, there exists for almost every $t \in \mathfrak T$ a walk $\wa \in {\Routes'}, j \leq |\wa|$ with $\wa[j] = \arc$ and $\Tilde{t} \in \arr_{\wa,j}^{-1}(t)$ such that ${h}_\wa(\Tilde{t})> 0$. 
        \label[thmpart]{lem: Relations:h>0u>0:u>0ExistsHw>0}        
       
        \item For all $\wa \in \Routes'$ and $j \leq \abs{\wa}$, an arbitrary representative of ${\g}_{\wa[j]}$ and  almost all $t \in \hori$ the following implication holds 
        \begin{align*}
             {h}_\wa(t) > 0 \implies {\g}_{\wa[j]} (\arr_{\wa,j}(t)) >0.
        \end{align*}
       \label[thmpart]{lem: Relations:h>0u>0:Pointwise}

        \item For any edge $\arc \in \GA$, any representative of $h$ and almost all $t\in\hori$, the following implication holds 
        \begin{align*} 
 					\g_\arc(t)>0 \implies \exists \wa \in {\Routes'}, j \leq |\wa| \text{ with } \wa[j] = \arc \text{ and } \Tilde{t} \in \arr_{\wa,j}^{-1}(t) \text{ such that } {h}_\wa(\Tilde{t})> 0. 
 		\end{align*}
        \label[thmpart]{lem: Relations:h>0u>0:u>0ExistsHw>0Pointwise}

        \item \label[thmpart]{lem: Relations:h>0u>0:GoodRepres} There exist representatives of $h$ and $\g$ that fulfill 
            for all $\wa \in \Routes'$,  $j\leq \abs{\wa}$, $\arc \in \GA$  
            and \emph{for all} $t \in \hori$ the implications in \ref{lem: Relations:h>0u>0:Pointwise} and \ref{lem: Relations:h>0u>0:u>0ExistsHw>0Pointwise}.

 		\item If $\Routes'$ is a collection of \stwalk s, then, for any representative of $h$,  
 		there exists for $\op_\dest\g$-almost every $t$ a walk $\wa \in {\Routes'}$  and $\Tilde{t} \in \arr_{\wa,\abs{\wa}+1}^{-1}(t)$ such that ${h}_\wa(\Tilde{t})> 0$. 
 		\label[thmpart]{lem: Relations:h>0u>0:u>0ExistsHw>0D}

        \item Consider an arbitrary representative of $h$, an arbitrary $\mathfrak T \in \mathcal{B}(\hori)$ with $\sigma(\mathfrak T)> 0$ and a countable index set $\hat{L}$ with corresponding walk-edge-index-pairs $(\wa^{\hat l},j_{\hat l}),\hat l  \in \hat{L}$ with $\wa^{\hat l}\in \Routes'$ and $j_{\hat l}\leq \abs{\wa^{\hat l}}+1$ such that the following holds: For almost every $t \in \mathfrak T$ there exists an index $\hat l \in \hat{L}$ with corresponding $\tilde{t} \in \arr_{\wa^{\hat l},j_{\hat l}}^{-1}(t)$ such that $h_{\wa^{\hat l}}(\tilde{t})>0$. 
        Then, 
        there exists a countable set $L$ with corresponding walks $\wa^{l}\in \Routes'$, indices $j_{l}\leq \abs{\wa^{l}}+1$ and departure time sets $\mathfrak D^{l} \in \mathcal{B}(\hori)$ with $\sigma(\mathfrak D^{l})> 0$ such that
        \begin{itemize}\label[thmpart]{lem: Relations:h>0u>0:u>0ExistsCountableMZwischen}  
            \item for every $l \in L$ there exists $\hat{l}\in \hat{L}$ with $(\wa^l,j_l) = (\wa^{\hat{l}},j_{\hat{l}})$, 
            \item for every $l \in L$ we have $h_{\wa^l}(t)>0$ for a.e.~$t \in \mathfrak D^l$, 
            \item the union  $\bigcup_{l \in L}  \arr_{\wa^{l},j_{l}}(\mathfrak D^{l})$ equals $\mathfrak T$ up to a null set,
            \item the union  $\bigcup_{l \in L}  \arr_{\wa^{l},j_{l}}(\mathfrak D^{l})$ is disjoint, 
            \item all $\arr_{\wa^{l},j_{l}}(\mathfrak D^{l}), l \in L$ have positive measure and 
            \item every pair $(\wa,j), \wa\in {\Routes'}, j\leq\abs{\wa}+1$  appears at most as often in~$L$ as in~$\hat{L}$, that is, $\abs{\{l \in L\mid (\wa^{l},j_{l})=(\wa,j)\}} \leq \abs{\{\hat l \in \hat{L}\mid (\wa^{\hat l},j_{\hat l})=(\wa,j)\}}$.
        \end{itemize} 

        \item \label[thmpart]{lem: Relations:h>0u>0:u>0ExistsCountableM}  
        Consider an arbitrary $\arc \in \GA$, $\mathfrak T \in \mathcal{B}(\hori)$ with $\sigma(\mathfrak T)> 0$ and ${\g}_\arc(t)>0$ for a.e.~$t \in \mathfrak T$. Then, 
        there exists a countable set $L$ with corresponding walks $\wa^{l}\in \Routes', j_{l}\leq \abs{\wa^{l}}$ and departure time sets $\mathfrak D^{l} \in \mathcal{B}(\hori)$ with $\sigma(\mathfrak D^{l})> 0$ such that 
        \begin{itemize}
            \item for every $l \in L$ we have $h_{\wa^l}(t)>0$ for a.e.~$t \in \mathfrak D^l$, 
            \item the union  $\bigcup_{l \in L}  \arr_{\wa^{l},j_{l}}(\mathfrak D^{l})$ equals $\mathfrak T$ up to a null set,
            \item the union  $\bigcup_{l \in L}  \arr_{\wa^{l},j_{l}}(\mathfrak D^{l})$ is disjoint, 
            \item all $\arr_{\wa^{l},j_{l}}(\mathfrak D^{l}), l \in L$ have positive measure and 
            \item every $\wa \in {\Routes'}$ appears at most $\abs{\wa}$ many times, i.e.~$\abs{\{l \in L\mid \wa^{l}=\wa\}} \leq \abs{\wa}$. 
        \end{itemize} 
\end{thmparts}
\end{lemma}

 \begin{lemma}[\ourref{lem: elluPropagation}]\label{lem: elluPropagation}
  Consider an arbitrary walk $\wa$, two indices $j_1\leq j_2\leq|\wa|+1$ and $h_\wa \in \edom{\wa,j_1}$. Then, 
     $\ell_{\wa_{\geq j_1},j_2-j_1+1}(\ell_{\wa,j_1}(h_\wa))$ exists if and only if $\ell_{\wa,j_2}(h_\wa)$ exists,  in which case they coincide, i.e.~$\ell_{\wa,j_2}(h_\wa) = \ell_{\wa_{\geq j_1},j_2-j_1+1}(\ell_{\wa,j_1}(h_\wa))$. 
\end{lemma}

\begin{lemma}[\ourref{lem: elluinj}]\label{lem: elluinj}
       Consider an arbitrary walk $\wa$, $j\in [|\wa|+1]$ and $h_\wa \in L_+(\hori)$. If $\ell_{\wa,j}(h_\wa)$ exists, then $h_\wa = 0$ on $\arr_{\wa,j}^{-1}(\arr_{\wa,j}(\mathfrak T))\setminus \mathfrak T$ and in particular $\int_{\mathfrak T} h_\wa \di\leb = \int_{\arr_{\wa,j}^{-1}(\arr_{\wa,j}(\mathfrak T))} h_\wa \di\leb$ for any $\mathfrak T\in \mathcal{B}(\hori)$.
\end{lemma}

\begin{lemma}[\ourref{lem: ellOrderPreserving}]\label{lem: ellOrderPreserving}
   Consider an arbitrary walk  $\wa$, $j\in [|\wa|+1]$ and $h_{\wa},\Tilde{h}_\wa \in \edom{\wa,j}$. Then, we have
        \[\ell_{\wa,j}(h_\wa) \leq \ell_{\wa,j}(\tilde{h}_\wa) \iff h_\wa \leq \Tilde{h}_\wa.\]
    The analogue statement holds for $<$ instead of $\leq$ where $\wflow < \tilde{\wflow}$ for $\wflow,\tilde{\wflow} \in L_+(\hori)$ means that $\wflow \leq \tilde{\wflow}$ and $\wflow \neq \tilde{\wflow}$. 
\end{lemma}
\begin{lemma}[\ourref{lem: elluindi}]\label{lem: elluindi}
       Consider an arbitrary walk $\wa$, $j\leq |\wa|+1$ and $h_\wa \in L_+(\hori)$ with $\ell_{\wa,j}(h_\wa)$ existing. 
       Then for any $\mathfrak T^* \in \mathcal{B}(\hori)$, the flow $\ell_{\wa,j}(1_{\mathfrak T^*}\cdot h_\wa)$ exists and  fulfills $\ell_{\wa,j}(1_{\mathfrak T^*}\cdot h_\wa) = 1_{\arr_{\wa,j}(\mathfrak T^*)}\cdot\ell_{\wa,j}(h_\wa)$.  
\end{lemma}

\begin{lemma}[\ourref{lem: ellOrderPreservingSharpened}]\label{lem: ellOrderPreservingSharpened}
    Consider an arbitrary walk  $\wa$, $j\in [|\wa|+1]$ and $h_{\wa},\Tilde{h}_\wa \in \edom{\wa}$. 
    For any set $\mathfrak D \in \mathcal{B}(\hori)$ the inequality 
     $\ell_{\wa,j}(h_\wa) \leq \ell_{\wa,j}(\tilde{h}_\wa)$ on $\arr_{\wa,j}(\mathfrak D)$ is equivalent to $h_\wa \leq \Tilde{h}_\wa$ on $\mathfrak D$. The analogue statement holds for $<$ instead of $\leq$. 
\end{lemma}

\begin{lemma}[\ourref{lem: FLowOnZeroTrav}]\label{lem: FLowOnZeroTrav}
Consider an arbitrary walk $\wa$, two edge indices $j_1<j_2\leq|\wa|+1$ and $h_\wa \in\edom{\wa,j_1}$.
 Furthermore, let $\mathfrak D \in \mathcal{B}(\hori)$ be a set for which for almost every $t \in \mathfrak D$, we have  $\arr_{\wa,j_1}(t) = \arr_{\wa,j_2}(t)$. 
Then, $\ell_{\wa,j'}(h_\wa\cdot 1_{\mathfrak{D}}), j'\in\{j_1+1,\ldots,j_2\}$ exist and fulfill  $\ell_{\wa,j'}(h_\wa \cdot 1_{\mathfrak{D}} ) = \ell_{\wa,j_1}(h_\wa) \cdot 1_{\arr_{\wa,j_1}(\mathfrak D)}$. 

In particular, if  $h_\wa = 0$ on $\hori\setminus \mathfrak D$, then $\ell_{\wa,j'}(h_\wa)$ exists and is equal to $\ell_{\wa,j_1}(h_\wa)$ on the whole set $\hori$ for all $j'\in\{j_1+1,\ldots,j_2\}$. For a zero-cycle inflow rate $h_c$ into a cycle $c$, the latter statement shows that $\ell_{c,\arc}(h_c) = \sum_{j\leq\abs{c}:c[j] = \arc}h_c$ for every $\arc \in c$. 
\end{lemma}

 \begin{lemma}[\ourref{lem: FlowConEveryNode}]\label{lem: FlowConEveryNode} 
    Let $\eflow \in L_+(\hori)^\GA$ be \aauto $\source$,$\dest$-flow fulfilling flow conservation also at $s$. 
    Then, $\eflow$ is \aauto dynamic circulation  and we have
    \begin{align*}
        {\eflow}_\arc(t) > 0 \implies \trav_\arc(t) = 0 \text{ for almost all } t \in \hori \text{ and all } \arc \in \GA.   
    \end{align*}
\end{lemma}

We say for any $\arc \in \GA$ that $\g_\arc \in L_+(\hori)$
 admits  an (edge) outflow rate $\g_\arc^-\in L_+(\hori)$ (\wrt $\trav$) if there exists $\g_\arc^-\in L_+(\hori)$ fulfilling 
\begin{align}
    \int_{\mathfrak T}\g_\arc^- \di\sigma =   \int_{\exit_\arc^{-1}(\mathfrak T)}\g_\arc \di\sigma \text{ for all }\mathfrak T \in \mathcal{B}(\hori).
\end{align}

\begin{lemma}[\ourref{lem: outflow}] \label{lem: outflow}
Let  $\eflow \in L_+(\hori)^\GA$ and $\arc \in \GA$ be arbitrary. The following statements are true:
\begin{thmparts}
    \item \label[thmpart]{lem: outflow: =ell+Cons}  The outflow rate $\eflow_\arc^-$ exists if and only if $\eflow_\arc \in \edom{\wa,2}$ for $\wa = (\arc)$, in which case $\eflow_\arc^- = \ell_{\wa,2}(\eflow_\arc)$. That is, the outflow rate $\eflow_\arc^-$ equals the inflow rate into the end node of the single-edge walk $\wa=(\arc)$. As a direct consequence, we get 
    \begin{enumerate}
        \item \label{lem: outflow: Exis}  The 
    outflow rate $\eflow_\arc^-$ exists and is then uniquely determined if
     and only if 
     \begin{align}\label{eq: outflowExists}
              \eflow_\arc = 0  \text{ on }\exit_\arc^{-1}(\mathfrak T) \text{ for any null set }\mathfrak T\subseteq \hori. 
     \end{align} 
     \item \label{lem: outflow: ExisLeq}  If $\eflow_\arc \leq \tilde{\eflow}_\arc$ for some $\tilde\eflow_\arc \in L_+(\hori)$ and $\tilde{\eflow}_\arc^-$ exists, then so does $\eflow_\arc^-$.
     \item  \label{lem: outflow: ExisLeqInverted} Every $\Tilde{\eflow}_\arc^-\in L_+(\hori)$  with $\Tilde{\eflow}_\arc^- = 0$ on ${\startint}\exit_\arc(-\infty))$ for $\exit_\arc(-\infty) := \lim_{t \to -\infty}\exit_\arc(t) $ has a corresponding inflow rate $\Tilde{\eflow}_\arc\in L_+(\hori)$. 
    \item \label{lem: outflow: ExisLeqInvertedCons} If $\eflow_\arc^-$ exists, then every $\Tilde{\eflow}_\arc^-\in L_+(\hori)$  with $\Tilde{\eflow}_\arc^- \leq \eflow_\arc^-$ has a corresponding inflow rate $\Tilde{\eflow}_\arc\in L_+(\hori)$ with $\Tilde{\eflow}_\arc \leq {\eflow}_\arc$.  
    \end{enumerate}

     \item  \label[thmpart]{lem: outflow: FlowCon} For any $v \in \GV$,  the edge outflow rates  $\eflow_\arc^-,\arc \in \delta^-(v)$ exist, if  $\eflow$ has a net node outflow rate $\inflow_v \in L_+(\hori)$ at $v$, i.e.~if \eqref{eq: FlowBalanceDerivative} holds w.r.t.~$\inflow_v$.  
    
    \item  \label[thmpart]{lem: outflow: In=OutIfD=0} If $\eflow_\arc^-$ exists, then $\eflow_\arc^-(t) = \eflow_\arc(t)$ for almost all $t$ with $\trav_\arc(t) = 0$. 
\end{thmparts}
\end{lemma}

\begin{theorem}[\ourref{lem: ZeroCycleDecomposition}]\label{lem: ZeroCycleDecomposition}
	Any  \aauto dynamic circulation $\eflow$ can be decomposed into zero-cycle inflow rates ${h}_c \in L_+(\hori), c \in \SimpCyc$ via $\eflow_\arc = \sum_{c \in \SimpCyc} \ell_{c,\arc}({h}_c) = \sum_{c \in \SimpCyc:\arc \in c}  {h}_c$ for all $\arc \in \GA$.
\end{theorem}

\subsection{Existence of Flow Decomposition} 

\begin{theorem}[\ourref{thm: FlowDecomp}]\label{thm: FlowDecomp}
    Every \auto $\source$,$\dest$-flow has \aauto flow decomposition, that is, a vector of walk inflow rates $h \in L_+(\hori)^{\hat{\Routes}}$ together with zero-cycle inflow rates $h \in L_+(\hori)^{\mathcal{C}}$ such that $\eflow = \sum_{\wa \in \hat{\Routes}}\ell_\wa(h_\wa) + \sum_{c \in \SimpCyc}\ell_c(h_c)$.
\end{theorem}

\begin{theorem}[\ourref{thm: PureFlowDecompIntuitive}]\label{thm: PureFlowDecompIntuitive} 
 \Aauto $\source$,$\dest$-flow $\eflow \in L_+(\hori)^\GA$ with net outflow rate $\inflow_\dest$ at $\dest$ has \aauto pure $\source$,$\dest$-flow decomposition   if and only if 
    for every zero-cycle inflow rate $h_c'\in L_+(\hori)$ into \emph{any} (not necessary simple) cycle $c$ with $h_c' \leq \eflow_\arc,\arc \in c$, we have for almost all $t \in \hori$ with $h_c'(t)>0$ that (at least) one of the following conditions is satisfied:
    \begin{thmparts}
        \item  $\dest \in c$ and  $ \inflow_\dest (t)<0$. \label[thmpart]{thm: PureFlowDecompIntuitive: Dest}
        \item  There exists an edge $\arc=(v,v') \notin c$ with $v \in c$ and  $\eflow_\arc(t)>0$. \label[thmpart]{thm: PureFlowDecompIntuitive: NotDest}
    \end{thmparts}
\end{theorem}

\begin{definition}[\ourref{def: ConnectedComp}]\label{def: ConnectedComp}
Consider a set of zero-cycle inflow rates $h_c,c \in \mathcal{C}$ and fix some (arbitrary) representatives of those. We define 
for all $t \in \hori$ the set $\mathcal{C}(t):=\{c \in \mathcal{C}\mid  {h}_c(t)>0 \text{ and }\trav_{\arc}(t) = 0, \arc \in c\}$. 
Let  $C_1^t,\ldots,C_{m(t)}^t$   for  a $\n2(t) \in \N$ be the 
partition of $\mathcal{C}(t)$ into maximal connected components, that is, a partition with the following two properties:
\begin{itemize}
    \item For every $j \in\{1,\ldots,m(t)\}$ the edge set $\GA_{C_j^t}:=\set{\arc \in \GA | \exists c \in C_j^t: e \in c}$ induces a connected subgraph of~$G$ and 
    \item the node sets $\GV_{C_j^t}:=\set{v \in \GV | \exists c \in C_j^t: v \in c}$ are disjoint, i.e.\ $\GV_{C_j^t} \cap \GV_{C_{j'}^t} = \emptyset$ for all $j'\neq j$. 
\end{itemize} 
Next, we consider the set $\{C \subseteq \mathcal{C} \mid \sigma(\mathfrak T_C)>0\} \subseteq 2^\mathcal{C}$ where $\mathfrak T_{C} := \{t \in \hori \mid  \exists\, j: C = C_j^t \}$. We denote this set via $\{C_{\n1}\}_{\n1 \in \capn1}$ 
where   $\capn1\subseteq \N$ is a finite family of indices.  Note that the sets $\mathfrak T_{C}$  are measurable as they can be written as follows: 
    $\mathfrak T_{C} = \bigcap_{c \in C}\mathfrak T_{c} \cap \bigcap_{c \in \bar{C}\setminus C}(\hori \setminus\mathfrak T_c)$ where
    $\mathfrak T_c := \{t \in \hori \mid {h}_c(t)>0 \text{ and }\trav_{\arc}(t) = 0, \arc \in c\}$ and $\Bar{C} = \{c \in \mathcal{C}\mid \exists \,c' \in C: c \text{ shares a node with }c'\}$ with $\mathfrak T_c$ being measurable due to ${h},\trav$ being measurable. Furthermore, we denote for any $\n1 \in \capn1$ by $\GV_{C_{\n1}}:=\{v \in \GV\mid\exists c \in C_{\n1}: v \in c \}$ the nodes contained in $C_{\n1}$ and analogously by  $\GA_{C_{\n1}}:=\{\arc \in \GA\mid\exists c \in C_{\n1}: \arc \in c \}$ the edges  contained in $C_{\n1}$.
    Finally, we remark that the above sets are uniquely determined (up to enumeration) by the support of the representative $h_c,c \in \mathcal{C}$.

    In the situation as described above, we call $\mathcal{C}(t)$ the set of active cycles at $t \in \hori$,    $\{C_{\n1}\}_{\n1 \in \capn1} $ the resulting  connected components and   $\mathfrak T_{C_{\n1}},\n1\in\capn1$ the set of times at which the connected components are active. Similarly, $\mathfrak T_{c}$ for any $c \in \mathcal{C}$ is the set of times at which the cycle $c$ is active.
\end{definition}

 \begin{theorem}[\ourref{thm: PureFlowDecomp}] \label{thm: PureFlowDecomp}
    Consider an  $\source$,$\dest$-flow $\eflow \in L_+(\hori)^\GA$ with a corresponding 
    flow decomposition $h_\wa,\wa \in  {\hat{\Routes}},h_c,c \in \mathcal{C}$, an outflow rate $\inflow_\dest$ and an arbitrary representative of $h$ together with the  sets defined in \Cref{def: ConnectedComp}. 
    Then  $\eflow$ has a flow decomposition   purely into \stwalk s if and only if for every $\n1 \in \capn1$ and almost all $t \in \mathfrak T_{C_{\n1}}$ (at least) one of the following statements is true 
\begin{thmparts}
    \item $\dest \in \GV_{C_{\n1}}$ and  $\inflow_\dest (t)<0$. \label[thmpart]{thm: PureFlowDecomp: Dest} 
    \item there exists an edge $\arc=(v,v') \notin \GA_{C_{\n1}}$ with $v \in \GV_{C_{\n1}}$ and  $\eflow_\arc(t)>0$. \label[thmpart]{thm: PureFlowDecomp: NotDest}
\end{thmparts}
\end{theorem}

\begin{corollary}[\ourref{cor: PureFlowDecomp}] \label{cor: PureFlowDecomp}
    Consider a  $\source$,$\dest$-flow $\eflow \in L_+(\hori)^\GA$ with outflow rate $\inflow_\dest$ and an arbitrary representative of a corresponding flow decomposition $h$ together with the sets defined in \Cref{def: ConnectedComp}. 
    Then  there exists 
    another flow decomposition $h'_\wa,\wa \in  {\hat{\Routes}},h'_c,c \in \mathcal{C}$ with a corresponding representative and sets $\mathcal{C}'(t),t \in \hori,C'_{\n1'},\mathfrak T'_{C'_{\n1'}},\n1' \in \capn1'$ such that for every $\n1' \in \capn1'$ 
    and almost every $t \in \mathfrak T_{C'_{\n1'}}$ neither of the following statements holds:
\begin{thmparts}
    \item $\dest \in \GV_{C_{\n1'}'}$ and  $\inflow_\dest (t)<0$. \label[thmpart]{cor: PureFlowDecomp: DestCor} 
    \item There exists an edge $\arc=(v,v') \notin \GA_{C_{\n1'}'}$ with $v \in \GV_{C_{\n1'}'}$ and  $\eflow_\arc(t)>0$.  \label[thmpart]{cor: PureFlowDecomp: NotDestCor}
\end{thmparts}
Moreover, $h_c'\leq h_c,c \in \mathcal{C}$ and for every $\n1' \in \capn1'$ there exists $\n1 \in \capn1$ with ${C}'_{{\n1}'} = C_{\n1}$ and ${\mathfrak T}'_{ C'_{{\n1'}}} \subseteq \mathfrak T_{C_{\n1}}$.
\end{corollary}

\section{Properties of Locally Absolutely Continuous Functions}

In this section, we gather several insights into locally absolutely continuous functions $\func:\hori\to\R$, i.e.~functions that are  absolutely continuous on every closed interval $[a,b]\subseteq \hori$ in the sense of \cite[Definition 5.3.1]{Bogachev2007I}. 
These results follow almost immediately by their corresponding analogues for absolutely continuous functions on an interval. 
We refer to \cite{GHS24FD} for a proof. 
 \Cref{lem: PropAbsCon:ImageMeas} states that the image of a measurable set  under $\func$ is measurable as well. 
\Cref{lem: PropAbsCon:Lus} states that 
$\func$ fulfills Lusin's property, i.e.~takes null sets to null sets. 
\Cref{lem: PropAbsCon:Der} states that
the derivative of 
$\func$ exists almost everywhere. 
\Cref{lem: PropAbsCon:DerNonDec}
shows that for $\func$ non-decreasing, its derivative is larger or equal to $0$ in case of existence. 
\Cref{lem: PropAbsCon:Est} establishes an estimate for 
the measure of the image $\func(\mathfrak T)$ for any measurable set $\mathfrak T \in \mathcal{B}(\hori)$. 
Finally, \Cref{lem: PropAbsCon:Conca} states that the concatenation $\func\circ \hat{\func}$ of two locally absolutely continuous functions $\func,\hat{\func}$ is again locally absolutely continuous, provided that $\hat{\func}$ is non-decreasing. 
\begin{lemma}[\ourref{lem: PropAbsCon}]\label{lem: PropAbsCon} 
    For an arbitrary locally absolutely continuous function $\func:\hori\to\R$, the following statements are valid: 
\begin{thmparts} 
    \item\label[thmpart]{lem: PropAbsCon:ImageMeas} The image $\func(\mathfrak T)$ of a measurable set $\mathfrak T \in \mathcal{B}(\hori)$ under $\func$ is again (Borel-)measurable.
    \item  \label[thmpart]{lem: PropAbsCon:Lus} $\func$ fulfills Lusin's property, i.e.~$\sigma(\func(\mathfrak T)) = 0$ for all null sets $\mathfrak T$.
    \item\label[thmpart]{lem: PropAbsCon:Der} The derivative $\func'$ exists almost everywhere, any measurable extension $\Tilde{\func}'$ of $\func'$ to the whole real line is locally integrable and fulfills for all $a<b \in \R$  the equality 
    \begin{align*}
       \func(b) - \func(a) =  \int_{[a,b]}{\tilde{\func}'}\di\sigma . 
    \end{align*}
    In particular, this allows us to write simply  $\int_{[a,b]}\func'\di\sigma$ for the right hand side. 
    \item \label[thmpart]{lem: PropAbsCon:DerNonDec}  If $\func$ is non-decreasing, then $\func'(t)\geq 0$ holds if $\func'(t)$ exists. 
    \item \label[thmpart]{lem: PropAbsCon:Est} The following estimate holds for all $\mathfrak T \in \mathcal{B}(\hori)$: 
    \begin{align*}
        \sigma(\func(\mathfrak T))) \leq \int_{\mathfrak T}\abs{\func'}\di\sigma.
    \end{align*}
    \item \label[thmpart]{lem: PropAbsCon:Conca}   If $\hat{\func}:\hori\to \hori$ is a non-decreasing locally absolutely continuous function, then $\func \circ \hat{\func}$ is also locally absolutely continuous. 
\end{thmparts}
     
\end{lemma} 

\section{List of Symbols}

{
	\newcommand{\losEntry}[2]{#1 & #2 \\}
	\renewcommand{\arraystretch}{1.2}
	
	\begin{longtable}{p{4cm}p{10cm}}
		Symbol				& Description \\\hline
		
		\hline\multicolumn{2}{l}{\textbf{General}}\\\hline
		
		\losEntry{$L(\hori)$}{space of integrable functions on $\hori$}
		\losEntry{$L_+(\hori)$}{non-negative functions in $L(\hori)$}
        \losEntry{$\seql[1][M][L(\hori)]$}{vectors  $(h_m)_{m\in M}\in (L(\hori))^M$ for an arbitrary countable set $M$ whose corresponding series $\sum_{m\in M}h_m$ converges absolutely in $L(\hori)$. }
        \losEntry{$\seql[\infty][M][L^\infty(\hori)]$}{vectors  $(h_m)_{m\in M}\in (L^\infty(\hori))^M$  whose entries are uniformly bounded, i.e.~$\sup_{m \in M}\norm{h_m}_{\infty}< \infty$. }
		\losEntry{$L^\infty(\hori)$}{space of measurable essentially bounded functions on $\hori$}
		\losEntry{$L^\infty_+(\hori)$}{non-negative functions in $L^\infty(\hori)$}

		\losEntry{$\sigma$}{the Lebesgue measure on $\hori$}
		\losEntry{$\mathfrak{T}, \mathfrak{D}$}{measurable subsets of $\hori$}

        \losEntry{$1_{\mathfrak T}$}{characteristic function of a (measurable) set~$\mathfrak T$, i.e.\ $1_{\mathfrak T}(t)=1$ if $t \in \mathfrak T$ and $0$, otherwise}

        \losEntry{$\dup{f}{g}$}{the bilinear form between the dual pair $(\seql[1][M][L(\hori)],\seql[\infty][M][L^\infty(\hori)])$, i.e.\ $\dup{f}{g} \coloneqq \sum_{m \in M}\int_\hori f_m\cdot g_m \sigma$}

		\hline\multicolumn{2}{l}{\textbf{Network}}\\\hline
		
		\losEntry{$G=(\GV,\GA)$}{directed graph with nodes $\GV$ and edges $\GA$}
		\losEntry{$\edgesFrom{v}$}{edge starting from node $v$}
		\losEntry{$\edgesTo{v}$}{edge ending at node $v$}
		\losEntry{$\source_i \in V$}{source node of commodity $i$}
		\losEntry{$\dest_i \in V$}{destination node of commodity $i$}
		\losEntry{$\hori$}{planning horizon}
		\losEntry{$t \in \hori$}{time}
		\losEntry{$\hat\wa=(\arc_1,\dots,\arc_k)$}{walk consisting of edges $\arc_j$}
		\losEntry{$\hat\wa[j]$}{$j$-th edge on walk $\wa$} 
        \losEntry{$\hat\wa_{\leq j}$ ($\hat\wa_{< j}$)}{subwalk of $\hat{\wa}$ up to and including (excluding) the $j$-th edge}
        \losEntry{$\hat\wa_{\geq j}$ ($\hat\wa_{> j}$)}{subwalk of $\hat{\wa}$ starting from the $j$-th ($j+1$-th) edge}
		\losEntry{$\hat\Routes_{v_1,v_2}$}{set of (finite) \stwalk[v_1][v_2]s (untyped)}
        \losEntry{$I$}{(finite) set of commodities}
        \losEntry{$\gamma_i > 0$}{value-of-time parameter of commodity~$i$}
		\losEntry{$\inflow_i \in L_+(\hori)$}{network inflow rate of commodity $i$}
        \losEntry{$\Routes_i:= \hat{\Routes}_{\source_i,\dest}\times \{i\}$}{set of commodity-typed \stwalk[\source_i][\dest_i]}
		\losEntry{$\Routes \coloneqq \bigcup_i\Routes_i$}{collection of all commodity-typed walks}
		\losEntry{$\Routes'$}{arbitrary collection of (finite)  walks}  
        \losEntry{$\SimpCyc$}{set of simple cycles}
        \losEntry{$\DestCyc$}{set of $\dest$-cycles (not necessarily simple!), i.e.\ cycles starting at the destination~$\dest$}
        \losEntry{$\hat{\Routes}^\dest$}{set of all finite walks starting at $\dest$}
		\losEntry{$\wir$}{set of feasible walk-inflows}

        \hline\multicolumn{2}{l}{\textbf{Initial Travel Times and Network Loading}}\\\hline
         \losEntry{$\trav(\cdot,\cdot): L_+(\hori)^\GA \times \hori \to \R_+$}{Initial flow-dependent travel time}       
	    \losEntry{$\trav_\arc(\g,t)$}{edge traversal time under $\g$ when entering edge $\arc$ at time $t$ (absolutely continuous)}
		\losEntry{$\exit_\arc(\g,t)$}{edge exit time when entering edge $\arc$ at time $t$ under $\g$: $\exit_\arc(\g,t) \coloneqq t+\trav_\arc(\g,t)$ (non-decreasing)}
		\losEntry{$\arr_{\wa,j}(\g,t)$}{arrival time in front of the $j$-th edge of walk $\wa$ when entering this walk at time $t$ under $h$}
        \losEntry{$\wir\subseteq  L_+(\hori)^\Routes$}{set of admissible walk inflows}
    	\losEntry{$\wttime_\wa(\g,t)$}{private costs for particles of commodity $i$ entering walk $\wa\in \Routes_i$ at time $t$ under $\g$ for any $i\in I$}
		\losEntry{$\Nl:\wir \to L_+(\hori)^\GA$}{Network loading operator \wrt (flow-dependent) travel time function $\trav(\cdot,\cdot)$}
		\losEntry{$h \in L_+(\hori)^{\Routes}$}{walk-inflow for the walk collection $\Routes$} 
		\losEntry{$\g = \Nl(h) \in L_+(\hori)^\GA$}{edge flow induced by $h$}

		\hline\multicolumn{2}{l}{\textbf{\Auto Network Loadings}}\\\hline  
        \losEntry{$\trav(\cdot): \hori \to \R_+$}{Some fixed flow-independent (but still time-dependent) travel time function}
		\losEntry{$\Nl[\trav(\cdot)]_{\Routes'}:\edom{\Routes'}[\trav(\cdot)] \to L_+(\hori)^\GA$}{\Auto network loading operator \wrt (flow-independent) travel time function $\trav(\cdot)$ and walk collection $\Routes'$}        
        \losEntry{$\edom{\Routes'}[\trav(\cdot)]$}{maximal set of walk inflow rates $h$ into $\Routes'$ whose induced \auto edge flow $\Nl[\trav(\cdot)]_{\Routes'}(h)$ exists} 
		\losEntry{$h \in L_+(\hori)^{\Routes'}$}{walk-inflow for the walk collection $\Routes'$} 
		\losEntry{$\g = \Nl[\trav(\cdot)](h) \in L(\hori)^\GA$}{\auto edge flow induced by $h$} 
		\losEntry{$\Nl[\trav(\cdot)]_{\wa,j}(h_\wa)$}{the flow on the $j$-th edge on walk $\wa$ induced by inflow into that walk under $ h_\wa$, without aggregating over multiple occurrences of that edge}
		\losEntry{$\Nl[\trav(\cdot)]_{\wa,\arc}(h_\wa)$}{the flow on edge $\arc$ on walk $\wa$ induced by inflow into that walk under $ h_\wa$,  aggregated over multiple occurrences of that edge}
        \losEntry{$\g^i = \Nl[\trav(\cdot)]_{\Routes_i}(h^i) \in L(\hori)^\GA$}{edge flow of commodity $i$ under $h\in L(\hori)^\Routes$ where $h^i:=(h_\wa)_{\wa \in \Routes_i}$}    
  
		\hline\multicolumn{2}{l}{\textbf{Tolls}}\\\hline

		\losEntry{$\prices: \hori \to \Rnn^\GA$}{time dependent edge-tolls}
		\losEntry{$\Pf^\prices_\wa(\g,t) $}{total toll along walk $\wa$ when entering at time $t$ under $\g$ and $\prices$, given by $\Pf^\prices_\wa(\g,t):= \sum_{j \leq|{\wa}|}  \prices_{{\wa}[j]}(\arr_{{\wa},j}(\g,t)) $}        
		
		\hline\multicolumn{2}{l}{\textbf{Special Notation}}\\\hline
		\losEntry{\ref*{opt: Master}}{the master problem \wrt $u$}
        \losEntry{$\mathfrak D_\wa^\g$}{defined \wrt~a representative of $\g \in L_+(\hori)^\GA$. Denotes the set of times at which a particle can enter the walk $\wa$ and arrive at all edges $\arc \in \wa$ only if there is inflow under $\g_\arc$.}
		
 	\end{longtable}
}

\bibliographystyle{plain}
\bibliography{master-bib}

\end{document}